\documentclass[11pt,a4paper]{article}

\usepackage[margin=1in]{geometry}
\usepackage{amsthm}
\usepackage[dvipsnames]{xcolor}

\usepackage{amsmath, amssymb, amsfonts}
\usepackage{bm}
\usepackage{mathrsfs}

\usepackage{graphicx}
\usepackage{booktabs}
\usepackage{multirow}
\usepackage{float}
\usepackage{wrapfig}
\usepackage[small]{caption}
\usepackage{subcaption}
\usepackage[most]{tcolorbox}

\usepackage{algorithm}
\usepackage{algorithmic}

\usepackage{arydshln}
\usepackage{soul}
\usepackage{url}
\usepackage{pifont}
\usepackage[round,authoryear]{natbib}
\usepackage[
  colorlinks,
  linkcolor=blue,
  citecolor=red,
  urlcolor=Emerald,
  pdftitle={Diffusion Bootstrap for High-Dimensional Linear Models},
  pdfauthor={Ce Liang and Wei Ma},
  pdfsubject={High-dimensional bootstrap inference using diffusion distribution estimators},
  pdfkeywords={bootstrap, diffusion models, high-dimensional inference}
]{hyperref}

\newenvironment{keywords}{\par\noindent\textbf{Keywords:}\ }{\par}

\theoremstyle{plain}
\newtheorem{mytheorem}{Theorem}[section]
\newtheorem{mylemma}[mytheorem]{Lemma}
\newtheorem{mycorollary}[mytheorem]{Corollary}
\newtheorem{myproposition}[mytheorem]{Proposition}
\theoremstyle{definition}
\newtheorem{myassumption}[mytheorem]{Assumption}
\newtheorem{mydefinition}[mytheorem]{Definition}
\theoremstyle{plain}
\newtheorem{remark}[mytheorem]{Remark}
\newtheorem{example}[mytheorem]{Example}

\newcommand{\op}{\mathrm{op}}
\DeclareMathOperator{\Var}{Var}
\DeclareMathOperator{\Cov}{Cov}
\DeclareMathOperator{\Lip}{Lip}
\DeclareMathOperator{\Ent}{Ent}
\DeclareMathOperator{\diag}{diag}
\DeclareMathOperator{\tr}{tr}
\DeclareMathOperator*{\argmin}{argmin}
\DeclareMathOperator{\sgn}{sgn}
\DeclareMathOperator{\supp}{supp}

\DeclareMathOperator{\Unif}{Unif}
\DeclareMathOperator{\Exp}{Exp}

\begin{document}

\title{Diffusion Bootstrap for High-Dimensional Linear Models}

\author{Ce Liang\\
       Institute of Statistics and Big Data, Renmin University of China\\
       \texttt{liangce158@ruc.edu.cn}
       \and
       Wei Ma\\
       Institute of Statistics and Big Data, Renmin University of China\\
       \texttt{mawei@ruc.edu.cn}}
\date{}

\maketitle

\begin{abstract}%
Classical bootstrap methods can behave poorly in high-dimensional linear models: the pairs bootstrap often yields overly conservative inference, whereas the residual bootstrap can be anti-conservative, reflecting systematic failures in variance calibration. We propose a diffusion-based pairs bootstrap that replaces the empirical joint distribution with a learned generative law. We establish variance consistency under a score approximation assumption, using complementary SDE and PDE arguments. Counterexamples show that terminal $W_4$ convergence alone is insufficient for variance consistency. Experiments indicate that diffusion pairs bootstrap improves variance calibration and generally improves Type~I error calibration, including in settings not covered by our theory.
\end{abstract}


\section{Introduction}
The bootstrap \citep{efron1979bootstrap} provides a general plug-in approach to approximating the sampling distribution of a statistic without requiring a fully specified parametric model. In the classical nonparametric bootstrap, the unknown underlying law is replaced by the empirical distribution, and the statistic is recomputed on samples drawn from this plug-in law. More generally, the validity of a plug-in bootstrap depends on whether the fitted distribution reproduces the features of the underlying law that govern the statistic of interest.

This principle becomes delicate in high-dimensional and non-regular problems, where standard bootstrap procedures may fail \citep{bickel1983bootstrapping,kosorok2008bootstrapping, groeneboom2024confidence}. A prominent example is linear regression with random design in the proportional regime $p/n\to\kappa\in(0,1)$. Even under a Gaussian design and Gaussian errors, classical residual and pairs bootstrap procedures exhibit systematic but opposite variance distortions. The residual bootstrap resamples fitted residuals whose variance is already shrunk by high-dimensional projection and therefore tends to underestimate uncertainty. The pairs bootstrap, by contrast, perturbs the geometry of the design and tends to overestimate the variance of the regression estimator \citep{el2018can}. Consequently, residual bootstrap confidence intervals can be anti-conservative, whereas pairs bootstrap intervals can be overly conservative.

The failure of empirical resampling does not by itself imply that the bootstrap principle must fail. It instead raises the possibility of replacing the empirical distribution by a more informative distribution estimator. Modern generative models provide one way to construct such an estimator. Although unrestricted high-dimensional distribution estimation is subject to the curse of dimensionality, structured distribution classes may admit much simpler score representations. Generative models based on neural networks can exploit smoothness, low-dimensional geometry, or compositional structure \citep{oko2023diffusion,bach2017breaking,chen2023score,cole2024score}. Recent work has begun to explore the use of generative models in bootstrap procedures \citep{tran2026generative}.

Among modern generative models, diffusion models are especially attractive because they learn the scores of a continuum of smoothed distributions and generate observations through a learned reverse-time stochastic differential equation \citep{song2020score,croitoru2023diffusion}. Recent studies have also explored the use of generative models and synthetic data for downstream statistical inference tasks \citep{liu2024novel,wang2025diffusion,ma2026synthetic, zhang2026doubly, ma2026does}.

In this paper, we study a bootstrap procedure based on diffusion, which we call diffusion pairs bootstrap. Rather than resampling the observed pairs $(X_i,Y_i)$ from their empirical distribution, we train a diffusion model on the joint observations and draw bootstrap samples from the terminal law of the learned reverse diffusion. We also consider a diffusion residual bootstrap, in which a diffusion model is trained only on the fitted residuals. These two constructions behave very differently in the proportional regime. Learning the joint law can correct the geometric distortion created by empirical pairs resampling, whereas learning the fitted residual law alone does not recover information that has already been lost through high-dimensional projection.

The role of diffusion in our analysis goes beyond providing a generic distribution estimator that is close to the target law in Wasserstein distance. Terminal Wasserstein convergence alone is not sufficient for bootstrap variance consistency; see Section~\ref{subsec:w4-counterexamples}. The variance of the bootstrap ordinary least squares (OLS) estimator depends on inverse moments of the generated Gram matrix and is therefore highly sensitive to rare, nearly singular bootstrap samples. Such degeneracies can have negligible average transportation cost and hence need not be ruled out by Wasserstein convergence of the fitted distribution to the target distribution.

Score approximation along the full Ornstein--Uhlenbeck interpolation controls the discrepancy between the exact and learned reverse drifts throughout the reverse evolution, rather than only at the terminal time. This control allows us to compare the exact and learned Fokker--Planck equations and to derive density regularity, density-ratio estimates, and anti-concentration bounds for the generated design. These estimates control the lower tail and inverse moments of the generated Gram matrix, which are the quantities needed for bootstrap variance consistency. The diffusion interpolation therefore supplies information that is absent from terminal Wasserstein convergence.

Our assumptions therefore focus on score approximation. To show that these assumptions are not vacuous in the proportional regime, we identify structured models for which the scores along the smoothed diffusion path can be estimated consistently even though direct high-dimensional distribution estimation remains difficult.

A basic example is the joint law of \(Y=X^\top\beta+\varepsilon\) under a standard Gaussian design. Although direct estimation of this joint law is difficult when \(p\) is proportional to \(n\), Section~\ref{sec:structured-minimax} gives a minimax upper bound of order \(\log p/n\) for the corresponding score estimation problem. The score error therefore vanishes even when \(p/n\to\kappa\in(0,1)\). We also treat several structured non-Gaussian models. In these examples, the design score depends on a fixed number of parameters, while the regression vector is controlled by an \(\ell_1\) bound. The resulting joint distributions still have full dimensional support, but their score estimation risk vanishes in the proportional regime. These examples address the attainability of the score assumptions. The variance theorem itself applies to a different and broader distributional class: it allows a non-Gaussian, strongly log-concave design with a general well-conditioned covariance matrix, while retaining Gaussian regression noise. Uniform curvature along the Ornstein--Uhlenbeck flow supplies dimension-free Poincar\'e and log-Sobolev inequalities, and one-dimensional log-concave small-ball bounds control the true Gram matrix. These properties replace the explicit Gaussian score and inverse-Wishart calculations used in the Gaussian benchmark. The diffusion model itself is still trained directly on the observed joint vectors and does not impose a parametric likelihood on $(X,Y)$.

Because the theorem does not cover every design used in practice, we also evaluate the same procedure across ten combinations of design and error distributions. These settings include i.i.d.\ Laplace covariates, Laplace errors, and heterogeneous elliptical designs with Gaussian, uniform, and exponential radial multipliers. We use the same diffusion architecture and training hyperparameters in all experiments. Figure~\ref{fig:type1-curve} summarizes the Type~I error and variance calibration for the Gaussian benchmark. Beyond this benchmark, diffusion pairs bootstrap keeps variance ratios close to one for designs with i.i.d.\ coordinates and often improves Type~I error calibration. It also substantially reduces the variance distortion under heterogeneous elliptical designs at moderate-to-large aspect ratios, although strong radial heterogeneity makes exact calibration more difficult. These findings suggest that the improvement is not limited to the smooth strongly log-concave models covered by our proportional-regime theory.

Our main contributions are summarized as follows.
\begin{itemize}
    \item We propose diffusion pairs bootstrap, which replaces empirical resampling with samples generated by a learned reverse diffusion. Our analysis uses score approximation along the full Ornstein--Uhlenbeck path, rather than only Wasserstein convergence of the terminal distribution. This additional information helps control rare degeneracies in the generated design matrix that empirical resampling and terminal distributional convergence may fail to capture.

    \item We establish variance consistency of diffusion pairs bootstrap in both fixed-dimensional and proportional high-dimensional regimes. When $p_n/n\to\kappa\in(0,1)$, the result applies to OLS contrasts under strongly log-concave designs that need not be Gaussian and may have general covariance matrices whose eigenvalues are uniformly bounded above and away from zero.

    \item We establish upper bounds of order $\log p/n$ for the minimax integrated fourth-moment score risk in the standard Gaussian model and several structured extensions, including Gaussian AR(1), covariance perturbations of fixed rank, and a product exponential family beyond Gaussian distributions. We also construct counterexamples showing that convergence in Wasserstein distance of order four alone does not guarantee bootstrap variance consistency.

    \item Our experiments cover ten combinations of design and error distributions. Using the same architecture and hyperparameters throughout, diffusion pairs bootstrap substantially improves calibration for Gaussian designs and designs with independent coordinates. Under heterogeneous elliptical designs, it also reduces the variance distortion of classical pairs bootstrap at moderate-to-large aspect ratios.
\end{itemize}

\begin{figure}[t]
    \centering
    \includegraphics[width=1.0\linewidth]{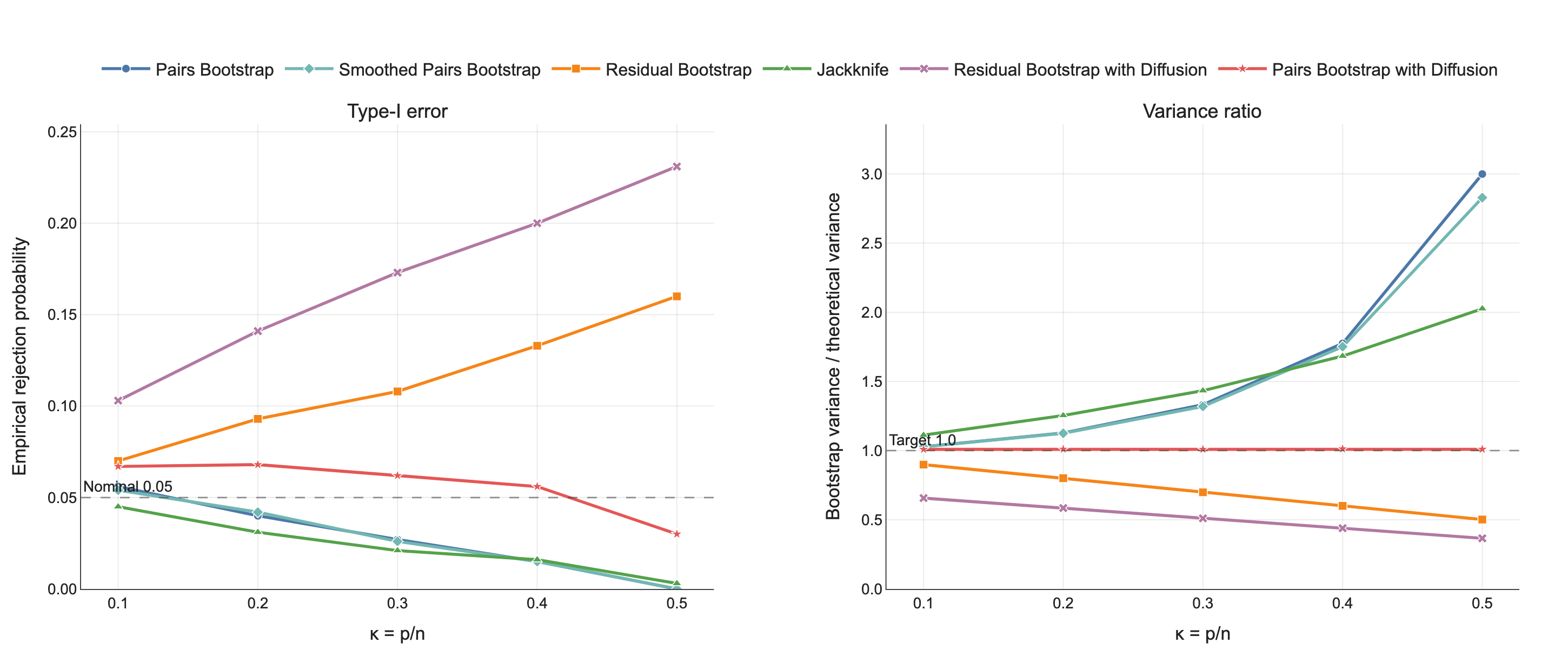}
    \caption{
    Type~I error and variance calibration in the proportional high-dimensional regime under Gaussian design and Gaussian errors. The sample size is $n=500$, the number of Monte Carlo replications is $R=1000$, and the nominal level is $\alpha=0.05$. The left panel reports the empirical Type~I error as the aspect ratio $\kappa=p/n$ varies; the horizontal dashed line marks the nominal level. The right panel reports the variance ratio, defined as the bootstrap variance estimate divided by the asymptotic theoretical variance $1/\{n(1-\kappa)\}$; the horizontal dashed line marks the target value one.}
    \label{fig:type1-curve}
\end{figure}

\subsection{Background: Bootstrap in linear models}
Consider the linear model with random design
\begin{equation*}
    Y_i=X_i^\top\beta+\varepsilon_i, \qquad i=1,\ldots,n,
\end{equation*}
where $X_i\in\mathbb R^p$, $\mathbb E[X_i]=0$, $\mathbb E[\varepsilon_i\mid X_i]=0$, and $\Var(\varepsilon_i)=\sigma_\varepsilon^2$. The unknown regression coefficient is $\beta\in\mathbb R^p$, and the ordinary least-squares estimator is
\begin{equation*}
    \widehat\beta=\argmin_{b\in\mathbb R^p}\sum_{i=1}^n(Y_i-X_i^\top b)^2.
\end{equation*}
Our primary interest is bootstrap inference for deterministic linear contrasts of $\widehat\beta$; in the proportional regime, both the dimension of $\widehat\beta$ and the contrast vector may depend on $n$.

The two classical procedures are the residual bootstrap and the pairs bootstrap. For the residual bootstrap, let $e_i=Y_i-X_i^\top\widehat\beta$ be the fitted residuals and let $\widehat G_n$ denote their centered empirical distribution. Bootstrap residuals $e_1^*,\ldots,e_n^*$ are drawn from $\widehat G_n$, the design is kept fixed, and the bootstrap responses are defined by
\begin{equation*}
    Y_i^*=X_i^\top\widehat\beta+e_i^*.
\end{equation*}
The pairs bootstrap instead draws the bootstrap observations directly from the empirical distribution of the observed pairs $(X_i,Y_i)$.

When $p$ is fixed, consistency of the pairs and residual bootstrap is well established under standard moment and nonsingularity conditions \citep{freedman1981bootstrapping,eck2018bootstrapping,bickel1981some}. Some validity results continue to hold when $p$ grows slowly relative to $n$. For example, consistency of pairs bootstrap is available under conditions such as $p^{1+\delta}/n\to0$ \citep{mammen1993bootstrap}, and fixed-design residual bootstrap is valid under a condition of the form $p^2/n\to0$ \citep{bickel1983bootstrapping}. A related line of work studies non-bootstrap inference in linear models with many nuisance covariates. In particular, \citet{cattaneo2018inference} develop heteroskedasticity-robust inference for a fixed-dimensional parameter of interest when the number of nuisance covariates may be a non-negligible fraction of the sample size.

Our focus is the proportional regime $p/n\to\kappa\in(0,1)$, where the design
geometry itself has a non-negligible effect on the sampling and bootstrap
variances \citep{el2018can}.

\subsection{Background: Diffusion models as distribution estimators}
Let $(W_t)_{t\in[0,\bar T]}$ be a standard Brownian motion and let $\beta_t$ be a positive noise schedule. Starting from a random vector $Z_0\sim p_0$, consider the Ornstein--Uhlenbeck forward process
\begin{equation*}
    \mathrm dZ_t=-\beta_tZ_t\,\mathrm dt+\sqrt{2\beta_t}\,\mathrm dW_t.
\end{equation*}
Writing
\begin{equation*}
    B(t)=\int_0^t\beta_s\,\mathrm ds,\qquad\alpha_t=e^{-B(t)},\qquad\sigma_t^2=1-e^{-2B(t)},
\end{equation*}
the transition law satisfies
\begin{equation*}
    Z_t\mid Z_0\sim N(\alpha_tZ_0,\sigma_t^2I_d).
\end{equation*}
If $p_t$ denotes the density of $Z_t$, the corresponding reverse-time process has drift determined by the score $\nabla\log p_t$. More precisely, under standard regularity conditions, the reverse process satisfies
\begin{equation*}
    \mathrm dY_t=\beta_{\bar T-t}\left\{Y_t+2\nabla\log p_{\bar T-t}(Y_t)\right\}\mathrm dt+\sqrt{2\beta_{\bar T-t}}\,\mathrm dW_t,\qquad Y_0\sim p_{\bar T}.
\end{equation*}
When $B(t)\to\infty$, the forward law approaches $N(0,I_d)$. A diffusion model replaces the exact score by a learned score $\widehat s_{n,t}$ and initializes the reverse process from the standard Gaussian law:
\begin{equation*}
    \mathrm d\widehat Y_t=\beta_{\bar T-t}\left\{\widehat Y_t+2\widehat s_{n,\bar T-t}(\widehat Y_t)\right\}\mathrm dt+\sqrt{2\beta_{\bar T-t}}\,\mathrm dW_t,\qquad \widehat Y_0\sim N(0,I_d).
\end{equation*}
The terminal law of this learned reverse process is the random distribution estimator used by our bootstrap procedure.

Existing diffusion theory studies errors in distribution estimation measured by metrics such as Wasserstein distance and total variation \citep{oko2023diffusion,de2021diffusion,chen2023improved, gao2025wasserstein}, improvements under assumptions on intrinsic dimension \citep{oko2023diffusion,chen2023score,de2022convergence}, and approximation under structural restrictions on the target score \citep{cole2024score}. One route to convergence assumes score approximation and derives stability of the learned reverse process \citep{gao2025wasserstein,chen2022sampling,lee2023convergence}; another derives score approximation from assumptions on the target distribution and the score estimator \citep{oko2023diffusion}. We use the first route as our starting point. A terminal distributional bound, however, does not by itself provide the lower-tail control needed for bootstrap inference. We therefore propagate the score approximation bound through the reverse SDE and Fokker--Planck equations to derive the required density and geometric estimates. For several structured Gaussian and non-Gaussian model classes, we also show that the assumed score accuracy is statistically attainable.

\subsection{Organization of the paper}
Section~2 develops diffusion pairs bootstrap in the fixed-dimensional regime. It proves convergence of the learned law, establishes density and anti-concentration estimates for the generated design, and derives variance and distributional bootstrap consistency. Section~3 studies the proportional regime under strongly log-concave designs with general well-conditioned covariance and proves high-dimensional variance consistency using density-ratio estimates and lower-tail bounds for the generated Gram matrix.

Section~4 explains why score estimation with diffusion models can remain feasible in structured high-dimensional models. It establishes upper bounds for the integrated score estimation risk and compares score estimation with direct approximation by an atomic distribution. Section~5 shows that terminal Wasserstein convergence alone does not guarantee variance consistency by constructing distribution estimators with vanishing Wasserstein error but divergent bootstrap variance. Section~6 discusses the failure of diffusion residual bootstrap. Section~7 evaluates the proposed method across ten combinations of Gaussian and non-Gaussian design and error distributions, including designs with i.i.d.\ Laplace coordinates and heterogeneous elliptical designs, and Appendix~\ref{app:simulation-tables} reports the complete numerical results.

\subsection{Notation}

For a positive integer $d$, let $I_d$ denote the $d\times d$ identity matrix. For a vector $v$, $\|v\|_2$ denotes its Euclidean norm and $\|v\|_1$ denotes its $\ell_1$-norm; when no ambiguity is possible, we abbreviate $\|v\|_2$ as $\|v\|$. For a matrix $A$, $\|A\|_{\op}$ denotes its operator norm. If $A$ is symmetric, $\lambda_{\min}(A)$ denotes its smallest eigenvalue. We write $a_n\lesssim b_n$ if $a_n\leq Cb_n$ for a constant $C$ independent of $n$, and $a_n\asymp b_n$ if both $a_n\lesssim b_n$ and $b_n\lesssim a_n$ hold. Convergence in probability is denoted by $\to_p$.

For probability measures $\nu$ and $\pi$ on a Euclidean space, $W_q(\nu,\pi)$ denotes the $q$-Wasserstein distance. If $Z$ is a random variable, $\mathcal L(Z)$ denotes its law. If $Z\sim\nu$, we write $\mathbb E_\nu f(Z)=\int f(z)\,\mathrm d\nu(z)$. Bootstrap probability, expectation, and variance are denoted by $\mathbb P^*$, $\mathbb E^*$, and $\Var^*$, respectively. Conditioning variables are stated explicitly when needed.

An observation is denoted by $Z_i=(X_i,Y_i)\in\mathbb R^{p+1}$, where $X_i\in\mathbb R^p$ is the covariate vector and $Y_i\in\mathbb R$ is the response. In the high-dimensional regime, $p=p_n$ may depend on $n$, and we write $d_n=p_n+1$ and $p_n/n\to\kappa\in(0,1)$. A deterministic contrast vector is denoted by $c_n\in\mathbb R^{p_n}$, with $\|c_n\|_2=1$. In the fixed-dimensional setting, we suppress the dependence on $n$ and write $p$ and $c$.

\section{Pairs bootstrap with diffusion distribution estimators}
A natural idea is to replace the empirical distribution used by a bootstrap procedure with a generative distribution estimator. We first develop this idea for the pairs bootstrap by fitting the joint law of $(X,Y)$; the residual version is considered separately in Section~\ref{sec:residual-bootstrap}.

Consider the pairs bootstrap under the linear model with random design
\begin{equation*}
    Y_i=X_i^\top\beta+\varepsilon_i,\qquad i=1,\ldots,n.
\end{equation*}
The joint observations $Z_i=(X_i,Y_i)\in\mathbb R^{p+1}$ are i.i.d.\ from a law $\mu$. Given $Z_1,\ldots,Z_n$, we train a diffusion model on these pairs and denote its fitted distribution by $\widehat\mu_n$.

Throughout this section, $\widehat\mu_n$ is viewed as a random probability measure depending on the observed data. We do not explicitly model the additional randomness arising from optimization or training. Equivalently, conditional on the observed data, the fitted diffusion model is treated as fixed.

The diffusion pairs bootstrap proceeds by drawing
\begin{equation*}
    Z_1^*,\ldots,Z_n^*\overset{\mathrm{iid}}{\sim}\widehat\mu_n,\qquad Z_i^*=(X_i^*,Y_i^*),
\end{equation*}
and computing the bootstrap OLS estimator
\begin{equation*}
    \widehat\beta^*=\left(\sum_{i=1}^nX_i^*X_i^{*\top}\right)^{-1}\left(\sum_{i=1}^nX_i^*Y_i^*\right).
\end{equation*}
Thus, unlike the classical pairs bootstrap, the diffusion pairs bootstrap draws new pairs from an estimated law of $(X,Y)$. To distinguish learning the joint law from merely smoothing the empirical distribution, our simulations also include a smoothed pairs bootstrap. Its conservative behavior persists in high dimensions, suggesting that the relevant issue is whether the fitted law recovers the features of the joint distribution that determine OLS variance.

This observation motivates assumptions on the diffusion dynamics rather than on the terminal law alone. Rather than assuming $W_q(\widehat\mu_n,\mu)\to_p 0$ as a generic property of the generated distribution, we impose approximation and stability assumptions on the learned score along the reverse dynamics. These assumptions imply $W_q$ consistency and provide the additional density control needed for variance consistency.

\begin{remark}[Scope of the OLS analysis]
Bootstrap methods for estimators defined through estimating equations have been studied in classical settings \citep{10.1093/biomet/82.2.263}. In the proportional asymptotic regime, \citet{lei2018asymptotics} establish coordinate-wise asymptotic normality for high-dimensional regression $M$-estimators under a fixed-design framework. In this paper, we restrict the theoretical analysis to OLS.
\end{remark}

We first formalize the argument in the classical fixed-dimensional regime, where $p$ is fixed and $n\to\infty$. Let $Z_i=(X_i,Y_i)\in\mathbb R^{p+1}$, $i=1,\dots,n$, be i.i.d.\ observations from $\mu$. Set $d=p+1$.

Following earlier bootstrap notation \citep{eck2018bootstrapping}, for any probability measure $\nu$ on $\mathbb R^{p+1}$, define
\begin{equation*}
    \Sigma_X(\nu)=\int xx^\top \mathrm d\nu(x,y),\qquad m_{XY}(\nu)=\int xy\,\mathrm d\nu(x,y),
\end{equation*}
and
\begin{equation*}
    \beta(\nu)=\Sigma_X(\nu)^{-1}m_{XY}(\nu),
\end{equation*}
whenever $\Sigma_X(\nu)$ is invertible. Let
\begin{equation*}
    e_\nu(x,y)=y-x^\top\beta(\nu).
\end{equation*}
For a fixed vector $c\in\mathbb R^p$ with $\|c\|_2=1$, define the influence function
\begin{equation*}
    S_\nu(x,y)=c^\top\Sigma_X(\nu)^{-1}x\,e_\nu(x,y),\qquad\tau^2(\nu)=\Var_\nu\{S_\nu(X,Y)\}.
\end{equation*}

Consider the Ornstein--Uhlenbeck forward process
\begin{equation*}
    \mathrm dZ_t=-\beta_tZ_t\,\mathrm dt+\sqrt{2\beta_t}\,\mathrm dW_t,\qquad Z_0\sim\mu.
\end{equation*}
Let $p_t$ denote the density of $Z_t$, define its score by $s_t=\nabla\log p_t$, and set $B(t)=\int_0^t\beta_s\,\mathrm ds$. Let $T_n\to\infty$ be the terminal time. The learned reverse process is
\begin{equation*}
    \mathrm d\widehat Y_t=\beta_{T_n-t}\left\{\widehat Y_t+2\widehat s_{n,T_n-t}(\widehat Y_t)\right\}\,\mathrm dt+\sqrt{2\beta_{T_n-t}}\,\mathrm dW_t,\qquad\widehat Y_0\sim N(0,I_d).
\end{equation*}
Let $\mathcal F_n$ denote the $\sigma$-field with respect to which the learned score $\widehat s_n$ is measurable, and define the learned generated law
\begin{equation*}
    \widehat \mu_n=\mathcal L(\widehat Y_{T_n}\mid \mathcal F_n).
\end{equation*}
We make the following assumption:
\begin{myassumption}[Target distribution]
    The target law $\mu$ has density
    \begin{equation*}
        \mathrm d\mu(z)= e^{-U_0(z)}\,\mathrm dz, \qquad U_0\in C^2(\mathbb R^d),
    \end{equation*}
    and there exist constants $0<m_0\leq L_0<\infty$ such that
    \begin{equation*}
        m_0I_d\preceq\nabla^2U_0(z)\preceq L_0I_d,\qquad z\in\mathbb R^d.
    \end{equation*}
    Moreover, for some $q>4$,
    \begin{equation*}
        \mathbb E_\mu\|Z\|_2^q<\infty,\qquad\Sigma_X(\mu)\succ0,\qquad\tau^2(\mu)>0.
    \end{equation*}
    \label{ass:target-distribution}
\end{myassumption}

\begin{remark}
    Strong log-concavity already implies $\mathbb E_{\mu}\|Z\|_2^q<\infty$ for every finite $q$. We nevertheless state the required moment condition explicitly for clarity.
\end{remark}

\begin{remark}
Assumption~\ref{ass:target-distribution} serves two related purposes in our fixed-dimensional analysis. First, it provides the tail control and stability conditions needed to convert score approximation along the reverse Ornstein--Uhlenbeck flow into convergence of the learned terminal law in \(W_q\) for some \(q>4\). Such higher-order Wasserstein convergence controls both weak convergence and the polynomial-moment components entering the OLS functional. Terminal Wasserstein convergence alone, however, is not sufficient for bootstrap validity, because the OLS estimator also depends on the inverse of the generated covariance and Gram matrices. We therefore additionally use the curvature assumption to establish density regularity and anti-concentration estimates for the learned law. These estimates rule out nearly singular generated designs and yield the inverse-moment control required by the bootstrap argument.

The uniform lower Hessian bound gives the target distribution Gaussian-type tails and supplies dissipativity and stability for the reverse diffusion dynamics. The uniform upper Hessian bound controls the growth and regularity of the exact score along the Ornstein--Uhlenbeck interpolation. Together, these conditions allow moment estimates, stability of the reverse flow, density regularity, and anti-concentration to be handled by the same set of assumptions. In particular, uniform strong log-concavity already implies the existence of moments of every finite order. Hence, the displayed \(q\)-moment condition is not an independent assumption; it is retained only to make explicit the moment order used in the subsequent bootstrap analysis.

Uniform strong log-concavity is a convenient sufficient condition rather than an intrinsic requirement of the bootstrap principle. Existing Wasserstein convergence results for diffusion models on non-compact spaces have weakened global strong log-concavity to some extent. For example, recent results allow potentials that are locally non-convex, or merely semiconvex, while typically retaining strong convexity at infinity or imposing an alternative coercive tail-dissipativity condition \citep{gentiloni2025beyond, bruno2025wasserstein}. Thus, global uniform strong convexity is not necessary for \(W_2\)-convergence itself, but currently available non-compact analyses generally continue to rely on some form of confining tail stability.

The assumption includes the linear model with additive Gaussian errors
\begin{equation*}
    X\sim N(0,\Sigma_X),\qquad\varepsilon\sim N(0,\sigma^2),\qquad Y=X^\top\beta+\varepsilon,
\end{equation*}
whenever \(\Sigma_X\succ0\) and \(\sigma^2>0\). More generally, suppose that the covariates \(X\) and the additive noise \(\varepsilon\) are independent and both admit smooth uniformly strongly log-concave densities whose potentials satisfy uniform upper and lower Hessian bounds. Then, after the invertible linear transformation
\begin{equation*}
    (X,\varepsilon)\longmapsto (X,X^\top\beta+\varepsilon),
\end{equation*}
the joint distribution of \((X,Y)\) satisfies the same type of curvature condition, with modified constants.
\end{remark}

\begin{myassumption}[Noise schedule]
    The noise schedule satisfies
    \begin{equation*}
        0<\underline\beta\leq\beta_t\leq\overline\beta<\infty,\qquad B(T_n)\to\infty.
    \end{equation*}
    \label{ass:noise-schedule}
\end{myassumption}
The uniform bounds give a nondegenerate time change, while $B(T_n)\to\infty$ ensures that the forward terminal law approaches the standard Gaussian reference law.

\begin{myassumption}[Score approximation and Lipschitz control]
\label{ass:score-approximation}
Conditionally on $\mathcal F_n$, the learned reverse SDE admits a unique non-explosive weak solution on $[0,T_n]$, and its time marginals have densities.
Let
\begin{equation*}
    e_{n,t}(z)=\widehat s_{n,t}(z)-s_t(z).
\end{equation*}
For some $T_n\to\infty$ and some $q>4$, define the $L^q$ score approximation error by
\begin{equation*}
    \mathcal A_{n,q}(T_n)=\int_0^{T_n}\beta_{T_n-t}\left(\mathbb E\left[\|e_{n,T_n-t}(\widehat Y_t)\|^q\,\middle|\,\mathcal F_n\right]\right)^{1/q}\mathrm dt.
\end{equation*}
Assume that
\begin{equation*}
    \mathcal A_{n,q}(T_n)\to_p 0.
\end{equation*}

In addition, assume that there exists a fixed constant $h_0>0$ such that $h_0<T_n$ for all sufficiently large $n$, and, for Lebesgue-a.e. $t\in[0,h_0]$, the map $e_{n,t}$ is globally Lipschitz. Moreover,
\begin{equation*}
    \mathrm{ess}\sup_{0\leq t\leq h_0}\Lip(e_{n,t})=O_p(1),
\end{equation*}
where
\begin{equation*}
    \Lip(f):=\sup_{z\neq z'}\frac{\|f(z)-f(z')\|}{\|z-z'\|}.
\end{equation*}
\end{myassumption}

\begin{remark}[Score approximation versus Wasserstein convergence]
\label{rem:score-versus-wasserstein}
Assumption~\ref{ass:score-approximation} is imposed on the learned score along the Ornstein--Uhlenbeck path, rather than directly assuming terminal Wasserstein convergence. Some analyses instead suppose directly that $W_k(\widehat\mu_n,\mu)\to_p0$ for some \(k\geq1\) \citep{tran2026generative,ma2026does}. Here we derive
\begin{equation*}
    W_q(\widehat\mu_n,\mu)\to_p0,\qquad q>4,
\end{equation*}
from score approximation and regularity of the learned reverse dynamics.

This distinction matters because terminal Wasserstein convergence alone does not control bootstrap variance. The latter depends on inverse moments of the generated Gram matrix and is therefore sensitive to rare nearly singular bootstrap designs. Section~\ref{subsec:w4-counterexamples} shows that, even for a Gaussian target, one may have $W_4(\widehat\mu_n,\mu)\to0$ while the corresponding bootstrap variance diverges.

Three features of diffusion enter the proof. The Ornstein--Uhlenbeck semigroup interpolates between the target and a Gaussian reference law; the score determines the reverse-time drift; and the exact and learned densities satisfy Fokker--Planck equations with the same diffusion operator. Together, these properties connect score error to the lower tail of the generated Gram matrix. A different generative estimator could be treated similarly if it provided the same analytic control.
\end{remark}

\begin{remark}[Lipschitz condition]
    The local Lipschitz condition in Assumption~\ref{ass:score-approximation} is imposed only near the target end of the forward flow, namely on \(t\in[0,h_0]\). This is the portion of the reverse dynamics closest to the target law, where the density and stability estimates are most delicate. Together with the remaining growth, moment, and dissipativity conditions, the local Lipschitz assumption ensures the required well-posedness and stability of the learned reverse process. Without sufficient regularity and growth control, the reverse dynamics may fail to be unique or stable and may, in extreme cases, become explosive; see \citep{beyler2025convergence} for related discussion.
\end{remark}

\begin{remark}[Cases not covered]
    We focus on continuous-time diffusion models and do not analyze discretization error. Compactly supported distributions are also outside the present assumptions.
\end{remark}

Variance consistency requires more than weak or Kolmogorov consistency of the bootstrap distribution. Besides showing that the bootstrap distribution is asymptotically correct on events with high probability, we also need to control the contribution from the exceptional events on which the bootstrap Gram matrix is nearly singular. In our setting, this requires quantitative lower-tail bounds for the generated Gram matrix. Using the Fokker--Planck equation for the learned reverse process, we first obtain an $L^\infty$ bound on the generated density by a Moser iteration. This density bound implies a slab anti-concentration inequality for the generated covariates, which in turn yields the required lower-tail estimate for the bootstrap Gram matrix. The key lemma is the following.

\begin{mylemma}
    Under Assumptions~\ref{ass:target-distribution}, \ref{ass:noise-schedule}, and~\ref{ass:score-approximation}, conditionally on $\widehat\mu_n$, let $(X_i^*,Y_i^*)\overset{\mathrm{iid}}{\sim}\widehat\mu_n$ and set
    \begin{equation*}
        \widehat\Sigma_n^*=\frac1n\sum_{i=1}^n X_i^*X_i^{*\top}.
    \end{equation*}
    Then, for every fixed $a>0$, there exist a deterministic constant $\eta_a>0$ and a random sequence $C_{n,a}=O_p(1)$ such that, with probability tending to one,
    \begin{equation*}
        \sup_{0<t<\eta_a}t^{-a}\mathbb P^*\bigl(\lambda_{\min}(\widehat\Sigma_n^*)\leq t\mid \widehat\mu_n\bigr)\leq C_{n,a}.
    \end{equation*}
\end{mylemma}

The proof is based on a density estimate obtained through a Moser iteration. On the last time interval $[T_n-h_0,T_n]$, the conditional density $\rho_{n,r}$ of the learned reverse process satisfies the renormalized energy inequality
\begin{equation*}
    \|\rho_{n,b}\|_{\ell}^{\ell}+\frac{4(\ell-1)}{\ell}\int_a^b a_r\|\nabla\rho_{n,r}^{\ell/2}\|_2^2\,\mathrm dr\leq\|\rho_{n,a}\|_{\ell}^{\ell}+(\ell-1)\int_a^b c_{n,r}\|\rho_{n,r}\|_{\ell}^{\ell}\,\mathrm dr ,
\end{equation*}
for every $\ell\geq2$ and for Lebesgue-a.e. pair $(a,b)$ satisfying $T_n-h_0<a<b<T_n$. Here $a_r$ is the diffusion coefficient and $c_{n,r}$ is the negative-divergence contribution of the learned reverse drift.

The term involving $c_{n,r}$ is removed by the integrating factor
\begin{equation*}
    C_n(r)=\int_{T_n-h_0}^r c_{n,s}\,\mathrm ds,\qquad u_r=e^{-C_n(r)}\rho_{n,r}.
\end{equation*}
The transformed density $u_r$ satisfies a dissipative energy inequality. Since $\rho_{n,r}$ is a probability density, $\|u_r\|_1\leq1$. Applying the Gagliardo--Nirenberg interpolation inequality to the $\ell=2$ energy estimate gives an $L^1$-to-$L^2$ smoothing bound after positive time. Starting from this $L^2$ bound, a standard Moser iteration over the exponents \citep{gilbarg1998elliptic}
\begin{equation*}
    L^1\xrightarrow[\text{Gagliardo--Nirenberg}]{\text{energy estimate}}\underbrace{L^2\to L^{2(1+2/d)}\to L^{2(1+2/d)^2}\to\cdots\to L^{2(1+2/d)^j} \xrightarrow{\;j\to\infty\;} L^\infty }_{\text{Moser iteration}}
\end{equation*}
yields the terminal-time estimate
\begin{equation*}
    \|\rho_{n,T_n}\|_{\infty}\leq C_n\|\rho_{n,T_n-h_0}\|_1,\qquad C_n=O_p(1).
\end{equation*}
Since $\rho_{n,T_n-h_0}$ is a probability density, $\|\rho_{n,T_n-h_0}\|_1=1$, and therefore
\begin{equation*}
    \|\rho_{n,T_n}\|_{\infty}=O_p(1).
\end{equation*}

Combining this $L^\infty$ bound with the $q$-moment control implied by $W_q(\widehat\mu_n,\mu)\to_p 0$, we obtain a uniform slab anti-concentration estimate: for some $\gamma>0$ and some $K_n=O_p(1)$,
\begin{equation*}
    \sup_{\|v\|=1}\sup_{b\in\mathbb R}\mathbb P_{\widehat\mu_n}\bigl(|v^\top X-b|\leq\epsilon\bigr)\leq K_n\epsilon^\gamma.
\end{equation*}
This anti-concentration inequality controls the smallest singular value of a fixed-size block of bootstrap covariates. A block amplification argument, combined with a fixed-dimensional moment bound for the sample covariance, then yields the stated polynomial lower-tail bound for $\lambda_{\min}(\widehat\Sigma_n^*)$.

The lower-tail estimate supplies the inverse-moment control needed for OLS. The other ingredient is convergence of the generated law, which follows from the same assumptions.
\begin{mylemma}
    Under Assumptions~\ref{ass:target-distribution}, \ref{ass:noise-schedule}, and~\ref{ass:score-approximation},
    \begin{equation*}
        W_q(\widehat\mu_n,\mu)\to_p 0.
    \end{equation*}
\end{mylemma}

Combining these two ingredients gives the fixed-dimensional variance result.
\begin{mytheorem}
    Suppose that $p$ is fixed and Assumptions~\ref{ass:target-distribution}, \ref{ass:noise-schedule}, and~\ref{ass:score-approximation} hold. Then
    \begin{equation*}
        \frac{\Var^*\!\left(\sqrt n\,c^\top\{\widehat\beta^*-\beta(\widehat\mu_n)\}\,\middle|\,\widehat\mu_n \right)}{\Var\!\left(\sqrt n\,c^\top\{\widehat\beta-\beta(\mu)\}\right)}\to_p 1.
    \end{equation*}
    \label{fixed-vari}
\end{mytheorem}

\begin{mycorollary}
Under the assumptions of Theorem~\ref{fixed-vari},
\begin{equation*}
    \begin{aligned}
        \sup_{t\in\mathbb R}\Bigl|&\mathbb P^*\!\left(\sqrt n\,c^\top\{\widehat\beta^*-\beta(\widehat\mu_n)\}\leq t\,\middle|\,\widehat\mu_n\right)-\mathbb P\!\left(\sqrt n\,c^\top\{\widehat\beta-\beta(\mu)\}\leq t\right)\Bigr|\to_p 0.
    \end{aligned}
\end{equation*}
\end{mycorollary}

A rigorous proof is provided in Appendix~\ref{app:theorem}. We briefly outline its main steps. First, the OU marginals of a strongly log-concave target remain strongly log-concave, with time-dependent curvature $m_t^{\mathrm{sc}}$. This controls the positive expansion of the exact reverse drift and gives the coupling estimate
\begin{equation*}
    W_q(\widehat\mu_n,\mu)\leq C W_q\{N(0,I_d),p_{T_n}\}+C\mathcal A_{n,q}(T_n).
\end{equation*}
The first term vanishes because $B(T_n)\to\infty$, and the second vanishes by score approximation. Hence $W_q(\widehat\mu_n,\mu)\to_p 0$.

Second, the preceding lower-tail lemma and the layer-cake formula imply that, for every $Q<a$,
\begin{equation*}
    \mathbb E^*\!\left[\lambda_{\min}(\widehat\Sigma_n^*)^{-Q}\,\middle|\,\widehat\mu_n\right]=O_p(1).
\end{equation*}
Third, the bootstrap OLS statistic admits the $L^2$ linearization
\begin{equation*}
    \sqrt n\,c^\top\{\widehat\beta^*-\beta(\widehat\mu_n)\}=\frac1{\sqrt n}\sum_{i=1}^nS_{\widehat\mu_n}(Z_i^*)+R_n^*,
\end{equation*}
where $\mathbb E^*\{(R_n^*)^2\mid\widehat\mu_n\}\to_p 0$. This proves variance consistency. Distributional consistency follows from the corresponding conditional central limit theorem and P\'olya's theorem.

\section{High-dimensional cases}
We now turn from fixed \(p\) to the proportional high-dimensional regime $p_n/n\to\kappa\in(0,1)$. Under the standard Gaussian model with random design, for any deterministic contrast $c_n\in\mathbb R^{p_n}$ with $\|c_n\|_2=1$,
\begin{equation*}
    \Var\left(c_n^\top\widehat\beta\mid X\right)=\sigma_n^2 c_n^\top(X^\top X)^{-1}c_n.
\end{equation*}
By classical Wishart theory, in the proportional regime,
\begin{equation*}
    n\,c_n^\top(X^\top X)^{-1}c_n\to_p\frac{1}{1-\kappa}.
\end{equation*}
This describes the typical scale of the conditional variance. For the unconditional variance, note that
\begin{equation*}
    \mathbb E(c_n^\top\widehat\beta\mid X)=c_n^\top\beta_n,
\end{equation*}
which does not depend on $X$. The law of total variance therefore gives
\begin{equation*}
    \Var\left(c_n^\top\widehat\beta\right)=\sigma_n^2\mathbb E\left[c_n^\top(X^\top X)^{-1}c_n\right].
\end{equation*}
Since $X^\top X\sim W_{p_n}(n,I_{p_n})$, the inverse-Wishart expectation formula yields, for all sufficiently large $n$,
\begin{equation*}
    \mathbb E\left[(X^\top X)^{-1}\right]=\frac{1}{n-p_n-1}I_{p_n}.
\end{equation*}
Consequently,
\begin{equation*}
    \Var\left(c_n^\top\widehat\beta\right)=\frac{\sigma_n^2}{n-p_n-1}=\frac{\sigma_n^2}{n(1-\kappa)}\{1+o(1)\}.
\end{equation*}

This variance scale provides the benchmark for bootstrap procedures in the high-dimensional regime. The classical pairs bootstrap does not generally reproduce this benchmark: its expected conditional bootstrap variance is systematically distorted and, for standard pairs bootstrap weights, tends to overestimate the correct variance. Thus, the pairs bootstrap becomes conservative as $\kappa$ increases \citep{el2018can}. Our goal is to show that replacing the empirical distribution by a sufficiently accurate diffusion distribution estimator can recover the correct high-dimensional variance scale. We also extend the analysis from the standard Gaussian design to a non-Gaussian strongly log-concave design with a general well-conditioned covariance matrix.

For simplicity, let $d_n=p_n+1$ and write $Z_i=(X_i,Y_i)$ for the $i$th observation. Let $Z_0$ denote a generic draw from their common law and consider the standard OU forward process
\begin{equation*}
    \mathrm dZ_t=-Z_t\,\mathrm dt+\sqrt{2}\,\mathrm dW_t,
\end{equation*}
or, equivalently,
\begin{equation*}
    Z_t=e^{-t}Z_0+\sqrt{1-e^{-2t}}\,\xi,\qquad\xi\sim N(0,I_{d_n}).
\end{equation*}
Let $p_{n,t}$ denote the density of $Z_t$, and let
\begin{equation*}
    s_{n,t}(z)=\nabla\log p_{n,t}(z)
\end{equation*}
be its score.

We now state the assumptions for the high-dimensional result. The first assumption specifies the underlying distribution, and the remaining assumptions specify the contrast and the required accuracy and stability of the fitted diffusion distribution estimator.

\begin{myassumption}[Target Distribution]
    \label{ass:hd-setting}

    Let $p=p_n$ and suppose that
    \begin{equation*}
        \frac{p_n}{n}\to\kappa\in(0,1).
    \end{equation*}
    The observations $Z_i=(X_i,Y_i)\in\mathbb R^{p_n+1}$, $i=1,\ldots,n$, are i.i.d.\ and satisfy the following conditions.
    \begin{enumerate}
        \item \textit{Linear model with Gaussian noise.}
        We have
        \begin{equation*}
            Y_i=X_i^\top\beta_n+\varepsilon_i,\qquad\varepsilon_i\sim N(0,\sigma_n^2),\qquad\varepsilon_i\perp X_i.
        \end{equation*}
        \item \textit{Design distribution.}
        The distribution of $X_i$ has density
        \begin{equation*}
            \mu_{n,X}(\mathrm dx)=Z_{n,X}^{-1}\exp\{-V_n(x)\}\,\mathrm dx,
        \end{equation*}
        and satisfies
        \begin{equation*}
            \mathbb E X_i=0,\qquad\mathbb E X_iX_i^\top=\Sigma_n.
        \end{equation*}
        \item \textit{Strong log-concavity and covariance regularity.}
        There exist constants
        \begin{equation*}
            0<c_\Sigma\leq C_\Sigma<\infty,\qquad 0<\underline m\leq\overline m<\infty,
        \end{equation*}
        independent of $n$, such that
        \begin{equation*}
            c_\Sigma I_{p_n}\preceq\Sigma_n\preceq C_\Sigma I_{p_n},
        \end{equation*}
        and, for every $x\in\mathbb R^{p_n}$,
        \begin{equation}
            \underline m\,\Sigma_n^{-1}\preceq \nabla^2V_n(x)\preceq\overline m\,\Sigma_n^{-1}. \label{eq:relative-curvature}
        \end{equation}

        \item \textit{Parameter bounds and diffusion horizon.}
        There exist constants $B<\infty$ and $0<\sigma_{\min}^2\leq\sigma_{\max}^2<\infty$, independent of $n$, such that
        \begin{equation*}
            \|\beta_n\|_2\leq B,\qquad\sigma_{\min}^2\leq\sigma_n^2\leq\sigma_{\max}^2.
        \end{equation*}
        Moreover, with $d_n=p_n+1$, the diffusion horizon satisfies
        \begin{equation*}
            T_n\to\infty,\qquad d_ne^{-2T_n}\to0.
        \end{equation*}

        \item \textit{Existence of a suitable score estimator.}
        Let $p_{n,0}$ denote the joint density of $Z_i=(X_i,Y_i)$ and let
        \begin{equation*}
            s_{n,0}(z)=\nabla\log p_{n,0}(z)
        \end{equation*}
        be its joint score. We assume that the joint distribution belongs to a class for which a fitted score can satisfy Assumption~\ref{ass:hd-score-approximation}. Section~\ref{sec:structured-minimax} gives concrete sufficient examples.
    \end{enumerate}
\end{myassumption}

\begin{myassumption}[Contrast sequence]
    \label{ass:hd-contrast}

    The contrast vectors $c_n\in\mathbb R^{p_n}$ are deterministic and satisfy
    \begin{equation*}
        \|c_n\|_2=1.
    \end{equation*}
\end{myassumption}

Before stating the score condition, we record a functional inequality implied by the distributional assumptions. Under Assumption~\ref{ass:hd-setting}, the exact OU marginals admit a uniform log-Sobolev constant. Let $C_{\mathrm{LS}}<\infty$ denote a deterministic constant, independent of $n$ and $t$, such that
    \begin{equation*}
        \operatorname{Ent}_{p_{n,t}}(g^2)\leq2C_{\mathrm{LS}}\int_{\mathbb R^{d_n}}\|\nabla g(z)\|_2^2p_{n,t}(z)\,\mathrm dz.
    \end{equation*}

\begin{myassumption}[Score approximation and Lipschitz Control]
    \label{ass:hd-score-approximation}

    Conditionally on $\mathcal F_n$, the map $(t,z)\mapsto\widehat s_{n,t}(z)$ is jointly Borel measurable. Let
    \begin{equation*}
        e_{n,t}(z)=\widehat s_{n,t}(z)-s_{n,t}(z),
    \end{equation*}
    and assume that
    \begin{equation}
        \int_0^{T_n}\mathbb E_{Z_t\sim p_{n,t}}\left[\|e_{n,t}(Z_t)\|_2^4\right]\,\mathrm dt\to_p 0. \label{eq:score-l4}
    \end{equation}
    Assume that there exist fixed deterministic constants $L_\star>0$ and $\lambda_\star>0$ satisfying
    \begin{equation}
        L_\star<\frac14,\qquad\lambda_\star>32C_{\mathrm{LS}},\qquad4\lambda_\star C_{\mathrm{LS}}L_\star^2<1, \label{eq:lipschitz-absorption}
    \end{equation}
    and
    \begin{equation}
        \mathbb P\left(\sup_{0\leq t\leq T_n}\Lip(e_{n,t})\leq L_\star\right)\longrightarrow1. \label{eq:small-lip}
    \end{equation}
\end{myassumption}

\begin{remark}[Structural requirements for high-dimensional score approximation]
Assumption~\ref{ass:hd-score-approximation} cannot hold uniformly over all high-dimensional models \citep{wibisono2024optimal}. Section~\ref{sec:structured-minimax} gives Gaussian and non-Gaussian examples for which the required score error does vanish in the proportional regime. Characterizing broader classes with this property remains an open problem \citep{cole2024score,gottwald2025localized}.
\end{remark}

\begin{remark}[Equivalence of the two curvature bounds]
The curvature condition in part~(3) of Assumption~\ref{ass:hd-setting} is expressed relative to the covariance matrix \(\Sigma_n\), and therefore differs in appearance from the usual Euclidean two-sided curvature condition
\begin{equation*}
    m_0I_{p_n}\preceq\nabla^2V_n(x)\preceq M_0I_{p_n}.
\end{equation*}
Under the uniform covariance bounds,
\begin{equation*}
    \frac{1}{C_\Sigma}I_{p_n}\preceq\Sigma_n^{-1}\preceq\frac{1}{c_\Sigma}I_{p_n}.
\end{equation*}
Consequently, the bounds in Assumption~\ref{ass:hd-setting}
\begin{equation*}
    \underline m\,\Sigma_n^{-1}\preceq\nabla^2V_n(x)\preceq \overline m\,\Sigma_n^{-1}
\end{equation*}
imply
\begin{equation*}
    \dfrac{\underline m}{C_{\Sigma}}I_{p_n}\preceq \nabla^2 V_n(x)\preceq \dfrac{\overline m}{c_{\Sigma}}I_{p_n}
\end{equation*}
Conversely, the Euclidean bounds imply
\begin{equation*}
    m_0c_\Sigma\,\Sigma_n^{-1}\preceq\nabla^2V_n(x)\preceq M_0C_\Sigma\,\Sigma_n^{-1}.
\end{equation*}
Thus, under the uniform spectral bounds on \(\Sigma_n\), the two formulations are equivalent up to constants. We use the formulation involving $\Sigma_n$ because, after whitening by \(\Sigma_n^{-1/2}\), it yields dimension-free Euclidean curvature bounds directly.
\end{remark}

\begin{remark}[Role of the Gaussian noise assumption]
\label{rem:hd-gaussian-noise}

The Gaussian assumption on the regression noise is imposed mainly for simplicity. In the proportional regime considered here, the increasing dimension comes from the covariate vector \(X\in\mathbb R^{p_n}\), whereas the regression noise remains one-dimensional. Accordingly, the dimension-dependent estimates in the proof are primarily determined by the design distribution.

The argument may extend to centered, smooth, strongly log-concave additive noise with uniformly controlled variance and moments. We do not pursue this extension here.
\end{remark}

These assumptions serve the same purposes as in the fixed-dimensional analysis; see Remark~\ref{rem:score-versus-wasserstein}. Here the corresponding constants and error bounds must hold uniformly in \(n\) because \(p_n/n\to\kappa\in(0,1)\).

We first record the global distributional consequence of score approximation.
\begin{mylemma}
\label{lem:hd-joint-W4}
    If Assumptions~\ref{ass:hd-setting} and~\ref{ass:hd-score-approximation} hold, then the fitted diffusion distribution estimator $\widehat \mu_n$ satisfies $W_4(\widehat \mu_n,\mu_n)\to_p 0$.
\end{mylemma}

Lemma~\ref{lem:hd-joint-W4} follows from a coupling of the exact and learned reverse SDEs and controls global moments of the fitted law. The variance proof also requires lower-tail control of the generated Gram matrix, which joint \(W_4\) convergence alone does not provide. For this sharper conclusion, we compare the corresponding Fokker--Planck equations through an entropy
argument.

We first recall the connection between stochastic differential equations and Fokker--Planck equations. The SDE
\begin{equation*}
    \mathrm dX_t=b_t(X_t)\,\mathrm dt+\sqrt2\,\mathrm dW_t
\end{equation*}
provides a Lagrangian description of the diffusion, whereas the Fokker--Planck equation provides an Eulerian description of its marginal laws. Formally, if $\rho_t$ denotes the density of $X_t$, then It\^o's formula gives
\begin{equation*}
    \partial_t\rho_t=\Delta\rho_t-\nabla\cdot(b_t\rho_t).
\end{equation*}
The converse direction is more delicate when the coefficients and densities are not smooth. Superposition principles show that, under suitable integrability conditions on the coefficients, a weak solution of the Fokker--Planck equation taking values in probability measures can be lifted to a probability measure on path space solving the associated martingale problem. Representation results of this type were established for bounded rough or degenerate coefficients and were subsequently extended to general integrable coefficients \citep{figalli2008existence,trevisan2016well}. The version of the superposition principle used in this paper is recorded in Lemma~\ref{lem:trevisan-superposition}. It provides a unique narrowly continuous representative of the Fokker--Planck solution and a martingale solution on path space whose marginal curve is this representative. If the associated martingale problem is well posed, Lemma~\ref{lem:superposition-uniqueness-transfer} shows that the lifted law on path space is unique and identifies the Fokker--Planck solution with the marginal curve of the corresponding SDE.

In the present setting, fix $n$, condition on the training $\sigma$-field $\mathcal F_n$, and write
\begin{equation*}
    d=d_n,\qquad T=T_n,\qquad\epsilon_{n,r}=e_{n,T-r}.
\end{equation*}
The exact and learned reverse drifts are
\begin{equation*}
    b_{n,r}(z)=z+2s_{n,T-r}(z),\qquad\widehat b_{n,r}(z)=b_{n,r}(z)+2\epsilon_{n,r}(z).
\end{equation*}
The covariance bounds and the score approximation assumption imply that these drifts are globally Lipschitz in the spatial variable and have time-integrable linear growth. Consequently, the corresponding SDEs and martingale problems are well posed.

More precisely, Lemma~\ref{lem:hd-reverse-regularity} shows that there exists an event $\mathcal E_n\in\mathcal F_n$, with $\mathbb P(\mathcal E_n)\to1$, such that, on $\mathcal E_n$, the following properties hold.

\begin{enumerate}
    \item The drift functions $b_{n,r}$ and $\widehat b_{n,r}$ are globally Lipschitz in the spatial variable. Moreover, there exists a nonnegative function $G_n\in L^4(0,T)$ such that, for Lebesgue-a.e. $r\in(0,T)$ and every $z\in\mathbb R^d$,
    \begin{equation*}
        \|b_{n,r}(z)\|+\|\widehat b_{n,r}(z)\|+\|\epsilon_{n,r}(z)\|\leq G_n(r)(1+\|z\|).
    \end{equation*}
    \item The exact and learned reverse SDEs admit unique non-explosive strong solutions and generate time-inhomogeneous Markov evolutions.
    \item Their marginal laws admit densities $\rho_{n,r}$ and $\widehat\rho_{n,r}$, respectively, where $\rho_{n,r}=p_{n,T-r}$.
    \item The densities admit representatives such that
    \begin{equation*}
        \rho_n,\widehat\rho_n\in C\bigl([0,T];L^1(\mathbb R^d)\bigr)\cap L^\infty\bigl((0,T);L^2(\mathbb R^d)\bigr) \cap L^2\bigl((0,T);H^1(\mathbb R^d)\bigr).
    \end{equation*}
    \item For every $0<\tau<R<T$ and every $L<\infty$,
    \begin{equation*}
        \partial_r\rho_n,\,\partial_r\widehat\rho_n\in L^2\bigl((\tau,R);H^{-1}(B_L)\bigr).
    \end{equation*}
    \item Since the exact OU density is smooth and strictly positive for $0<r<T$, the density ratio
    \begin{equation*}
        f_{n,r}=\frac{\widehat\rho_{n,r}}{\rho_{n,r}}
    \end{equation*}
    is well defined and satisfies
    \begin{equation*}
        f_n\in L_{\mathrm{loc}}^2\bigl((0,T);H_{\mathrm{loc}}^1(\mathbb R^d)\bigr).
    \end{equation*}
\end{enumerate}

The comparison requires enough regularity to justify testing the two weak Fokker--Planck equations. Smooth approximation and an energy estimate give
\begin{equation*}
    \frac{\mathrm d}{\mathrm dr}\|q_r^{(k)}\|_2^2 + 2\|\nabla q_r^{(k)}\|_2^2 \leq \Lambda_k(r)\|q_r^{(k)}\|_2^2, \qquad \sup_k\int_0^T\Lambda_k(r)\,\mathrm dr<\infty.
\end{equation*}
Together with the linear-growth moment bound, this estimate gives weak compactness. The superposition principle and uniqueness identify the limits with the SDE marginals, while the mild representation and positivity of the exact OU density give the stated time and density-ratio regularity. The full argument appears in Lemma~\ref{lem:hd-reverse-regularity}.

With these regularity properties in hand, the exact and learned Fokker--Planck equations can be written as
\begin{equation*}
    \begin{cases}
        \partial_r\rho_{n,r} = \Delta\rho_{n,r} - \nabla\cdot(b_{n,r}\rho_{n,r}),\\ \partial_r\widehat\rho_{n,r}=\Delta\widehat\rho_{n,r}-\nabla\cdot(\widehat b_{n,r}\widehat\rho_{n,r}).
    \end{cases}
\end{equation*}
Define
\begin{equation*}
    f_{n,r}=\frac{\widehat\rho_{n,r}}{\rho_{n,r}},\qquad D_n(r)=\int_{\mathbb R^d}f_{n,r}^2\rho_{n,r}\,\mathrm dz,
\end{equation*}
and
\begin{equation*}
    H_n(r)=D_n(r)-1=\chi^2\left(\widehat\rho_{n,r}(z)\,\mathrm dz\,\middle\|\,\rho_{n,r}(z)\,\mathrm dz\right).
\end{equation*}
The regularity above permits a rigorous density ratio calculation after introducing smooth truncations of \(f_{n,r}\) and spatial cutoffs. Passing first to the whole space and then removing the truncation gives the absolutely continuous identity
\begin{equation*}
    \begin{aligned}
        D_n'(r)&=-2\int_{\mathbb R^d}\|\nabla f_{n,r}\|^2\rho_{n,r}\,\mathrm dz+4\int_{\mathbb R^d}f_{n,r}\epsilon_{n,r}^\top\nabla f_{n,r}\rho_{n,r}\,\mathrm dz.
    \end{aligned}
\end{equation*}

The Gibbs variational principle and the uniform log-Sobolev inequality control the mixed term by the Fisher information and the fourth moment score error. After absorption and the Poincar\'e inequality, one obtains, for some \(c_0>0\),
\begin{equation*}
    H_n'(r)\leq-c_0H_n(r)+C\alpha_n(r)\{1+H_n(r)\}.
\end{equation*}
Here \(\alpha_n(r)= \{\int\|e_{n,T-r}(z)\|^4p_{n,T-r}(z)\,\mathrm dz\}^{1/2}\). Assumption~\ref{ass:hd-score-approximation}, the initial \(\chi^2\) bound, and a stopping argument then give
\begin{equation*}
    \sup_{0\leq r<T_n} H_n(r)=o_p(1).
\end{equation*}
Lower semicontinuity of \(\chi^2\) extends the same conclusion to the terminal generated and target laws. The truncation argument and endpoint limits are given in Lemmas~\ref{lem:hd-reverse-regularity} and~\ref{lem:hd-general-pde}.

The terminal joint $\chi^2$ control has two consequences. First, after whitening by $\Sigma_n^{-1/2}$, one-dimensional small-ball estimates for isotropic log-concave projections and a net argument give polynomial lower-tail bounds for the true Gram matrix. A change-of-measure argument then transfers these bounds to the generated Gram matrices. For
\begin{equation*}
    \widehat\Sigma_m^*=\frac1m\sum_{i=1}^mX_i^*X_i^{*\top},\qquad m\in\{n-1,n\},
\end{equation*}
and every fixed $a>0$, there exist a deterministic constant $\eta_a>0$ and a random sequence $C_{n,a}=O_p(1)$ such that
\begin{equation*}
    \sup_{0<t<\eta_a}t^{-a}\mathbb P^*\left(\lambda_{\min}(\widehat\Sigma_m^*)\leq t\,\middle|\,\widehat\mu_n\right)\leq C_{n,a}.
\end{equation*}
In particular, for every fixed $Q>0$,
\begin{equation*}
    \mathbb E^*\left[\lambda_{\min}^{-Q}(\widehat\Sigma_m^*)\,\middle|\,\widehat\mu_n \right]=O_p(1).
\end{equation*}

\begin{remark}
A lower-tail bound of this form is related to the high-dimensional result of \citep{mourtada2022exact}, which is established under a condition \(p/n<c\) for some constant \(c<1\). In the strongly log-concave generated-law setting considered here, the corresponding lower-tail control holds for every fixed proportional limit \(p/n\to\kappa\in(0,1)\).
\end{remark}

Second, the same joint $\chi^2$ control yields conditional Wasserstein consistency.
\begin{mylemma}[Conditional $W_4$ convergence]
\label{lem:hd-cond-W4}
Suppose Assumptions~\ref{ass:hd-setting} and ~\ref{ass:hd-score-approximation} hold. Let $\widehat K_n(x)=\widehat\mu_n(Y\in\cdot\mid X=x)$ and $K_n(x)=\mu_n(Y\in\cdot\mid X=x)$ denote the regular conditional laws of $Y$ given $X=x$ under $\widehat\mu_n$ and $\mu_n$, respectively. Then
\begin{equation*}
    K_n(x)=N(x^\top\beta_n,\sigma_n^2)
\end{equation*}
and
\begin{equation*}
    \mathbb E_{X\sim\widehat\mu_{n,X}}\left[W_4^4\left(\widehat K_n(X),K_n(X)\right)\right]\to_p 0.
\end{equation*}
\end{mylemma}
We can now state the high-dimensional variance consistency theorem.
\begin{mytheorem}
\label{thm:hd-actual-var}
Suppose Assumptions~\ref{ass:hd-setting}, ~\ref{ass:hd-contrast}, and ~\ref{ass:hd-score-approximation} hold. Conditional on the observed data, generate
\begin{equation*}
    Z_1^*,\dots,Z_n^*\overset{\mathrm{iid}}{\sim}\widehat\mu_n,\qquad Z_i^*=(X_i^*,Y_i^*),
\end{equation*}
and define the diffusion pairs bootstrap OLS estimator
\begin{equation*}
    \widehat\beta^*=\left(\sum_{i=1}^nX_i^*X_i^{*\top}\right)^{-1}\left(\sum_{i=1}^nX_i^*Y_i^*\right).
\end{equation*}
Let $\widehat\beta$ denote the OLS estimator computed from the original sample. Then
\begin{equation}
    \frac{\Var^*\left(c_n^\top\widehat\beta^*\,\middle|\,\widehat\mu_n\right)}{\Var\left(c_n^\top\widehat\beta\right)}\to_p 1. \label{eq:variance-ratio}
\end{equation}
\end{mytheorem}

We give a brief proof outline. The law of total variance separates the conditional response variance from the variation of the conditional mean. Let
\begin{equation*}
    r_n(x)=\mathbb E_{\widehat K_n(x)}[Y]-x^\top\beta_n,\qquad v_n^*(x)=\Var_{\widehat K_n(x)}(Y).
\end{equation*}
Lemma~\ref{lem:hd-cond-W4} implies
\begin{equation*}
    \mathbb E_{\widehat\mu_{n,X}}|r_n(X)|^4\to_p0,\qquad \mathbb E_{\widehat\mu_{n,X}}|v_n^*(X)-\sigma_n^2|^2\to_p0.
\end{equation*}
Thus the generated conditional mean and variance approach those of the true linear model.

Write \(S_n^*=\sum_{i=1}^nX_i^*X_i^{*\top}\). The convergence of \(v_n^*\), together with the inverse Gram moment bounds, shows that the first term in the variance decomposition equals
\begin{equation*}
    \sigma_n^2\mathbb E^*\left[c_n^\top(S_n^*)^{-1}c_n\,\middle|\,\widehat\mu_n\right]+ o_p(n^{-1}).
\end{equation*}
A rowwise \(W_4\) coupling, the resolvent identity, and the same inverse moment bounds then give
\begin{equation*}
    \mathbb E^*\left[c_n^\top(S_n^*)^{-1}c_n\,\middle|\,\widehat\mu_n\right]=\mathbb E\left[c_n^\top S_n^{-1}c_n\right]+o_p(n^{-1}).
\end{equation*}
For the conditional mean term, the convergence of \(r_n\), a conditional Efron--Stein inequality, and the Sherman--Morrison formula show that its contribution is \(o_p(n^{-1})\). Consequently,
\begin{equation*}
    \Var^*\left(c_n^\top\widehat\beta^*\,\middle|\,\widehat\mu_n\right)=\sigma_n^2\mathbb E\left[c_n^\top S_n^{-1}c_n\right]+o_p(n^{-1}).
\end{equation*}
For the original sample, independence and homoskedasticity give
\begin{equation*}
    \Var\left(c_n^\top\widehat\beta\right)=\sigma_n^2\mathbb E\left[c_n^\top S_n^{-1}c_n\right].
\end{equation*}
This quantity is of order \(n^{-1}\), so the variance ratio converges to one.

\section{Why diffusion estimators can outperform empirical measures}
\label{sec:structured-minimax}
The high-dimensional variance theorem is conditional on the score approximation assumption. We now show, first, that this assumption is attainable for several structured models and, second, that the empirical distribution faces a separate geometric obstruction in the same regime.

Assumption~\ref{ass:hd-score-approximation} requires the integrated fourth-moment score error to vanish. Such a conclusion cannot hold uniformly over an unrestricted class of high-dimensional densities \citep{wibisono2024optimal}. The purpose of this section is therefore not to claim that diffusion estimation is universally easy, but to show that the score condition is statistically attainable for several structured distribution classes in the proportional regime.

The empirical distribution provides the natural comparison, but it faces a geometric limitation in high dimensions. When $p/n\to\kappa\in(0,1)$, it is supported on only $n$ observed points and is not even $W_4$-consistent for a standard Gaussian target. More generally, Proposition~\ref{prop:minimax-atomic-w4-gaussian} shows that this failure applies to every distribution estimator supported on at most $n$ points, regardless of how its support points and weights are chosen. A diffusion estimator is not restricted to a measure supported on finitely many points. It instead estimates the score along the OU flow and generates a distribution with a density. For the structured models considered below, this score can be estimated consistently even though the empirical distribution cannot consistently approximate the target law in $W_4$.

Throughout this section,
\begin{equation*}
    Y=X^\top\beta+\varepsilon,\qquad\varepsilon\sim N(0,\sigma^2), \qquad\varepsilon\perp X,
\end{equation*}
where
\begin{equation*}
    \|\beta\|_1\le B_1,\qquad\sigma^2\in[\sigma_{\min}^2,\sigma_{\max}^2].
\end{equation*}
The $\ell_1$ bound is stronger than the $\ell_2$ bound used in Assumption~\ref{ass:hd-setting}. It is imposed here only to exhibit structured subclasses on which the score approximation requirement can be verified. Although the design distributions below are described by a fixed number of parameters, the joint score also depends on the $p$-dimensional regression vector $\beta$; the $\ell_1$ bound controls this part of the estimation error.

For a class $\mathcal X_p$ of design distributions, let $\mathcal M_p(\mathcal X_p)$ be the corresponding class of joint laws of $Z=(X,Y)$. For $\mu\in\mathcal M_p(\mathcal X_p)$, let $p_{\mu,t}$ and $s_{\mu,t}=\nabla\log p_{\mu,t}$ denote the density and score along the standard OU flow. Define
\begin{equation}
    \begin{aligned}
        \mathcal R_{n,4}(\mathcal X_p;T_n):=\inf_{\widehat s}\sup_{\mu\in\mathcal M_p(\mathcal X_p)}\mathbb E_\mu\left[\int_0^{T_n}\mathbb E_{Z_t\sim p_{\mu,t}}\|\widehat s_t(Z_t)-s_{\mu,t}(Z_t)\|_2^4\,\mathrm dt\right], \label{eq:structured-minimax-risk}
    \end{aligned}
\end{equation}
where the infimum is over all score estimators based on $n$ independent observations from $\mu$, and the inner expectation is evaluated on an independent OU trajectory. The results below are uniform upper bounds for this risk. They show the existence of score estimators satisfying the statistical part of Assumption~\ref{ass:hd-score-approximation}; they do not assert that an arbitrary neural network architecture or training algorithm necessarily attains these bounds. Without matching lower bounds, we also do not claim that the displayed rates are minimax optimal.
\subsection{Gaussian and other structured distributions}
We begin with the simplest structured class, $\mathcal X_p^{\mathrm I}=\{N(0,I_p)\}$. Although the joint score is $(p+1)$-dimensional, its covariance is determined by the regression vector and the noise variance. For this class, abbreviate
\begin{equation*}
    \mathcal R_{n,4}(T_n):=\mathcal R_{n,4}(\mathcal X_p^{\mathrm I};T_n).
\end{equation*}
\begin{mytheorem}
\label{thm:gaussian-linear-score-upper-unknown-var}
Suppose $p/n\to\kappa\in(0,1)$. There exists a constant $C<\infty$, depending only on $B_1$, $\sigma_{\min}$, and $\sigma_{\max}$, such that, for every deterministic sequence $T_n\in(0,\infty]$,
\begin{equation*}
    \mathcal R_{n,4}(T_n)\leq C\frac{\log p}{n}.
\end{equation*}
Moreover, the plug-in score estimator constructed in the proof satisfies, for every $\eta>0$,
\begin{equation*}
    \sup_{\mu\in\mathcal M_p(\mathcal X_p^{\mathrm I})}\mathbb P_\mu\left\{\int_0^{T_n}\mathbb E_{Z_t\sim p_{\mu,t}}\|\widehat s_t(Z_t)-s_{\mu,t}(Z_t)\|_2^4\,\mathrm dt>\eta\right\}\to0,
\end{equation*}
and
\begin{equation*}
    \sup_{0\leq t\leq T_n}\Lip(\widehat s_t-s_{\mu,t})\to_p0.
\end{equation*}
\end{mytheorem}

In this benchmark, the rate is governed by estimation of the unknown \(\ell_1\)-bounded regression vector. If $\beta$ and $\sigma^2$ were known, the score would be known exactly.

We next consider three additional design classes.

\paragraph{Gaussian AR(1).}
Fix $0<\rho_0<1$ and let
\begin{equation*}
    \mathcal X_p^{\mathrm{AR}}=\left\{N(0,\Sigma_p(\rho)): \Sigma_p(\rho)_{jk}=\rho^{|j-k|},\ |\rho|\leq\rho_0\right\}.
\end{equation*}

\paragraph{Fixed rank perturbations with known directions.}
Fix $r<\infty$ and a deterministic matrix $U_p=(u_{1,p},\ldots,u_{r,p})\in\mathbb R^{p\times r}$ with $U_p^\top U_p=I_r$. Let
\begin{equation*}
    \Sigma_p(\lambda)=I_p+U_p\diag(\lambda_1,\ldots,\lambda_r)U_p^\top,
\end{equation*}
where $\lambda\in\Lambda\subset(-1,\infty)^r$ and $\Lambda$ is compact and convex. Assume that
\begin{equation}
    \max_{1\leq j\leq p}\sum_{\ell=1}^r|u_{\ell,p,j}|\,\|u_{\ell,p}\|_1\leq C_U. \label{eq:low-rank-row-sum}
\end{equation}
Let $\mathcal X_p^{\mathrm{LR}}$ be the resulting Gaussian design class. This is a full-rank covariance model with a fixed-rank perturbation, not a rank-deficient Gaussian distribution.

\begin{remark}[Scope of the fixed-rank design class]
The directions $U_p$ in $\mathcal X_p^{\mathrm{LR}}$ are known, or are specified by a fixed-dimensional parametrization satisfying the same bounds. An arbitrary unknown dense $p\times r$ loading matrix contains order $pr$ unknown parameters and is not covered merely because $r$ is fixed. A rank-deficient covariance has no density on $\mathbb R^p$ and is also outside the strong log-concavity and Fokker--Planck framework used in this paper.
\end{remark}

\paragraph{A non-Gaussian product exponential family.}
Let $q<\infty$ be fixed and let $\Theta\subset\mathbb R^q$ be compact and convex. For $\theta\in\Theta$, define
\begin{equation}
    f_{\theta,p}(x)=\prod_{j=1}^p \frac{1}{Z_\theta}\exp\left\{-\frac{x_j^2}{2}-\sum_{\ell=1}^q\theta_\ell\phi_\ell(x_j)\right\}. \label{eq:product-exponential-family}
\end{equation}
Assume that the functions $\phi_1,\ldots,\phi_q$ are even and belong to $C^4(\mathbb R)$, their derivatives of orders one through four are uniformly bounded, and
\begin{equation}
    0<m\leq1+\sum_{\ell=1}^q\theta_\ell\phi_\ell''(u)\leq L<\infty \label{eq:product-exponential-curvature}
\end{equation}
uniformly over $u\in\mathbb R$ and $\theta\in\Theta$. Finally, assume uniform identifiability:
\begin{equation}
    cI_q\preceq\Cov_\theta\bigl\{(\phi_1(X_1),\ldots,\phi_q(X_1))^\top\bigr\}\preceq CI_q. \label{eq:product-exponential-information}
\end{equation}
Let $\mathcal X_p^{\mathrm{EXP}}$ denote this design class. To quantify the size of the non-Gaussian perturbation, define
\begin{equation}
    \delta_{\mathrm{NG}}:=\sup_{\theta\in\Theta}\sup_{u\in\mathbb R}\left|\sum_{\ell=1}^q\theta_\ell\phi_\ell''(u)\right|. \label{eq:product-nongaussian-curvature}
\end{equation}
No smallness condition on $\delta_{\mathrm{NG}}$ is needed for the bound on the integrated score estimation risk. A sufficiently small fixed value is used only to verify the Lipschitz threshold in Assumption~\ref{ass:hd-score-approximation}.

\begin{mytheorem}[Upper bounds for score estimation under structured designs]
\label{thm:structured-score-upper}
Suppose $p/n\to\kappa\in(0,1)$. For each
\begin{equation*}
    \mathcal X_p\in\left\{\mathcal X_p^{\mathrm{AR}},\mathcal X_p^{\mathrm{LR}},\mathcal X_p^{\mathrm{EXP}}\right\},
\end{equation*}
there exists a constant $C<\infty$, independent of $n$, $p$, and $T_n$, such that
\begin{equation}
    \mathcal R_{n,4}(\mathcal X_p;T_n)\leq C\left(\frac{\log p}{n}+\frac1{n^2}\right)\leq C\frac{\log p}{n}. \label{eq:structured-score-upper}
\end{equation}
Moreover, for each of the three classes, the estimator constructed in the proof satisfies, for every $\eta>0$,
\begin{equation*}
    \sup_{\mu\in\mathcal M_p(\mathcal X_p)}\mathbb P_\mu\left\{\int_0^{T_n}\mathbb E_{Z_t\sim p_{\mu,t}}\|\widehat s_t(Z_t)-s_{\mu,t}(Z_t)\|_2^4\,\mathrm dt>\eta\right\}\to0.
\end{equation*}

For the two Gaussian classes, the plug-in estimator additionally satisfies
\begin{equation*}
    \sup_{0\leq t\leq T_n}\Lip(\widehat s_t-s_{\mu,t})\to_p0.
\end{equation*}
For $\mathcal X_p^{\mathrm{EXP}}$, there are constants $\delta_0>0$ and $C_{\mathrm{Lip}}<\infty$, independent of $n$ and $p$, such that, whenever $\delta_{\mathrm{NG}}\leq\delta_0$,
\begin{equation}
    \sup_{0\leq t\leq T_n}\Lip(\widehat s_t-s_{\mu,t})\leq 2C_{\mathrm{Lip}}\delta_{\mathrm{NG}}+o_p(1). \label{eq:product-small-lipschitz}
\end{equation}
Consequently, if
\begin{equation*}
    \delta_{\mathrm{NG}}<\min\left\{\delta_0,\frac{L_\star}{4C_{\mathrm{Lip}}}\right\},
\end{equation*}
then the plug-in estimator satisfies the Lipschitz condition
\eqref{eq:small-lip}.
\end{mytheorem}

For all three classes, estimation of the design and noise parameters contributes $O(n^{-2})$ to the fourth-moment score risk. The larger term $\log(p)/n$ comes from estimation of the unknown regression vector $\beta$. The estimators and the corresponding moment calculations are given in the proof.

The non-Gaussian class includes, for example,
\begin{equation*}
    \phi(u)=\log\cosh(u),\qquad\theta\in[0,M],
\end{equation*}
and
\begin{equation*}
    \phi(u)=1-\cos(u),\qquad|\theta|\leq\theta_0<1.
\end{equation*}
Both examples define non-Gaussian distributions on $\mathbb R^p$. Their scores remain $p$-dimensional, but the same finite-dimensional parameter determines their coordinatewise form. The log-cosh family has $\delta_{\mathrm{NG}}\leq M$, and the periodic family has $\delta_{\mathrm{NG}}\leq\theta_0$. Hence each family satisfies the Lipschitz condition in Assumption~\ref{ass:hd-score-approximation} when its fixed amplitude is sufficiently small. The amplitude need not vanish with $n$, so these models remain non-Gaussian throughout the asymptotic sequence.

\subsection{Atomic approximation and the empirical distribution}

The preceding results show that the required score can be estimated for the structured classes above. We next contrast this conclusion with the geometric limitation of empirical resampling. A probability measure is called \(n\)-atomic if its support contains at most \(n\) points. The empirical distribution is $n$-atomic and is the distribution estimator used by the ordinary nonparametric bootstrap. The same description applies to bootstrap variants that only resample or reweight the observed data points. It does not apply to smoothed, parametric, or generative procedures that can assign probability to new points.

Let $\gamma_d$ be the standard Gaussian law on $\mathbb R^d$, and let $\mathcal A_{n,d}$ be the class of all possibly randomized probability measures $\widehat Q_n$ whose support contains at most $n$ points almost surely. The randomness may depend on $n$ observations from $\gamma_d$ and on additional randomization.

\begin{myproposition}[Optimal atomic approximation of a Gaussian]
\label{prop:minimax-atomic-w4-gaussian}
There exist universal constants $c,C>0$ such that, whenever $n\leq\exp(cd)$,
\begin{equation}
    cd^2\leq\inf_{\widehat Q_n\in\mathcal A_{n,d}}\mathbb E_{\gamma_d}W_4^4(\widehat Q_n,\gamma_d)\leq Cd^2. \label{eq:atomic-gaussian-rate}
\end{equation}
Consequently, if $d=d_n\asymp n$, no sequence of $n$-atomic distribution estimators is $W_4$-consistent for $\gamma_{d_n}$.
\end{myproposition}

Unless the number of support points is exponential in the dimension, balls of radius $a\sqrt d$ centered at those points cannot cover a nonnegligible part of a Gaussian shell. The upper bound shows that the order $d^2$ is sharp for the fourth power transportation loss. Thus, the failure is not specific to empirical weights or to a particular choice of support points.

The score estimation results and Proposition~\ref{prop:minimax-atomic-w4-gaussian} concern different problems. Consider the standard Gaussian linear model with $\beta=0$ and $\sigma^2=1$, for which the joint law is $\gamma_{p+1}$. When $p/n\to\kappa\in(0,1)$, the proposition gives an atomic $W_4^4$ risk of order $p^2$, whereas Theorem~\ref{thm:gaussian-linear-score-upper-unknown-var} gives an integrated fourth-moment score risk of order at most $\log(p)/n=o(1)$. Thus, for the structured models considered above, the OU score can be estimated consistently even though approximation by a distribution supported on at most $n$ points is not $W_4$-consistent. This comparison explains how a diffusion distribution estimator can have an advantage over the empirical distribution in high dimensions. This conclusion is limited to the structured classes above and does not apply to every fitted diffusion model.

\section{Why terminal \texorpdfstring{$W_4$}{W4} consistency is insufficient}
\label{subsec:w4-counterexamples}
The results established so far may suggest that the diffusion estimator is simply a distribution estimator that converges in \(W_4\). Under that interpretation, it could be replaced by any Wasserstein consistent estimator. The examples below show why this is not enough.

Bootstrap variance depends not only on the proximity between the fitted law and the target law, but also on the geometric properties of the generated design. In particular, the variance of the bootstrap ordinary least-squares estimator contains inverse moments of the generated Gram matrix. These quantities are highly sensitive to rare events on which the generated Gram matrix is nearly singular, whereas the Wasserstein distance only measures an average transportation cost between probability distributions. Consequently, terminal Wasserstein convergence alone does not rule out the rare geometric degeneracies that are responsible for bootstrap variance failure.

The following examples demonstrate this difference.
\begin{example}[$W_4$ convergence but divergence of the bootstrap variance:
fixed-dimensional case]
    Consider the simple model $X\sim N(0,1)$, $Y=\varepsilon$, and $\varepsilon\sim N(0,1)$, with $X\perp Y$. Then the joint law is $\mu=N(0,I_2)$, which is the standard two-dimensional Gaussian distribution and is strongly log-concave. Define the distribution estimator $\widehat\mu_n=(1-\delta_n)\mu+\delta_n\nu_n$, where $\nu_n=N(0,a_n^2)\otimes N(0,1)$. Let $\delta_n=n^{-2}$ and $a_n^2=\delta_n^{2n}$. We will prove that
    \begin{itemize}
        \item[(1)] $W_4(\widehat\mu_n,\mu)\to0$;
        \item[(2)] $\dfrac{\Var(\widehat\beta^*)}{\Var(\widehat\beta)}\to\infty$.
    \end{itemize}

    In fact, couple $(a_nG,\eta)$ with $(G,\eta)$, where $G$ and $\eta$ are independent standard Gaussian random variables. Then
    \begin{equation*}
        W_4^4(\widehat\mu_n,\mu)\leq\delta_n(1-a_n)^4\mathbb E|G|^4\leq3\delta_n\to0.
    \end{equation*}

    Now generate $n$ bootstrap samples $(X_i^*,Y_i^*)$ from $\widehat\mu_n$ and compute the OLS estimator
    \begin{equation*}
        \widehat\beta^*=\frac{\sum_{i=1}^nX_i^*Y_i^*}{\sum_{i=1}^n(X_i^*)^2}.
    \end{equation*}
    Since $Y_i^*$ is independent of $X_i^*$ and has variance one under both mixture components,
    \begin{equation*}
        \Var(\widehat\beta^*)=\mathbb E\left[\frac{1}{\sum_{i=1}^n(X_i^*)^2}\right].
    \end{equation*}
    Let $E_n$ denote the event that all $n$ bootstrap samples are drawn from the contamination component. Then $\mathbb P(E_n)=\delta_n^n$. On $E_n$, we have
    \begin{equation*}
        X_i^*=a_nG_i \quad\Longrightarrow\quad \sum_{i=1}^n(X_i^*)^2 = a_n^2\chi_n^2.
    \end{equation*}
    Therefore,
    \begin{equation*}
        \mathbb E\left[\frac{1}{\sum_i(X_i^*)^2}\,\middle|\,E_n\right]=\frac{1}{a_n^2(n-2)},
    \end{equation*}
    and hence
    \begin{equation*}
        \Var(\widehat\beta^*)\geq\frac{\delta_n^n}{a_n^2(n-2)}.
    \end{equation*}

    Under the true model,
    \begin{equation*}
        \Var(\widehat\beta)=\frac{1}{n-2}.
    \end{equation*}
    Thus,
    \begin{equation*}
        \frac{\Var(\widehat\beta^*)}{\Var(\widehat\beta)}\geq\frac{\delta_n^n}{a_n^2}=\delta_n^{-n}=n^{2n}\to\infty.
    \end{equation*}
\end{example}
The same mechanism persists in the proportional high-dimensional regime.
\begin{example}[$W_4$ convergence but divergence of the bootstrap variance: high-dimensional case]
    Let $p=p_n$ and suppose $p_n/n\to\kappa\in(0,1)$. Consider the model $X\sim N(0,I_{p_n})$, $Y=\varepsilon$, and $\varepsilon\sim N(0,1)$, with $X\perp Y$. Then $\mu_n=N(0,I_{p_n+1})$ is a standard Gaussian distribution and is strongly log-concave. Define $\widehat\mu_n=(1-\delta_n)\mu_n+\delta_n\nu_n$, where $\nu_n=N(0,a_n^2I_{p_n})\otimes N(0,1)$. Let $\delta_n=n^{-3}$ and $a_n^2=\delta_n^{2n}$. We will prove that
    \begin{itemize}
        \item[(1)] $W_4(\widehat\mu_n,\mu_n)\to0$;
        \item[(2)] $\dfrac{\Var(c_n^\top\widehat\beta^*)}{\Var(c_n^\top\widehat\beta)}\to\infty$.
    \end{itemize}

    Take the same coupling as in the fixed-dimensional case. Then
    \begin{equation*}
        W_4^4(\widehat\mu_n,\mu_n)\leq\delta_n(1-a_n)^4\mathbb E\|G\|^4.
    \end{equation*}
    Since $\mathbb E\|G\|^4=p_n(p_n+2)$, we have
    \begin{equation*}
        W_4^4(\widehat\mu_n,\mu_n)\leq\delta_np_n(p_n+2).
    \end{equation*}
    Since $p_n\asymp n$ and $\delta_n=n^{-3}$,
    \begin{equation*}
        \delta_np_n(p_n+2)=O(n^{-1})\to0.
    \end{equation*}
    Therefore,
    \begin{equation*}
        W_4(\widehat\mu_n,\mu_n)\to0.
    \end{equation*}

    Now generate $n$ bootstrap observations from $\widehat\mu_n$, let $X^*$ denote the resulting design matrix, and define
    \begin{equation*}
        \widehat\beta^*=(X^{*\top}X^*)^{-1}X^{*\top}Y^*.
    \end{equation*}
    For any deterministic unit vector $c_n$, since $Y^*$ is independent of $X^*$ and has variance one, we have
    \begin{equation*}
        \Var(c_n^\top\widehat\beta^*)=\mathbb E\left[c_n^\top(X^{*\top}X^*)^{-1}c_n\right].
    \end{equation*}

    Let $E_n$ denote the event that all bootstrap observations are drawn from the contamination component. Then $\mathbb P(E_n)=\delta_n^n$. On $E_n$, we have
    \begin{equation*}
        X^*=a_nG_n \quad\Longrightarrow\quad X^{*\top}X^*=a_n^2G_n^\top G_n,
    \end{equation*}
    where $G_n\in\mathbb R^{n\times p_n}$ has i.i.d. $N(0,1)$ entries. Hence
    \begin{equation*}
        \mathbb E\left[c_n^\top(G_n^\top G_n)^{-1}c_n\right]=\frac{1}{n-p_n-1}.
    \end{equation*}
    It follows that
    \begin{equation*}
        \Var(c_n^\top\widehat\beta^*)\geq\frac{\delta_n^n}{a_n^2(n-p_n-1)}.
    \end{equation*}
    Under the true model,
    \begin{equation*}
        \Var(c_n^\top\widehat\beta)=\frac{1}{n-p_n-1}.
    \end{equation*}
    Therefore,
    \begin{equation*}
        \frac{\Var(c_n^\top\widehat\beta^*)}{\Var(c_n^\top\widehat\beta)}\geq\frac{\delta_n^n}{a_n^2}=\delta_n^{-n}=n^{3n}\to\infty.
    \end{equation*}
\end{example}

The preceding counterexamples share a common mechanism. The fitted law contains a contamination component with vanishing probability mass, so its contribution to the Wasserstein distance disappears. The contaminating covariate distribution, however, is nearly degenerate. A bootstrap sample drawn from this component may therefore have a nearly singular Gram matrix, and the resulting inverse moments of the Gram matrix can dominate the bootstrap variance despite the rarity of the event.

This is why our score approximation condition is used for more than proving terminal \(W_4\) convergence. In the fixed-dimensional regime, the Fokker--Planck equation and Moser iteration yield density and lower-tail estimates for the generated Gram matrix. In the high-dimensional regime, analogous control follows from entropy estimates for the learned Fokker--Planck equation. These estimates provide the inverse-moment bounds needed for bootstrap variance consistency.

\section{Why not the residual bootstrap}
\label{sec:residual-bootstrap}
The preceding analysis shows why learning the joint law can repair the geometric failure of pairs resampling. A natural alternative is to keep the design fixed and fit a diffusion model only to the empirical residuals. The obstacle is that, in the proportional regime, fitted residuals generally do not have the same distribution as the regression errors. To see the mechanism, let $\widehat\beta_{(i)}$ be the estimator computed without observation $i$ and write
\begin{equation*}
    \widetilde e_{j(i)}=Y_j-X_j^\top\widehat\beta_{(i)}.
\end{equation*}
Suppose that the design is elliptical,
\begin{equation*}
    X_i=\lambda_i\Gamma_i,\qquad\Gamma_i\sim N(0,I_p),\qquad\mathbb E\lambda_i^2=1,
\end{equation*}
where $\lambda_i$ is independent of $\Gamma_i$. For an $M$-estimator with loss $\rho$ and score $\psi=\rho'$, assume in addition that the regression errors are independent of the design. The residual and leave-one-out expansions take the form \citep{el2018can}
\begin{equation*}
    \begin{aligned}
        \widetilde e_{i(i)}&=\varepsilon_i+|\lambda_i|\,\|\widehat\beta_{(i)}-\beta\|_2 Z_i+o_p(u_n),\\e_i+c_i\lambda_i^2\psi(e_i)&=\widetilde e_{i(i)}+o_p(u_n),
    \end{aligned}
\end{equation*}
where $Z_i\sim N(0,1)$ is independent of $\varepsilon_i$, $u_n\to0$, and
\begin{equation*}
    c_i=\frac1n\tr\left\{\left(\frac1n\sum_{j\neq i}\psi'(\widetilde e_{j(i)})X_jX_j^\top\right)^{-1}\right\}.
\end{equation*}
For least squares, $\psi(x)=x$ and the exact leave-one-out identity is
\begin{equation*}
    e_i=(1-h_i)\widetilde e_{i(i)},\qquad h_i = X_i^\top(X^\top X)^{-1}X_i.
\end{equation*}
Under homoskedastic errors and conditional on a full-rank design matrix,
\begin{equation*}
    \mathbb E\left[ \sum_{i=1}^n e_i^2 \,\middle|\, X \right] = \sigma_\varepsilon^2(n-p).
\end{equation*}
Thus, the expected residual sum of squares per observation is
$\sigma_\varepsilon^2(1-p/n)$. In the proportional regime, a generative model
that learns the fitted-residual law therefore targets a variance-shrunk
distribution rather than the true error law. Improving the residual
distribution estimator alone cannot remove this high-dimensional distortion.

\section{Simulation}
\label{sec:simulation}
\subsection{Experimental setting}
We evaluate the finite-sample performance of the proposed generative bootstrap procedures in proportional high-dimensional linear models. The simulation design follows the standard high-dimensional linear-model setting used in the bootstrap literature, where $p<n$ but $p/n$ is not close to zero. We consider the linear model
\begin{equation*}
    Y_i=X_i^\top\beta+\varepsilon_i,\qquad i=1,\dots,n.
\end{equation*}
Throughout the simulations, we set $\beta=0$. For unpenalized least squares, the estimation error obeys the exact identity
\begin{equation*}
    \widehat\beta-\beta=(X^\top X)^{-1}X^\top\varepsilon.
\end{equation*}
Thus, conditional on the design, the sampling variance does not depend on the value of $\beta$ \citep{el2018can}. Setting $\beta=0$ isolates the variance calibration question and makes the null hypothesis for the first coordinate exact, so the empirical rejection probability directly measures Type~I error.

We focus on inference for the first coordinate of $\beta$. Equivalently, the contrast vector is $c=e_1$, the first canonical basis vector. Even when full high-dimensional distributional approximation is unavailable, inference for a fixed coordinate or deterministic contrast may still be accurate.

Unless otherwise stated, the sample size is $n=500$. For each setting and each value of $\kappa$, we use $R=1000$ Monte Carlo replications. Within each replication, we use $B=1000$ bootstrap samples to estimate the bootstrap distribution and the bootstrap variance. We vary the aspect ratio over
\begin{equation*}
    \kappa=p/n\in\{0.1,0.2,0.3,0.4,0.5\},
\end{equation*}
and take $p=\operatorname{round}(n\kappa)$.

We consider ten settings obtained by combining five design distributions with two error distributions. The five design distributions are Normal, Laplace, ENL, EUL, and EEL. The Normal design is the Gaussian benchmark:
\begin{equation*}
    X_i\sim N(0,I_p).
\end{equation*}
The Laplace design has independent entries $X_{ij}\sim \operatorname{Laplace}(0,1/\sqrt2)$, so that each coordinate has variance one. The remaining three designs are elliptical scale-mixture designs of the form
\begin{equation*}
    X_i=\lambda_i Z_i,\qquad Z_i\sim N(0,I_p),
\end{equation*}
where $Z_i$ is independent of the scalar radial variable $\lambda_i$, and the scaling is chosen so that $\mathbb E[\lambda_i^2]=1$. Specifically, ENL uses $\lambda_i\sim N(0,1)$, EUL uses $\lambda_i=\sqrt{12/13}\,U_i$ with $U_i\sim \Unif(0.5,1.5)$, and EEL uses $\lambda_i\sim \Exp(\sqrt2)$. These elliptical designs preserve the marginal covariance $\mathbb E[X_iX_i^\top]=I_p$, but introduce heterogeneous row norms and leverage behavior, providing a test beyond the i.i.d. coordinate setting.

For the errors, we consider two distributions. The first is the Gaussian error $\varepsilon_i\sim N(0,1)$. The second is the Laplace error $\varepsilon_i\sim \operatorname{Laplace}(0,1/\sqrt2)$, again normalized to have variance one. The setting with Gaussian errors matches the noise assumption in our high-dimensional theory, whereas the setting with Laplace errors examines robustness to non-Gaussian errors with the same variance.

We compare the following procedures: the classical pairs bootstrap, the classical residual bootstrap, the jackknife variance estimator \citep{el2018can}, the smoothed pairs bootstrap \citep{hall1989smoothing,silverman1987bootstrap}, and generative bootstrap procedures based on diffusion models. The diffusion pairs bootstrap estimates the joint law of $(X,Y)$ and generates bootstrap pairs from the fitted distribution. The diffusion residual bootstrap estimates the residual distribution and generates bootstrap errors from the fitted residual law. The smoothed pairs bootstrap is included as a baseline in order to distinguish the effect of merely smoothing the empirical distribution from the effect of learning a more accurate generative distribution estimator.

For each method, we report two metrics. The first is the empirical Type~I error at the nominal level $\alpha=0.05$, computed as the fraction of Monte Carlo replications in which the confidence interval for the first coordinate fails to cover zero. The second is the variance ratio
\begin{equation*}
    \mathrm{VR}=\frac{\widehat V}{1/\{n(1-\kappa)\}}.
\end{equation*}
Here $\widehat V$ denotes the average method-specific variance estimate across Monte Carlo replications. The denominator $1/\{n(1-\kappa)\}$ is the theoretical variance benchmark in the proportional regime for the first coordinate under Gaussian design and error variance one. A variance ratio close to one indicates accurate variance calibration. Ratios larger than one correspond to conservative variance estimation, whereas ratios smaller than one correspond to anti-conservative variance estimation.

\subsection{Theoretical variance}
Before presenting the results, we specify the denominators used in the variance ratios. For the Gaussian design, the theoretical variance benchmark has a particularly simple form. When $X_i\sim N(0,I_p)$, $\varepsilon_i\sim N(0,\sigma_\varepsilon^2)$, and $p/n\to\kappa\in(0,1)$, standard Wishart theory gives, for any deterministic contrast $c_n$ with $\|c_n\|_2=1$,
\begin{equation*}
    n c_n^\top (X^\top X)^{-1} c_n\to_p\frac{1}{1-\kappa}.
\end{equation*}
Consequently,
\begin{equation*}
    \Var\left(c_n^\top\widehat\beta\mid X\right)= \sigma_\varepsilon^2 c_n^\top (X^\top X)^{-1} c_n=\frac{\sigma_\varepsilon^2}{n(1-\kappa)}(1+o_p(1)).
\end{equation*}
Thus, when $\sigma_\varepsilon^2=1$, the natural theoretical variance benchmark is $1/\{n(1-\kappa)\}$. This benchmark is used for the Gaussian design experiments when computing variance ratios.

For elliptical designs, the correct variance benchmark in the proportional regime depends on the radial distribution of the design. Therefore, in addition to the Gaussian benchmark $1/\{n(1-\kappa)\}$, we compute the corresponding asymptotic variance constants $m_\lambda(\kappa)$. In the elliptical model $X_i=\lambda_i Z_i$, where $Z_i\sim N(0,I_p)$ and $\mathbb E[\lambda_i^2]=1$, the asymptotic variance of a fixed normalized contrast takes the form
\begin{equation*}
    \Var(c_n^\top\widehat\beta)\sim\frac{\sigma_\varepsilon^2}{n}m_\lambda(\kappa).
\end{equation*}
The constant $m_\lambda(\kappa)$ is determined by the fixed-point equation
\begin{equation*}
    \frac{1}{m_\lambda(\kappa)}=\mathbb E\left[\frac{\lambda^2}{1+\kappa m_\lambda(\kappa)\lambda^2}\right].
\end{equation*}
Table~\ref{tab:elliptical-asymptotic-variance} reports these constants for the Gaussian design and for the three elliptical designs used in our simulations. The constants increase substantially as the radial distribution becomes more heterogeneous, especially for the elliptical exponential design. This reflects the fact that row-norm heterogeneity increases the variance of OLS contrasts in the proportional regime. Accordingly, for experiments with elliptical designs and $\sigma_\varepsilon^2=1$, we use $m_\lambda(\kappa)/n$ as the theoretical variance benchmark when computing variance ratios.
\begin{table*}[ht]
\centering
\begin{tabular}{lccccc}
\toprule
\multirow{2}{*}{Design}
& \multicolumn{5}{c}{Asymptotic variance constant $m_\lambda(\kappa)$} \\ \cmidrule(lr){2-6}
& 0.1 & 0.2 & 0.3 & 0.4 & 0.5 \\ \midrule
Gaussian Design $(\lambda\equiv1)$
& 1.111 & 1.250 & 1.429 & 1.667 & 2.000 \\
Elliptical Uniform
& 1.145 & 1.328 & 1.568 & 1.894 & 2.356 \\
Elliptical Normal
& 1.342 & 1.810 & 2.494 & 3.557 & 5.340 \\
Elliptical Exponential
& 1.615 & 2.456 & 3.737 & 5.831 & 9.539 \\
\bottomrule
\end{tabular}
\caption{Asymptotic variance constants for OLS under Gaussian and elliptical designs.  We consider $X_i=\lambda_i Z_i$, where $Z_i\sim N(0,I_p)$, $\mathbb E\lambda_i^2=1$, and $p/n\to\kappa\in(0,1)$. For any deterministic contrast $c_n$ with $\|c_n\|=1$, \(\Var(c_n^\top\widehat\beta)\sim \frac{\sigma_\varepsilon^2}{n}m_\lambda(\kappa)\), where $m_\lambda(\kappa)$ solves $\frac{1}{m_\lambda(\kappa)} =\mathbb E\left[\frac{\lambda^2}{1+\kappa m_\lambda(\kappa)\lambda^2}\right]$. When $\sigma_\varepsilon^2=1$, the theoretical variance used for calibration is $m_\lambda(\kappa)/n$.}
\label{tab:elliptical-asymptotic-variance}
\end{table*}

\subsection{Simulation results}
We organize the results by the error distribution, beginning with the Normal-Normal benchmark and then considering Laplace errors and non-Gaussian designs. The main results appear in Figures~\ref{fig:type1-curve}, \ref{fig:normal-laplace}, \ref{fig:laplace-normal}, and \ref{fig:laplace-laplace}; the complete numerical values are reported in Appendix~\ref{app:simulation-tables}. Across the ten combinations of design and error distributions, we compare the classical pairs bootstrap, the classical residual bootstrap, the jackknife variance estimator, and the diffusion pairs bootstrap. In the Normal-Normal benchmark setting, we additionally include the smoothed pairs bootstrap and the diffusion residual bootstrap as diagnostic baselines. The smoothed pairs bootstrap helps separate the effect of smoothing the empirical distribution from the effect of learning a generative distribution estimator, whereas the diffusion residual bootstrap tests whether improving the residual
distribution alone is sufficient.

\begin{figure}[t]
    \centering
    \includegraphics[width=0.95\textwidth]{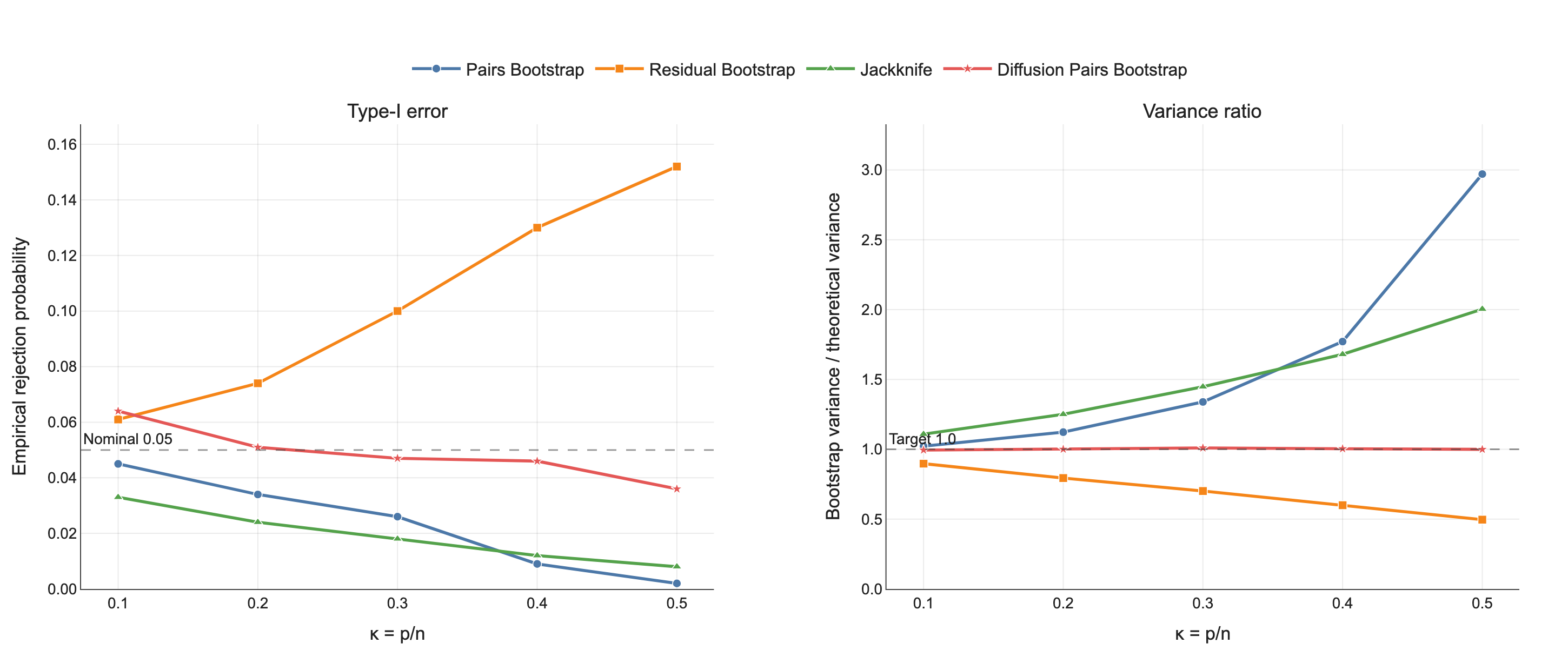}
    \caption{
    Type~I error and variance calibration with Gaussian design and Laplace errors. We take $n=500$, $p/n=\kappa$, and $\varepsilon_i\sim \operatorname{Laplace}(0,1/\sqrt2)$, so that $\Var(\varepsilon_i)=1$. Variance ratios are computed using the benchmark for Gaussian design $1/\{n(1-\kappa)\}$. This setting tests robustness to non-Gaussian errors while keeping the design distribution
    Gaussian.}
    \label{fig:normal-laplace}
\end{figure}

\begin{figure}[t]
    \centering
    \includegraphics[width=0.95\textwidth]{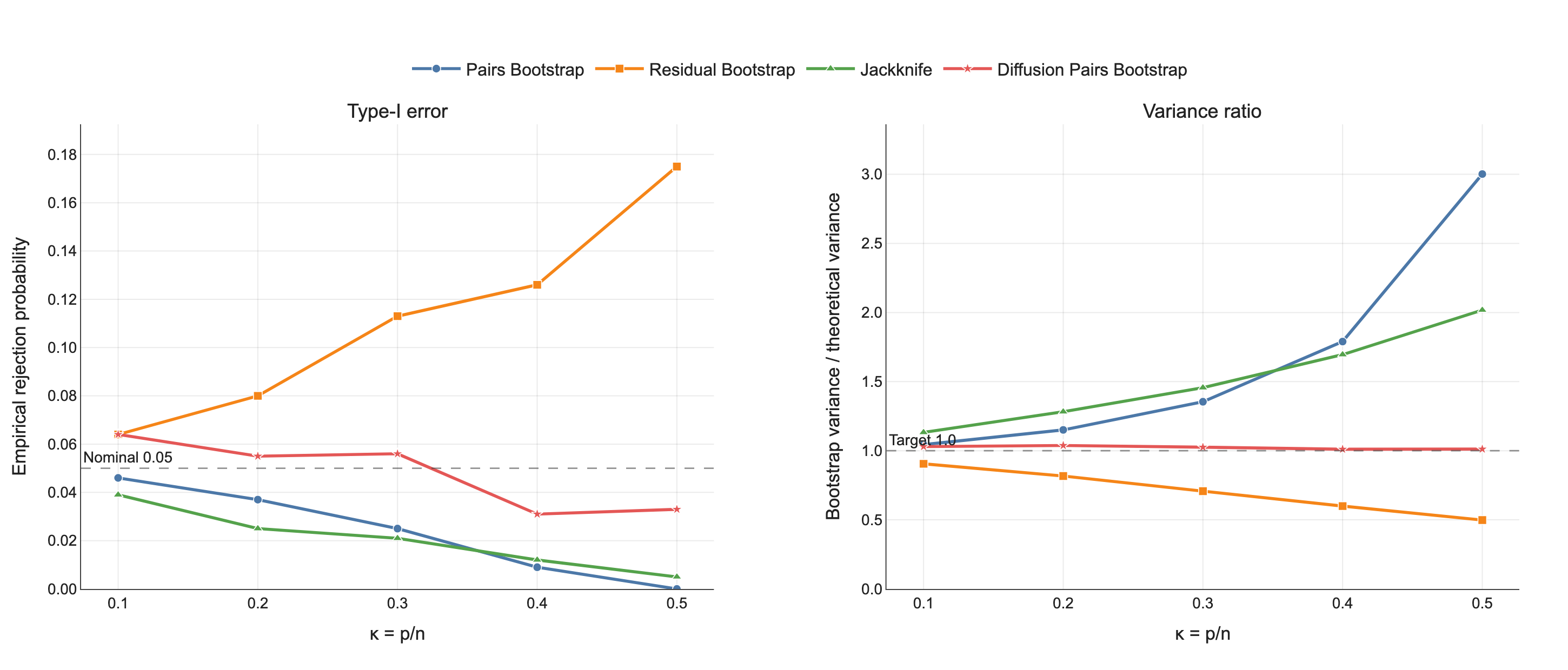}
    \caption{
    Type~I error and variance calibration with Laplace design and Gaussian errors. The entries of the design matrix are i.i.d. $\operatorname{Laplace}(0,1/\sqrt2)$, hence have mean zero and variance one, and the errors satisfy $\varepsilon_i\sim N(0,1)$. Variance ratios are computed using the benchmark $1/\{n(1-\kappa)\}$. This setting tests robustness to non-Gaussian i.i.d.\ covariates while keeping the error distribution Gaussian.}
    \label{fig:laplace-normal}
\end{figure}

\begin{figure}[t]
    \centering
    \includegraphics[width=0.95\textwidth]{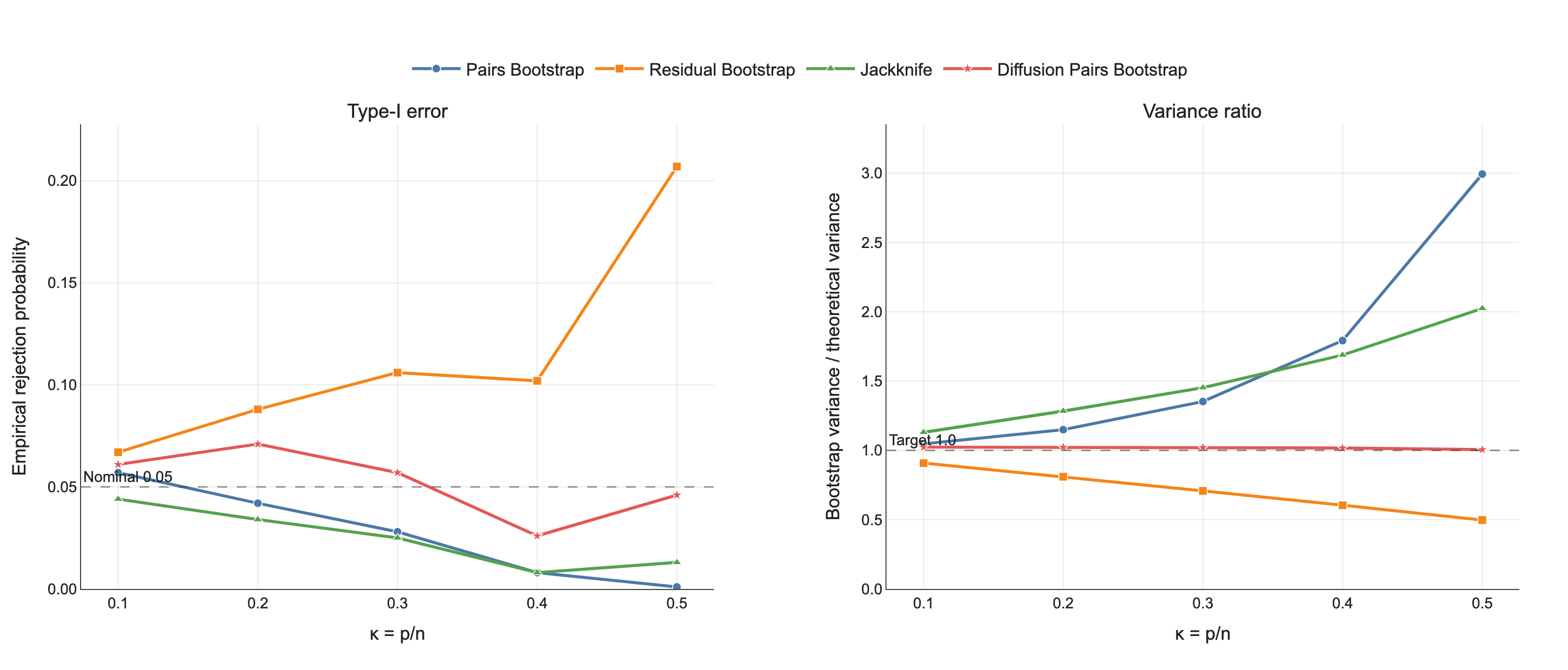}
    \caption{
    Type~I error and variance calibration with Laplace design and Laplace errors. The design entries are i.i.d. $\operatorname{Laplace}(0,1/\sqrt2)$, and the errors are also drawn from $\operatorname{Laplace}(0,1/\sqrt2)$, so both the covariates and the errors are non-Gaussian but standardized to have variance one. Variance ratios are computed using the benchmark $1/\{n(1-\kappa)\}$. This setting evaluates whether the calibration improvement persists when both the design and the errors depart from Gaussianity.}
    \label{fig:laplace-laplace}
\end{figure}

Across all combinations of design and error distributions, we observe a consistent high-dimensional pattern. The classical pairs bootstrap and the jackknife tend to overestimate the variance as $\kappa$ increases, whereas the residual bootstrap tends to underestimate it. These two failure modes lead to opposite inferential behavior: pairs bootstrap and jackknife produce conservative intervals, while residual bootstrap can become anti-conservative. The diffusion pairs bootstrap substantially reduces this variance distortion in most settings.

The Normal-Normal setting provides the benchmark comparison. In this setting, the classical pairs bootstrap becomes increasingly conservative as $\kappa$ grows, and the residual bootstrap increasingly underestimates the variance. This agrees with the high-dimensional theory: pairs resampling changes the effective geometry of the design matrix, while residual resampling uses fitted residuals whose empirical distribution no longer matches the true error distribution in the proportional regime. The diagnostic baselines also show that adding noise or smoothing the empirical distribution does not by itself produce the same improvement.

The same qualitative behavior persists under non-Gaussian i.i.d.\ designs and non-Gaussian errors. In the settings with Laplace designs and Laplace errors, the coordinates of $X_i$ and the errors are both non-Gaussian but standardized to have unit variance. The classical pairs bootstrap remains conservative, and the residual bootstrap remains anti-conservative. The diffusion pairs bootstrap keeps its variance ratio close to one and generally improves Type~I error calibration, so the improvement is not specific to Gaussian designs or Gaussian errors.

The elliptical designs are more challenging because row norms are heterogeneous. In these settings, the correct variance benchmark is $m_\lambda(\kappa)/n$ rather than $1/\{n(1-\kappa)\}$, as reported in Table~\ref{tab:elliptical-asymptotic-variance}. The elliptical exponential design is the most difficult case among those considered. In this setting, the classical pairs bootstrap and the jackknife can severely overestimate the variance, while the residual bootstrap continues to underestimate it. At moderate-to-large values of $\kappa$, the diffusion pairs bootstrap reduces the variance distortion relative to classical pairs bootstrap, but its calibration is less accurate than in the Gaussian designs and designs with i.i.d.\ coordinates. Thus, learning the joint distribution reduces, but does not eliminate, the variance distortion under strong radial heterogeneity.

Overall, standard resampling methods exhibit systematic high-dimensional variance distortions in these experiments. Diffusion pairs bootstrap is most accurate for Gaussian designs and designs with i.i.d.\ coordinates, and it reduces, but does not fully remove, the distortion under strongly heterogeneous elliptical designs.

\section*{Acknowledgments}
We thank Fang Han, Lin Liu, and Anru Zhang for helpful comments.

\vskip 0.2in
\bibliographystyle{plainnat}
\bibliography{sample}

\newpage

\appendix
\section{Details of the numerical experiments}
\label{app:numerical-experiments}
This appendix provides additional implementation details for the numerical experiments reported in Section~\ref{sec:simulation}. We first describe the Monte Carlo procedure, then specify the simulation parameters and the implementation of each bootstrap method. Finally, we describe the construction of confidence intervals, the computation of the reported performance metrics, and the numerical evaluation of the theoretical variance benchmarks.

\subsection{Simulation description}

All numerical experiments were implemented in Python. Data generation and classical statistical computations were carried out using NumPy and SciPy, while the diffusion models and batched linear algebra computations were implemented in PyTorch. Figures were produced using Matplotlib. The diffusion models were trained and sampled on NVIDIA A100 GPUs.

For each combination of design distribution, error distribution, and aspect ratio $\kappa=p/n$, we perform $R$ independent Monte Carlo replications. In each replication, a dataset $\{(X_i,Y_i)\}_{i=1}^n$ is first generated from the specified linear model, and the ordinary least-squares estimator $\widehat\beta$ is computed from the full sample. All competing bootstrap procedures are then applied to exactly the same dataset.

The classical pairs, residual, and smoothed pairs bootstrap samples are constructed directly from the observed data. The jackknife variance estimator is computed from leave-one-out fits. For the diffusion pairs bootstrap, a diffusion model is trained on the observed joint vectors $(X_i,Y_i)$, and bootstrap samples are generated from the learned reverse diffusion process. For the diffusion residual bootstrap, the diffusion model is trained on the centered fitted residuals and is used to generate bootstrap errors. A separate diffusion model is trained for every Monte Carlo replication.

Within each Monte Carlo replication, every bootstrap procedure uses $B$ bootstrap samples. For each method, we record its variance estimate for the first regression coefficient and whether the resulting confidence interval contains the true value $\beta_1=0$. After all Monte Carlo replications have been completed, we compute the empirical Type~I error and the average variance estimate reported in Section~\ref{sec:simulation}.

\subsection{Simulation parameters}

Unless otherwise stated, the sample size is fixed at $n=500$, and the aspect ratio varies over
\begin{equation*}
    \kappa\in\{0.1,0.2,0.3,0.4,0.5\},
\end{equation*}
with $p=\operatorname{round}(n\kappa)$. For every combination of design and error distribution and every value of $\kappa$, we perform $R=1000$ independent Monte Carlo replications. Within each replication, each bootstrap method uses $B=1000$ bootstrap samples. The nominal significance level is fixed at $\alpha=0.05$, and all reported results correspond to inference on the first regression coefficient.

Five covariate distributions are considered. 
\begin{itemize}
    \item The Gaussian design satisfies $X_i\sim N(0,I_p)$.
    \item The Laplace design has independent coordinates $X_{ij}\sim\operatorname{Laplace}(0,1/\sqrt2)$, so that each coordinate has variance one.
    \item The remaining three designs are elliptical models of the form $X_i=\lambda_iZ_i$, where $Z_i\sim N(0,I_p)$, the radial variable $\lambda_i$ is independent of $Z_i$, and $\mathbb E\lambda_i^2=1$. Specifically,
    \begin{itemize}
        \item The elliptical normal design uses $\lambda_i\sim N(0,1)$;
        \item the elliptical uniform design uses $\lambda_i=\sqrt{12/13}\,U_i$ with $U_i\sim\Unif(0.5,1.5)$;
        \item the elliptical exponential design uses $\lambda_i\sim\Exp(\sqrt2)$.
    \end{itemize}
\end{itemize}

The error distribution is either $N(0,1)$ or $\operatorname{Laplace}(0,1/\sqrt2)$, both having variance one.

For each diffusion procedure, the same architecture and training hyperparameters are used across all design distributions and all values of $\kappa$. No hyperparameter tuning specific to individual settings or early stopping based on validation is performed.

\begin{table}[t]
\centering
\small
\begin{tabular}{lcc}
\toprule
Hyperparameter
& Joint diffusion
& Residual diffusion \\
\midrule
Model type
& DDPM with noise prediction
& DDPM with noise prediction \\
Denoising network
& Multilayer perceptron
& Multilayer perceptron \\
Hidden width
& 256
& 128 \\
Network depth
& 2
& 3 \\
Activation function
& SiLU
& SiLU \\
Dimension of time embedding
& 64
& 32 \\
Optimizer
& AdamW
& AdamW \\
Learning rate
& $10^{-4}$
& $10^{-4}$ \\
Batch size
& 256
& 256 \\
Training iterations
& 4000
& 1500 \\
Number of diffusion steps
& 200
& 200 \\
$\beta_{\min}$
& $10^{-4}$
& $10^{-4}$ \\
$\beta_{\max}$
& $10^{-2}$
& $2\times10^{-2}$ \\
Early stopping
& No
& No \\
Tuning for individual settings
& No
& No \\
\bottomrule
\end{tabular}
\caption{Diffusion model hyperparameters used in the numerical experiments. The same joint diffusion configuration is used across all design distributions and values of $\kappa$, and the same is true for the residual diffusion configuration.}
\label{tab:diffusion-hyperparameters}
\end{table}

\subsection{Implementation of the bootstrap procedures}
\paragraph{Classical pairs bootstrap.}
Bootstrap observations are generated by sampling the observed pairs $(X_i,Y_i)$ independently with replacement.

\paragraph{Classical residual bootstrap.}
The OLS estimator is first computed from the full sample, and the fitted residuals are centered by subtracting their sample mean. Bootstrap residuals are then sampled with replacement and combined with the original design matrix. No leverage correction or degrees-of-freedom correction is applied.

\paragraph{Jackknife variance estimator.}
The delete-one jackknife estimator is computed using the exact leave-one-out identity for ordinary least squares, avoiding repeated model fitting.

\paragraph{Smoothed pairs bootstrap.}
Observed pairs are first resampled with replacement. Independent Gaussian perturbations are then added to the covariates and responses after scaling by their empirical standard deviations: \(X_i^*=X_{I_i}+h\widehat D_X\xi_i\) and \(Y_i^*=Y_{I_i}+h\widehat\sigma_Y\zeta_i\), where \(\xi_i\sim N(0,I_p)\) and \(\zeta_i\sim N(0,1)\). Throughout the experiments, the smoothing bandwidth is fixed at $h=0.05$.

\paragraph{Diffusion pairs bootstrap.}
The diffusion model is trained on the standardized joint observations $(X_i,Y_i)$. We use a denoising diffusion probabilistic model with the noise prediction objective. The architecture and training hyperparameters are reported in Table~\ref{tab:diffusion-hyperparameters}. After training, the reverse chain is initialized from a standard Gaussian distribution, and the generated observations are transformed back to the original coordinate scale.

\paragraph{Diffusion residual bootstrap.}
The centered fitted residuals are standardized and used to train a one-dimensional denoising diffusion probabilistic model. Its architecture and training hyperparameters are reported in Table~\ref{tab:diffusion-hyperparameters}. Generated residuals are transformed back to the original scale, centered within each bootstrap sample, and combined with the original design matrix.

\subsection{Confidence intervals and evaluation metrics}
For the classical pairs bootstrap, residual bootstrap, smoothed pairs bootstrap, diffusion pairs bootstrap, and diffusion residual bootstrap, confidence intervals are constructed using percentile bootstrap intervals. If $\widehat\beta_1^{*(1)},\ldots,\widehat\beta_1^{*(B)}$ denote the bootstrap estimates of the first regression coefficient, the confidence interval is given by the empirical $\alpha/2$ and $1-\alpha/2$ quantiles of these bootstrap estimates.

For the jackknife, we use the Gaussian confidence interval $\widehat\beta_1\pm z_{1-\alpha/2}\sqrt{\widehat V_{\mathrm{jack}}}$.

Bootstrap variances are computed using the unbiased sample variance of the bootstrap estimates with denominator $B-1$. The reported bootstrap variance is the average of these variance estimates across the Monte Carlo replications.

The empirical Type~I error is the proportion of Monte Carlo replications for which the corresponding confidence interval fails to contain the true value $\beta_1=0$. The reported variance ratio is defined as
\begin{equation*}
    \mathrm{VR}=\frac{ R^{-1}\sum_{r=1}^R\widehat V_r}{V_{\mathrm{theory}}},
\end{equation*}
where $V_{\mathrm{theory}}$ denotes the theoretical variance benchmark.

\subsection{Computation of the theoretical variance benchmarks}

For Gaussian design, the theoretical benchmark in the proportional regime is
\begin{equation*}
    V_{\mathrm{theory}} =\frac{1}{n(1-\kappa)}.
\end{equation*}
The same first-order benchmark is used for the standardized i.i.d.\ Laplace design. For elliptical designs, the theoretical variance benchmark is
\begin{equation*}
    V_{\mathrm{theory}}=\frac{m_\lambda(\kappa)}{n},
\end{equation*}
where $m_\lambda(\kappa)$ is the unique positive solution of
\begin{equation*}
    \frac{1}{m_\lambda(\kappa)}=\mathbb E\left[\frac{\lambda^2}{1+\kappa m_\lambda(\kappa)\lambda^2}\right].
\end{equation*}
The fixed-point equation is solved numerically for each radial distribution and each value of $\kappa$, and the resulting constants are reported in Table~\ref{tab:elliptical-asymptotic-variance}.

\subsection{Computational details}

Each Monte Carlo replication uses an independent random seed. Within the same replication, all competing methods are applied to exactly the same generated dataset. No Monte Carlo replication is discarded on the basis of numerical results. All diffusion models are trained for a fixed number of optimization iterations, and no early stopping based on validation is used.

The diffusion training, reverse sampling, and batched least-squares computations are performed on GPUs. Numerical results from individual Monte Carlo replications are aggregated after all parallel jobs have completed to produce the reported Type~I errors, variance ratios, and variance estimates.
\clearpage
\section{Notation}
The following tables summarize the notation used repeatedly in the fixed-dimensional, high-dimensional, and structured score estimation arguments. The notation is grouped by its role in the proofs. Variables local to individual proofs and truncation indices are defined at their first occurrence and are omitted here.
\vspace{0.5\baselineskip}

\begin{table}[H]
    \vspace*{0.35\baselineskip}
    \centering
    \caption{Notation used in the fixed-dimensional proofs.}
    \label{tab:notation-fixed-dimensional}
    \small
    \renewcommand{\arraystretch}{1.15}
    \setlength{\tabcolsep}{5pt}
    \begin{tabular}{@{}p{0.25\textwidth}p{0.69\textwidth}@{}}
        \toprule
        Notation & Meaning \\
        \midrule
        $d=p+1$ & Fixed dimension of the joint vector $Z=(X,Y)\in\mathbb R^d$. \\
        $\mu$ & Target law of $Z=(X,Y)$. \\
        $\beta_t$, $B(t)$ & Forward noise schedule and its cumulative value,
        $B(t)=\int_0^t\beta_s\,\mathrm ds$. \\
        $p_t=e^{-U_t}$ & Density of the forward Ornstein--Uhlenbeck process at time $t$ and its negative log-density. \\
        $s_t=\nabla\log p_t$ & Exact score of the forward density. \\
        $m_t^{\mathrm{sc}}$ & Strong log-concavity parameter of $p_t$. \\
        $\widehat s_{n,t}$, $e_{n,t}$ & Learned score and score error,
        $e_{n,t}=\widehat s_{n,t}-s_t$. \\
        $b_{n,r}$, $\widehat b_{n,r}$ & Exact and learned reverse drifts. \\
        $Y_r$, $\widehat Y_r$ & Exact and learned reverse processes. \\
        $\widehat\mu_n$ & Terminal law of the learned reverse process. \\
        $\rho_{n,r}$ & Conditional density of the learned reverse process $\widehat Y_r$ in the fixed-dimensional PDE argument. \\
        $\mathfrak C_n$ & Integral of the negative part of the divergence,
        $\int\|(\nabla\!\cdot\widehat b_{n,r})_-\|_\infty\,\mathrm dr$. \\
        $L_n$, $G_n$ & Lipschitz and linear growth envelopes, depending on time, for the learned reverse drift. \\
        $\Sigma_X(\nu)$ & Design second-moment matrix,
        $\Sigma_X(\nu)=\mathbb E_\nu[XX^\top]$. \\
        $\beta(\nu)$ & Population least-squares coefficient under the law $\nu$. \\
        $\widehat\Sigma_n^*$ & Normalized generated Gram matrix,
        $\widehat\Sigma_n^*=n^{-1}\sum_{i=1}^nX_i^*X_i^{*\top}$. \\
        $U_n^*$ & Generated empirical score,
        $U_n^*=n^{-1/2}\sum_{i=1}^nX_i^*\{Y_i^*-X_i^{*\top}\beta(\widehat\mu_n)\}$. \\
        $S_\nu(z)$, $\tau^2(\nu)$ & Influence function for the linear contrast and its variance under $Z\sim\nu$. \\
        $\mathbb P^*$, $\mathbb E^*$, $o_{p^*}(1)$ & Probability, expectation, and stochastic order conditional on the learned law $\widehat\mu_n$. \\
        \bottomrule
    \end{tabular}
\end{table}

\clearpage
\begin{table}[H]
    \vspace*{0.35\baselineskip}
    \centering
    \caption{Model, diffusion, and PDE notation used in the high-dimensional proofs.}
    \label{tab:notation-high-dimensional-pde}
    \small
    \renewcommand{\arraystretch}{1.15}
    \setlength{\tabcolsep}{5pt}
    \begin{tabular}{@{}p{0.25\textwidth}p{0.69\textwidth}@{}}
        \toprule
        Notation & Meaning \\
        \midrule
        $d_n=p_n+1$ & Dimension of the high-dimensional joint vector $Z=(X,Y)$. \\
        $T_n$ & Diffusion horizon, with $T_n\to\infty$ and $d_ne^{-2T_n}\to0$. \\
        $\mathcal F_n$ & $\sigma$-field with respect to which the fitted score and generated law are measurable. \\
        $\mu_n$, $\mu_{n,X}$ & Target joint law of $(X,Y)$ and its design marginal. \\
        $\Sigma_n$ & Design covariance, $\Sigma_n=\mathbb E[XX^\top]$. \\
        $c_\Sigma$, $C_\Sigma$ & Uniform lower and upper spectral bounds for $\Sigma_n$. \\
        $\underline m$, $\overline m$ & Relative curvature constants in
        $\underline m\,\Sigma_n^{-1}\preceq\nabla^2V_n\preceq \overline m\,\Sigma_n^{-1}$. \\
        $V_n$, $U_n$ & Design and joint potentials, with
        $U_n(x,y)=V_n(x)+(y-x^\top\beta_n)^2/(2\sigma_n^2)$. \\
        $p_{n,t}$ & Density of the exact forward Ornstein--Uhlenbeck law at time $t$. \\
        $s_{n,t}(z)$ & Exact score, $s_{n,t}=\nabla\log p_{n,t}$. \\
        $\widehat s_{n,t}$, $e_{n,t}$ & Learned score and score error,
        $e_{n,t}=\widehat s_{n,t}-s_{n,t}$. \\
        $\epsilon_{n,r}$ & Score error evaluated at reverse time,
        $\epsilon_{n,r}=e_{n,T_n-r}$. \\
        $b_{n,r}$, $\widehat b_{n,r}$ & Exact and learned reverse drifts, with
        $\widehat b_{n,r}=b_{n,r}+2\epsilon_{n,r}$. \\
        $C_{\mathrm{LS}}$ & Uniform log-Sobolev and Poincar\'e constant for the exact OU marginals. \\
        $L_\star$, $\lambda_\star$ & Constants appearing in
        \eqref{eq:lipschitz-absorption}; $L_\star$ bounds
        $\Lip(e_{n,t})$. \\
        $\mathcal E_n$ & Event of probability tending to one on which the score error satisfies the required Lipschitz and integrability bounds. \\
        $\rho_{n,r}$, $\widehat\rho_{n,r}$ & Exact and learned reverse densities, respectively; $\rho_{n,r}=p_{n,T_n-r}$. \\
        $f_{n,r}$ & Density ratio,
        $f_{n,r}=\widehat\rho_{n,r}/\rho_{n,r}$. \\
        $D_n(r)$ & The integral
        $D_n(r)=\int f_{n,r}^2\rho_{n,r}\,\mathrm dz$. \\
        $H_n(r)$ & Chi-square divergence,
        $H_n(r)=D_n(r)-1=\chi^2(\widehat\rho_{n,r},\rho_{n,r})$. \\
        $a_n(t)$ & Fourth moment of score error,
        $a_n(t)=\int\|e_{n,t}(z)\|^4p_{n,t}(z)\,\mathrm dz$. \\
        $\alpha_n(r)$ & Coefficient used in the chi-square energy estimate,
        $\alpha_n(r)=a_n(T_n-r)^{1/2}$. \\
        $T_M$, $\Phi_M$, $F_M$ & Truncation, its primitive, and the associated perspective function used in the density-ratio identity. \\
        $E_M(r)$ & The truncated integral
        $E_M(r)=\int\Phi_M(f_{n,r})\rho_{n,r}\,\mathrm dz$. \\
        \bottomrule
    \end{tabular}
\end{table}

\clearpage
\begin{table}[H]
    \vspace*{0.35\baselineskip}
    \centering
    \caption{Statistical notation used in the high-dimensional proofs.}
    \label{tab:notation-high-dimensional-statistics}
    \small
    \renewcommand{\arraystretch}{1.15}
    \setlength{\tabcolsep}{5pt}
    \begin{tabular}{@{}p{0.25\textwidth}p{0.69\textwidth}@{}}
        \toprule
        Notation & Meaning \\
        \midrule
        $\kappa$, $\beta_n$, $\sigma_n^2$ & Aspect-ratio limit $p_n/n\to\kappa$, regression coefficient, and noise variance. \\
        $c_n$ & Deterministic contrast with $\|c_n\|_2=1$. \\
        $\widehat\mu_{n,X}$ & Design marginal of the generated joint law $\widehat\mu_n$. \\
        $\widehat K_n(x)$, $K_n(x)$ & Generated and target conditional laws of $Y$ given $X=x$, with
        $K_n(x)=N(x^\top\beta_n,\sigma_n^2)$. \\
        $R_n$ & Likelihood ratio in the conditional law argument,
        $R_n=\mathrm d\widehat\mu_n/\mathrm d\mu_n$. \\
        $v_n^*(x)$ & Generated conditional variance,
        $v_n^*(x)=\Var_{\widehat K_n(x)}(Y)$. \\
        $m_n(x)$, $r_n(x)$ & Generated conditional mean and its deviation from the target regression function,
        $m_n(x)=\mathbb E_{\widehat\mu_n}[Y\mid X=x]$ and
        $r_n(x)=m_n(x)-x^\top\beta_n$. \\
        $\eta_n$, $\zeta_n$ & Average fourth power error in the conditional mean and root mean square error in the conditional variance. \\
        $\widehat\Sigma_m^*$ & Normalized generated Gram matrix,
        $\widehat\Sigma_m^*=m^{-1}\sum_{i=1}^mX_i^*X_i^{*\top}$. \\
        $S_n$, $S_n^*$ & Unnormalized true and generated Gram matrices. \\
        $S_{-i}^*$ & Leave-one-out Gram matrix,
        $S_{-i}^*=\sum_{j\neq i}X_j^*X_j^{*\top}$. \\
        $W_i$, $W_i^*$ & Whitened true and generated designs,
        $W_i=\Sigma_n^{-1/2}X_i$ and $W_i^*=\Sigma_n^{-1/2}X_i^*$. \\
        $A_n$, $A_n^*$ & Normalized Gram matrices of the whitened true and generated designs. \\
        $\Delta_i$ & Rowwise coupling error, $W_i^*=W_i+\Delta_i$. \\
        $h^*$ & Weight vector for the generated linear contrast,
        $h^*=X^*(S_n^*)^{-1}c_n\in\mathbb R^n$. \\
        $M_{4,n}$ & Uniform directional fourth moment,
        $M_{4,n}=\sup_{\|u\|_2=1}\mathbb E_{\widehat\mu_{n,X}}(u^\top X)^4$. \\
        $\varepsilon$, $\varepsilon_i$ & Regression noise; the symbol is reserved for model noise rather than analytic perturbations. \\
        $S_\lambda$ & Coordinatewise soft-thresholding operator. \\
        $\widehat m$ & Moment estimator
        $\widehat m=n^{-1}\sum_{i=1}^nX_iY_i$ used in the upper bound for
        the score estimation risk. \\
        $\overline\sigma^2$, $\widehat\sigma^2$ & Raw residual-moment estimator and its projection onto
        $[\sigma_{\min}^2,\sigma_{\max}^2]$. \\
        \bottomrule
    \end{tabular}
\end{table}

\clearpage
\begin{table}[H]
    \vspace*{0.35\baselineskip}
    \centering
    \caption{Notation used in the structured score estimation and atomic approximation results.}
    \label{tab:notation-structured-minimax}
    \small
    \renewcommand{\arraystretch}{1.15}
    \setlength{\tabcolsep}{5pt}
    \begin{tabular}{@{}p{0.25\textwidth}p{0.69\textwidth}@{}}
        \toprule
        Notation & Meaning \\
        \midrule
        $\mathcal X_p$, $\mathcal M_p(\mathcal X_p)$ & A class of design distributions and the induced class of joint laws of $Z=(X,Y)$. \\
        $p_{\mu,t}$, $s_{\mu,t}$ & Density and score along the standard OU flow starting from $\mu\in\mathcal M_p(\mathcal X_p)$. \\
        $\mathcal R_{n,4}(\mathcal X_p;T_n)$ & Minimax risk for integrated
        fourth-moment score estimation over $[0,T_n]$. \\
        $\gamma_d$, $\mathcal A_{n,d}$ & Standard Gaussian law on $\mathbb R^d$ and the class of possibly randomized probability measures supported on at most $n$ points. \\
        $\mathcal X_p^{\mathrm I}$, $\mathcal X_p^{\mathrm{AR}}$, $\mathcal X_p^{\mathrm{LR}}$, $\mathcal X_p^{\mathrm{EXP}}$ & Standard Gaussian, Gaussian AR(1), fixed-rank Gaussian with known directions, and product exponential family design classes. \\
        $B_1$ & Uniform $\ell_1$ bound on the regression vector in the structured examples, $\|\beta\|_1\leq B_1$. \\
        $\vartheta$, $\Sigma_\vartheta$ & Generic fixed-dimensional structural parameter and its associated design covariance. \\
        $\rho$, $\widehat\rho$ & AR(1) correlation parameter and its estimator based on adjacent pairs. \\
        $U_p$, $r$, $\lambda$, $\Lambda$ & Known loading directions, fixed perturbation rank, spike parameter, and its compact convex parameter space in the fixed-rank model. \\
        $\theta$, $\Theta$, $q$, $\phi_\ell$ & Parameter, parameter space, fixed parameter dimension, and coordinate potentials of the product exponential family. \\
        $\delta_{\mathrm{NG}}$ & Uniform bound defined in
        \eqref{eq:product-nongaussian-curvature}. \\
        $a_{\ell,p}$ & Quantity satisfying
        $\|\partial_{\vartheta_\ell}\Sigma_p(\vartheta)\|_{\mathrm F}^2 \leq Ca_{\ell,p}$ in
        Lemma~\ref{lem:gaussian-structured-score-transfer}. \\
        $\lambda_n$ & Threshold level used for nuisance estimation,
        $\lambda_n=A\sqrt{\log p/n}$; it is unrelated to the fixed-rank parameter $\lambda$. \\
        $C_{\mathrm{Lip}}$, $\delta_0$ & Constants in
        \eqref{eq:product-small-lipschitz}. \\
        \bottomrule
    \end{tabular}
\end{table}

\begin{table}[H]
    \vspace*{0.35\baselineskip}
    \centering
    \caption{Shared analytic notation used in the auxiliary proofs.}
    \label{tab:notation-auxiliary}
    \small
    \renewcommand{\arraystretch}{1.15}
    \setlength{\tabcolsep}{5pt}
    \begin{tabular}{@{}p{0.25\textwidth}p{0.69\textwidth}@{}}
        \toprule
        Notation & Meaning \\
        \midrule
        $v_t(x)$ & Generic drift in the SDE and Markov kernel lemmas. \\
        $\ell(t)$, $G(t)$ & Lipschitz and linear growth envelopes, depending on time, for $v_t$. \\
        $P_{s,t}(x,A)$ & Transition kernel of the SDE started from $x$ at time $s$. \\
        $P_{s,t}\varphi$ & Backward action of the transition kernel on a test function. \\
        $P_{s,t}^*\nu$ & Forward image of a finite measure $\nu$ under the transition kernel. \\
        $b_t(x)$, $a_t(x)$ & Drift and diffusion covariance in the generic Fokker--Planck operator. \\
        $L_t$ & Kolmogorov operator,
        $L_t\phi=b_t^\top\nabla\phi+\tfrac12a_t:\nabla^2\phi$. \\
        $\nu_t$, $\widetilde\nu_t$ & A weak Fokker--Planck solution taking values in probability measures and its narrowly continuous representative. \\
        $\boldsymbol\eta$ & Probability measure on the path space $C([0,T];\mathbb R^d)$ supplied by the superposition principle. \\
        $e_t$ & Evaluation map on path space, $e_t(\omega)=\omega(t)$; hence $(e_t)_\#\boldsymbol\eta$ is the time-$t$ marginal. \\
        $M_t^\psi$ & Martingale associated with the test function $\psi$ and the operator $L_t$. \\
        $E(r)$, $D(r)$ & Generic energy and dissipation in the integral energy lemmas. \\
        $\widetilde E$ & Right-continuous non-increasing representative of an almost-everywhere defined dissipative energy. \\
        $\mathcal P(\mathbb R^d)$ & Space of Borel probability measures on $\mathbb R^d$. \\
        $\|\cdot\|_{\mathrm{TV}}$ & Total variation norm of a finite signed measure. \\
        \bottomrule
    \end{tabular}
\end{table}
\clearpage

\section{Proof of Theorem~\ref{fixed-vari}}
\label{app:theorem}
In this section, we prove Theorem~\ref{fixed-vari}. We first establish the $W_q$ convergence result under the score approximation assumption.
\begin{mylemma}[Strong log-concavity]
\label{strong-log-concavity}
Under Assumption~\ref{ass:target-distribution}, write $p_t=e^{-U_t}$. Then
\begin{equation*}
    \nabla^2U_t(z)\succeq m_t^{\mathrm{sc}}I_d,\qquad m_t^{\mathrm{sc}}=\frac{m_0}{e^{-2B(t)}+m_0\{1-e^{-2B(t)}\}}.
\end{equation*}
\end{mylemma}

\begin{proof}
By Lemma~\ref{lem:aux-strong-log-concavity-preservation}, the OU-smoothed law remains strongly log-concave. Indeed,
\begin{equation*}
    Z_t=e^{-B(t)}Z_0+\sqrt{1-e^{-2B(t)}}G,\qquad G\sim N(0,I_d).
\end{equation*}
Since $\nabla^2U_0\succeq m_0I_d$, the law of $e^{-B(t)}Z_0$ has covariance proxy $\frac{e^{-2B(t)}}{m_0}I_d$. The Gaussian noise term $\sqrt{1-e^{-2B(t)}}G$ has covariance proxy $\{1-e^{-2B(t)}\}I_d$. Therefore the law $p_t=e^{-U_t}$ of $Z_t$ has covariance proxy $(\dfrac{e^{-2B(t)}}{m_0}+1-e^{-2B(t)})I_d$. Equivalently,
\begin{equation*}
    \nabla^2U_t(z) \succeq\left(\frac{e^{-2B(t)}}{m_0}+1-e^{-2B(t)}\right)^{-1}I_d=\frac{m_0}{e^{-2B(t)}+m_0\{1-e^{-2B(t)}\}}I_d.
\end{equation*}
This is the asserted bound.
\end{proof}

\begin{mylemma}[$W_q$ convergence]
    Under Assumptions~\ref{ass:target-distribution}, \ref{ass:noise-schedule}, and \ref{ass:score-approximation}, 
    \begin{equation*}
        W_q(\widehat\mu_n,\mu)\to_p 0.
    \end{equation*}
    \label{fix_wq}
\end{mylemma}

\begin{proof}
    By Lemma~\ref{strong-log-concavity}, $p_t=e^{-U_t}$ satisfies
    \begin{equation*}
        \nabla^2 U_t(z)\succeq m_t^{\mathrm{sc}}I_d,m_t^{\mathrm{sc}}=\dfrac{m_0}{e^{-2B(t)}+m_0(1-e^{-2B(t)})}.
    \end{equation*}
    The exact reverse drift is
    \begin{equation*}
        b_{n,t}(z)=\beta_{T_n-t}(z+2s_{T_n-t}(z))=\beta_{T_n-t}(z-2\nabla U_{T_n-t}(z)).
    \end{equation*}
    Therefore, for $v=z-z'$, we have
    \begin{equation*}
        \begin{aligned}
            \langle v,b_{n,t}(z)-b_{n,t}(z')\rangle &=\beta_{T_n-t}(\langle v, z-z' \rangle-2\langle v,\nabla U_{T_n-t}(z)-\nabla U_{T_n-t}(z')\rangle)\\ &\leq \beta_{T_n-t}(1-2m_{T_n-t}^{\mathrm{sc}})\|v\|^2.
        \end{aligned}
    \end{equation*}
    Hence the positive expansion budget of the exact reverse drift is bounded:
    \begin{equation*}
        \Lambda_0:=\sup_{T>0}\int_0^T \beta_u(1-2m_u^{\mathrm{sc}})^+\mathrm du<\infty,
    \end{equation*}
    Indeed, with the change of variables $r=B(u)$ and $\mathrm dr=\beta_u\,\mathrm du$,
    \begin{equation*}
        \begin{aligned}
            \int_0^T \beta_u(1-2m_u^{\mathrm{sc}})^+\mathrm du=\int_0^{B(T)}(1-2\dfrac{m_0}{e^{-2r}+m_0(1-e^{-2r})})^+\mathrm dr,
        \end{aligned}
    \end{equation*}
    The integrand is continuous and vanishes for all sufficiently large $r$ because the displayed curvature converges to one. Hence $\Lambda_0<\infty$.

    Let $Y_t$ be the exact reverse process initialized from $p_{T_n}$:
    \begin{equation*}
        \mathrm dY_t = b_{n,t}(Y_t)\,\mathrm dt + \sqrt{2\beta_{T_n-t}}\,\mathrm dW_t, \qquad Y_0\sim p_{T_n}.
    \end{equation*}
    Couple it synchronously with the learned reverse process initialized from $\widehat Y_0\sim N(0,I_d)$, and set $\Delta_t=\widehat Y_t-Y_t$. Then
    \begin{equation*}
        \mathrm d\Delta_t=(b_{n,t}(\widehat Y_t)-b_{n,t}(Y_t))\mathrm dt+2\beta_{T_n-t}(\widehat s_{n,T_n-t}(\widehat Y_t)-s_{T_n-t}(\widehat Y_t))\mathrm dt,
    \end{equation*}
    Define
    \begin{equation*}
        a_t = \beta_{T_n-t}(1-2m_{T_n-t}^{\mathrm{sc}})^+, \qquad r_t = 2\beta_{T_n-t} \|\widehat s_{n,T_n-t}(\widehat Y_t) -s_{T_n-t}(\widehat Y_t)\|_2.
    \end{equation*}
    Hence
    \begin{equation*}
        \|\Delta_t\|\leq \|\Delta_0\|+\int_0^t a_s \|\Delta_s\|\mathrm ds+\int_0^t r_s\mathrm ds
    \end{equation*}
    for almost every path, conditionally on $\mathcal F_n$. Gronwall's inequality and $\int_0^{T_n}a_s\,\mathrm ds\leq\Lambda_0$ therefore give
    \begin{equation*}
        \|\Delta_{T_n}\|\leq e^{\Lambda_0}\|\Delta_0\|+e^{\Lambda_0}\int_0^{T_n}r_t \mathrm dt.
    \end{equation*}
    Indeed, $t\mapsto\Delta_t$ is absolutely continuous. At times for which $\Delta_t\neq0$,
    \begin{equation*}
        \begin{aligned}
            \frac{\mathrm d}{\mathrm dt}\|\Delta_t(\omega)\| &=\frac{\langle\Delta_t,\dot\Delta_t\rangle}{\|\Delta_t\|}\\ &=\frac{\langle\Delta_t,b_{n,t}(\widehat Y_t)-b_{n,t}(Y_t)\rangle} {\|\Delta_t\|}\\ &\quad+ \frac{2\beta_{T_n-t}\langle\Delta_t, \widehat s_{n,T_n-t}(\widehat Y_t)-s_{T_n-t}(\widehat Y_t)\rangle} {\|\Delta_t\|}\\ &\leq a_t\|\Delta_t\|+r_t.
        \end{aligned}
    \end{equation*}
    The same upper bound holds almost everywhere on the zero set by the standard chain rule for the norm. Integrating proves the claimed pathwise inequality.

    Taking conditional $L^q$ norms and applying Minkowski's inequality gives
    \begin{equation*}
        \begin{aligned}
            (\mathbb E[\|\Delta_{T_n}\|^q\mid \mathcal F_n])^{1/q}&\leq e^{\Lambda_0}\|\Delta_0\|_{L^q\mid \mathcal F_n}+e^{\Lambda_0}\|\int_0^{T_n}r_t \mathrm dt\|_{L^q\mid \mathcal F_n}\\&\leq e^{\Lambda_0}\|\Delta_0\|_{L^q\mid \mathcal F_n}+e^{\Lambda_0}\int_0^{T_n}\|r_t\|_{L^q\mid \mathcal F_n}\mathrm dt
        \end{aligned}
    \end{equation*}
    The second term is
    \begin{equation*}
        \|r_t\|_{L^q\mid \mathcal F_n}=2\beta_{T_n-t}(\mathbb E[\|\widehat s_{n,T_n-t}(\widehat Y_t)-s_{T_n-t}(\widehat Y_t)\|^q\mid \mathcal F_n])^{1/q},
    \end{equation*}
    so we have
    \begin{equation*}
        \|\int_0^{T_n}r_t \mathrm dt\|_{L^q\mid \mathcal F_n}\leq 2\int_0^{T_n}\beta_{T_n-t}(\mathbb E[\|\widehat s_{n,T_n-t}(\widehat Y_t)-s_{T_n-t}(\widehat Y_t)\|^q\mid \mathcal F_n])^{1/q}\mathrm dt=2\mathcal A_{n,q}(T_n).
    \end{equation*}
    Choose $(\widehat Y_0,Y_0)$ to be an optimal $W_q$ coupling of $N(0,I_d)$ and $p_{T_n}$; such a coupling exists by \citep[Theorem~4.1]{villani2009optimal}. Then
    \begin{equation*}
        \mathbb E[\|\widehat Y_0-Y_0\|^q\mid \mathcal F_n]^{1/q}=W_q(N(0,I_d),p_{T_n}).
    \end{equation*}
    Therefore,
    \begin{equation*}
        W_q(\widehat \mu_n,\mu)\leq e^{\Lambda_0}W_q(N(0,I_d),p_{T_n})+2e^{\Lambda_0}\mathcal A_{n,q}(T_n).
    \end{equation*}
    Finally, use the coupling
    \begin{equation*}
        Z_t = e^{-B(t)}Z_0 + \sqrt{1-e^{-2B(t)}}G, \qquad Z_0\sim\mu, \quad G\sim N(0,I_d), \quad Z_0\perp G.
    \end{equation*}
    It yields
    \begin{equation*}
        \begin{aligned}
            (\mathbb E\|Z_t-G\|^q)^{1/q}&=(\mathbb E\|e^{-B(t)}Z_0+(\sqrt{1-e^{-2B(t)}}-1)G\|^q)^{1/q}\\ &\leq e^{-B(t)}(\mathbb E\|Z_0\|^q)^{1/q}+|1-\sqrt{1-e^{-2B(t)}}|(\mathbb E\|G\|^q)^{1/q}
        \end{aligned}
    \end{equation*}
    Since
    \begin{equation*}
        |1-\sqrt{1-e^{-2B(t)}}|\leq e^{-2B(t)}\leq e^{-B(t)}
    \end{equation*}
    we obtain
    \begin{equation*}
        W_q(p_t,N(0,I_d))\leq e^{-B(t)}(\mathbb E\|Z_0\|^q)^{1/q}+e^{-B(t)}(\mathbb E\|G\|^q)^{1/q}\leq Ce^{-B(t)}(1+(\mathbb E\|Z\|^q)^{1/q})\to 0
    \end{equation*}
    Combining this bound at $t=T_n$ with $\mathcal A_{n,q}(T_n)\to_p 0$ proves $W_q(\widehat\mu_n,\mu)\to_p 0$.
\end{proof}

\begin{mylemma}
    Under Assumptions~\ref{ass:target-distribution}, \ref{ass:noise-schedule}, and \ref{ass:score-approximation}, let $p_t=e^{-U_t}$ be the density of the forward OU process and write
    \begin{equation*}
        \widehat b_{n,T_n-u}(z) = \beta_u\{z-2\nabla U_u(z)+2e_{n,u}(z)\}, \qquad 0\leq u\leq h_0.
    \end{equation*}
    Then the following statements hold.
    \begin{enumerate}
        \item With $\alpha_t=e^{-B(t)}$, $\sigma_t^2=1-\alpha_t^2$, and $\overline L=\max\{L_0,1\}$,
        \begin{equation*}
            \sup_{z\in\mathbb R^d} \|\nabla^2U_t(z)\|_{\op} \leq \overline L, \qquad t\geq0.
        \end{equation*}
        \item Define $\mathfrak C_n:=\int_{T_n-h_0}^{T_n} \|(\nabla\cdot\widehat b_{n,r})_-\|_\infty\,\mathrm dr$. Then
        \begin{equation*}
            \mathfrak C_n \leq d(2\overline L-1)_+ \int_0^{h_0}\beta_u\,\mathrm du + 2d\int_0^{h_0}\beta_u\Lip(e_{n,u})\,\mathrm du.
        \end{equation*}
        In particular, $\mathfrak C_n=O_p(1)$;
        \item With probability tending to one, there exist nonnegative random functions $L_n,G_n\in L^1([0,h_0])$ such that, for almost every $u\in[0,h_0]$ and all $z,z'\in\mathbb R^d$,
\begin{equation*}
    \|\widehat b_{n,T_n-u}(z) -\widehat b_{n,T_n-u}(z')\| \leq L_n(u)\|z-z'\|,
\end{equation*}
and
\begin{equation*}
    \|\widehat b_{n,T_n-u}(z)\| \leq G_n(u)(1+\|z\|).
\end{equation*}
Moreover,
\begin{equation*}
    \int_0^{h_0}L_n(u)\,\mathrm du=O_p(1), \qquad \int_0^{h_0}G_n(u)\,\mathrm du=O_p(1).
\end{equation*}
In particular, for every $R>0$,
\begin{equation*}
    \int_{T_n-h_0}^{T_n} \|\widehat b_{n,r}\|_{W^{1,\infty}(B_R)}\,\mathrm dr = O_p(1),
\end{equation*}
and, with probability tending to one,
\begin{equation*}
    \widehat b_n \in L^1_{\mathrm{loc}} \bigl( [T_n-h_0,T_n]; W^{1,\infty}_{\mathrm{loc}}(\mathbb R^d) \bigr).
\end{equation*}
    \end{enumerate}
\end{mylemma}

\begin{proof}
     Write $d=p+1$, which is fixed throughout the proof. 

    (1) Let $p_t=e^{-U_t}$ be the law of the forward OU process and set $\alpha_t=e^{-B(t)}$ and $\sigma_t^2=1-\alpha_t^2$. Then $Z_t=\alpha_tZ_0+\sigma_tG$, where $G\sim N(0,I_d)$. For $t>0$, Tweedie's formula gives
    \begin{equation*}
        \nabla^2 U_t(z)=\dfrac 1{\sigma_t^2}I_d-\dfrac{\alpha_t^2}{\sigma_t^4}\Cov(Z_0\mid Z_t=z).
    \end{equation*}
    Conditionally on $Z_t=z$, the law of $Z_0$ has potential
    \begin{equation*}
        V_{t,z}(x)=U_0(x)+\dfrac{\|z-\alpha_t x\|^2}{2\sigma_t^2}
    \end{equation*}
    Indeed, we have
    \begin{equation*}
        p(x|z)\propto p(z|x)p_0(x)\propto \exp(-\dfrac{\|z-\alpha_t x\|^2}{2\sigma_t^2})e^{-U_0(x)}\propto \exp(-(U_0(x)+\dfrac{\|z-\alpha_t x\|^2}{2\sigma_t^2})).
    \end{equation*}
    Therefore,
    \begin{equation*}
        \nabla^2 V_{t,z}(x)\preceq (L_0+\dfrac{\alpha_t^2}{\sigma_t^2})I_d
    \end{equation*}
    By the Cram\'er--Rao inequality,
    \begin{equation*}
        \Cov(Z_0\mid Z_t=z)\succeq (\mathbb E[\nabla^2 V_{t,z}(Z_0)\mid Z_t=z])^{-1}.
    \end{equation*}
    Since
    \begin{equation*}
        \mathbb E[\nabla^2 V_{t,z}(Z_0)\mid Z_t=z]\preceq (L_0+\dfrac{\alpha_t^2}{\sigma_t^2})I_d,
    \end{equation*}
    we obtain that
    \begin{equation*}
        \Cov(Z_0\mid Z_t=z)\succeq (L_0+\dfrac{\alpha_t^2}{\sigma_t^2})^{-1}I_d.
    \end{equation*}
    Thus $\nabla^2 U_t\preceq L_t I_d$, where $L_t:=(\sigma_t^2+\alpha_t^2/L_0)^{-1}$. Together with the bound at $t=0$, this gives
    \begin{equation*}
        \sup_{t\geq 0}\sup_{z\in \mathbb R^d}\|\nabla^2 U_t(z)\|_{\op}\leq \overline L:=\max \{L_0,1\}.
    \end{equation*}
    By Lemma~\ref{strong-log-concavity}, $\nabla ^2U_t(z)\succeq m_t^{\rm sc}I_d\succeq 0$. Combining this lower bound with $\nabla^2 U_t(z)\preceq \overline L I_d$ proves the claim.

    (2) The learned reverse drift is $\widehat b_{n,r}(z)=\beta_{T_n-r}(z+2\widehat s_{n,T_n-r}(z))$. For $r\in[T_n-h_0,T_n]$, put $u=T_n-r\in [0,h_0]$. Since $\widehat s_{n,u}=s_u+e_{n,u}=-\nabla U_u+e_{n,u}$, we have, at every differentiability point of $e_{n,u}$,
    \begin{equation*}
        \nabla \cdot \widehat b_{n,T_n-u}(z)=\beta_u(d-2\Delta U_u(z)+2\nabla \cdot e_{n,u}(z)).
    \end{equation*}
    By part~(1), $\Delta U_u(z)\leq d\overline L$. Moreover, Rademacher's theorem gives
    \begin{equation*}
        \|\nabla e_{n,u}(z)\|_{\op}\leq \Lip(e_{n,u}), \qquad \text{a.e.}
    \end{equation*}
    Therefore,
    \begin{equation*}
        \|(\nabla \cdot \widehat b_{n,T_n-u})_-\|_{\infty}\leq \beta_u (d(2\overline L-1)_++2d\Lip(e_{n,u})).
    \end{equation*}
    Consequently,
    \begin{equation*}
        \mathfrak C_n:=\int_{T_n-h_0}^{T_n}\|(\nabla \cdot \widehat b_{n,r})_-\|_{\infty}\mathrm dr\leq d(2\overline L-1)_+\int_0^{h_0}\beta_u \mathrm du+2d\int_0^{h_0}\beta_u \Lip(e_{n,u})\mathrm du.
    \end{equation*}

    (3) For $u\in[0,h_0]$ and $z,z'\in\mathbb R^d$, we have
\begin{align*}
    &\|\widehat b_{n,T_n-u}(z) -\widehat b_{n,T_n-u}(z')\|\leq \beta_u \left[ \|z-z'\| + 2\|\nabla U_u(z)-\nabla U_u(z')\| + 2\|e_{n,u}(z)-e_{n,u}(z')\| \right].
\end{align*}
By the global Hessian bound from part~(1),
\begin{equation*}
    \|\nabla U_u(z)-\nabla U_u(z')\| \leq \overline L\|z-z'\|.
\end{equation*}
Since $e_{n,u}$ is globally Lipschitz,
\begin{equation*}
    \|e_{n,u}(z)-e_{n,u}(z')\| \leq \Lip(e_{n,u})\|z-z'\|.
\end{equation*}
Therefore, $\|\widehat b_{n,T_n-u}(z) -\widehat b_{n,T_n-u}(z')\| \leq L_n(u)\|z-z'\|$, where $L_n(u) = C\beta_u \left\{ 1+\Lip(e_{n,u}) \right\}$.

Below, we will bound the growth rate of $\widehat b$. Note that
\begin{align*}
    \|\widehat b_{n,T_n-u}(z)\| &\leq \beta_u \left\{ \|z\| + 2\|\nabla U_u(z)\| + 2\|e_{n,u}(z)\| \right\}.
\end{align*}
By the global Hessian bound, $\|\nabla U_u(z)\|\leq\|\nabla U_u(0)\|+\overline L\|z\|$.
On the compact interval $[0,h_0]$, set $C_U=\sup_{0\leq u\leq h_0}\|\nabla U_u(0)\|<\infty$. Moreover,
\begin{equation*}
    \|e_{n,u}(z)\| \leq \|e_{n,u}(0)\| + \Lip(e_{n,u})\|z\|.
\end{equation*}
Hence
\begin{equation*}
    \|\widehat b_{n,T_n-u}(z)\| \leq G_n(u)(1+\|z\|),
\end{equation*}
where
\begin{equation*}
    G_n(u) = C\beta_u \left\{ 1+C_U +\|e_{n,u}(0)\| +\Lip(e_{n,u}) \right\}.
\end{equation*}

It remains to verify that $L_n,G_n\in L^1([0,h_0])$ with the stated stochastic bounds. By Assumption~\ref{ass:score-approximation},
\begin{equation*}
    \sup_{0\leq u\leq h_0} \Lip(e_{n,u}) = O_p(1),
\end{equation*}
and by Assumption~\ref{ass:noise-schedule},
\begin{equation*}
    \int_0^{h_0}\beta_u\,\mathrm du<\infty.
\end{equation*}
Therefore,
\begin{equation*}
    \int_0^{h_0}L_n(u)\,\mathrm du = O_p(1).
\end{equation*}

To control the remaining term in $G_n$, by Lipschitz continuity,
\begin{equation*}
    \|e_{n,u}(0)\| \leq \|e_{n,u}(x)\| + \Lip(e_{n,u})\|x\|
\end{equation*}
for every $x\in\mathbb R^d$. Taking $x=\widehat Y_{T_n-u}$, conditioning on $\mathcal F_n$, and integrating, we obtain
\begin{align*}
    \int_0^{h_0} \beta_u\|e_{n,u}(0)\|\,\mathrm du &\leq \int_0^{h_0} \beta_u \left( \mathbb E \left[ \|e_{n,u}(\widehat Y_{T_n-u})\|^q \mid\mathcal F_n \right] \right)^{1/q} \mathrm du\\ &\quad+ \left( \sup_{0\leq u\leq h_0} \Lip(e_{n,u}) \right) \int_0^{h_0} \beta_u \mathbb E \left[ \|\widehat Y_{T_n-u}\| \mid\mathcal F_n \right] \mathrm du.
\end{align*}
The first term is bounded by $\mathcal A_{n,q}(T_n)$ and hence is $o_p(1)$. By the coupling estimate from Lemma~\ref{fix_wq},
\begin{equation*}
    \sup_{0\leq u\leq h_0} \mathbb E \left[ \|\widehat Y_{T_n-u}\| \mid\mathcal F_n \right] = O_p(1).
\end{equation*}
Consequently,
\begin{equation*}
    \int_0^{h_0} \beta_u \mathbb E \left[ \|\widehat Y_{T_n-u}\| \mid\mathcal F_n \right] \mathrm du = O_p(1),
\end{equation*}
and therefore
\begin{equation*}
    \int_0^{h_0} \beta_u\|e_{n,u}(0)\|\,\mathrm du = O_p(1).
\end{equation*}
It follows that
\begin{equation*}
    \int_0^{h_0}G_n(u)\,\mathrm du = O_p(1).
\end{equation*}

Finally, for every $R>0$,
\begin{equation*}
    \|\widehat b_{n,T_n-u}\|_{L^\infty(B_R)} \leq G_n(u)(1+R),
\end{equation*}
while Rademacher's theorem and the preceding Lipschitz estimate give
\begin{equation*}
    \|\nabla\widehat b_{n,T_n-u}\|_{L^\infty(B_R)} \leq L_n(u)
\end{equation*}
for almost every $u$. Hence
\begin{equation*}
    \int_0^{h_0} \|\widehat b_{n,T_n-u}\|_{W^{1,\infty}(B_R)} \,\mathrm du \leq C_R \int_0^{h_0} \{L_n(u)+G_n(u)\}\,\mathrm du = O_p(1).
\end{equation*}
Changing variables $r=T_n-u$ yields the desired conclusion.
\end{proof}

\begin{mylemma}
    Under Assumptions~\ref{ass:target-distribution}, \ref{ass:noise-schedule}, and~\ref{ass:score-approximation}, let $\rho_{n,r}$ denote the conditional density of $\widehat Y_r$ on $[T_n-h_0,T_n]$. Set $r_0=T_n-h_0$, $a_r=\beta_{T_n-r}$, and $c_{n,r}=\|(\nabla \cdot \widehat b_{n,r})_-\|_{\infty}$. Then, for every $\ell\geq2$ and Lebesgue-a.e.\ pair $(a,b)$ satisfying $r_0<a<b<T_n$,
    \begin{equation*}
        \|\rho_{n,b}\|_{\ell}^{\ell}+\frac {4(\ell -1)}{\ell}\int_a^b a_r \|\nabla \rho_{n,r}^{\ell/2}\|_2^2\mathrm dr\leq \|\rho_{n,a}\|_{\ell}^{\ell}+(\ell -1)\int_a^b c_{n,r}\|\rho_{n,r}\|_{\ell}^{\ell}\mathrm dr.
    \end{equation*}
\end{mylemma}

\begin{proof}
    Conditionally on $\mathcal F_n$, let $\rho_{n,r}$ denote the density of $\widehat Y_r$. On $[T_n-h_0,T_n]$, it satisfies
    \begin{equation*}
        \partial_r \rho_{n,r}=\beta_{T_n-r}\Delta \rho_{n,r}-\nabla \cdot(\widehat b_{n,r}\rho_{n,r}).
    \end{equation*}
    Formally, multiplying the equation by $\ell\rho_{n,r}^{\ell-1}$ gives the desired inequality:
    \begin{equation*}
        \dfrac{\mathrm d}{\mathrm dr}\|\rho_{n,r}\|_{\ell}^{\ell}+\dfrac{4(\ell -1)}{\ell}\beta_{T_n-r}\|\nabla \rho_{n,r}^{\ell/2}\|_{2}^2\leq (\ell -1)\|(\nabla \cdot \widehat b_{n,r})_-\|_{\infty}\|\rho_{n,r}\|_{\ell}^{\ell}.
    \end{equation*}
    We now justify this rigorously by a truncation and renormalization argument. Denote $r_0=T_n-h_0,a_r=\beta_{T_n-r},c_{n,r}=\|(\nabla \cdot \widehat b_{n,r})_-\|_{\infty}$. Fix $\ell \geq 2$. We first establish the truncated energy inequality. For $k>0$, define
    \begin{equation*}
        \Phi_{k,\ell}'(s)=\ell(s\wedge k)^{\ell -1},\Phi_{k,\ell}(0)=0\Rightarrow \Phi_{k,\ell}(s)=
        \begin{cases}
            s^{\ell}&0\leq s\leq k\\\ell k^{\ell -1}s-(\ell -1)k^{\ell}&s>k
        \end{cases}
        ,
    \end{equation*}
    In particular, $0\leq \Phi_{k,\ell}(s)\leq \ell k^{\ell -1}s$. Writing $\Psi_{k,\ell}(s):=s\Phi_{k,\ell}'(s)-\Phi_{k,\ell}(s)$, we have
    \begin{equation*}
        \Psi_{k,\ell}(s)=(\ell -1)(s\wedge k)^{\ell}.
    \end{equation*}
    Moreover, $\Phi_{k,\ell}\in C^{1,1}([0,\infty))$, and 
    \begin{equation*}
        \Phi_{k,\ell}''(s)=\ell(\ell -1)s^{\ell -2}\mathbf 1_{\{0<s<k\}}, \qquad \text{for a.e. }s>0.
    \end{equation*}
    Extend it to $\mathbb R$ by setting it equal to zero on $(-\infty,0]$. Let $\eta_{\delta}$ be a standard nonnegative mollifier and define $F_{k,\ell}^{\delta}=\bar \Phi_{k,\ell}*\eta_{\delta}$ and $\Phi_{k,\ell}^{\delta}(s)=F_{k,\ell}^{\delta}(s) -F_{k,\ell}^{\delta}(0)-(F_{k,\ell}^{\delta})'(0)s$. Then $\Phi_{k,\ell}^\delta$ is smooth and convex, $\Phi_{k,\ell}^\delta(0)=(\Phi_{k,\ell}^{\delta})'(0)=0$, and, locally uniformly for $s\geq0$, $\Phi_{k,\ell}^{\delta}\to \Phi_{k,\ell}$ and $(\Phi_{k,\ell}^{\delta})'\to \Phi_{k,\ell}'$ as $\delta\downarrow0$. For simplicity, write $\rho=\rho_{n,r}$, $b=\widehat b_{n,r}$, and $a_r=\beta_{T_n-r}$. Let $\rho^{\epsilon}=\zeta_{\epsilon}*\rho$, where $\zeta_{\epsilon}$ is a standard mollifier on $\mathbb R^d$. Convolving the Fokker--Planck equation gives
    \begin{equation*}
        \partial_r \rho^{\epsilon}=a_r \Delta \rho^{\epsilon}-\nabla \cdot(b\rho^{\epsilon})+R_{\epsilon},R_{\epsilon}:=\nabla \cdot(b\rho^{\epsilon})-\zeta_{\epsilon}*\nabla \cdot (b\rho).
    \end{equation*}
    By the previous lemma, $b\in L_{\mathrm{loc}}^1(I;W_{\mathrm{loc}}^{1,\infty}(\mathbb R^d))$. Since $\rho$ is a probability density, $\rho\in L_{\mathrm{loc}}^1$. Hence
    \begin{equation*}
        R_{\epsilon}=(b_r \cdot \nabla \rho^{\epsilon}-\zeta_{\epsilon}*(b_r\cdot \nabla \rho))+((\nabla \cdot b_r)\rho^{\epsilon}-\zeta_{\epsilon}*((\nabla \cdot b_r)\rho))
    \end{equation*}
    Fix a ball $B_R$ and suppose that $\epsilon<1$. On the larger ball $B_{R+1}$, we have
    \begin{equation*}
        \begin{aligned}
            b_r(x) \cdot \nabla \rho^{\epsilon}(x) &= \int b_r(x)\cdot \nabla \zeta_{\epsilon}(x-y)\rho (y)\mathrm dy,\\ \zeta_{\epsilon}*(b_r\cdot \nabla \rho) &= \int \zeta_{\epsilon}(x-y)b_r(y)\cdot \nabla \rho(y)\mathrm dy\\ &= \int b_r(y)\nabla \zeta_{\epsilon}(x-y)\rho(y) \mathrm dy -\int \zeta_{\epsilon}(x-y)(\nabla \cdot b_r)(y)\rho(y) \mathrm dy
        \end{aligned}
    \end{equation*}
    The first equality holds in the sense of distributions. Therefore,
    \begin{equation*}
        \begin{aligned}
            Q_{\epsilon} &= \int (b_r(x)-b_r(y))\cdot\nabla \zeta_{\epsilon}(x-y)\rho (y)\mathrm dy+ \int \zeta_{\epsilon}(x-y)(\nabla \cdot b_r)(y)\rho(y)\mathrm dy.
        \end{aligned}
    \end{equation*}
    Local Lipschitz continuity gives
    \begin{equation*}
        \begin{aligned}
            &\int_{B_R} \left|\int (b_r(x)-b_r(y))\cdot\nabla\zeta_{\epsilon}(x-y) \rho(y)\,\mathrm dy\right|\mathrm dx\\ &\leq \|\nabla b_r\|_{L^\infty(B_{R+1})} \int_{B_{R+1}}|\rho(y)| \int_{\mathbb R^d}\|x-y\|\,\|\nabla \zeta_{\epsilon}(x-y)\| \,\mathrm dx\,\mathrm dy\\ &\leq C\|\nabla b_r\|_{L^\infty(B_{R+1})} \|\rho_r\|_{L^1(B_{R+1})}.
        \end{aligned}
    \end{equation*}
    The second term satisfies
    \begin{equation*}
        \begin{aligned}
            &\int_{B_R} \left|\int \zeta_{\epsilon}(x-y)(\nabla \cdot b_r)(y)\rho(y)\mathrm dy\right|\mathrm dx\leq \|\nabla \cdot b_r\|_{L^\infty(B_{R+1})} \|\rho_r\|_{L^1(B_{R+1})}.
        \end{aligned}
    \end{equation*}
    Thus,
    \begin{equation*}
        \|Q_{\epsilon}(r,\cdot)\|_{L^1(B_R)}\leq C\|\nabla b_r\|_{L^\infty(B_{R+1})}\|\rho_r\|_{L^1(B_{R+1})}.
    \end{equation*}
    For smooth compact support $\rho$, we have
    \begin{equation*}
        b_r \cdot \nabla \rho^{\epsilon}\to b_r \cdot \nabla \rho \text{ in }L^1(B_R)
    \end{equation*}
    and $\zeta_{\epsilon}*(b_r \cdot \nabla \rho)\to b_r\cdot\nabla\rho$ in $L^1(B_R)$. Hence $Q_{\epsilon}\to0$. For general $\rho\in L_{\mathrm{loc}}^1$, choose smooth, compactly supported $\rho^m\to\rho$ in $L^1(B_{R+1})$. The uniform bound gives
    \begin{equation*}
        \|Q_{\epsilon}(\rho -\rho^m)\|_{L^1(B_R)}\leq C\|\nabla b_r\|_{L^\infty(B_{R+1})}\|\rho-\rho^m\|_{L^1(B_{R+1})}.
    \end{equation*}
    Letting first $\epsilon\to0$ and then $m\to\infty$ gives
    \begin{equation*}
        Q_{\epsilon}\to 0 \text{ in }L^1(B_R).
    \end{equation*}
    For the second term, we have $M_{\epsilon}=(\nabla \cdot b_r)\rho^{\epsilon}-\zeta_{\epsilon}*((\nabla \cdot b_r)\rho)$. Let $f_r=\nabla \cdot b_r$. Since $b_r \in W_{\mathrm{loc}}^{1,\infty}$, we have $f_r \in L_{\mathrm{loc}}^\infty$. Decompose $M_{\epsilon}=f_r(\rho^{\epsilon}-\rho)+(f_r\rho-\zeta_{\epsilon}*(f_r\rho))$. The first term satisfies
    \begin{equation*}
        \|f_r(\rho^{\epsilon}-\rho)\|_{L^1(B_R)}\leq \|f_r\|_{L^\infty(B_R)}\|\rho^{\epsilon}-\rho\|_{L^1(B_R)}\to 0.
    \end{equation*}
    The second term is
    \begin{equation*}
        \|f_r\rho -\zeta_{\epsilon}*(f_r\rho)\|_{L^1(B_R)}\to 0.
    \end{equation*}
    Since $f_r\rho \in L_{\mathrm{loc}}^1$, we have $M_{\epsilon}\to 0$ in $L^1(B_R)$. Hence $R_{\epsilon}=Q_{\epsilon}+M_{\epsilon}\to 0$ in $L^1(B_R)$. Integrating in time gives
    \begin{equation*}
        \|\rho_r\|_{L^1(B_{R+1})}\leq 1,\int_I \|\nabla b_r\|_{L^\infty(B_{R+1})}\mathrm dr<\infty,
    \end{equation*}
    The dominated convergence theorem therefore gives
    \begin{equation*}
        R_{\epsilon}\to 0 \text{ in }L^1(I\times B_R).
    \end{equation*}
    Since $R>0$ is arbitrary, we have $R_{\epsilon}\to 0$ in $L_{\mathrm{loc}}^1(I\times \mathbb R^d)$. Now fix $H=\Phi_{k,\ell}^{\delta}$. Since $H\in C^\infty(\mathbb R)$, $H$ is convex, $H'$ is bounded, and $H''$ is bounded, we may multiply the mollified equation by $H'(\rho^{\epsilon})\chi_R$. This gives
    \begin{equation*}
        \int \partial_r \rho^{\epsilon}H'(\rho^{\epsilon})\chi_R= a_r\int \Delta \rho^{\epsilon}H'(\rho^{\epsilon})\chi_R -\int \nabla\cdot (b\rho^{\epsilon})H'(\rho^{\epsilon})\chi_R +\int R_{\epsilon}H'(\rho^{\epsilon})\chi_R.
    \end{equation*}
    The time derivative term is
    \begin{equation*}
        \int \partial_r \rho^{\epsilon}H'(\rho^{\epsilon})\chi_R =\dfrac{\mathrm d}{\mathrm dr}\int H(\rho^{\epsilon})\chi_R.
    \end{equation*}
    For the diffusion term, we have
    \begin{equation*}
        \begin{aligned}
            a_r \int \Delta \rho^{\epsilon}H'(\rho^{\epsilon})\chi_R &=-a_r\int H''(\rho^{\epsilon})|\nabla \rho^{\epsilon}|^2 \chi_R-a_r\int \nabla H(\rho^{\epsilon})\cdot \nabla \chi_R\\ &=-a_r \int H''(\rho^{\epsilon})|\nabla \rho^{\epsilon}|^2 \chi_R +a_r \int H(\rho^{\epsilon})\Delta \chi_R.
        \end{aligned}
    \end{equation*}
    For the drift term
    \begin{equation*}
        \begin{aligned}
            -\int \nabla \cdot(b\rho^{\epsilon})H'(\rho^{\epsilon})\chi_R &=\int b\rho^{\epsilon}\cdot \nabla (H'(\rho^{\epsilon})\chi_R)\\ &=\int \rho^{\epsilon}H''(\rho^{\epsilon})b\cdot \nabla \rho^{\epsilon}\chi_R +\int \rho^{\epsilon}H'(\rho^{\epsilon})b\cdot \nabla \chi_R.
        \end{aligned}
    \end{equation*}
    Define $\Psi^{\delta}(s)=sH'(s)-H(s),(\Psi^{\delta})'(s)=sH''(s)$, and hence $\nabla \Psi^{\delta}(\rho^{\epsilon})=\rho^{\epsilon}H''(\rho^{\epsilon})\nabla \rho^{\epsilon}$. Therefore
    \begin{equation*}
        \int \rho^{\epsilon}H''(\rho^{\epsilon})b\cdot \nabla \rho^{\epsilon}\chi_R=\int b \cdot \nabla\Psi^{\delta}(\rho^{\epsilon})\chi_R =-\int (\nabla \cdot b)\Psi^{\delta}(\rho^{\epsilon})\chi_R -\int \Psi^{\delta}(\rho^{\epsilon})b \cdot \nabla \chi_R.
    \end{equation*}
    Combining this with the second drift contribution, and using $\rho^{\epsilon}H'(\rho^{\epsilon})-\Psi^{\delta}(\rho^{\epsilon})=H(\rho^{\epsilon})$, we get
    \begin{equation*}
        -\int \nabla \cdot(b\rho^{\epsilon})H'(\rho^{\epsilon})\chi_R =-\int (\nabla \cdot b)\Psi^{\delta}(\rho^{\epsilon})\chi_R +\int H(\rho^{\epsilon})b\cdot \nabla \chi_R.
    \end{equation*}
    Putting the three terms together yields
    \begin{equation*}
        \begin{aligned}
            \dfrac{\mathrm d}{\mathrm dr}\int H(\rho^{\epsilon})\chi_R +a_r \int H''(\rho^{\epsilon})|\nabla \rho^{\epsilon}|^2 \chi_R &=-\int (\nabla \cdot b)\Psi^{\delta}(\rho^{\epsilon})\chi_R +\int H(\rho^{\epsilon})b\cdot \nabla \chi_R\\ &\quad+a_r \int H(\rho^{\epsilon})\Delta \chi_R +\int R_{\epsilon}H'(\rho^{\epsilon})\chi_R.
        \end{aligned}
    \end{equation*}
    Let $\epsilon \downarrow 0$. Since $H'$ is bounded and $R_{\epsilon}\to 0$ in $L_{\mathrm{loc}}^1$, $\int R_{\epsilon}H'(\rho^{\epsilon})\chi_R \to 0$ in $L_{\mathrm{loc}}^1(I)$. Moreover, $\rho^{\epsilon}\to \rho$ in $L_{\mathrm{loc}}^1(I\times \mathbb R^d)$. For fixed $k,\ell,\delta$, $H$ and $\Psi^{\delta}$ are globally Lipschitz with at most linear growth. Hence $H(\rho^{\epsilon})\to H(\rho)$ and $\Psi^{\delta}(\rho^{\epsilon})\to \Psi^{\delta}(\rho)$ in $L_{\mathrm{loc}}^1$.
    To pass to the dissipation term, set
    \begin{equation*}
        \Gamma_H(s)=\int_0^s\sqrt{H''(v)}\,\mathrm dv.
    \end{equation*}
    Then $|\nabla\Gamma_H(\rho^\epsilon)|^2=H''(\rho^\epsilon)|\nabla\rho^\epsilon|^2$. The preceding identity gives a local $L^2(I;H^1)$ bound for $\Gamma_H(\rho^\epsilon)$. After taking a subsequence, weak compactness and the $L_{\mathrm{loc}}^1$ convergence of $\rho^\epsilon$ identify the weak limit as $\Gamma_H(\rho)$. Weak lower semicontinuity therefore gives
    for every compact time interval $J\Subset I$,
    \begin{equation*}
        \int_J a_r\int H''(\rho)|\nabla\rho|^2\chi_R\,\mathrm dx\,\mathrm dr \leq \liminf_{\epsilon\downarrow0} \int_J a_r\int H''(\rho^\epsilon)|\nabla\rho^\epsilon|^2\chi_R\,\mathrm dx\,\mathrm dr.
    \end{equation*}
    For later use, collect the two cutoff terms in
    \begin{equation*}
        \mathcal R_R(H;r) := \int H(\rho_r)b_r\cdot\nabla\chi_R\,\mathrm dx + a_r\int H(\rho_r)\Delta\chi_R\,\mathrm dx.
    \end{equation*}
    Consequently, we have
    \begin{equation*}
        \begin{aligned}
            \frac{\mathrm d}{\mathrm dr}\int H(\rho_r)\chi_R\,\mathrm dx +a_r\int H''(\rho_r)|\nabla\rho_r|^2\chi_R\,\mathrm dx &\leq -\int(\nabla\cdot b_r)\Psi^\delta(\rho_r)\chi_R\,\mathrm dx +\mathcal R_R(H;r).
        \end{aligned}
    \end{equation*}
    in the sense of distributions in $r$. Equivalently, for Lebesgue-a.e. $a<b$, we have
    \begin{equation*}
        \begin{aligned}
            \int H(\rho_b)\chi_R\,\mathrm dx +\int_a^b a_r\int H''(\rho_r)|\nabla\rho_r|^2 \chi_R\,\mathrm dx\,\mathrm dr &\leq \int H(\rho_a)\chi_R\,\mathrm dx\\ &\quad-\int_a^b\int(\nabla\cdot b_r) \Psi^\delta(\rho_r)\chi_R\,\mathrm dx\,\mathrm dr\\ &\quad+\int_a^b\mathcal R_R(H;r)\,\mathrm dr.
        \end{aligned}
    \end{equation*}
    Finally, substituting back $H=\Phi_{k,\ell}^{\delta}$ and $\Psi^{\delta}=\Psi_{k,\ell}^{\delta}=s(\Phi_{k,\ell}^{\delta})'(s)-\Phi_{k,\ell}^{\delta}(s)$, we obtain the desired localized renormalized inequality:
    \begin{equation*}
        \begin{aligned}
            \int\Phi_{k,\ell}^\delta(\rho_b)\chi_R\,\mathrm dx +\int_a^b a_r\int(\Phi_{k,\ell}^\delta)''(\rho_r) |\nabla\rho_r|^2\chi_R\,\mathrm dx\,\mathrm dr &\leq \int\Phi_{k,\ell}^\delta(\rho_a)\chi_R\,\mathrm dx\\ &\quad-\int_a^b\int(\nabla\cdot b_r) \Psi_{k,\ell}^\delta(\rho_r)\chi_R\,\mathrm dx\,\mathrm dr\\ &\quad+\int_a^b\mathcal R_R(\Phi_{k,\ell}^\delta;r)\,\mathrm dr.
        \end{aligned}
    \end{equation*}
    Letting $\delta\downarrow0$ and applying Fatou's lemma gives
    \begin{equation*}
        \begin{aligned}
            \int\Phi_{k,\ell}(\rho_b)\chi_R\,\mathrm dx +\int_a^b a_r\int\Phi_{k,\ell}''(\rho_r) |\nabla\rho_r|^2\chi_R\,\mathrm dx\,\mathrm dr &\leq \int\Phi_{k,\ell}(\rho_a)\chi_R\,\mathrm dx\\ &\quad-\int_a^b\int(\nabla\cdot b_r) \Psi_{k,\ell}(\rho_r)\chi_R\,\mathrm dx\,\mathrm dr\\ &\quad+\int_a^b\mathcal R_R(\Phi_{k,\ell};r)\,\mathrm dr.
        \end{aligned}
    \end{equation*}
    Since $\Psi_{k,\ell}(\rho)=(\ell -1)(\rho \wedge k)^{\ell}$ and $-\nabla\cdot b_r \leq (\nabla \cdot b_r)_-$, we have
    \begin{equation*}
        \begin{aligned}
            \int\!\Phi_{k,\ell}(\rho_b)\chi_R\,\mathrm dx +\frac{4(\ell-1)}{\ell}\!\int_a^b\!\!\int a_r|\nabla(\rho_r\wedge k)^{\ell/2}|^2\chi_R\,\mathrm dx\,\mathrm dr &\leq \int\Phi_{k,\ell}(\rho_a)\chi_R\,\mathrm dx\\ &\quad+(\ell-1)\int_a^b\!\!\int c_r(\rho_r\wedge k)^\ell\chi_R\,\mathrm dx\,\mathrm dr\\ &\quad+\int_a^b\mathcal R_R(\Phi_{k,\ell};r)\,\mathrm dr.
        \end{aligned}
    \end{equation*}
    where $c_r =\|(\nabla \cdot b_r)_-\|_{\infty}$. Let $R\to \infty$, keeping $k$ fixed, and let $A_R=\{x\in \mathbb R^d:R<\|x\|\leq 2R\}$. We show that both components of $\mathcal R_R(\Phi_{k,\ell};r)$ vanish after integration in $r$. For the diffusion cutoff term, since
    \begin{equation*}
        0\leq \Phi_{k,\ell}(s)\leq \ell k^{\ell -1}s,
    \end{equation*}
    and since $\rho_{n,r}$ is a probability density, we have
    \begin{align*}
        \left|\int_a^b a_r \int_{\mathbb R^d} \Phi_{k,\ell}(\rho_{n,r})\Delta \chi_R\,\mathrm dx\,\mathrm dr\right| &\leq \frac{C\ell k^{\ell -1}}{R^2} \int_a^b a_r \int_{A_R}\rho_{n,r}\,\mathrm dx\,\mathrm dr\\ &\leq \frac{C\ell k^{\ell -1}}{R^2} \int_a^b a_r\,\mathrm dr\to 0.
    \end{align*}
    Next we control the drift cutoff term. Write $u=T_n-r$. Since $\nabla \chi_R$ is supported on $A_R$ and satisfies $\|\nabla \chi_R\|\leq C/R$, and by Hessian bound on $U_u$, $\|\nabla U_u(x)\|\leq C_U+\overline L \|x\|$ for $0\leq u\leq h_0$. Hence on $A_R$, $(\|x\|+\|\nabla U_u(x)\|)/R\leq C$. Then we have
    \begin{equation*}
        \begin{aligned}
            \left|\int_{\mathbb R^d}\Phi_{k,\ell}(\rho_{n,r})\widehat b_{n,r}\cdot \nabla \chi_R \right|&\leq C\ell k^{\ell -1}\beta _u \mathbb P(R<\|\widehat Y_r\|\leq 2R\mid \mathcal F_n)\\ &\qquad+ \dfrac{C\ell k^{\ell -1}\beta_u}R \mathbb E[\|e_{n,u}(\widehat Y_r)\|1_{R<\|\widehat Y_r\|\leq 2R}\mid \mathcal F_n].
        \end{aligned}
    \end{equation*}
    For the first term, for each fixed $r$, $\mathbb P(R<\|\widehat Y_r\|\leq 2R\mid \mathcal F_n)\to 0$ as $R\to\infty$. Since this probability is bounded by \(1\) and $\int_0^{h_0}\beta_u\,\mathrm du<\infty$, the dominated convergence theorem gives
    \begin{equation*}
        \int_a^b \beta_{T_n-r}\mathbb P(R<\|\widehat Y_r\|\leq 2R\mid \mathcal F_n)\mathrm dr\to 0.
    \end{equation*}
    For the second term, H\"older's inequality gives
    \begin{equation*}
        \begin{aligned}
            \dfrac 1R \int_a^b \beta_{T_n-r} \mathbb E[\|e_{n,T_n-r}(\widehat Y_r)\|1_{\{R<\|\widehat Y_r\|\leq 2R\}}\mid \mathcal F_n]\mathrm dr&\leq \dfrac 1R \int_a^b \beta_{T_n-r} \left(\mathbb E[\|e_{n,T_n-r}(\widehat Y_r)\|^q\mid \mathcal F_n]\right)^{1/q}\mathrm dr\\ &\leq \dfrac{\mathcal A_{n,q}(T_n)}{R}\to 0.
        \end{aligned}
    \end{equation*}
    Thus
    \begin{equation*}
        \int_a^b \int_{\mathbb R^d}\Phi_{k,\ell}(\rho_{n,r})\widehat b_{n,r}\cdot \nabla \chi_R \mathrm dr \to 0.
    \end{equation*}
    The dominated convergence theorem and Fatou's lemma yield, as $R\to\infty$,
    \begin{equation*}
        \begin{aligned}
            \int_{\mathbb R^d}\Phi_{k,\ell}(\rho_{n,b}) +\dfrac{4(\ell -1)}{\ell} \int_a^b a_r\|\nabla (\rho_{n,r}\wedge k)^{\ell/2}\|_2^2 \mathrm dr &\leq \int_{\mathbb R^d}\Phi_{k,\ell}(\rho_{n,a})\\ &\quad+(\ell -1)\int_a^b c_{n,r}\|\rho_{n,r}\wedge k\|_{\ell}^{\ell}\mathrm dr.
        \end{aligned}
    \end{equation*}
    Let $k\to\infty$. Since $\Phi_{k,\ell}(s)\uparrow s^{\ell}$ and $(s\wedge k)^{\ell}\uparrow s^{\ell}$, the monotone convergence theorem gives
    \begin{equation*}
        \int_{\mathbb R^d}\Phi_{k,\ell}(\rho_{n,b})\to \|\rho_{n,b}\|_{\ell}^{\ell},\int_{\mathbb R^d}\Phi_{k,\ell}(\rho_{n,a})\to \|\rho_{n,a}\|_{\ell}^{\ell},\|\rho_{n,r}\wedge k\|_{\ell}^{\ell}\to \|\rho_{n,r}\|_{\ell}^{\ell}.
    \end{equation*}
    Lower semicontinuity gives
    \begin{equation*}
        \int_a^b a_r \|\nabla \rho_{n,r}^{\ell/2}\|_2^2 \mathrm dr\leq \liminf_{k\to \infty}\int_a^b a_r\|\nabla (\rho_{n,r}\wedge k)^{\ell/2}\|_2^2 \mathrm dr.
    \end{equation*}
    Therefore
    \begin{equation*}
        \|\rho_{n,b}\|_{\ell}^{\ell}+\dfrac{4(\ell -1)}{\ell}\int_a^b a_r \|\nabla \rho_{n,r}^{\ell/2}\|_2^2 \mathrm dr\leq \|\rho_{n,a}\|_{\ell}^{\ell}+(\ell -1)\int_a^b c_{n,r}\|\rho_{n,r}\|_{\ell}^{\ell}\mathrm dr.
    \end{equation*}
    Equivalently, in the sense of distributions in time, we have
    \begin{equation*}
        \dfrac{\mathrm d}{\mathrm dr}\|\rho_{n,r}\|_{\ell}^{\ell}+\dfrac{4(\ell -1)}{\ell}\beta_{T_n-r}\|\nabla \rho_{n,r}^{\ell/2}\|_2^2 \leq (\ell -1)\|(\nabla \cdot \widehat b_{n,r})_-\|_{\infty}\|\rho_{n,r}\|_{\ell}^{\ell}.
    \end{equation*}
\end{proof}

\begin{mylemma}
\label{lem:L1-L2-smoothing}
Let $u^{(N)}$ be a nonnegative finite-mass solution on $(r_0,r_1)$ such that
\begin{equation*}
    \|u_r^{(N)}\|_1\leq M
\end{equation*}
for every $r\in(r_0,r_1)$. Assume that its canonical $L^2$ energy $\mathcal E_{N,2}$ satisfies, for every $r_0<a<b<r_1$,
\begin{equation*}
    \mathcal E_{N,2}(b) +2\underline\beta\int_a^b\|\nabla u_r^{(N)}\|_2^2\,\mathrm dr \leq \mathcal E_{N,2}(a),
\end{equation*}
where $\mathcal E_{N,2}(r)=\|u_r^{(N)}\|_2^2$ for almost every $r$. Then, for every $b\in(r_0,r_1)$,
\begin{equation*}
    \mathcal E_{N,2}(b)^{1/2} \leq C_d M\{\underline\beta(b-r_0)\}^{-d/4}.
\end{equation*}
\end{mylemma}

\begin{proof}
By Nash's inequality, for almost every $r$,
\begin{equation*}
    \|u_r^{(N)}\|_2^{2+4/d} \leq C_d\|\nabla u_r^{(N)}\|_2^2\|u_r^{(N)}\|_1^{4/d}.
\end{equation*}
Since $\|u_r^{(N)}\|_1\leq M$ and $\mathcal E_{N,2}(r)=\|u_r^{(N)}\|_2^2$ almost everywhere,
\begin{equation*}
    \|\nabla u_r^{(N)}\|_2^2 \geq c_dM^{-4/d}\mathcal E_{N,2}(r)^{1+2/d}
\end{equation*}
for almost every $r$. Consequently,
\begin{equation*}
    \mathcal E_{N,2}(b) +c_d\underline\beta M^{-4/d} \int_a^b\mathcal E_{N,2}(r)^{1+2/d}\,\mathrm dr \leq \mathcal E_{N,2}(a)
\end{equation*}
for every $a<b$. Applying Lemma~\ref{lem:coming-down-from-infinity} with $\alpha=2/d$ gives
\begin{equation*}
    \mathcal E_{N,2}(b) \leq C_dM^2\{\underline\beta(b-r_0)\}^{-d/2},
\end{equation*}
which is the claim.
\end{proof}

\begin{mylemma}
\label{lem:L2-Linf-moser}
Let $s<\tau$ be in $(r_0,r_1)$. Let $u^{(N)}$ be a nonnegative solution for which, for every $\ell\geq2$, there exists a canonical energy $\mathcal E_{N,\ell}$ satisfying
\begin{equation*}
    \mathcal E_{N,\ell}(r)=\|u_r^{(N)}\|_\ell^\ell
\end{equation*}
for almost every $r$, and, for every $s\leq a<b\leq \tau$,
\begin{equation*}
    \mathcal E_{N,\ell}(b) +\frac{4(\ell-1)}{\ell}\underline\beta \int_a^b \left\|\nabla (u_r^{(N)})^{\ell/2}\right\|_2^2\,\mathrm dr \leq \mathcal E_{N,\ell}(a).
\end{equation*}
Then, we have
\begin{equation*}
    \|u_\tau^{(N)}\|_\infty \leq C_d\{\underline\beta(\tau-s)\}^{-d/4} \mathcal E_{N,2}(s)^{1/2}.
\end{equation*}
\end{mylemma}

\begin{proof}
Define
\begin{equation*}
    t_j=\tau-(\tau-s)2^{-j}, \qquad j=0,1,2,\ldots.
\end{equation*}
Then $t_0=s$, $t_j\uparrow\tau$, and $t_{j+1}-t_j=(\tau-s)2^{-j-1}$. Set $A_j^{(N)}=\mathcal E_{N,\ell_j}(t_j)^{1/\ell_j}$. Let $\kappa =1+2/d,\ell_j=2\kappa^j$. Since $4(\ell_j-1)/\ell_j\geq2$, the canonical energy inequality gives
\begin{equation*}
    2\underline\beta \int_{t_j}^{\tau} \left\|\nabla (u_r^{(N)})^{\ell_j/2}\right\|_2^2\,\mathrm dr \leq \mathcal E_{N,\ell_j}(t_j) = \bigl(A_j^{(N)}\bigr)^{\ell_j}.
\end{equation*}
Moreover, for almost every $r\geq t_j$,
\begin{equation*}
    \|u_r^{(N)}\|_{\ell_j}^{\ell_j} =\mathcal E_{N,\ell_j}(r) \leq \mathcal E_{N,\ell_j}(t_j) = \bigl(A_j^{(N)}\bigr)^{\ell_j}.
\end{equation*}
Applying the Gagliardo--Nirenberg inequality to $f=(u_r^{(N)})^{\ell_j/2}$ yields, for almost every $r$,
\begin{equation*}
    \|u_r^{(N)}\|_{\ell_{j+1}}^{\ell_{j+1}} \leq C_d \left\|\nabla (u_r^{(N)})^{\ell_j/2}\right\|_2^2 \|u_r^{(N)}\|_{\ell_j}^{2\ell_j/d}.
\end{equation*}
Therefore,
\begin{equation*}
    \int_{t_j}^{\tau} \|u_r^{(N)}\|_{\ell_{j+1}}^{\ell_{j+1}}\,\mathrm dr \leq \frac{C_d}{\underline\beta} \bigl(A_j^{(N)}\bigr)^{\ell_j+2\ell_j/d} = \frac{C_d}{\underline\beta} \bigl(A_j^{(N)}\bigr)^{\ell_{j+1}}.
\end{equation*}
Since $\mathcal E_{N,\ell_{j+1}}$ is non-increasing and agrees almost everywhere with the actual $L^{\ell_{j+1}}$ energy, for almost every $r\in[t_j,t_{j+1}]$,
\begin{equation*}
    \|u_r^{(N)}\|_{\ell_{j+1}}^{\ell_{j+1}} =\mathcal E_{N,\ell_{j+1}}(r) \geq \mathcal E_{N,\ell_{j+1}}(t_{j+1}).
\end{equation*}
Hence
\begin{equation*}
    (t_{j+1}-t_j) \mathcal E_{N,\ell_{j+1}}(t_{j+1}) \leq \frac{C_d}{\underline\beta} \bigl(A_j^{(N)}\bigr)^{\ell_{j+1}}.
\end{equation*}
Equivalently,
\begin{equation*}
    A_{j+1}^{(N)} \leq \left\{ \frac{C_d2^j}{\underline\beta(\tau-s)} \right\}^{1/\ell_{j+1}} A_j^{(N)}.
\end{equation*}
Iteration gives
\begin{equation*}
    A_m^{(N)} \leq A_0^{(N)} \prod_{j=0}^{m-1} \left\{ \frac{C_d2^j}{\underline\beta(\tau-s)} \right\}^{1/\ell_{j+1}}.
\end{equation*}
Because
\begin{equation*}
    \sum_{j=0}^\infty\frac1{\ell_{j+1}}=\frac d4, \qquad \sum_{j=0}^\infty\frac j{\ell_{j+1}}<\infty,
\end{equation*}
we obtain
\begin{equation*}
    \sup_m A_m^{(N)} \leq C_d\{\underline\beta(\tau-s)\}^{-d/4} \mathcal E_{N,2}(s)^{1/2}.
\end{equation*}
Finally, by monotonicity of the canonical energy,
\begin{equation*}
    \|u_\tau^{(N)}\|_{\ell_j}^{\ell_j} = \mathcal E_{N,\ell_j}(\tau) \leq \mathcal E_{N,\ell_j}(t_j) = \bigl(A_j^{(N)}\bigr)^{\ell_j}.
\end{equation*}
Therefore,
\begin{equation*}
    \|u_\tau^{(N)}\|_{\ell_j} \leq A_j^{(N)} \leq K,
\end{equation*}
where
\begin{equation*}
    K = C_d \{\underline\beta(\tau-s)\}^{-d/4} \mathcal E_{N,2}(s)^{1/2}.
\end{equation*}

We claim that $\|u_\tau^{(N)}\|_\infty\leq K$. Suppose otherwise. Then there exists $\epsilon>0$ such that the set
\begin{equation*}
    A_\epsilon = \left\{ x\in\mathbb R^d: u_\tau^{(N)}(x)>K+\epsilon \right\}
\end{equation*}
has positive Lebesgue measure. Hence, for every $j$,
\begin{equation*}
    \|u_\tau^{(N)}\|_{\ell_j} \geq (K+\epsilon) |A_\epsilon|^{1/\ell_j}.
\end{equation*}
Since $\ell_j\to\infty$, we have $|A_\epsilon|^{1/\ell_j}\to1$\, and therefore
\begin{equation*}
    \liminf_{j\to\infty} \|u_\tau^{(N)}\|_{\ell_j} \geq K+\epsilon,
\end{equation*}
contradicting $\|u_\tau^{(N)}\|_{\ell_j}\leq K$ for every $j$. Thus
\begin{equation*}
    \|u_\tau^{(N)}\|_\infty \leq C_d \{\underline\beta(\tau-s)\}^{-d/4} \mathcal E_{N,2}(s)^{1/2}.
\end{equation*}
\end{proof}

\begin{remark}
In the terminal-density argument below, Lemma~\ref{lem:L2-Linf-moser} is used only at times belonging to a common full-measure set on which all canonical energies equal the corresponding actual $L^{\ell_j}$ norms. No estimate at the exceptional times is needed.
\end{remark}

\begin{mylemma}
    Under Assumptions~\ref{ass:target-distribution}, \ref{ass:noise-schedule}, and~\ref{ass:score-approximation}, using the notation of the previous lemma, there exists a constant $C_d>0$, depending only on the fixed dimension $d$, such that
    \begin{equation*}
        \|\rho_{n,T_n}\|_{\infty} \leq C_d(\underline \beta h_0)^{-d/2} \exp(C_d\mathfrak C_n) \|\rho_{n,T_n-h_0}\|_1.
    \end{equation*}
    In particular, $\|\rho_{n,T_n}\|_{\infty}=O_p(1)$.
\end{mylemma}

\begin{proof}
Write
\begin{equation*}
    r_0=T_n-h_0, \qquad r_1=T_n.
\end{equation*}
Condition throughout on $\mathcal F_n$. Recall that
\begin{equation*}
    a_r=\beta_{T_n-r}, \qquad c_{n,r} = \|(\nabla\cdot\widehat b_{n,r})_-\|_\infty, \qquad C_n(r)=\int_{r_0}^r c_{n,s}\,\mathrm ds.
\end{equation*}
Set
\begin{equation*}
    \mathfrak C_n=C_n(r_1), \qquad u_r=e^{-C_n(r)}\rho_{n,r}, \qquad M=\|\rho_{n,r_0}\|_1.
\end{equation*}

By the preceding drift-regularity result, on \([r_0,r_1]\) the learned drift is globally Lipschitz in the spatial variable with a time-integrable Lipschitz coefficient and has at most linear growth with a time-integrable growth coefficient. Since \(0<\underline\beta\leq a_r\leq\bar\beta\), Lemma~\ref{lem:time-inhomogeneous-sde-wellposed} yields a unique non-explosive strong solution with volatility \(\sqrt{2a_r}\).

By Lemma~\ref{lem:sde-markov-kernel}, this solution generates a time-inhomogeneous Markov evolution
\begin{equation*}
    (P_{s,r})_{r_0\leq s\leq r\leq r_1}.
\end{equation*}
Its dual evolution preserves finite nonnegative mass and is contractive in total variation. In particular, whenever the involved measures have densities,
\begin{equation}
    \|P_{s,r}^*f-P_{s,r}^*g\|_1 \leq \|f-g\|_1, \qquad r_0\leq s\leq r\leq r_1. \label{eq:L1-contraction-markov}
\end{equation}

The original learned process restricted to \([r_0,r_1]\) is a weak solution of this same SDE. Pathwise uniqueness implies uniqueness in law, and hence
\begin{equation}
    \rho_{n,r}\,\mathrm dx = P_{r_0,r}^*(\rho_{n,r_0}\,\mathrm dx), \qquad r_0\leq r\leq r_1. \label{eq:original-semigroup-identification}
\end{equation}
We next verify that the renormalization argument of the preceding energy lemma extends to every nonnegative finite-mass solution of the same Fokker--Planck equation. The only additional point is the cutoff at spatial infinity.

Writing $r=T_n-u$, the global Hessian bound and the global Lipschitz property of $e_{n,u}$ imply
\begin{equation*}
    \|\widehat b_{n,r}(x)\| \leq g_{n,r}(1+\|x\|),
\end{equation*}
where
\begin{equation*}
    g_{n,T_n-u} = C\beta_u \left\{ 1+C_U +\|e_{n,u}(0)\| +\Lip(e_{n,u}) \right\},
\end{equation*}
and, for the conditional realization under consideration,
\begin{equation*}
    \int_{r_0}^{r_1}g_{n,r}\,\mathrm dr<\infty.
\end{equation*}

Let $\varrho_r$ be any nonnegative distributional solution of the same Fokker--Planck equation satisfying $\|\varrho_r\|_1=M_\varrho<\infty$ for almost every $r\in(r_0,r_1)$. Let $\chi_R$ be the standard cutoff such that
\begin{equation*}
    0\leq\chi_R\leq1, \qquad \chi_R=1\text{ on }B_R, \qquad \supp(\chi_R)\subset B_{2R},
\end{equation*}
and
\begin{equation*}
    \|\nabla\chi_R\|\leq\frac CR, \qquad |\Delta\chi_R|\leq\frac C{R^2}.
\end{equation*}
Its gradient is supported on $A_R=\{x\in\mathbb R^d:R<\|x\|<2R\}$. Since $0\leq\Phi_{k,\ell}(s)\leq\ell k^{\ell-1}s$, we have
\begin{align*}
    \left| \int_{\mathbb R^d} \Phi_{k,\ell}(\varrho_r) \widehat b_{n,r}\cdot\nabla\chi_R\,\mathrm dx \right| &\leq C\ell k^{\ell-1}g_{n,r} \int_{A_R} \left( \frac1R+\frac{\|x\|}{R} \right) \varrho_r(x)\,\mathrm dx\\ &\leq C\ell k^{\ell-1}g_{n,r} \left\{ \frac{M_\varrho}{R} + 2\int_{\{\|x\|>R\}}\varrho_r(x)\,\mathrm dx \right\}.
\end{align*}
For almost every $r$, the quantity in braces converges to zero as $R\to\infty$. For $R\geq1$, it is bounded by $3M_\varrho$. Since $g_n\in L^1(r_0,r_1)$, dominated convergence gives
\begin{equation*}
    \int_a^b\int_{\mathbb R^d} \Phi_{k,\ell}(\varrho_r) \widehat b_{n,r}\cdot\nabla\chi_R \,\mathrm dx\,\mathrm dr \longrightarrow0
\end{equation*}
as $R\to\infty$. Similarly,
\begin{align*}
    \left| \int_a^b a_r \int_{\mathbb R^d} \Phi_{k,\ell}(\varrho_r)\Delta\chi_R \,\mathrm dx\,\mathrm dr \right|\leq \frac{C\ell k^{\ell-1}M_\varrho}{R^2} \int_a^b a_r\,\mathrm dr \longrightarrow0.
\end{align*}
All local commutator and mollification arguments use only
\begin{equation*}
    \varrho\in L^1_{\mathrm{loc}}, \qquad \widehat b_n \in L^1_{\mathrm{loc}} \bigl((r_0,r_1);W^{1,\infty}_{\mathrm{loc}}(\mathbb R^d)\bigr),
\end{equation*}
and therefore apply unchanged to $\varrho$. The lower-semicontinuity and truncation limits are also identical. Hence the same renormalized energy inequality holds for every such finite-mass solution.

For $N\geq1$, define the bounded initial densities
\begin{equation*}
    \rho_{n,r_0}^{(N)} = M \frac{ \rho_{n,r_0}\wedge N }{ \displaystyle \int_{\mathbb R^d} (\rho_{n,r_0}(z)\wedge N)\,\mathrm dz }.
\end{equation*}
Then
\begin{equation*}
    \rho_{n,r_0}^{(N)} \in L^1(\mathbb R^d)\cap L^\infty(\mathbb R^d), \qquad \|\rho_{n,r_0}^{(N)}\|_1=M,
\end{equation*}
and
\begin{equation*}
    \rho_{n,r_0}^{(N)} \to \rho_{n,r_0} \qquad \text{in }L^1(\mathbb R^d).
\end{equation*}
Put
\begin{equation*}
    q_N = \frac{M}{\int(\rho_{n,r_0}\wedge N)\,\mathrm dz}.
\end{equation*}
Then \(\rho_{n,r_0}^{(N)}\leq q_N\rho_{n,r_0}\). Positivity of the dual Markov evolution and \eqref{eq:original-semigroup-identification} give
\begin{equation*}
    P_{r_0,r}^*(\rho_{n,r_0}^{(N)}\,\mathrm dx) \leq q_N\rho_{n,r}\,\mathrm dx.
\end{equation*}
Thus the measure on the left has a density, which we denote by \(\rho_{n,r}^{(N)}\). By Itô's formula it is a distributional solution of the same Fokker--Planck equation. Set $u_r^{(N)}= e^{-C_n(r)}\rho_{n,r}^{(N)}.$ Mass preservation implies
\begin{equation*}
    \|u_r^{(N)}\|_1 = e^{-C_n(r)} \|\rho_{n,r}^{(N)}\|_1 \leq M.
\end{equation*}

Since \(\rho_{n,r_0}^{(N)}\in L^\infty\), taking $M_r=\|\rho_{n,r_0}^{(N)}\|_{\infty}\exp(C_n(r))$ as a supersolution, the weak maximum principle gives
\begin{equation}
    \|\rho_{n,r}^{(N)}\|_\infty \leq e^{C_n(r)}\|\rho_{n,r_0}^{(N)}\|_\infty. \label{eq:bounded-initial-maximum-principle}
\end{equation}
Indeed, the right-hand side is the spatially constant supersolution associated with \(-\nabla\cdot\widehat b_{n,r}\leq c_{n,r}\). Consequently, for every \(\ell\geq2\),
\begin{equation}
    \|u_r^{(N)}\|_\ell^\ell \leq M\|\rho_{n,r_0}^{(N)}\|_\infty^{\ell-1}, \label{eq:bounded-initial-energy-finite}
\end{equation}
so the energy is locally integrable in time for each fixed \(N\).

For every $\ell\geq2$, the renormalized energy inequality applied to $\rho_{n,\cdot}^{(N)}$ gives, for Lebesgue-a.e. $r_0<a<b<r_1$,
\begin{align*}
    \|\rho_{n,b}^{(N)}\|_\ell^\ell &+ \frac{4(\ell-1)}{\ell} \int_a^b a_r \left\| \nabla (\rho_{n,r}^{(N)})^{\ell/2} \right\|_2^2\,\mathrm dr\leq \|\rho_{n,a}^{(N)}\|_\ell^\ell + (\ell-1) \int_a^b c_{n,r} \|\rho_{n,r}^{(N)}\|_\ell^\ell\,\mathrm dr.
\end{align*}
Since
\begin{equation*}
    (\ell-1)c_{n,r} \|\rho_{n,r}^{(N)}\|_\ell^\ell \leq \ell c_{n,r} \|\rho_{n,r}^{(N)}\|_\ell^\ell,
\end{equation*}
the integral Gronwall lemma yields
\begin{align*}
    e^{-\ell C_n(b)} \|\rho_{n,b}^{(N)}\|_\ell^\ell &+ \frac{4(\ell-1)}{\ell} \int_a^b e^{-\ell C_n(r)} a_r \left\| \nabla (\rho_{n,r}^{(N)})^{\ell/2} \right\|_2^2\,\mathrm dr\leq e^{-\ell C_n(a)} \|\rho_{n,a}^{(N)}\|_\ell^\ell.
\end{align*}
Equivalently,
\begin{align*}
    \|u_b^{(N)}\|_\ell^\ell &+ \frac{4(\ell-1)}{\ell} \int_a^b a_r \left\| \nabla (u_r^{(N)})^{\ell/2} \right\|_2^2\,\mathrm dr\leq \|u_a^{(N)}\|_\ell^\ell.
\end{align*}
Using $a_r\geq\underline\beta$, we obtain
\begin{align}
    \|u_b^{(N)}\|_\ell^\ell &+ \frac{4(\ell-1)}{\ell} \underline\beta \int_a^b \left\| \nabla (u_r^{(N)})^{\ell/2} \right\|_2^2\,\mathrm dr \leq \|u_a^{(N)}\|_\ell^\ell \label{eq:dissipative-approx-ae}
\end{align}
for Lebesgue-a.e. $r_0<a<b<r_1$.

By \eqref{eq:bounded-initial-energy-finite}, the energy term in \eqref{eq:dissipative-approx-ae} belongs to $L^1_{\mathrm{loc}}$. For any compact subinterval of $(r_0,r_1)$, choose admissible endpoints outside that interval in \eqref{eq:dissipative-approx-ae}; the same inequality then shows that its dissipation term also belongs to $L^1_{\mathrm{loc}}$. Applying Lemma~\ref{lem:canonical-dissipative-representative}, let $\mathcal E_{N,\ell}$ denote the associated canonical energy. Then $\mathcal E_{N,\ell}(r)=\|u_r^{(N)}\|_\ell^\ell$ for Lebesgue-a.e. $r\in(r_0,r_1)$, and for every $r_0<a<b<r_1$,
\begin{align}
    \mathcal E_{N,\ell}(b) &+ \frac{4(\ell-1)}{\ell} \underline\beta \int_a^b \left\| \nabla (u_r^{(N)})^{\ell/2} \right\|_2^2\,\mathrm dr\leq \mathcal E_{N,\ell}(a). \label{eq:canonical-energy-final-proof}
\end{align}

By Lemma~\ref{lem:L1-L2-smoothing}, for every $s\in(r_0,r_1)$,
\begin{equation}
    \mathcal E_{N,2}(s)^{1/2} \leq C_dM \{\underline\beta(s-r_0)\}^{-d/4}, \label{eq:L2-uniform-N}
\end{equation}
uniformly in $N$. Fix $s=r_0+\frac{h_0}{2}$. Let
\begin{equation*}
    \kappa=1+\frac2d, \qquad \ell_j=2\kappa^j.
\end{equation*}
Since $N\in\mathbb N$ and $j\in\mathbb N$ range over countable sets, there exists a common full-measure set $\mathcal T\subset(s,r_1)$ such that $\mathcal E_{N,\ell_j}(\tau)=\|u_\tau^{(N)}\|_{\ell_j}^{\ell_j}$ for every $\tau\in\mathcal T$, every $N$, and every $j$.

For every $\tau\in\mathcal T$, Lemma~\ref{lem:L2-Linf-moser} and \eqref{eq:L2-uniform-N} give
\begin{align*}
    \|u_\tau^{(N)}\|_\infty &\leq C_d \{\underline\beta(\tau-s)\}^{-d/4} \mathcal E_{N,2}(s)^{1/2}\leq C_dM \{\underline\beta(\tau-s)\}^{-d/4} (\underline\beta h_0)^{-d/4},
\end{align*}
uniformly in $N$. In particular, if
\begin{equation*}
    \tau \in \mathcal T \cap \left[ r_0+\frac{3h_0}{4}, r_1 \right),
\end{equation*}
then $\tau-s\geq\frac{h_0}{4}$, and hence
\begin{equation}
    \|u_\tau^{(N)}\|_\infty \leq C_dM (\underline\beta h_0)^{-d/2}. \label{eq:uniform-Linf-N}
\end{equation}

By \eqref{eq:L1-contraction-markov}, for every $r\in[r_0,r_1]$,
\begin{equation*}
    \|\rho_{n,r}^{(N)}-\rho_{n,r}\|_1 \leq \|\rho_{n,r_0}^{(N)}-\rho_{n,r_0}\|_1 \longrightarrow0.
\end{equation*}
Therefore, for every fixed $\tau\in\mathcal T\cap\left[r_0+\frac{3h_0}{4},r_1\right)$, we have $\|u_\tau^{(N)}-u_\tau\|_1\to0$. After passing to a subsequence, $u_\tau^{(N)}(x)\to u_\tau(x)$ for Lebesgue-a.e. $x\in\mathbb R^d$. The uniform estimate \eqref{eq:uniform-Linf-N} therefore implies
\begin{equation*}
    \|u_\tau\|_\infty \leq C_dM (\underline\beta h_0)^{-d/2}.
\end{equation*}
Since $\rho_{n,\tau}=e^{C_n(\tau)}u_\tau$ and $C_n(\tau)\leq C_n(r_1)=\mathfrak C_n$, we obtain
\begin{equation}
    \|\rho_{n,\tau}\|_\infty \leq C_dM (\underline\beta h_0)^{-d/2} e^{\mathfrak C_n} \label{eq:preterminal-density-bound}
\end{equation}
for every $\tau\in\mathcal T\cap\left[r_0+\frac{3h_0}{4},r_1\right).$

Choose a sequence $\tau_k\in\mathcal T,\tau_k\uparrow r_1=T_n.$ By continuity of the learned reverse SDE paths, conditionally on $\mathcal F_n$, $\widehat Y_{\tau_k}\to \widehat Y_{T_n}$ almost surely. Hence
\begin{equation*}
    \mathcal L(\widehat Y_{\tau_k}\mid\mathcal F_n) \Rightarrow \mathcal L(\widehat Y_{T_n}\mid\mathcal F_n) = \widehat\mu_n.
\end{equation*}
Set $B_n=C_dM(\underline\beta h_0)^{-d/2}e^{\mathfrak C_n}.$ By \eqref{eq:preterminal-density-bound}, we have $\|\rho_{n,\tau_k}\|_\infty\leq B_n$. Therefore, for every nonnegative
$\varphi\in C_c(\mathbb R^d)$,
\begin{align*}
    \int_{\mathbb R^d} \varphi(z)\,\widehat\mu_n(\mathrm dz) &= \lim_{k\to\infty} \int_{\mathbb R^d} \varphi(z)\rho_{n,\tau_k}(z)\,\mathrm dz\leq B_n \int_{\mathbb R^d} \varphi(z)\,\mathrm dz.
\end{align*}
Thus $\widehat\mu_n$ is absolutely continuous with density $\rho_{n,T_n}$ and
\begin{equation*}
    \|\rho_{n,T_n}\|_\infty \leq C_d (\underline\beta h_0)^{-d/2} e^{\mathfrak C_n} \|\rho_{n,T_n-h_0}\|_1.
\end{equation*}
After enlarging $C_d$ if necessary,
\begin{equation*}
    \|\rho_{n,T_n}\|_\infty \leq C_d (\underline\beta h_0)^{-d/2} \exp(C_d\mathfrak C_n) \|\rho_{n,T_n-h_0}\|_1.
\end{equation*}
Since $\rho_{n,T_n-h_0}$ is a probability density and $\mathfrak C_n=O_p(1)$, we conclude that
\begin{equation*}
    \|\rho_{n,T_n}\|_\infty=O_p(1).
\end{equation*}
\end{proof}

\begin{mylemma}
    Under Assumptions~\ref{ass:target-distribution}, \ref{ass:noise-schedule}, and~\ref{ass:score-approximation}, conditionally on $\widehat \mu_n$, let $(X_i^*,Y_i^*)\overset{\mathrm{iid}}{\sim}\widehat \mu_n$, and let $\widehat \Sigma_n^*=n^{-1}\sum_{i=1}^n X_i^*X_i^{*\top}$. Then for every fixed $a>0$, there exist a deterministic constant $\eta_a>0$ and a random sequence $C_{n,a}=O_p(1)$ such that, with probability tending to one,
    \begin{equation*}
        \sup_{0<t<\eta_a}t^{-a}\mathbb P^*(\lambda_{\min}(\widehat \Sigma_n^*)\leq t\mid \widehat \mu_n)\leq C_{n,a}.
    \end{equation*}
\end{mylemma}

\begin{proof}
    We first prove an anti-concentration inequality. By assumption,
    \begin{equation*}
        W_q(\widehat \mu_n,\mu)\to_p 0.
    \end{equation*}
    Hence
    \begin{equation*}
        M_{q,n}:=\int_{\mathbb R^d}\|z\|^q \mathrm d\widehat \mu_n(z)=O_p(1),\Sigma_n:=\int xx^{\top}\mathrm d\widehat \mu_n(x,y)\to_p \Sigma_X(\mu).
    \end{equation*}
    Since $\Sigma_X(\mu)$ is positive definite, there exists a deterministic $\lambda_0>0$ such that, with probability tending to one, $\lambda_{\min}(\Sigma_n)\geq \lambda_0$. Fix $v\in \mathbb S^{p-1}$, $b\in \mathbb R$, and $\epsilon\in(0,1)$, and put $R=\epsilon^{-1/(q+d-1)}$. The intersection of the ball $\{z\in \mathbb R^d:\|z\|\leq R\}$ with the slab $\{(x,y):|v^{\top}x-b|\leq \epsilon\}$ has Lebesgue volume at most $C_d\epsilon R^{d-1}$. Hence, by Markov's inequality and $M_{\infty,n}=\|\rho_{n,T_n}\|_{\infty}=O_p(1)$, we have
    \begin{equation*}
        \mathbb P_{\widehat \mu_n}(|v^{\top}X-b|\leq \epsilon)\leq C_dM_{\infty,n}\epsilon R^{d-1}+M_{q,n}R^{-q}\leq K_n\epsilon^{\gamma},
    \end{equation*}
    where
    \begin{equation*}
        \gamma =\dfrac q{q+d-1}>0, \qquad K_n=C_d(1+M_{\infty,n}+M_{q,n})=O_p(1).
    \end{equation*}
    Thus
    \begin{equation*}
        \sup_{\|v\|=1}\sup_{b\in \mathbb R}\mathbb P_{\widehat \mu_n}(|v^{\top}X-b|\leq \epsilon)\leq K_n\epsilon^{\gamma}.
    \end{equation*}
    Conditionally on $\widehat \mu_n$, let $S_m^*=m^{-1}\sum_{i=1}^m X_i^*(X_i^*)^{\top}$, where $m\geq p$. Form the $p\times p$ matrix $A=(X_1^{*\top},X_2^{*\top},\ldots,X_p^{*\top})$. Since $S_m^*=m^{-1}A^{\top}A+m^{-1}\sum_{i=p+1}^m X_i^*X_i^{*\top} \succeq m^{-1}A^{\top}A$, we have
    \begin{equation*}
        \lambda_{\min}(S_m^*)\geq \dfrac 1m s_{\min}(A)^2.
    \end{equation*}
    Consequently,
    \begin{equation*}
        \{\lambda_{\min}(S_m^*)\leq s\}\subset \{s_{\min}(A)\leq \sqrt{ms}\}.
    \end{equation*}
    For $i=1,\ldots,p$, let $H_i=\mathrm{span}\{X_j^*:j\neq i,\ 1\leq j\leq p\}$ and $d_i=\mathrm{dist}(X_i^*,H_i)$. If $A$ is singular, the estimate below is trivial. Otherwise, the negative second-moment identity gives
    \begin{equation*}
        \sum_{j=1}^p s_j(A)^{-2}=\sum_{i=1}^p d_i^{-2}.
    \end{equation*}
    Hence $s_{\min}(A)\leq r$ implies $\min_{1\leq i\leq p}d_i\leq\sqrt p\,r$. Conditionally on all rows except $X_i^*$, the subspace $H_i$ is fixed and proper. Choose a unit vector $v_i$ normal to $H_i$. Since $H_i$ is a linear subspace, $d_i\leq\delta$ implies $|v_i^{\top}X_i^*|\leq\delta$. Therefore, the anti-concentration inequality gives
    \begin{equation*}
        \mathbb P^*(d_i \leq \sqrt p r|\{X_j^*:j\neq i\},\widehat \mu_n)\leq K_n(\sqrt p r)^\gamma.
    \end{equation*}
    Taking a union bound, we have
    \begin{equation*}
        \mathbb P^*(\lambda_{\min}(S_m^*)\leq s\mid \widehat \mu_n)\leq A_n(ms)^{\gamma/2},A_n=O_p(1).
    \end{equation*}
    We also need a bound away from the extreme small-ball regime. Since $p$ is fixed, the matrix Rosenthal inequality and $M_{q,n}=O_p(1)$ give
    \begin{equation*}
        \mathbb E^*[\|S_m^*-\Sigma_n\|_{\op}^{q/2}\mid \widehat \mu_n]\leq H_n m^{-q/4},H_n=O_p(1).
    \end{equation*}
    On the event $\{\lambda_{\min}(\Sigma_n)\geq \lambda_0\}$, if $s<\lambda_0/2$, then Weyl's inequality implies
    \begin{equation*}
        \lambda_{\min}(S_m^*)\leq s\Rightarrow \|S_m^*-\Sigma_n\|_{\op}\geq \lambda_0/2.
    \end{equation*}
    Therefore, by Markov's inequality, we have
    \begin{equation*}
        \mathbb P^*(\lambda_{\min}(S_m^*)\leq s\mid \widehat \mu_n)\leq B_n m^{-q/4},B_n=O_p(1).
    \end{equation*}
    Combining the two inequalities, for $0<s<\lambda_0/2$ we have
    \begin{equation*}
        \mathbb P^*(\lambda_{\min}(S_m^*)\leq s\mid \widehat \mu_n)\leq \min \{A_n(ms)^{\gamma/2},B_n m^{-q/4}\}.
    \end{equation*}
    Fix $a>0$. Choose an integer $k$ sufficiently large that
    \begin{equation*}
        k>\dfrac{2a}{\gamma}+\dfrac{4a}q.
    \end{equation*}
    Set $m=\lfloor n/k\rfloor$ and partition the first $km$ observations into $k$ blocks of size $m$. Let $S_{m,j}^*$ denote the sample covariance in block $j$. Since the remaining observations contribute a positive semidefinite matrix,
    \begin{equation*}
        \widehat \Sigma_n^*\succeq \dfrac mn \sum_{j=1}^k S_{m,j}^*.
    \end{equation*}
    If $\lambda_{\min}(\widehat \Sigma_n^*)\leq t$, then there exists $v\in \mathbb S^{p-1}$ such that $m/n \cdot \sum_{j=1}^k v^{\top}S_{m,j}^*v\leq t$. All terms in the sum are nonnegative. Hence for every $j$, $\lambda_{\min}(S_{m,j}^*)\leq v^{\top}S_{m,j}^* v\leq \dfrac nm t \leq 2kt$ for all sufficiently large $n$. Therefore
    \begin{equation*}
        \{\lambda_{\min}(\widehat \Sigma_n^*)\leq t\}\subset \bigcap_{j=1}^k \{\lambda_{\min}(S_{m,j}^*)\leq 2kt\}.
    \end{equation*}
    The blocks are conditionally independent given $\widehat \mu_n$. Choose $\eta_a=\lambda_0/4k$. For $0<t<\eta_a$, we have
    \begin{equation*}
        \mathbb P^*(\lambda_{\min}(\widehat \Sigma_n^*)\leq t\mid \widehat \mu_n)\leq (\min(\tilde A_n(mt)^{\gamma/2},B_nm^{-q/4}))^k,
    \end{equation*}
    where $\widetilde A_n=O_p(1)$. Let $\theta=2a/(k\gamma)\in(0,1)$. Using $\min \{x,y\}\leq x^{\theta}y^{1-\theta}$, we have
    \begin{equation*}
        \mathbb P^*(\lambda_{\min}(\widehat \Sigma_n^*)\leq t\mid \widehat \mu_n)\leq C_{n,a}t^a m^{a+\frac{aq}{2\gamma}-\frac{kq}{4}},
    \end{equation*}
    where $C_{n,a}=\widetilde A_n^{k\theta}B_n^{k(1-\theta)}=O_p(1)$. Since $k>2a/\gamma+4a/q$, we have
    \begin{equation*}
        a+\dfrac{aq}{2\gamma}-\dfrac{kq}4<0.
    \end{equation*}
    Thus the power of $m$ in the previous equation is nonpositive for all sufficiently large $n$. Hence
    \begin{equation*}
        \mathbb P^*(\lambda_{\min}(\widehat \Sigma_n^*)\leq t\mid \widehat \mu_n) \leq C_{n,a}t^a, \qquad 0<t<\eta_a.
    \end{equation*}
    Since $a>0$ was arbitrary, the displayed bound holds in particular for some $a>2q/(q-4)$.
\end{proof}

\begin{mylemma}
    For every fixed $Q>0$, we have
    \begin{equation*}
        \mathbb E^*[\lambda_{\min}(\widehat \Sigma_n^*)^{-Q}\mid \widehat \mu_n]=O_p(1).
    \end{equation*}
    \label{conc}
\end{mylemma}

\begin{proof}
    Fix $Q>0$ and choose $a>Q$. By the preceding lemma, there exist $\eta_a>0$ and $C_{n,a}=O_p(1)$ such that
    \begin{equation*}
        \mathbb P^*(\lambda_{\min}(\widehat\Sigma_n^*)<t \mid\widehat\mu_n) \leq C_{n,a}t^a, \qquad 0<t<\eta_a.
    \end{equation*}
    Let $Z_n=\lambda_{\min}(\widehat \Sigma_n^*)$. Then
    \begin{equation*}
        \begin{aligned}
            \mathbb E^*[Z_n^{-Q}\mid \widehat \mu_n]&=Q\int_0^\infty t^{-Q-1}\mathbb P^*(Z_n<t\mid \widehat \mu_n)\mathrm dt\\ &=Q\left(\int_0^{\eta_a} t^{-Q-1}\mathbb P^*(Z_n<t\mid \widehat \mu_n)\mathrm dt +\int_{\eta_a}^\infty t^{-Q-1}\mathbb P^*(Z_n<t\mid \widehat \mu_n)\mathrm dt\right)\\ &\leq \eta_a^{-Q}+QC_{n,a}\int_0^{\eta_a} t^{a-Q-1}\mathrm dt =\dfrac{QC_{n,a}}{a-Q}\eta_a^{a-Q}+\eta_a^{-Q}.
        \end{aligned}
    \end{equation*}
    Since $C_{n,a}=O_p(1)$, the result follows.
\end{proof}

\begin{mylemma}
    Under Assumptions~\ref{ass:target-distribution}, \ref{ass:noise-schedule}, and \ref{ass:score-approximation}, we have
    \begin{equation*}
        \mathbb E^*[(R_n^*)^2\mid \widehat \mu_n]\to_p 0,
    \end{equation*}
    where
    \begin{equation*}
        R_n^*=c^{\top}((\widehat \Sigma_n^*)^{-1}-\Sigma_X(\widehat \mu_n)^{-1})U_n^*,\quad U_n^*=\dfrac 1{\sqrt n}\sum_{i=1}^n X_i^*(Y_i^*-X_i^{*\top}\beta(\widehat \mu_n)).
    \end{equation*}
    \label{le7}
\end{mylemma}

\begin{proof}
    Denote $\Sigma_n=\Sigma_X(\widehat \mu_n)$ and $\beta_n=\beta(\widehat \mu_n)$, and define $D_n^*=\|(\widehat \Sigma_n^*)^{-1}-\Sigma_n^{-1}\|_{\op}$. Then $(R_n^*)^2\leq \|c\|^2(D_n^*)^2\|U_n^*\|^2$. Set $s=q/4>1$ and $r=s/(s-1)=q/(q-4)$, and choose any fixed $Q>2r$. Lemmas~\ref{fix_wq} and~\ref{wmm} give
    \begin{equation*}
        \int \|z\|^q\,\mathrm d\widehat \mu_n(z)=O_p(1), \qquad \Sigma_n \to_p \Sigma_X(\mu).
    \end{equation*}
    Since $\Sigma_X(\mu)\succ 0$, with probability tending to one, we have
    \begin{equation*}
        \lambda_{\min}(\Sigma_n)\geq c_0 >0.
    \end{equation*}
    Moreover, conditionally on $\widehat \mu_n$, we have
    \begin{equation*}
        \mathbb E^*[\|\widehat \Sigma_n^*-\Sigma_n\|^2\mid \widehat \mu_n]\leq \mathbb E^*[\|\widehat \Sigma_n^*-\Sigma_n\|_{F}^2\mid \widehat \mu_n]\leq \dfrac Cn \int\|x\|^4 \mathrm d\widehat \mu_n(x,y)=O_p(n^{-1})
    \end{equation*}
    Thus $\widehat \Sigma_n^*\to_{p^*}\Sigma_n$ in conditional probability. Moreover,
    \begin{equation*}
        D_n^*=\|(\widehat\Sigma_n^*)^{-1}-\Sigma_n^{-1}\|_{\op}\leq\|\Sigma_n^{-1}\|_{\op}\| \widehat\Sigma_n^*-\Sigma_n\|_{\op}\|(\widehat\Sigma_n^*)^{-1}\|_{\op}.
    \end{equation*}
    Since $\Sigma_X(\mu)\succ 0$, there exists a constant $c_0>0$ such that, with probability tending to one, $\lambda_{\min}(\Sigma_n)\geq c_0$. On this event, if $\|\widehat \Sigma_n^*-\Sigma_n\|_{\op}\leq c_0/2$, Weyl's inequality gives
    \begin{equation*}
        \lambda_{\min}(\widehat \Sigma^*_n)\geq \lambda_{\min}(\Sigma_n)-\|\widehat \Sigma_n^*-\Sigma_n\|_{\op}\geq \dfrac{c_0}2
    \end{equation*}
    Hence $\widehat \Sigma_n^*$ is invertible, with $\|(\widehat \Sigma_n^*)^{-1}\|_{\op}\leq 2/c_0$ and $\|\Sigma_n^{-1}\|_{\op}\leq 1/c_0$. Therefore,
    \begin{equation*}
        D_n^* \leq \dfrac 2{c_0^2}\|\widehat \Sigma_n^*-\Sigma_n\|_{\op}.
    \end{equation*}
    Since $\|\widehat \Sigma_n^*-\Sigma_n\|_{\op}=o_{p^*}(1)$, it follows that $D_n^*=o_{p^*}(1)$; that is,
    \begin{equation*}
        \mathbb P^*(D_n^*>\epsilon \mid \widehat \mu_n)\to_p 0,\forall \epsilon.
    \end{equation*}
    The triangle inequality gives
    \begin{equation*}
        D_n^* \leq \lambda_{\min}(\widehat \Sigma_n^*)^{-1}+\lambda_{\min}(\Sigma_n)^{-1}
    \end{equation*}
    and hence
    \begin{equation*}
        \mathbb E^*[(D_n^*)^Q\mid \widehat \mu_n]=O_p(1).
    \end{equation*}
    Since $D_n^*=o_{p^*}(1)$ and $Q>2r$, by Hölder's inequality, it follows that
    \begin{equation*}
        \mathbb E^*[(D_n^*)^{2r}\mid \widehat \mu_n]=\mathbb E^*[(D_n^*)^{2r}I_{D_n^*\leq \epsilon}]+\mathbb E^*[(D_n^*)^{2r}I_{D_n^*>\epsilon}]\to_p 0
    \end{equation*}
    Let $\xi_i^*=X_i^*(Y_i^*-X_i^{*\top}\beta_n)$. Conditionally on $\widehat \mu_n$, the $\xi_i^*$'s are i.i.d.\ and have mean zero. Since $2s>2$, Rosenthal's inequality (Lemma~\ref{Rosenthal}) gives
    \begin{equation*}
        \mathbb E^*[\|U_n^*\|^{2s}\mid \widehat \mu_n]\leq C((\mathbb E^*\|\xi_1^*\|^2\mid \widehat \mu_n)^s+n^{1-s}\mathbb E^*[\|\xi_1^*\|^{2s}\mid \widehat \mu_n])
    \end{equation*}
    Since $4s=q$, Lemma~\ref{wmm} and $W_q(\widehat \mu_n,\mu)\to_p0$ imply $\beta_n=\beta(\widehat \mu_n)\to_p\beta(\mu)$. Together with \(\|\xi_1^*\|^{2s} \leq C(1+\|\beta_n\|^{2s})\|(X_1^*,Y_1^*)\|^{4s}\), this yields
    \begin{equation*}
        \mathbb E^*[\|\xi_1^*\|^{2s}\mid \widehat \mu_n]\leq C(1+\|\beta_n\|^{2s})\int \|z\|^q \mathrm d\widehat \mu_n(z)=O_p(1)
    \end{equation*}
    Moreover,
    \begin{equation*}
        (\mathbb E^*[\|\xi_1^*\|^2\mid \widehat \mu_n])^s\leq \mathbb E^*[\|\xi_1^*\|^{2s}\mid \widehat \mu_n],
    \end{equation*}
    so
    \begin{equation*}
        \mathbb E^*[\|U_n^*\|^{2s}\mid \widehat \mu_n]=O_p(1).
    \end{equation*}
    By H\"older's inequality,
    \begin{equation*}
        \begin{aligned}
            \mathbb E^*[(R_n^*)^2 \mid \widehat \mu_n]&\leq \|c\|^2 \mathbb E^*[(D_n^*)^2\|U_n^*\|^2\mid \widehat \mu_n]\\ &\leq \|c\|^2(\mathbb E^*[(D_n^*)^{2r}\mid \widehat \mu_n])^{1/r}(\mathbb E^*[\|U_n^*\|^{2s}\mid \widehat \mu_n])^{1/s}\to_p 0
        \end{aligned}
    \end{equation*}
\end{proof}

\begin{mylemma}[OLS linearization for the original sample]
\label{lem:original-ols-linearization}
Under Assumption~\ref{ass:target-distribution}, let $Z_i=(X_i,Y_i)\overset{\mathrm{iid}}{\sim}\mu$, $i=1,\ldots,n$. Then
\begin{equation*}
    \sqrt n\,c^\top(\widehat\beta-\beta(\mu)) = \frac1{\sqrt n}\sum_{i=1}^n S_\mu(Z_i)+r_n,
\end{equation*}
where
\begin{equation*}
    r_n = c^\top \left( \widehat\Sigma_n^{-1}-\Sigma_X(\mu)^{-1} \right) U_n, \qquad U_n = \frac1{\sqrt n} \sum_{i=1}^n X_i\{Y_i-X_i^\top\beta(\mu)\},
\end{equation*}
and
\begin{equation*}
    \mathbb E r_n^2\to0.
\end{equation*}
Consequently,
\begin{equation*}
    \Var \left( \sqrt n\,c^\top(\widehat\beta-\beta(\mu)) \right) = \tau^2(\mu)+o(1).
\end{equation*}
\end{mylemma}

\begin{proof}
By the normal equation,
\begin{equation*}
    \sqrt n(\widehat\beta-\beta(\mu)) = \widehat\Sigma_n^{-1}U_n.
\end{equation*}
Therefore,
\begin{equation*}
    \begin{aligned}
        \sqrt n\,c^\top(\widehat\beta-\beta(\mu)) &= c^\top\Sigma_X(\mu)^{-1}U_n + c^\top \left( \widehat\Sigma_n^{-1}-\Sigma_X(\mu)^{-1} \right)U_n = \frac1{\sqrt n}\sum_{i=1}^n S_\mu(Z_i)+r_n.
    \end{aligned}
\end{equation*}

It remains to show that $\mathbb E r_n^2\to0$. Set $D_n=\left\|\widehat\Sigma_n^{-1}-\Sigma_X(\mu)^{-1}\right\|_{\op}.$ Then
\begin{equation*}
    r_n^2\leq \|c\|^2D_n^2\|U_n\|^2.
\end{equation*}
Let
\begin{equation*}
    s=\frac q4>1, \qquad r=\frac{s}{s-1} = \frac q{q-4}.
\end{equation*}
By Hölder's inequality,
\begin{equation}
    \mathbb E r_n^2 \leq \|c\|^2 \left(\mathbb E D_n^{2r}\right)^{1/r} \left(\mathbb E\|U_n\|^{2s}\right)^{1/s}. \label{eq:original-rem-holder}
\end{equation}

We first prove $\mathbb E D_n^{2r}\to0$. Since $p$ is fixed and Assumption~\ref{ass:target-distribution} gives a finite $q$-moment with $q>4$,
\begin{equation*}
    \mathbb E \|\widehat\Sigma_n-\Sigma_X(\mu)\|_{\mathrm F}^2 \leq \frac Cn \int\|x\|^4\,\mathrm d\mu(x,y) = O(n^{-1}).
\end{equation*}
Hence $\widehat\Sigma_n-\Sigma_X(\mu)=o_p(1)$. Since $\Sigma_X(\mu)\succ0$, the inverse perturbation identity gives $D_n=o_p(1)$.

Moreover, by the same argument for the lower tail as in the generated Gram matrix lemma, applied now to the fixed law $\mu$, for every fixed $Q>0$ we have $\mathbb E\lambda_{\min}(\widehat\Sigma_n)^{-Q}=O(1)$. Choosing $Q>2r$, and using
\begin{equation*}
    D_n\leq\lambda_{\min}(\widehat\Sigma_n)^{-1}+\lambda_{\min}(\Sigma_X(\mu))^{-1},
\end{equation*}
we get $\mathbb E D_n^Q=O(1)$. Together with $D_n=o_p(1)$, this implies $\mathbb E D_n^{2r}\to0$. Indeed, for every $\epsilon>0$,
\begin{equation*}
    \mathbb E[D_n^{2r}\mathbf 1_{\{D_n\leq \epsilon\}}] \leq \epsilon^{2r},
\end{equation*}
while
\begin{equation*}
    \mathbb E[D_n^{2r}\mathbf 1_{\{D_n>\epsilon\}}] \leq (\mathbb E D_n^Q)^{2r/Q} \mathbb P(D_n>\epsilon)^{1-2r/Q} = o(1).
\end{equation*}
Letting $\epsilon\downarrow0$ gives $\mathbb E D_n^{2r}\to0$.

Next we prove $\mathbb E\|U_n\|^{2s}=O(1)$. Let
\begin{equation*}
    \xi_i=X_i\{Y_i-X_i^\top\beta(\mu)\}.
\end{equation*}
Then
\begin{equation*}
    U_n=\frac1{\sqrt n}\sum_{i=1}^n\xi_i, \qquad \mathbb E\xi_i=0.
\end{equation*}
By Rosenthal's inequality,
\begin{equation*}
    \mathbb E\|U_n\|^{2s} \leq C \left[ \left(\mathbb E\|\xi_1\|^2\right)^s + n^{1-s}\mathbb E\|\xi_1\|^{2s} \right].
\end{equation*}
Moreover,
\begin{equation*}
    \|\xi_1\|^{2s} \leq C(1+\|\beta(\mu)\|^{2s})\|Z_1\|^{4s}.
\end{equation*}
Since $4s=q$, Assumption~\ref{ass:target-distribution} implies $\mathbb E\|\xi_1\|^{2s}<\infty$. Thus $\mathbb E\|U_n\|^{2s}=O(1)$. Combining this with \eqref{eq:original-rem-holder}, we obtain $\mathbb E r_n^2\to0$.

Finally, set
\begin{equation*}
    A_n=\frac1{\sqrt n}\sum_{i=1}^n S_\mu(Z_i).
\end{equation*}
Then $\mathbb E A_n=0$ and $\Var(A_n)=\tau^2(\mu)$. Since $\Var(r_n)\leq \mathbb E r_n^2\to0$ and
\begin{equation*}
    |\Cov(A_n,r_n)| \leq \Var(A_n)^{1/2} \Var(r_n)^{1/2} \to0,
\end{equation*}
we conclude that
\begin{equation*}
    \Var \left( \sqrt n\,c^\top(\widehat\beta-\beta(\mu)) \right) = \Var(A_n+r_n) = \tau^2(\mu)+o(1).
\end{equation*}
\end{proof}

\begin{proof}[Proof of Theorem~\ref{fixed-vari}]
    Conditional on $\widehat \mu_n$, write $\Sigma_n=\Sigma_X(\widehat \mu_n)$ and $\beta_n=\beta(\widehat \mu_n)$. The normal equation gives
    \begin{equation*}
        \sqrt n(\widehat \beta^*-\beta_n)=(\widehat \Sigma_n^*)^{-1}U_n^*.
    \end{equation*}
    Therefore,
    \begin{equation*}
        \begin{aligned}
            \sqrt n c^{\top}(\widehat \beta^*-\beta_n)=c^{\top}\Sigma_n^{-1}U_n^*+c^{\top}((\widehat \Sigma_n^*)^{-1}-\Sigma_n^{-1})U_n^*=\dfrac 1{\sqrt n}\sum_{i=1}^n S_{\widehat \mu_n}(Z_i^*)+R_n^*.
        \end{aligned}
    \end{equation*}
    By Lemma~\ref{le7}, $\mathbb E^*[(R_n^*)^2\mid \widehat \mu_n]\to_p 0$. We next show that
    \begin{equation*}
        \Var^*(\sqrt nc^{\top}(\widehat \beta^*-\beta(\widehat \mu_n))\mid \widehat \mu_n)=\Var^*(\dfrac 1{\sqrt n}\sum_{i=1}^n S_{\widehat \mu_n}(Z_i^*)\mid \widehat \mu_n)+o_p(1).
    \end{equation*}
    Indeed, define $A_n^*=n^{-1/2}\sum_{i=1}^nS_{\widehat\mu_n}(Z_i^*)$. Then $T_n^*:=\sqrt n\,c^\top\{\widehat\beta^*-\beta(\widehat\mu_n)\} =A_n^*+R_n^*$, and
    \begin{equation*}
        \Var^*(A_n^*+R_n^*\mid \widehat \mu_n)=\Var^*(A_n^*\mid \widehat \mu_n)+\Var^*(R_n^*\mid \widehat \mu_n)+2\Cov^*(A_n^*,R_n^*\mid \widehat \mu_n)
    \end{equation*}
    It follows that
    \begin{equation*}
        \Var^*(T_n^*\mid \widehat \mu_n)-\Var^*(A_n^*\mid \widehat \mu_n)=\Var^*(R_n^*\mid \widehat \mu_n)+2\Cov^*(A_n^*,R_n^*\mid \widehat \mu_n).
    \end{equation*}
    Moreover,
    \begin{equation*}
        \Var^*(R_n^*\mid \widehat \mu_n)\leq \mathbb E^*[(R_n^*)^2\mid \widehat \mu_n]\to_p 0
    \end{equation*}
    and
    \begin{equation*}
        |\Cov^*(A_n^*,R_n^*\mid \widehat \mu_n)|\leq (\Var^*(A_n^*\mid \widehat \mu_n))^{1/2}(\Var^*(R_n^*\mid \widehat \mu_n))^{1/2}
    \end{equation*}
    If $W_4(\widehat \mu_n,\mu)\to_p0$, then $\Sigma_X(\widehat\mu_n)\to_p\Sigma_X(\mu)$ and $m_{XY}(\widehat\mu_n)\to_p m_{XY}(\mu)$. Since $\Sigma_X(\mu)$ is positive definite,
    \begin{equation*}
        \beta(\widehat \mu_n)=\Sigma_X(\widehat \mu_n)^{-1}m_{XY}(\widehat \mu_n)\to_p \Sigma_X(\mu)^{-1}m_{XY}(\mu)=\beta(\mu)
    \end{equation*}
    Moreover,
    \begin{equation*}
        \begin{aligned}
            A_n&=\Sigma_X(\widehat\mu_n)^{-1} \to_p\Sigma_X(\mu)^{-1}=A, \qquad \|A\|_{\op}=O(1),\\ b_n&=\beta(\widehat\mu_n) \to_p\beta(\mu)=b, \qquad \|b\|<\infty.
        \end{aligned}
    \end{equation*}
    For every $\epsilon>0$, there exists a deterministic constant $M<\infty$ and $E_n$ with $\mathbb P(E_n)\to 1$, such that on $E_n$, we have
    \begin{equation*}
        \begin{aligned}
            \|A_n\|_{\op},\ \|A\|_{\op},\ \|b_n\|,\ \|b\| &\leq M,\\ \|A_n-A\|_{\op},\ \|b_n-b\| &\leq\epsilon, \qquad \int\|z\|^4\,\mathrm d\widehat\mu_n(z)\leq M.
        \end{aligned}
    \end{equation*}
    On \(E_n\),
    \begin{equation*}
        \begin{aligned}
            |S_{\widehat \mu_n}(z)|&=|c^{\top}A_nx(y-x^\top b_n)|\leq \|c\|\,\|A_n\|_{\op}\,\|x\|(|y|+|x^\top b_n|)\\ &\leq M\|x\|(|y|+M\|x\|)\leq C_M\|z\|^2
        \end{aligned}
    \end{equation*}
    Similarly, $|S_{\mu}(z)|\leq C_M\|z\|^2$. Also,
    \begin{equation*}
        \begin{aligned}
            S_{\widehat \mu_n}(z)-S_{\mu}(z)&=c^{\top}A_nx(y-x^\top b_n)- c^{\top}Ax(y-x^{\top}b)\\ &=c^{\top}(A_n-A)x(y-x^{\top}b_n)+c^\top A x(x^{\top}(b-b_n))
        \end{aligned}
    \end{equation*}
    Therefore
    \begin{equation*}
        |S_{\widehat \mu_n}(z)-S_{\mu}(z)|\leq \|A_n-A\|_{\op}\,\|x\|\,(|y|+\|x\|\,\|b_n\|)+\|A\|_{\op}\|x\|^2\,\|b_n-b\|\leq C_M\epsilon \|z\|^2,
    \end{equation*}
    hence, after changing the constant \(C_M\) if necessary,
    \begin{equation*}
        |S_{\widehat \mu_n}(z)^2-S_{\mu}(z)^2|=|S_{\widehat \mu_n}(z)-S_{\mu}(z)|\,|S_{\widehat \mu_n}(z)+S_{\mu}(z)|\leq C_M\epsilon \|z\|^2\cdot C_M\|z\|^2\leq C_M \epsilon \|z\|^4.
    \end{equation*}
    Integrating with respect to $\widehat \mu_n$ gives
    \begin{equation*}
        \int |S_{\widehat \mu_n}(z)^2-S_{\mu}(z)^2|\mathrm d\widehat \mu_n(z)\leq C_M\epsilon \int \|z\|^4 \mathrm d\widehat \mu_n(z)\leq C_M\epsilon,
    \end{equation*}
    Since $\epsilon>0$ is arbitrary and $\mathbb P(E_n)\to1$,
    \begin{equation*}
        \int S_{\widehat \mu_n}(z)^2 \mathrm d\widehat \mu_n(z)-\int S_{\mu}(z)^2 \mathrm d\widehat \mu_n(z)\to_p 0.
    \end{equation*}
    Since $|S_{\mu}(z)|^2\leq C\|z\|^4$ and
    $W_4(\widehat\mu_n,\mu)\to_p0$,
    \begin{equation*}
        \int S_{\mu}(z)^2 \mathrm d\widehat \mu_n(z)\to_p \int S_{\mu}(z)^2 \mathrm d\mu(z).
    \end{equation*}
    Combining the preceding two limits gives
    \begin{equation*}
        \int S_{\widehat \mu_n}(z)^2 \mathrm d\widehat \mu_n(z)\to_p \int S_{\mu}(z)^2 \mathrm d\mu(z)
    \end{equation*}
    For every admissible \(\nu\),
    \begin{equation*}
        \int S_{\nu}(z)\mathrm d\nu(z)=c^{\top}\Sigma_X(\nu)^{-1}(m_{XY}(\nu)-\Sigma_X(\nu)\beta(\nu))=0,
    \end{equation*}
    and hence
    \begin{equation*}
        \tau^2(\nu)=\Var_{\nu}(S_{\nu}(Z))=\int S_{\nu}(z)^2 \mathrm d\nu(z)\Rightarrow \tau^2(\widehat \mu_n)\to_p \tau^2(\mu).
    \end{equation*}
    Therefore,
    \begin{equation*}
        \Var^*(A_n^*\mid \widehat \mu_n)=\tau^2(\widehat \mu_n)\to_p \tau^2(\mu)\Rightarrow \Var^*(A_n^*\mid \widehat \mu_n)=O_p(1).
    \end{equation*}
    Consequently,
    \begin{equation*}
        |\Cov^*(A_n^*,R_n^*\mid \widehat \mu_n)|=o_p(1).
    \end{equation*}
    We conclude that
    \begin{equation*}
        \Var^*(A_n^*+R_n^*\mid \widehat \mu_n)=\Var^*(A_n^*\mid \widehat \mu_n)+o_p(1).
    \end{equation*}
    Finally,
    \begin{equation*}
        \begin{aligned}
            \Var^*(A_n^*\mid\widehat\mu_n) &=\Var^*\!\left( \frac1{\sqrt n}\sum_{i=1}^nS_{\widehat\mu_n}(Z_i^*) \mathrel{\Big|}\widehat\mu_n\right)\\ &=\frac1n\sum_{i=1}^n \Var^*\!\left(S_{\widehat\mu_n}(Z_i^*)\mid\widehat\mu_n\right)\\ &=\Var_{\widehat\mu_n}\!\left(S_{\widehat\mu_n}(Z)\right)\\ &=\tau^2(\widehat\mu_n).
        \end{aligned}
    \end{equation*}
    Thus,
    \begin{equation*}
        \Var^*(\sqrt n c^{\top}(\widehat \beta^*-\beta(\widehat \mu_n))\mid \widehat \mu_n)=\tau^2(\mu)+o_p(1).
    \end{equation*}
    By Lemma~\ref{lem:original-ols-linearization},
    \begin{equation*}
        \Var \left( \sqrt n\,c^\top\{\widehat\beta-\beta(\mu)\} \right) = \tau^2(\mu)+o(1).
    \end{equation*}
    Since $\tau^2(\mu)>0$, the desired variance-ratio consistency follows.
\end{proof}

\begin{proof}[Proof of Corollary]
    By the $L^2$ linearization,
    \begin{equation*}
        \sqrt n\,c^\top \{\widehat\beta^*-\beta(\widehat\mu_n)\} = \frac1{\sqrt n} \sum_{i=1}^nS_{\widehat\mu_n}(Z_i^*) + o_{p^*}(1),
    \end{equation*}
    and $\tau^2(\widehat\mu_n)\to_p\tau^2(\mu)$. Choose $\delta>0$ so that $2(2+\delta)\leq q$. The $W_q$ convergence implies the conditional Lyapunov condition and hence
    \begin{equation*}
        \mathcal L^*\!\left( \sqrt n\,c^\top \{\widehat\beta^*-\beta(\widehat\mu_n)\} \,\middle|\, \widehat\mu_n \right) \Rightarrow_p N\{0,\tau^2(\mu)\}.
    \end{equation*}
    Indeed, the Lyapunov ratio satisfies
    \begin{equation*}
        \begin{aligned}
            L_n^* &= \dfrac{ \sum_{i=1}^n \mathbb E^*[|n^{-1/2}S_{\widehat \mu_n}(Z^*_i)|^{2+\delta}\mid \widehat \mu_n] }{(\tau^2(\widehat \mu_n))^{1+\delta/2}}\\ &= \dfrac{ n\cdot n^{-(2+\delta)/2} \mathbb E^*[|S_{\widehat \mu_n}(Z^*)|^{2+\delta}\mid \widehat \mu_n] }{(\tau^2(\widehat \mu_n))^{1+\delta/2}}\\ &= \dfrac{O_p(1)n^{-\delta/2}}{(\tau^2(\widehat \mu_n))^{1+\delta/2}}.
        \end{aligned}
    \end{equation*}
    Since $\tau^2(\widehat\mu_n)\to_p\tau^2(\mu)>0$, it follows that $L_n^*\to_p 0$. Similarly,
    \begin{equation*}
        \sqrt n c^{\top}(\widehat \beta-\beta(\mu))\Rightarrow N(0,\tau^2(\mu)).
    \end{equation*}
    Because the Gaussian limit has a continuous distribution function, Lemma~\ref{lem:conditional-polya} and the triangle inequality give
    \begin{equation*}
        \sup_{t\in \mathbb R}|\mathbb P^*(\sqrt nc^{\top}(\widehat \beta^*-\beta(\widehat \mu_n))\leq t\mid \widehat \mu_n)-\mathbb P(\sqrt n c^{\top}(\widehat \beta-\beta(\mu))\leq t)|\to_p 0.
    \end{equation*}
\end{proof}

\section{Proof of Theorem~\ref{thm:hd-actual-var}}
In this section, we prove Theorem~\ref{thm:hd-actual-var} under the general covariance and strong log-concavity assumptions.

\begin{mylemma}
\label{lem:hd-general-curvature}
Under Assumption~\ref{ass:hd-setting}, the joint law of \(Z_0=(X,Y)\) has density proportional to \(\exp\{-U_n(x,y)\}\), where
\begin{equation*}
    U_n(x,y) = V_n(x) + \frac{(y-x^\top\beta_n)^2}{2\sigma_n^2}.
\end{equation*}
There exist constants \(0<m_0\leq L_0<\infty\), independent of \(n\), such that
\begin{equation*}
    m_0I_{d_n} \preceq \nabla^2U_n(x,y) \preceq L_0I_{d_n}.
\end{equation*}
Moreover, every OU marginal has a potential \(U_{n,t}\) satisfying
\begin{equation}
    m_{\mathrm{OU}}I_{d_n} \preceq \nabla^2U_{n,t}(z) \preceq L_{\mathrm{OU}}I_{d_n}, \label{eq:ou-curvature}
\end{equation}
where the constants are independent of \(n\) and \(t\). Consequently, the exact score \(s_{n,t}=-\nabla U_{n,t}\) is uniformly globally Lipschitz, and
\begin{align}
    \Var_{p_{n,t}}(g) &\leq C_{\mathrm{LS}} \int_{\mathbb R^{d_n}} \|\nabla g(z)\|_2^2p_{n,t}(z)\,\mathrm dz, \label{eq:poincare}\\ \operatorname{Ent}_{p_{n,t}}(g^2) &\leq 2C_{\mathrm{LS}} \int_{\mathbb R^{d_n}} \|\nabla g(z)\|_2^2p_{n,t}(z)\,\mathrm dz. \label{eq:lsi}
\end{align}
\end{mylemma}

\begin{proof}
The covariance bounds and the relative curvature assumption give
\begin{equation*}
    \frac{1}{C_\Sigma}I_{p_n} \preceq \Sigma_n^{-1} \preceq \frac{1}{c_\Sigma}I_{p_n}.
\end{equation*}
Hence
\begin{equation*}
    m_XI_{p_n} \preceq \nabla^2V_n(x) \preceq L_XI_{p_n}, \qquad m_X:=\frac{\underline m}{C_\Sigma}, \qquad L_X:=\frac{\overline m}{c_\Sigma}.
\end{equation*}
For \((u,v)\in\mathbb R^{p_n}\times\mathbb R\), direct differentiation gives
\begin{equation}
    \begin{aligned}
        \begin{pmatrix}
            u\\v
        \end{pmatrix}
        ^{\!\top} \nabla^2U_n(x,y)
        \begin{pmatrix}
            u\\v
        \end{pmatrix}
        &= u^\top\nabla^2V_n(x)u + \frac{(v-\beta_n^\top u)^2}{\sigma_n^2}.
    \end{aligned}
    \label{eq:joint-hessian-quadratic}
\end{equation}

Since \(\|\beta_n\|_2\leq B\),
\begin{equation*}
    \begin{aligned}
        v^2 &= \bigl\{v-\beta_n^\top u+\beta_n^\top u\bigr\}^2\\ &\leq 2(v-\beta_n^\top u)^2 + 2(\beta_n^\top u)^2\\ &\leq 2(v-\beta_n^\top u)^2 + 2B^2\|u\|_2^2.
    \end{aligned}
\end{equation*}
Therefore,
\begin{equation*}
    \|u\|_2^2+v^2 \leq (1+2B^2)\|u\|_2^2 + 2(v-\beta_n^\top u)^2.
\end{equation*}
Using \(\sigma_n^2\leq\sigma_{\max}^2\) in
\eqref{eq:joint-hessian-quadratic}, we obtain
\begin{equation*}
    \begin{aligned}
        \begin{pmatrix}
            u\\v
        \end{pmatrix}
        ^{\!\top} \nabla^2U_n(x,y)
        \begin{pmatrix}
            u\\v
        \end{pmatrix}
        &\geq m_X\|u\|_2^2 + \frac{(v-\beta_n^\top u)^2}{\sigma_{\max}^2}\\ &\geq m_0(\|u\|_2^2+v^2),
    \end{aligned}
\end{equation*}
where one may take
\begin{equation*}
    m_0 = \min\left\{ \frac{m_X}{1+2B^2}, \frac{1}{2\sigma_{\max}^2} \right\}.
\end{equation*}

For the upper bound,
\begin{equation*}
    (v-\beta_n^\top u)^2 \leq 2v^2+2B^2\|u\|_2^2.
\end{equation*}
Using \(\sigma_n^2\geq\sigma_{\min}^2\) in
\eqref{eq:joint-hessian-quadratic}, we obtain
\begin{equation*}
    \begin{aligned}
        \begin{pmatrix}
            u\\v
        \end{pmatrix}
        ^{\!\top} \nabla^2U_n(x,y)
        \begin{pmatrix}
            u\\v
        \end{pmatrix}
        &\leq L_X\|u\|_2^2 + \frac{ 2v^2+2B^2\|u\|_2^2 }{\sigma_{\min}^2}\\ &\leq L_0(\|u\|_2^2+v^2),
    \end{aligned}
\end{equation*}
where
\begin{equation*}
    L_0 = \max\left\{ L_X+\frac{2B^2}{\sigma_{\min}^2}, \frac{2}{\sigma_{\min}^2} \right\}.
\end{equation*}
Thus
\begin{equation*}
    m_0I_{d_n} \preceq \nabla^2U_n(x,y) \preceq L_0I_{d_n}.
\end{equation*}

We next propagate these bounds along the OU flow. Fix \(t>0\) and write
\begin{equation*}
    \alpha_t=e^{-t}, \qquad \tau_t^2=1-e^{-2t}, \qquad R=\alpha_tZ_0, \qquad Z_t=R+\tau_t\xi,
\end{equation*}
where \(\xi\sim N(0,I_{d_n})\) is independent of \(Z_0\).

Up to an additive constant, the potential of \(R\) is
\begin{equation*}
    \widetilde U_t(r) = U_n(r/\alpha_t).
\end{equation*}
Consequently,
\begin{equation}
    \frac{m_0}{\alpha_t^2}I_{d_n} \preceq \nabla^2\widetilde U_t(r) \preceq \frac{L_0}{\alpha_t^2}I_{d_n}. \label{eq:scaled-curvature}
\end{equation}

Let \(\pi_{t,z}\) denote the conditional distribution of \(R\) given \(Z_t=z\). Its potential is
\begin{equation*}
    H_{t,z}(r) = \widetilde U_t(r) + \frac{\|z-r\|_2^2}{2\tau_t^2}.
\end{equation*}
Hence
\begin{equation}
    K_-I_{d_n} \preceq \nabla^2H_{t,z}(r) \preceq K_+I_{d_n}, \label{eq:posterior-curvature}
\end{equation}
where
\begin{equation*}
    K_- = \frac{m_0}{\alpha_t^2} + \frac{1}{\tau_t^2}, \qquad K_+ = \frac{L_0}{\alpha_t^2} + \frac{1}{\tau_t^2}.
\end{equation*}

By the second-order Tweedie formula for Gaussian convolution,
\begin{equation}
    \nabla^2U_{n,t}(z) = \frac{1}{\tau_t^2}I_{d_n} - \frac{1}{\tau_t^4} \Cov_{\pi_{t,z}}(R). \label{eq:gaussian-channel-hessian}
\end{equation}

The lower Hessian bound in \eqref{eq:posterior-curvature} and the Brascamp--Lieb inequality give
\begin{equation}
    \Cov_{\pi_{t,z}}(R) \preceq K_-^{-1}I_{d_n}. \label{eq:posterior-cov-upper}
\end{equation}
On the other hand, the matrix Cram\'er--Rao inequality and the upper Hessian bound in \eqref{eq:posterior-curvature} yield
\begin{equation}
    \begin{aligned}
        \Cov_{\pi_{t,z}}(R) &\succeq \left\{ \mathbb E_{\pi_{t,z}} \nabla^2H_{t,z}(R) \right\}^{-1}\\ &\succeq K_+^{-1}I_{d_n}.
    \end{aligned}
    \label{eq:posterior-cov-lower}
\end{equation}

Substituting \eqref{eq:posterior-cov-upper} into \eqref{eq:gaussian-channel-hessian}, we obtain
\begin{equation*}
    \begin{aligned}
        \nabla^2U_{n,t}(z) &\succeq \left( \frac{1}{\tau_t^2} - \frac{1}{\tau_t^4}K_-^{-1} \right)I_{d_n}\\ &= \left( \frac{\alpha_t^2}{m_0} + \tau_t^2 \right)^{-1}I_{d_n}.
    \end{aligned}
\end{equation*}
Similarly, substituting \eqref{eq:posterior-cov-lower} into \eqref{eq:gaussian-channel-hessian} gives
\begin{equation*}
    \nabla^2U_{n,t}(z) \preceq \left( \frac{\alpha_t^2}{L_0} + \tau_t^2 \right)^{-1}I_{d_n}.
\end{equation*}
Therefore,
\begin{equation}
    \left( \frac{\alpha_t^2}{m_0} + \tau_t^2 \right)^{-1}I_{d_n} \preceq \nabla^2U_{n,t}(z) \preceq \left( \frac{\alpha_t^2}{L_0} + \tau_t^2 \right)^{-1}I_{d_n}. \label{eq:ou-time-dependent-curvature}
\end{equation}

The same bounds hold at \(t=0\) by the already established estimates for \(U_n\). Since
\begin{equation*}
    \alpha_t^2+\tau_t^2=1,
\end{equation*}
we have
\begin{equation*}
    \left( \frac{\alpha_t^2}{m_0} + \tau_t^2 \right)^{-1} \geq \min\{m_0,1\},
\end{equation*}
and
\begin{equation*}
    \left( \frac{\alpha_t^2}{L_0} + \tau_t^2 \right)^{-1} \leq \max\{L_0,1\}.
\end{equation*}
Thus one may take
\begin{equation*}
    m_{\mathrm{OU}} = \min\{m_0,1\}, \qquad L_{\mathrm{OU}} = \max\{L_0,1\},
\end{equation*}
which proves \eqref{eq:ou-curvature}.

Since
\begin{equation*}
    s_{n,t}(z) = \nabla\log p_{n,t}(z) = -\nabla U_{n,t}(z),
\end{equation*}
we have
\begin{equation*}
    \|\nabla s_{n,t}(z)\|_{\op} = \|\nabla^2U_{n,t}(z)\|_{\op} \leq L_{\mathrm{OU}}.
\end{equation*}
Hence
\begin{equation*}
    \sup_{n\geq1}\sup_{t\geq0} \Lip(s_{n,t}) \leq L_{\mathrm{OU}}.
\end{equation*}

Finally, the lower curvature bound
\begin{equation*}
    \nabla^2U_{n,t} \succeq m_{\mathrm{OU}}I_{d_n}
\end{equation*}
implies the log-Sobolev inequality
\begin{equation*}
    \operatorname{Ent}_{p_{n,t}}(g^2) \leq \frac{2}{m_{\mathrm{OU}}} \int_{\mathbb R^{d_n}} \|\nabla g(z)\|_2^2p_{n,t}(z)\,\mathrm dz.
\end{equation*}
Its linearization gives
\begin{equation*}
    \Var_{p_{n,t}}(g) \leq \frac{1}{m_{\mathrm{OU}}} \int_{\mathbb R^{d_n}} \|\nabla g(z)\|_2^2p_{n,t}(z)\,\mathrm dz.
\end{equation*}
Thus one may choose
\begin{equation*}
    C_{\mathrm{LS}} = m_{\mathrm{OU}}^{-1}.
\end{equation*}

\end{proof}
\begin{mylemma}
\label{lem:hd-general-weighted}
Let $\rho$ satisfy a log-Sobolev inequality with constant $C_{\mathrm{LS}}$.  Let $e:\mathbb R^d\to\mathbb R^d$ satisfy $\Lip(e)\leq L_\star$, put $g=\|e\|_2$, and set
\begin{equation*}
    a=\int g^4\rho\,\mathrm dz.
\end{equation*}
If~\eqref{eq:lipschitz-absorption} holds, then
\begin{equation}
    \int g^2h^2\rho\,\mathrm dz \leq \frac{4C_{\mathrm{LS}}}{\lambda_\star} \int\|\nabla h\|_2^2\rho\,\mathrm dz +Ca^{1/2}\int h^2\rho\,\mathrm dz \label{eq:weighted}
\end{equation}
for all $h\in H^1(\rho)$.
\end{mylemma}

\begin{proof}
Let $\bar g=\mathbb E_\rho g$ and $Y=g-\bar g$. Since $g$ is $L_\star$-Lipschitz, the Herbst part of Lemma~\ref{lem:aux-bakry-emery-herbst} gives
\begin{equation*}
    \mathbb E_\rho e^{tY} \leq \exp\left(\frac{C_{\mathrm{LS}}L_\star^2t^2}{2}\right).
\end{equation*}
Gaussian randomization and $4\lambda_\star C_{\mathrm{LS}}L_\star^2<1$ give
\begin{equation}
    \mathbb E_\rho e^{2\lambda_\star Y^2} \leq \left(1-4\lambda_\star C_{\mathrm{LS}}L_\star^2\right)^{-1/2}. \label{eq:square-exp}
\end{equation}
Moreover, $\bar g^2\leq a^{1/2}$ and $\mathbb E_\rho Y^4\leq16a$.  Using $e^x-1\leq xe^x$, Cauchy--Schwarz, and~\eqref{eq:square-exp},
\begin{equation*}
    \log\mathbb E_\rho e^{\lambda_\star Y^2} \leq Ca^{1/2}.
\end{equation*}
The entropy variational inequality and the log-Sobolev inequality imply
\begin{equation*}
    \int Y^2h^2\rho \leq \frac{2C_{\mathrm{LS}}}{\lambda_\star} \int\|\nabla h\|_2^2\rho +Ca^{1/2}\int h^2\rho.
\end{equation*}
Finally, $g^2\leq2Y^2+2\bar g^2$, which proves~\eqref{eq:weighted}.
\end{proof}

\begin{mylemma}
\label{lem:hd-reverse-regularity}
Fix $n$, condition on $\mathcal F_n$, and work on the event
\begin{equation}
    \begin{aligned}
        \mathcal E_n = \Biggl\{ &\sup_{0\leq t\leq T_n}\Lip(e_{n,t})\leq L_\star,\int_0^{T_n}\int_{\mathbb R^{d_n}} \|e_{n,t}(z)\|_2^4p_{n,t}(z)\,\mathrm dz\,\mathrm dt <\infty \Biggr\}.
    \end{aligned}
    \label{eq:regularity-event}
\end{equation}
Set $T=T_n$ and, for $0\leq r\leq T$,
\begin{equation*}
    b_{n,r}(z)=z+2s_{n,T-r}(z), \qquad \widehat b_{n,r}(z)=b_{n,r}(z)+2e_{n,T-r}(z).
\end{equation*}
More explicitly, conditionally on $\mathcal F_n$, the exact and learned reverse processes are initialized according to
\begin{align*}
    \mathrm dY_r &=b_{n,r}(Y_r)\,\mathrm dr+\sqrt{2}\,\mathrm dW_r, &Y_0&\sim p_{n,T},\\ \mathrm d\widehat Y_r &=\widehat b_{n,r}(\widehat Y_r)\,\mathrm dr +\sqrt{2}\,\mathrm d\widehat W_r, &\widehat Y_0&\sim N(0,I_{d_n}).
\end{align*}
Then both reverse SDEs have unique non-explosive strong solutions and generate time-dependent Markov evolutions. Their dual evolutions preserve finite nonnegative mass and are contractions in total variation. The marginal laws have densities $\rho_{n,r}$ and $\widehat\rho_{n,r}$, where
\begin{equation*}
    \rho_{n,r}=p_{n,T-r}.
\end{equation*}
The exact density $\rho_{n,r}$ is strictly positive. The densities admit representatives such that
\begin{equation*}
    \rho_n,\widehat\rho_n \in C([0,T];L^1(\mathbb R^{d_n})) \cap L^\infty((0,T);L^2(\mathbb R^{d_n})) \cap L^2((0,T);H^1(\mathbb R^{d_n})).
\end{equation*}
In particular, their marginal curves are narrowly continuous. Moreover, for every $0<\tau<R<T$ and $L<\infty$,
\begin{equation*}
    \partial_r\rho_n,\ \partial_r\widehat\rho_n \in L^2((\tau,R);H^{-1}(B_L)),
\end{equation*}
and
\begin{equation*}
    f_n=\frac{\widehat\rho_n}{\rho_n} \in L^2_{\mathrm{loc}} ((0,T);H^1_{\mathrm{loc}}(\mathbb R^{d_n})).
\end{equation*}
Moreover, for
\begin{equation*}
    f_{n,r}=\frac{\widehat\rho_{n,r}}{\rho_{n,r}}, \qquad D_n(r)=\int f_{n,r}^2\rho_{n,r}\,\mathrm dz,
\end{equation*}
if $0\leq a<R<T$ and $D_n(a)<\infty$, then $D_n$ admits an absolutely continuous representative on $[a,R]$. For almost every $r\in(a,R)$, this representative satisfies
\begin{equation}
    D_n'(r) = -2\int\|\nabla f_{n,r}\|_2^2\rho_{n,r}\,\mathrm dz +4\int f_{n,r}e_{n,T-r}^\top\nabla f_{n,r} \rho_{n,r}\,\mathrm dz. \label{eq:regularity-density-ratio}
\end{equation}
\end{mylemma}

\begin{proof}
Fix $n$, condition on $\mathcal F_n$, and suppress $n$ from the notation when no ambiguity can arise. Constants in this proof may depend on this fixed $n$.

\medskip
\noindent
\textbf{Step 1: Lipschitz and linear-growth bounds.}

By Lemma~\ref{lem:hd-general-curvature},
\begin{equation*}
    \sup_{0\leq t\leq T}\Lip(s_{n,t}) \leq L_{\mathrm{OU}}.
\end{equation*}
Hence
\begin{equation*}
    \Lip(b_{n,r})\leq1+2L_{\mathrm{OU}}, \qquad \Lip(\widehat b_{n,r}) \leq1+2L_{\mathrm{OU}}+2L_\star.
\end{equation*}
Let $Z_t\sim p_{n,t}$. Since the score has mean zero under its own density,
\begin{align*}
    \|s_{n,t}(0)\|_2 &= \left\| \int\{s_{n,t}(0)-s_{n,t}(z)\}p_{n,t}(z)\,\mathrm dz \right\|_2\\ &\leq L_{\mathrm{OU}}\mathbb E\|Z_t\|_2.
\end{align*}
Moreover,
\begin{equation*}
    \|e_{n,t}(0)\|_2 \leq \mathbb E\|e_{n,t}(Z_t)\|_2 +L_\star\mathbb E\|Z_t\|_2.
\end{equation*}
By $(x+y)^4\leq8(x^4+y^4)$ and Jensen's inequality,
\begin{align*}
    \|e_{n,t}(0)\|_2^4 &\leq 8\bigl(\mathbb E\|e_{n,t}(Z_t)\|_2\bigr)^4 +8L_\star^4\bigl(\mathbb E\|Z_t\|_2\bigr)^4\\ &\leq 8\mathbb E\|e_{n,t}(Z_t)\|_2^4 +8L_\star^4\mathbb E\|Z_t\|_2^4.
\end{align*}
The joint curvature bound implies that $Z_0$ has moments of every finite order. Since
\begin{equation*}
    Z_t=e^{-t}Z_0+\sqrt{1-e^{-2t}}\,\xi,
\end{equation*}
its fourth moment is uniformly finite over $t\geq0$. Consequently, the definition of $\mathcal E_n$ implies
\begin{equation*}
    t\longmapsto\|e_{n,t}(0)\|_2\in L^4(0,T).
\end{equation*}
It follows that there exists $G_n\in L^4(0,T)$ such that, for almost every $r$ and every $z$,
\begin{equation}
    \|b_{n,r}(z)\|_2 +\|\widehat b_{n,r}(z)\|_2 +\|e_{n,T-r}(z)\|_2 \leq G_n(r)(1+\|z\|_2). \label{eq:reverse-linear-growth}
\end{equation}

\medskip
\noindent
\textbf{Step 2: strong solutions and Markov evolutions.}

The preceding global Lipschitz bounds and \eqref{eq:reverse-linear-growth}, with $G_n\in L^1(0,T)$, verify the assumptions of Lemma~\ref{lem:time-inhomogeneous-sde-wellposed}. Hence both reverse SDEs admit unique non-explosive strong solutions. Applying It\^o's formula to $(1+\|Z_r\|_2^2)^{q/2}$, localizing, and then using Gronwall's inequality yields
\begin{equation}
    \mathbb E\sup_{0\leq r\leq T}\|Z_r\|_2^q<\infty \label{eq:reverse-moment-bound}
\end{equation}
for every fixed finite $q$, for both reverse processes.

Lemma~\ref{lem:sde-markov-kernel} further gives time-inhomogeneous Markov evolutions $(P_{s,r})_{0\leq s\leq r\leq T}$. Their dual evolutions preserve finite nonnegative mass and satisfy
\begin{equation*}
    \|P_{s,r}^*\nu-P_{s,r}^*\nu'\|_{\mathrm{TV}} \leq \|\nu-\nu'\|_{\mathrm{TV}}.
\end{equation*}

\medskip
\noindent
\textbf{Step 3: global parabolic energy estimates.}

We give the approximation argument because it will also justify the later renormalized density-ratio calculation. Let $v$ denote either $b$ or $\widehat b$, and let $q_0$ denote the corresponding initial density. For the learned process, $q_0$ is standard Gaussian. For the exact process, $q_0=p_{n,T}$ and Lemma~\ref{lem:hd-general-curvature} implies
\begin{equation}
    q_0\in L^1(\mathbb R^{d_n})\cap L^2(\mathbb R^{d_n}), \qquad \int\|z\|_2q_0(z)\,\mathrm dz<\infty. \label{eq:regularity-initial-class}
\end{equation}

Extend $v$ by zero outside $(0,T)$, convolve it in space and time, and multiply by a cutoff $\chi_k$ satisfying
\begin{equation*}
    \chi_k=1\ \text{on }B_k, \qquad \supp(\chi_k)\subset B_{2k}, \qquad \|\nabla\chi_k\|_\infty\leq C/k.
\end{equation*}
Denote the resulting smooth compactly supported drift by $v^{(k)}$. Then, for every $L<\infty$,
\begin{equation}
    v^{(k)}\longrightarrow v \quad\text{in }L^4((0,T);L^\infty(B_L)). \label{eq:regularity-local-drift-convergence}
\end{equation}
If
\begin{equation*}
    \ell=\operatorname*{ess\,sup}_{0<r<T}\Lip(v_r), \qquad h(r)=\|v_r(0)\|_2,
\end{equation*}
then $h\in L^4(0,T)$. Let $h_k$ denote its time mollification after the zero extension outside $(0,T)$; in particular,
\begin{equation*}
    \sup_k\|h_k\|_{L^4(0,T)}<\infty.
\end{equation*}
The construction gives functions $\Lambda_k,H_k\geq0$ satisfying
\begin{align}
    \|(\nabla\cdot v_r^{(k)})^-\|_\infty &\leq\Lambda_k(r), & \sup_k\int_0^T\Lambda_k(r)\,\mathrm dr&<\infty, \label{eq:regularity-divergence-bound}\\ \|v_r^{(k)}(z)\|_2 &\leq H_k(r)(1+\|z\|_2), & \sup_k\|H_k\|_{L^1(0,T)}&<\infty. \label{eq:regularity-approx-growth}
\end{align}
Indeed, the mollified drift has divergence bounded below by $-d_n\ell$, while the derivative of the cutoff contributes at most $C\{1+\ell+h_k(r)/k\}$.

Let $q_0^{(k)}=K_{1/k}*q_0$, where $K_t$ is the heat kernel for $\Delta$. The approximation-of-the-identity property and Young's inequality give
\begin{equation}
    q_0^{(k)}\to q_0 \quad\text{in }L^1\cap L^2, \qquad \|q_0^{(k)}\|_2\leq\|q_0\|_2. \label{eq:regularity-initial-approximation}
\end{equation}
If $X_0\sim q_0$ and $\xi_k\sim N(0,2k^{-1}I_{d_n})$ are independent, then $q_0^{(k)}$ is the law of $X_0+\xi_k$. Consequently,
\begin{equation}
    \sup_k\int\|z\|_2q_0^{(k)}(z)\,\mathrm dz<\infty. \label{eq:regularity-initial-first-moment}
\end{equation}
Let $q^{(k)}$ solve
\begin{equation}
    \partial_rq^{(k)} = \Delta q^{(k)}-\nabla\cdot(v_r^{(k)}q^{(k)}), \qquad q^{(k)}|_{r=0}=q_0^{(k)}. \label{eq:regularity-approximating-fp}
\end{equation}
Since $v^{(k)}$ is smooth and compactly supported, standard parabolic theory and the classical Fokker--Planck correspondence show that $q^{(k)}$ is the marginal density of the unique solution to
\begin{equation*}
    \mathrm dX_r^{(k)} =v_r^{(k)}(X_r^{(k)})\,\mathrm dr +\sqrt{2}\,\mathrm dW_r, \qquad X_0^{(k)}\sim q_0^{(k)}.
\end{equation*}
Testing by $q^{(k)}$ and first inserting a spatial cutoff gives
\begin{equation*}
    \frac12\frac{\mathrm d}{\mathrm dr}\|q_r^{(k)}\|_2^2 = -\|\nabla q_r^{(k)}\|_2^2 -\frac12\int (\nabla\cdot v_r^{(k)})(q_r^{(k)})^2.
\end{equation*}
Therefore
\begin{equation}
    \frac{\mathrm d}{\mathrm dr}\|q_r^{(k)}\|_2^2 +2\|\nabla q_r^{(k)}\|_2^2 \leq \Lambda_k(r)\|q_r^{(k)}\|_2^2. \label{eq:regularity-energy-inequality}
\end{equation}
Writing $\Gamma_k(r)=\int_0^r\Lambda_k(s)\,\mathrm ds$ and integrating,
\begin{equation}
    e^{-\Gamma_k(r)}\|q_r^{(k)}\|_2^2 +2\int_0^r e^{-\Gamma_k(s)}\|\nabla q_s^{(k)}\|_2^2\,\mathrm ds \leq \|q_0^{(k)}\|_2^2. \label{eq:regularity-weighted-energy}
\end{equation}
Young's convolution inequality and~\eqref{eq:regularity-initial-class} now imply
\begin{equation}
    \sup_k\|q^{(k)}\|_{L^\infty(0,T;L^2)}<\infty, \qquad \sup_k\|q^{(k)}\|_{L^2(0,T;H^1)}<\infty. \label{eq:regularity-uniform-energy}
\end{equation}

The growth bound~\eqref{eq:regularity-approx-growth}, \eqref{eq:regularity-initial-first-moment}, and Gronwall's inequality for the approximating SDEs give
\begin{equation}
    \sup_k\sup_{0\leq r\leq T} \int\|z\|_2q_r^{(k)}(z)\,\mathrm dz \leq C_n<\infty. \label{eq:regularity-uniform-first-moment}
\end{equation}
In particular,
\begin{equation}
    \sup_k\sup_{0\leq r\leq T} \int_{\{\|z\|_2>L\}}q_r^{(k)}(z)\,\mathrm dz \leq\frac{C_n}{L}. \label{eq:regularity-uniform-tightness}
\end{equation}

By Banach--Alaoglu, after passing to a subsequence,
\begin{equation*}
    q^{(k)}\overset{*}{\rightharpoonup}\overline q \quad\text{in }L^\infty(0,T;L^2), \qquad q^{(k)}\rightharpoonup\overline q \quad\text{in }L^2(0,T;H^1).
\end{equation*}
The limit is nonnegative. For $\phi\in C_c^\infty([0,T)\times\mathbb R^{d_n})$, passing to the limit in the weak formulation of~\eqref{eq:regularity-approximating-fp} gives
\begin{align}
    &\int_0^T\int\overline q \{-\partial_r\phi-\Delta\phi-v_r^\top\nabla\phi\} \,\mathrm dz\,\mathrm dr = \int q_0(z)\phi(0,z)\,\mathrm dz. \label{eq:regularity-limit-weak-fp}
\end{align}
Indeed, on the compact support of $\phi$,
\begin{equation*}
    v^{(k)}q^{(k)}-v\overline q = (v^{(k)}-v)q^{(k)}+v(q^{(k)}-\overline q).
\end{equation*}
The first term tends to zero in $L^1$ by \eqref{eq:regularity-local-drift-convergence} and the uniform local $L^2$ bound, whereas the second converges distributionally by weak $L^2$ convergence. The initial term converges by \eqref{eq:regularity-initial-approximation}.

We next verify that no mass is lost. Choose increasing cutoffs $\chi_L\in C_c^\infty(\mathbb R^{d_n})$ with $0\leq\chi_L\leq1$, $\chi_L=1$ on $B_L$, and $\supp(\chi_L)\subset B_{2L}$. For every nonnegative $\vartheta\in C_c^\infty(0,T)$,
\begin{equation*}
    \begin{aligned}
        \int_0^T\vartheta(r)\int\chi_L(z)\overline q_r(z) \,\mathrm dz\,\mathrm dr = \lim_{k\to\infty} \int_0^T\vartheta(r)\int\chi_L(z)q_r^{(k)}(z) \,\mathrm dz\,\mathrm dr.
    \end{aligned}
\end{equation*}
Since every $q_r^{(k)}$ has mass one, \eqref{eq:regularity-uniform-tightness} implies
\begin{equation*}
    \left| \int_0^T\vartheta(r) \left\{\int\chi_Lq_r^{(k)}-1\right\}\,\mathrm dr \right| \leq \frac{C_n}{L}\|\vartheta\|_{L^1(0,T)}.
\end{equation*}
Letting first $k\to\infty$ and then $L\to\infty$ yields
\begin{equation}
    \int\overline q_r(z)\,\mathrm dz=1 \qquad\text{for almost every }r\in(0,T). \label{eq:regularity-limit-mass}
\end{equation}

Choose a jointly Borel representative of $\overline q$ and let $N$ be the Borel null set on which~\eqref{eq:regularity-limit-mass} fails. Define
\begin{equation*}
    \nu_r(\mathrm dz) =
    \begin{cases}
        \overline q_r(z)\,\mathrm dz,&r\notin N,\\ \delta_0(\mathrm dz),&r\in N.
    \end{cases}
\end{equation*}
Changing the curve on $N$ does not affect the time-integrated weak equation, so $(\nu_r)_{r\in(0,T)}$ is a Borel probability-valued weak solution.

The first-moment estimate is inherited as well. Let $\psi_L\in C_c^\infty(\mathbb R^{d_n})$ be nonnegative and increase to $\|z\|_2$. Testing against a nonnegative $\vartheta\in C_c^\infty(0,T)$, passing first $k\to\infty$ and then $L\to\infty$, and using~\eqref{eq:regularity-uniform-first-moment} gives
\begin{equation}
    \int\|z\|_2\overline q_r(z)\,\mathrm dz \leq C_n \qquad\text{for almost every }r\in(0,T). \label{eq:regularity-limit-first-moment}
\end{equation}
Since the diffusion covariance is $2I_{d_n}$, and since \eqref{eq:reverse-linear-growth} and \eqref{eq:regularity-limit-first-moment} imply
\begin{equation}
    \begin{aligned}
        \int_0^T\int \{\|v_r(z)\|_2+\|2I_{d_n}\|_{\mathrm F}\} \,\mathrm d\nu_r(z)\,\mathrm dr <\infty,
    \end{aligned}
    \label{eq:regularity-superposition-integrability}
\end{equation}
all assumptions of Lemma~\ref{lem:trevisan-superposition} are satisfied. It follows that $(\nu_r)$ has a unique narrowly continuous representative $(\widetilde\nu_r)_{r\in[0,T]}$ and admits a superposition solution of the associated martingale problem. The boundary term in \eqref{eq:regularity-limit-weak-fp} gives
\begin{equation*}
    \widetilde\nu_0(\mathrm dz)=q_0(z)\,\mathrm dz.
\end{equation*}

The first-moment estimate extends from almost every time to this entire narrowly continuous representative. Indeed, for any $r\in[0,T]$, choose $r_j\to r$ outside the exceptional null set. Narrow continuity and the Portmanteau theorem yield
\begin{equation}
    \int\|z\|_2\,\mathrm d\widetilde\nu_r(z) \leq \liminf_{j\to\infty} \int\|z\|_2\,\mathrm d\widetilde\nu_{r_j}(z) \leq C_n. \label{eq:regularity-all-time-first-moment}
\end{equation}

The drift $v$ is globally Lipschitz with an integrable Lipschitz envelope and satisfies~\eqref{eq:reverse-linear-growth}. Therefore the final assertion of Lemma~\ref{lem:superposition-uniqueness-transfer} identifies $\widetilde\nu_r$, for every $r\in[0,T]$, with the marginal law of the unique strong SDE from Step~2. In particular, the passage from the space--time weak limit to fixed-time SDE marginals is justified by the superposition principle and well-posedness of the martingale problem.

Weak and weak-$*$ lower semicontinuity in
\eqref{eq:regularity-uniform-energy} now give
\begin{equation}
    \overline q\in L^\infty((0,T);L^2(\mathbb R^{d_n})) \cap L^2((0,T);H^1(\mathbb R^{d_n})). \label{eq:regularity-global-energy}
\end{equation}
We henceforth write $q$ for $\overline q$. This construction applies to both reverse equations and yields their density representatives.

For $\phi\in C_c^\infty(\mathbb R^{d_n})$, the two density curves satisfy
\begin{align}
    \frac{\mathrm d}{\mathrm dr}\int\phi\rho_r &= \int\{\Delta\phi+b_r^\top\nabla\phi\}\rho_r, \label{eq:regularity-exact-weak-fp}\\ \frac{\mathrm d}{\mathrm dr}\int\phi\widehat\rho_r &= \int\{\Delta\phi+\widehat b_r^\top\nabla\phi\} \widehat\rho_r. \label{eq:regularity-learned-weak-fp}
\end{align}
For the exact equation, direct substitution of the forward OU equation shows that $p_{n,T-r}$ solves the same Fokker--Planck equation with initial density $p_{n,T}$. Uniqueness from  Lemma~\ref{lem:superposition-uniqueness-transfer} therefore gives
\begin{equation*}
    \rho_{n,r}=p_{n,T-r} \qquad\text{for every }r\in[0,T].
\end{equation*}

\medskip
\noindent
\textbf{Step 4: continuity in $L^1$.}

Let $F_r=v_rq_r$ and $m(r)=\|F_r\|_{L^1}$. By \eqref{eq:reverse-linear-growth}, \eqref{eq:regularity-all-time-first-moment}, and $G_n\in L^4(0,T)$,
\begin{equation}
    m(r) \leq G_n(r)\left(1+\int\|z\|_2q_r(z)\,\mathrm dz\right), \qquad m\in L^4(0,T). \label{eq:regularity-flux-L4}
\end{equation}
The weak Fokker--Planck equation has the Duhamel representation
\begin{equation}
    \widetilde q_r = K_r*q_0 -\int_0^r\nabla K_{r-s}*F_s\,\mathrm ds, \qquad 0\leq r\leq T, \label{eq:regularity-duhamel}
\end{equation}
where equality with $q_r$ initially holds for almost every $r$. Since
\begin{equation}
    \|\nabla K_u\|_{L^1}\leq C_{d_n}u^{-1/2}, \label{eq:regularity-heat-gradient}
\end{equation}
the integral is finite in $L^1$. At $r=0$, the heat-semigroup term converges to $q_0$ in $L^1$, while H\"older's inequality gives
\begin{equation*}
    \begin{aligned}
        \int_0^r(r-s)^{-1/2}m(s)\,\mathrm ds &\leq \left\{\int_0^r(r-s)^{-2/3}\,\mathrm ds\right\}^{3/4} \left\{\int_0^rm(s)^4\,\mathrm ds\right\}^{1/4} \longrightarrow0.
    \end{aligned}
\end{equation*}
Hence $\widetilde q_r\to q_0$ in $L^1$ as $r\downarrow0$.

Fix $u\in(0,T)$ and let $r>u$. Then
\begin{align}
    \widetilde q_r-\widetilde q_u &=(K_r-K_u)*q_0 \notag\\ &\quad+ \int_0^u (\nabla K_{u-s}-\nabla K_{r-s})*F_s\,\mathrm ds \notag\\ &\quad- \int_u^r\nabla K_{r-s}*F_s\,\mathrm ds. \label{eq:regularity-right-continuity}
\end{align}
The first and last terms converge to zero in $L^1$ as $r\downarrow u$. For the middle term, split the integral at $u-\eta$. Dominated convergence applies on $[0,u-\eta]$, whereas on $[u-\eta,u]$ the $L^1$ norm is bounded by
\begin{equation*}
    C_{d_n} \left\{\int_{u-\eta}^u(u-s)^{-2/3}\,\mathrm ds\right\}^{3/4} \left\{\int_{u-\eta}^um(s)^4\,\mathrm ds\right\}^{1/4}.
\end{equation*}
First let $r\downarrow u$ and then $\eta\downarrow0$.

For left continuity, let $0<r<u\leq T$. We have
\begin{align}
    \widetilde q_r-\widetilde q_u &=(K_r-K_u)*q_0 \notag\\ &\quad+ \int_0^r (\nabla K_{u-s}-\nabla K_{r-s})*F_s\,\mathrm ds \notag\\ &\quad+ \int_r^u\nabla K_{u-s}*F_s\,\mathrm ds. \label{eq:regularity-left-continuity}
\end{align}
Splitting the middle integral at $u-\eta$, dominated convergence applies on $[0,u-\eta]$. On $[u-\eta,r]$, the $L^1$ norm is bounded by
\begin{equation*}
    \begin{aligned}
        &C_{d_n} \left\{\int_{u-\eta}^u(u-s)^{-2/3}\,\mathrm ds\right\}^{3/4} \left\{\int_{u-\eta}^um(s)^4\,\mathrm ds\right\}^{1/4}\\ &\quad+ C_{d_n} \left\{\int_{u-\eta}^r(r-s)^{-2/3}\,\mathrm ds\right\}^{3/4} \left\{\int_{u-\eta}^rm(s)^4\,\mathrm ds\right\}^{1/4}.
    \end{aligned}
\end{equation*}
First let $r\uparrow u$ and then $\eta\downarrow0$. Thus every term in
\eqref{eq:regularity-left-continuity} tends to zero in $L^1$. This also
gives left continuity at $u=T$. Consequently,
\begin{equation}
    \rho_n,\widehat\rho_n \in C([0,T];L^1(\mathbb R^{d_n})). \label{eq:regularity-L1-continuity}
\end{equation}
For every $r\in[0,T]$, choose $r_j\to r$ from the full-measure set on which $\widetilde q_{r_j}=q_{r_j}$ is a probability density. The $L^1$-continuity implies that $\widetilde q_r\geq0$ and $\|\widetilde q_r\|_{L^1}=1$. Hence $\widetilde q_r$ is a probability density at every time. The measure curve defined by $\widetilde q_r$ is continuous in total variation and agrees almost everywhere with the narrowly continuous curve $\widetilde\nu_r$ from Step~3. The two curves therefore agree for every $r\in[0,T]$. Thus the $L^1$ representative is precisely the fixed-time SDE marginal density, including both endpoint traces.

\medskip
\noindent
\textbf{Step 5: time derivatives.}

Fix $0<\tau<R<T$ and $L<\infty$. On $B_L$,
\begin{equation*}
    \|v_r\|_{L^\infty(B_L)}\leq C_LG_n(r).
\end{equation*}
Since $G_n\in L^4(0,T)$ and $q\in L^\infty(0,T;L^2)$,
\begin{equation*}
    vq\in L^2((\tau,R);L^2(B_L;\mathbb R^{d_n})).
\end{equation*}
Together with~\eqref{eq:regularity-global-energy}, the equation
\begin{equation*}
    \partial_rq=\Delta q-\nabla\cdot(vq)
\end{equation*}
therefore gives
\begin{equation}
    \partial_rq\in L^2((\tau,R);H^{-1}(B_L)). \label{eq:regularity-time-derivative}
\end{equation}

\medskip
\noindent
\textbf{Step 6: local Sobolev regularity of the density ratio.}

The exact density $\rho_{n,r}=p_{n,T-r}$ is strictly positive. On every compact cylinder $[\tau,R]\times B_L\subset(0,T)\times\mathbb R^{d_n}$, OU smoothing and Lemma~\ref{lem:hd-general-curvature} imply that the exact density is smooth and strictly positive. Hence
\begin{equation*}
    0<c\leq\rho_{n,r}(z)\leq C,
\end{equation*}
and $\rho_n^{-1}$, $\nabla\rho_n$, $\rho_n^{-2}\nabla\rho_n$, and $\rho_n^{-2}\partial_r\rho_n$ are bounded on the cylinder. Hence
\begin{equation*}
    f_n=\frac{\widehat\rho_n}{\rho_n} \in L^\infty((\tau,R);L^2(B_L))
\end{equation*}
and, in the sense of distributions,
\begin{equation*}
    \nabla f_n = \frac{\nabla\widehat\rho_n}{\rho_n} -\frac{\widehat\rho_n\nabla\rho_n}{\rho_n^2}.
\end{equation*}
Both terms belong to $L^2((\tau,R)\times B_L)$, proving
\begin{equation}
    f_n\in L^2_{\mathrm{loc}} ((0,T);H^1_{\mathrm{loc}}(\mathbb R^{d_n})). \label{eq:regularity-ratio-H1}
\end{equation}

\medskip
\noindent
\textbf{Step 7: the density-ratio identity.}

Fix $0\leq a<R<T$ such that $D_n(a)<\infty$, and choose $R_1\in(R,T)$. All local calculations below are carried out for $b<R_1$ and are then restricted to $[a,R]$. For $M>0$, define
\begin{equation*}
    T_M(s)=s\wedge M, \qquad \Phi_M(s)=\int_0^s2T_M(u)\,\mathrm du =
    \begin{cases}
        s^2,&0\leq s\leq M,\\ 2Ms-M^2,&s>M.
    \end{cases}
\end{equation*}
Then
\begin{equation}
    T_M(s)^2\leq\Phi_M(s)\leq s^2, \qquad \Phi_M(s)\leq2Ms. \label{eq:regularity-truncation-bounds}
\end{equation}

The trace at $r=0$ must be defined using $\int\Phi_M(f_{n,r})\rho_{n,r}\,\mathrm dz$, rather than through separate $L^2$ traces of the two densities. Define the perspective function
\begin{equation*}
    F_M(x,y) =
    \begin{cases}
        x\Phi_M(y/x),&x>0,\\ 2My,&x=0.
    \end{cases}
\end{equation*}
For $x>0$, direct differentiation gives
\begin{equation*}
    \partial_yF_M(x,y)=\Phi_M'(y/x), \qquad \partial_xF_M(x,y) = \Phi_M(y/x)-\frac yx\Phi_M'(y/x).
\end{equation*}
Since $|\Phi_M'|\leq2M$ and
\begin{equation*}
    \left| \Phi_M(s)-s\Phi_M'(s) \right| \leq M^2,
\end{equation*}
$F_M$ extends to a globally Lipschitz function on $\mathbb R_+^2$. Therefore Step~4 implies that
\begin{equation}
    E_M(r) := \int\Phi_M(f_r)\rho_r = \int F_M(\rho_r,\widehat\rho_r) \label{eq:regularity-perspective-energy}
\end{equation}
is continuous on $[0,T]$. In particular, this supplies the correct one-sided density-ratio trace at $r=0$.

Let $\eta_L\in C_c^\infty$ satisfy
\begin{equation*}
    0\leq\eta_L\leq1, \quad \eta_L=1\ \text{on }B_L, \quad \supp(\eta_L)\subset B_{2L}, \quad \|\nabla\eta_L\|_\infty\leq C/L, \quad \|\nabla^2\eta_L\|_\infty\leq C/L^2,
\end{equation*}
and put $\chi_L=\eta_L^2$. Then $\chi_L$ has the same support properties and satisfies $\|\nabla\chi_L\|_\infty\leq C/L$ and $\|\Delta\chi_L\|_\infty\leq C/L^2$. We first justify the renormalized test used below. On every compact cylinder $[a,b]\times B_{3L}\subset(0,T)\times\mathbb R^{d_n}$, the exact density and its space--time derivatives are bounded, and $\rho$ is bounded away from zero. Moreover, $\rho_r^{-1}\in W^{1,\infty}(B_{3L})$ uniformly for $r\in[a,b]$, so multiplication by $\rho_r^{-1}$ is uniformly bounded on $H^{-1}(B_{3L})$. Hence Step~5 and the identity
\begin{equation*}
    \partial_r f = \rho^{-1}\partial_r\widehat\rho -\widehat\rho\,\rho^{-2}\partial_r\rho
\end{equation*}
give
\begin{equation*}
    f\in L^2((a,b);H^1(B_{3L})), \qquad \partial_r f\in L^2((a,b);H^{-1}(B_{3L})).
\end{equation*}
Choose $\zeta_L\in C_c^\infty(B_{3L})$ such that $0\leq\zeta_L\leq1$ and $\zeta_L=1$ on $B_{2L}$. Then
\begin{equation*}
    \zeta_Lf \in L^2((a,b);H_0^1(B_{3L})), \qquad \partial_r(\zeta_Lf) = \zeta_L\partial_rf \in L^2((a,b);H^{-1}(B_{3L})).
\end{equation*}
Here multiplication by $\zeta_L$ is bounded on $H^{-1}(B_{3L})$. The Lions--Magenes lemma \citep{lions2012non} therefore gives, after choosing the canonical representative,
\begin{equation*}
    \zeta_Lf\in C([a,b];L^2(B_{3L})).
\end{equation*}
Since $\zeta_L=1$ on $B_{2L}$, it follows that
\begin{equation*}
    f\in C([a,b];L^2(B_{2L})).
\end{equation*}
Hence all endpoint terms below are well defined. Subtracting $f$ times the exact Fokker--Planck equation from the learned equation gives, in distributions,
\begin{equation}
    \partial_r f = \Delta f+(2\nabla\log\rho-b_r)^\top\nabla f -2\rho^{-1}\nabla\cdot(e_{n,T-r}f\rho). \label{eq:regularity-ratio-equation}
\end{equation}

We next justify the local parabolic chain rule used in the renormalized calculation. Fix $0<a<b<R_1$ and set
\begin{equation*}
    \Omega_L=B_{2L}, \qquad w_r=\rho_r\chi_L.
\end{equation*}
The preceding argument gives
\begin{equation*}
    f\in L^2((a,b);H^1(\Omega_L)) \cap C([a,b];L^2(\Omega_L)), \qquad \partial_r f\in L^2((a,b);H^{-1}(\Omega_L)).
\end{equation*}
Moreover, the exact density is smooth and strictly positive on the compact cylinder $[a,b]\times\overline{\Omega_L}$. In particular,
\begin{equation*}
    \rho,\,\partial_r\rho,\,\nabla\rho,\,\rho^{-1} \in L^\infty((a,b)\times\Omega_L),
\end{equation*}
and $w$, $\partial_rw$, and $\nabla w$ are bounded on the same cylinder.

For $0<\varepsilon\leq1$, let $\Phi_{M,\varepsilon}\in C^\infty([0,\infty))$ be a smooth convex approximation of $\Phi_M$ satisfying
\begin{equation*}
    \Phi_{M,\varepsilon}(0)=0, \qquad 0\leq\Phi_{M,\varepsilon}'\leq2M, \qquad 0\leq\Phi_{M,\varepsilon}''\leq2, \qquad \Phi_{M,\varepsilon}''(s)=0 \quad\text{for }s\geq M+\varepsilon.
\end{equation*}
For every $s\geq0$, it also satisfies
\begin{equation*}
    0\leq\Phi_{M,\varepsilon}(s)\leq2Ms, \qquad 0\leq s\Phi_{M,\varepsilon}'(s)\leq2Ms.
\end{equation*}
Moreover, as $\varepsilon\downarrow0$,
\begin{equation*}
    \Phi_{M,\varepsilon}\to\Phi_M, \qquad \Phi_{M,\varepsilon}'\to2T_M, \qquad \Phi_{M,\varepsilon}'' \to2\mathbf1_{\{s<M\}}
\end{equation*}
pointwise, with the last convergence understood away from $s=M$. The Sobolev chain rule gives
\begin{equation*}
    \nabla\{\Phi_{M,\varepsilon}'(f)w\} = \Phi_{M,\varepsilon}''(f)\nabla f\,w +\Phi_{M,\varepsilon}'(f)\nabla w.
\end{equation*}
Since $\chi_L$ is compactly supported in $\Omega_L$ and the first two derivatives of $\Phi_{M,\varepsilon}$ are bounded, this proves
\begin{equation}
    \Phi_{M,\varepsilon}'(f)\rho\chi_L \in L^2((a,b);H_0^1(\Omega_L)). \label{eq:regularity-admissible-renormalized-test}
\end{equation}

We now use the nonlinear integration-by-parts formula of Alt and Luckhaus \citeyearpar{alt1983quasilinear}; see their Lemma~1.5. Applied after localization with the smooth time-dependent weight $w_r$, it shows that the map
\begin{equation*}
    r\longmapsto \int\Phi_{M,\varepsilon}(f_r)\rho_r\chi_L
\end{equation*}
has an absolutely continuous representative and, for almost every $r\in(a,b)$, satisfies
\begin{align}
    \frac{\mathrm d}{\mathrm dr} \int\Phi_{M,\varepsilon}(f_r)\rho_r\chi_L &= \left\langle \partial_r f_r, \Phi_{M,\varepsilon}'(f_r)\rho_r\chi_L \right\rangle_{H^{-1},H_0^1} + \int\Phi_{M,\varepsilon}(f_r) \partial_r\rho_r\chi_L. \label{eq:regularity-parabolic-chain-rule}
\end{align}
This is the weighted version of the standard nonlinear integration-by-parts formula. For completeness, it may be obtained directly by replacing $f$ with its one-sided time-Steklov averages, applying the classical chain rule to the averaged functions, and then passing to the limit in the $L^2_tH^1_x$--$L^2_tH^{-1}_x$ duality. The endpoint terms converge by $f\in C([a,b];L^2(\Omega_L))$, while the term containing $\partial_rw$ converges by dominated convergence. This is also the standard renormalization procedure for weak Fokker--Planck equations; see \citet[Section~5.3]{bris2008existence}.

All terms in \eqref{eq:regularity-parabolic-chain-rule} are integrable. Indeed, the coefficients involving $\rho$ and their derivatives are bounded on the compact cylinder, while $b_r$ and $e_{n,T-r}$ belong locally to $L^4$ in time with values in $L^\infty$ in space, and
\begin{equation*}
    f\in L^\infty((a,b);L^2(\Omega_L)) \cap L^2((a,b);H^1(\Omega_L)).
\end{equation*}

Substituting~\eqref{eq:regularity-ratio-equation} and $\partial_r\rho=\Delta\rho-\nabla\cdot(b_r\rho)$ into \eqref{eq:regularity-parabolic-chain-rule}, and integrating by parts in space, yields
\begin{align}
    &\int\Phi_{M,\varepsilon}(f_b)\rho_b\chi_L +\int_a^b\int \Phi_{M,\varepsilon}''(f_r) \|\nabla f_r\|_2^2\rho_r\chi_L \notag\\ &= \int\Phi_{M,\varepsilon}(f_a)\rho_a\chi_L + 2\int_a^b\int f_r\Phi_{M,\varepsilon}''(f_r) e_{n,T-r}^\top\nabla f_r\rho_r\chi_L +\int_a^b\mathcal R_{M,\varepsilon,L}(r)\,\mathrm dr, \label{eq:regularity-smooth-localized-ratio-identity}
\end{align}
where
\begin{align}
    \mathcal R_{M,\varepsilon,L}(r) &= \int \Phi_{M,\varepsilon}(f_r)\rho_r \{\Delta\chi_L+b_r^\top\nabla\chi_L\}+ 2\int f_r\Phi_{M,\varepsilon}'(f_r)\rho_r e_{n,T-r}^\top\nabla\chi_L. \label{eq:regularity-smooth-cutoff-remainder}
\end{align}
Here and below, spatial integrals without an indicated domain are over $\mathbb R^{d_n}$. The exact-drift terms in the interior cancel because $\partial_r\rho=\Delta\rho-\nabla\cdot(b_r\rho)$.

It remains to remove the smoothing parameter. A Sobolev function has zero gradient almost everywhere on each of its level sets, so $\nabla f=0$ almost everywhere on $\{f=M\}$. Consequently,
\begin{equation*}
    \nabla T_M(f)=\mathbf1_{\{f<M\}}\nabla f \qquad\text{a.e.},
\end{equation*}
and
\begin{equation*}
    \Phi_{M,\varepsilon}''(f)\|\nabla f\|_2^2 \longrightarrow 2\|\nabla T_M(f)\|_2^2 \qquad\text{a.e.}
\end{equation*}
The bound $0\leq\Phi_{M,\varepsilon}''\leq2$ and the local $L^2$-integrability of $\nabla f$ justify dominated convergence in the dissipation term. Similarly,
\begin{equation*}
    f\Phi_{M,\varepsilon}''(f)\nabla f \longrightarrow 2T_M(f)\nabla T_M(f) \qquad\text{a.e.}
\end{equation*}
On the support of $\Phi_{M,\varepsilon}''$, the factor $f$ is bounded by a constant depending only on $M$. The local $L^4_tL^\infty_x$ bound for $e_{n,T-r}$ and the $L^2_tL^2_x$ bound for $\nabla f$ therefore provide an integrable majorant for the mixed term. Finally,
\begin{equation*}
    \Phi_{M,\varepsilon}(f)\rho \leq2Mf\rho=2M\widehat\rho, \qquad f\Phi_{M,\varepsilon}'(f)\rho \leq2Mf\rho=2M\widehat\rho,
\end{equation*}
which permits passage to the limit in the endpoint and cutoff terms. Letting $\varepsilon\downarrow0$ in \eqref{eq:regularity-smooth-localized-ratio-identity}, we obtain, for $0<a<b<R_1$,
\begin{align}
    &\int\Phi_M(f_b)\rho_b\chi_L +2\int_a^b\int \|\nabla T_M(f_r)\|_2^2\rho_r\chi_L \notag\\ &= \int\Phi_M(f_a)\rho_a\chi_L + 4\int_a^b\int T_M(f_r)e_{n,T-r}^{\top}\nabla T_M(f_r) \rho_r\chi_L +\int_a^b\mathcal R_{M,L}(r)\,\mathrm dr. \label{eq:regularity-localized-ratio-identity}
\end{align}
where the cutoff remainder is
\begin{align}
    \mathcal R_{M,L}(r) &= \int \Phi_M(f_r)\rho_r \{\Delta\chi_L+b_r^\top\nabla\chi_L\} + 2\int f_r\Phi_M'(f_r)\rho_r e_{n,T-r}^\top\nabla\chi_L. \label{eq:regularity-cutoff-remainder-formula}
\end{align}
By \eqref{eq:regularity-truncation-bounds}, $\Phi_M(f)\rho\leq2M\widehat\rho$ and $f\Phi_M'(f)\rho=2fT_M(f)\rho\leq2M\widehat\rho$.  Therefore the linear-growth bound~\eqref{eq:reverse-linear-growth} and the support properties of $\nabla\chi_L$ and $\Delta\chi_L$ give
\begin{equation}
    |\mathcal R_{M,L}(r)| \leq CM\{1+G_n(r)\} \int_{\{\|z\|_2>L\}}\widehat\rho_r(z)\,\mathrm dz. \label{eq:regularity-cutoff-remainder-bound}
\end{equation}
Since $G_n\in L^1(0,T)$ and $\widehat\rho_r$ is a probability density, dominated convergence implies
\begin{equation}
    \int_a^b|\mathcal R_{M,L}(r)|\,\mathrm dr\longrightarrow0 \qquad\text{as }L\to\infty, \label{eq:regularity-cutoff-remainder}
\end{equation}
for each fixed $M$.

Put
\begin{equation*}
    a_n(t)=\int\|e_{n,t}(z)\|_2^4p_{n,t}(z)\,\mathrm dz, \qquad \alpha_n(r)=a_n(T-r)^{1/2}.
\end{equation*}
Define
\begin{equation*}
    E_{M,L}(r) = \int\Phi_M(f_r)\rho_r\chi_L, \qquad I_{M,L}(r) = \int\|\nabla T_M(f_r)\|_2^2\rho_r\chi_L.
\end{equation*}
At this point only local $H^1$ regularity of $f$ is known. Therefore Lemma~\ref{lem:hd-general-weighted} must first be applied to the compactly supported function
\begin{equation*}
    h=T_M(f_r)\eta_L,
\end{equation*}
which belongs to $H^1(\rho_r)$ for almost every $r$. Since
\begin{equation*}
    \|\nabla(T_M(f_r)\eta_L)\|_2^2 \leq 2\eta_L^2\|\nabla T_M(f_r)\|_2^2 +2T_M(f_r)^2\|\nabla\eta_L\|_2^2,
\end{equation*}
$T_M(f_r)\leq M$, and $T_M(f_r)^2\leq\Phi_M(f_r)$, the weighted estimate gives
\begin{align}
    \int T_M(f_r)^2\|e_{n,T-r}\|_2^2\rho_r\chi_L &\leq \frac{8C_{\mathrm{LS}}}{\lambda_\star}I_{M,L}(r) +\frac{CM^2}{L^2} +C\alpha_n(r)E_{M,L}(r). \label{eq:regularity-localized-weighted-bound}
\end{align}
Here we used $\int\rho_r=1$ to control the cutoff-gradient term. On the other hand, Young's inequality gives
\begin{equation*}
    4T_M(f_r)e_{n,T-r}^\top\nabla T_M(f_r) \leq \|\nabla T_M(f_r)\|_2^2 +4T_M(f_r)^2\|e_{n,T-r}\|_2^2.
\end{equation*}
Combining this inequality,~\eqref{eq:regularity-localized-weighted-bound}, and~\eqref{eq:regularity-localized-ratio-identity} yields
\begin{align}
    E_{M,L}(b) &+ \theta_\star\int_a^b I_{M,L}(r)\,\mathrm dr \notag\\ &\leq E_{M,L}(a) +C\int_a^b\alpha_n(r)E_{M,L}(r)\,\mathrm dr +\frac{CM^2}{L^2}(b-a) +\int_a^b|\mathcal R_{M,L}(r)|\,\mathrm dr, \label{eq:regularity-localized-energy-bound}
\end{align}
where
\begin{equation}
    \theta_\star = 1-\frac{32C_{\mathrm{LS}}}{\lambda_\star}>0 \label{eq:regularity-absorption-constant}
\end{equation}
by Assumption~\ref{ass:hd-score-approximation}.

For the endpoint $a=0$, apply \eqref{eq:regularity-localized-energy-bound} first with $a=\epsilon>0$ and then let $\epsilon\downarrow0$. The initial energy converges by the same perspective-continuity argument as in \eqref{eq:regularity-perspective-energy}, now with the fixed cutoff $\chi_L$; Fatou's lemma applies to the dissipation, while dominated convergence applies to the remaining terms. Thus \eqref{eq:regularity-localized-energy-bound} is valid also for $a=0$. At this stage we claim only the inequality, not yet the full identity \eqref{eq:regularity-localized-ratio-identity} at the initial endpoint.

Now let $L\to\infty$. The endpoint energies converge by dominated convergence, since $0\leq\Phi_M(f_r)\rho_r\chi_L\leq2M\widehat\rho_r$. Fatou's lemma applies to the dissipation, while \eqref{eq:regularity-cutoff-remainder} eliminates the remainder. We obtain
\begin{align}
    E_M(b) &+ \theta_\star\int_a^b\int \|\nabla T_M(f_r)\|_2^2\rho_r \,\mathrm dz\,\mathrm dr \leq E_M(a) +C\int_a^b\alpha_n(r)E_M(r)\,\mathrm dr. \label{eq:regularity-global-truncated-bound}
\end{align}
Because $E_M$ is continuous also at $a=0$, Gronwall's inequality applies on $[a,R]$. Since $E_M(a)\leq D_n(a)$ and $\alpha_n\in L^1(0,T)$, it gives, uniformly in $M$,
\begin{align}
    \sup_{a\leq r\leq R}E_M(r) +\theta_\star\int_a^R\int \|\nabla T_M(f_r)\|_2^2\rho_r\,\mathrm dz\,\mathrm dr &\leq C_R D_n(a), \label{eq:regularity-uniform-truncated-energy}
\end{align}
where $C_R<\infty$ because $\alpha_n\in L^1(0,T)$.

This is the point at which global weighted Sobolev regularity becomes a conclusion rather than an assumption. For almost every $r$, \eqref{eq:regularity-uniform-truncated-energy} and $T_M(f_r)^2\leq\Phi_M(f_r)$ imply
\begin{equation*}
    T_M(f_r)\in H^1(\rho_r).
\end{equation*}
Indeed, $T_M(f_r)\eta_L\to T_M(f_r)$ in $H^1(\rho_r)$: the tail of $\nabla T_M(f_r)$ vanishes by integrability, and the cutoff-gradient term is bounded by $CM^2L^{-2}\int\rho_r$.

It is now legitimate to apply Lemma~\ref{lem:hd-general-weighted} globally with $h=T_M(f_r)$. We obtain
\begin{align}
    \int T_M(f_r)^2\|e_{n,T-r}\|_2^2\rho_r &\leq \frac{4C_{\mathrm{LS}}}{\lambda_\star} \int\|\nabla T_M(f_r)\|_2^2\rho_r + C\alpha_n(r)E_M(r). \label{eq:regularity-global-weighted-bound}
\end{align}
After integration in $r$, the right-hand side is bounded uniformly in $M$ by~\eqref{eq:regularity-uniform-truncated-energy}.

The mixed term is consequently integrable by Cauchy--Schwarz. For $0<a<b\leq R$, return to \eqref{eq:regularity-localized-ratio-identity} and let $L\to\infty$ in the equality itself. If $a=0$, first use the resulting equality on $(\epsilon,b)$ and then let $\epsilon\downarrow0$; the endpoint energy converges by \eqref{eq:regularity-perspective-energy}, and the dissipation and mixed terms are now absolutely integrable by \eqref{eq:regularity-uniform-truncated-energy} and \eqref{eq:regularity-global-weighted-bound}. We thereby obtain, for $0\leq a<b\leq R$, the global truncated identity
\begin{align}
    E_M(b)-E_M(a) &=-2\int_a^b\int \|\nabla T_M(f_r)\|_2^2\rho_r \,\mathrm dz\,\mathrm dr \notag\\ &\quad+ 4\int_a^b\int T_M(f_r)e_{n,T-r}^\top\nabla T_M(f_r)\rho_r \,\mathrm dz\,\mathrm dr. \label{eq:regularity-global-truncated-identity}
\end{align}

Finally let $M\to\infty$. Monotone convergence gives
\begin{equation*}
    \Phi_M(f_r)\uparrow f_r^2, \qquad \|\nabla T_M(f_r)\|_2^2 \uparrow\|\nabla f_r\|_2^2.
\end{equation*}
Integrating~\eqref{eq:regularity-global-weighted-bound} over $[a,R]$ and using monotone convergence gives
\begin{equation}
    \int_a^R\int f_r^2\|e_{n,T-r}\|_2^2\rho_r \,\mathrm dz\,\mathrm dr<\infty. \label{eq:regularity-full-error-integrability}
\end{equation}
Consequently, $T_M(f)e_{n,T-r}\to fe_{n,T-r}$ and $\nabla T_M(f)\to\nabla f$ strongly in their corresponding weighted $L^2$ spaces. Indeed, the first convergence follows by dominated convergence from~\eqref{eq:regularity-full-error-integrability}, while
\begin{equation*}
    \nabla T_M(f) = \mathbf 1_{\{f<M\}}\nabla f
\end{equation*}
and the limiting dissipation bound give the second. Cauchy--Schwarz therefore passes the mixed term to the limit in $L^1(a,R)$. We obtain, for $0\leq a<b\leq R$,
\begin{align}
    D_n(b)-D_n(a) &= -2\int_a^b\int \|\nabla f_{n,r}\|_2^2\rho_{n,r} \,\mathrm dz\,\mathrm dr \notag\\ &\quad+ 4\int_a^b\int f_{n,r}e_{n,T-r}^\top\nabla f_{n,r}\rho_{n,r} \,\mathrm dz\,\mathrm dr. \label{eq:regularity-integrated-ratio-identity}
\end{align}
Both integrands are in $L^1(a,R)$ by \eqref{eq:regularity-full-error-integrability} and the limiting dissipation bound. Hence the right-hand side defines an absolutely continuous representative of $D_n$ on $[a,R]$, and differentiation proves \eqref{eq:regularity-density-ratio}.
\end{proof}

\begin{mylemma}
\label{lem:hd-general-pde}
Under Assumptions~\ref{ass:hd-setting} and~\ref{ass:hd-score-approximation},
\begin{equation}
    \chi^2(\widehat\mu_n,\mu_n)\to_p 0, \qquad \chi^2(\widehat\mu_{n,X},\mu_{n,X})\to_p 0, \qquad W_4(\widehat\mu_n,\mu_n)\to_p 0. \label{eq:pde-conclusions}
\end{equation}
\end{mylemma}

\begin{proof}
Write $T=T_n$ and work on the event $\mathcal E_n$ in \eqref{eq:regularity-event}. By~\eqref{eq:score-l4} and \eqref{eq:small-lip}, this event has probability tending to one. Set
\begin{equation*}
    \rho_{n,r}=p_{n,T-r}, \qquad \epsilon_{n,r}=e_{n,T-r},
\end{equation*}
and let $\widehat\rho_{n,r}$ be the learned reverse density.  The exact and learned Fokker--Planck equations are
\begin{align*}
    \partial_r\rho_{n,r} &= \Delta\rho_{n,r}-\nabla\cdot(b_{n,r}\rho_{n,r}),\\ \partial_r\widehat\rho_{n,r} &= \Delta\widehat\rho_{n,r} -\nabla\cdot\{(b_{n,r}+2\epsilon_{n,r})\widehat\rho_{n,r}\}.
\end{align*}
All the regularity needed below, including the validity of the density-ratio calculation, is supplied by Lemma~\ref{lem:hd-reverse-regularity}.  Define
\begin{equation*}
    f_{n,r}=\frac{\widehat\rho_{n,r}}{\rho_{n,r}}, \qquad D_n(r)=\int f_{n,r}^2\rho_{n,r}\,\mathrm dz, \qquad H_n(r)=D_n(r)-1.
\end{equation*}
Since both densities integrate to one,
\begin{equation}
    \int f_{n,r}\rho_{n,r}\,\mathrm dz=1, \qquad H_n(r)=\chi^2(\widehat\rho_{n,r},\rho_{n,r}). \label{eq:H-is-chi}
\end{equation}
At reverse time zero, $\widehat\rho_{n,0}=\gamma_{d_n}$ and $\rho_{n,0}=p_{n,T_n}$. Moreover,
\begin{equation*}
    \mathbb E\|Z_0\|_2^2 = \tr(\Sigma_n)+\beta_n^\top\Sigma_n\beta_n+\sigma_n^2 \leq Cd_n.
\end{equation*}
Lemma~\ref{lem:hd-general-init} therefore gives $H_n(0)=o(1)$ and, in particular, $D_n(0)<\infty$ for all sufficiently large $n$. Thus Lemma~\ref{lem:hd-reverse-regularity} makes the following differential calculation valid on every compact subinterval of $[0,T)$. Let
\begin{equation*}
    a_n(u) = \int\|e_{n,u}(z)\|_2^4p_{n,u}(z)\,\mathrm dz, \qquad \alpha_n(r)=a_n(T-r)^{1/2}.
\end{equation*}
From~\eqref{eq:regularity-density-ratio} and $4xy\leq x^2+4y^2$,
\begin{equation*}
    D_n'(r) \leq -\int\|\nabla f_{n,r}\|_2^2\rho_{n,r} +4\int\|\epsilon_{n,r}\|_2^2f_{n,r}^2\rho_{n,r}.
\end{equation*}
Apply Lemma~\ref{lem:hd-general-weighted} with $g=\|\epsilon_{n,r}\|_2$ and $h=f_{n,r}$. Since the Euclidean norm is one-Lipschitz,
\begin{equation*}
    \Lip\bigl(\|\epsilon_{n,r}\|_2\bigr) \leq \Lip(\epsilon_{n,r}) \leq L_\star.
\end{equation*}
We therefore obtain
\begin{equation*}
    D_n'(r) \leq -\left(1-\frac{16C_{\mathrm{LS}}}{\lambda_\star}\right) \int\|\nabla f_{n,r}\|_2^2\rho_{n,r} +C\alpha_n(r)D_n(r).
\end{equation*}
Since $\lambda_\star>32C_{\mathrm{LS}}$, the Poincar\'e inequality yields
\begin{equation*}
    \int\|\nabla f_{n,r}\|_2^2\rho_{n,r} \geq \frac1{C_{\mathrm{LS}}} \int(f_{n,r}-1)^2\rho_{n,r} =\frac{H_n(r)}{C_{\mathrm{LS}}}.
\end{equation*}
Therefore, with
\begin{equation*}
    c_0 := \frac{1}{C_{\mathrm{LS}}} \left(1-\frac{16C_{\mathrm{LS}}}{\lambda_\star}\right)>0,
\end{equation*}
we have
\begin{equation}
    H_n'(r) \leq -c_0H_n(r)+C\alpha_n(r)\{1+H_n(r)\}. \label{eq:H-diff}
\end{equation}
By Cauchy--Schwarz and~\eqref{eq:score-l4},
\begin{equation}
    \mathfrak a_n := \sup_{0\leq r<T} \int_0^re^{-c_0(r-s)}\alpha_n(s)\,\mathrm ds \to_p 0. \label{eq:an-small}
\end{equation}
Indeed, the left-hand side is bounded by
\begin{equation*}
    \left(\int_0^\infty e^{-c_0u}\,\mathrm du\right)^{1/2} \left(\int_0^{T_n}a_n(u)\,\mathrm du\right)^{1/2}.
\end{equation*}
We spell out the stopping argument. Let
\begin{equation*}
    \tau_n=\inf\{r<T:H_n(r)\geq1\},
\end{equation*}
with $\inf\varnothing=T$.  On $[0,\tau_n]$, inequality~\eqref{eq:H-diff} and variation of constants give
\begin{equation*}
    H_n(r) \leq e^{-c_0r}H_n(0) +2C\int_0^re^{-c_0(r-s)}\alpha_n(s)\,\mathrm ds \leq H_n(0)+2C\mathfrak a_n.
\end{equation*}
The last expression is smaller than one with probability tending to one, so $\tau_n=T$ on that event.  Consequently,
\begin{equation*}
    \sup_{0\leq r<T_n}H_n(r)=o_p(1).
\end{equation*}

It remains to include the terminal time. The two marginal curves are narrowly continuous by Lemma~\ref{lem:hd-reverse-regularity}. For arbitrary probability measures $Q$ and $P$, adopt the convention $\chi^2(Q,P)=+\infty$ when $Q\not\ll P$. Then the variational representation
\begin{equation}
    \chi^2(Q,P) = \sup_{\varphi\in C_b(\mathbb R^{d_n})} \left\{ 2\int\varphi\,\mathrm dQ -\int\varphi^2\,\mathrm dP -1 \right\} \label{eq:chi-variational}
\end{equation}
holds without any prior assumption of absolute continuity. If $Q\ll P$, it follows by completing the square and truncating $\mathrm dQ/\mathrm dP$; if $Q\not\ll P$, regular approximation of a $P$-null set carrying positive $Q$-mass shows that the supremum is infinite. Taking $r\uparrow T$ for each bounded continuous test function in~\eqref{eq:chi-variational}, and then taking the supremum, gives joint lower semicontinuity of $\chi^2$. Hence
\begin{equation*}
    \chi^2(\widehat\mu_n,\mu_n) \leq \liminf_{r\uparrow T}H_n(r) =o_p(1).
\end{equation*}

If $R=\mathrm d\widehat\mu_n/\mathrm d\mu_n$, then the likelihood ratio of the $X$-marginal is $R_X(X)=\mathbb E_{\mu_n}[R(X,Y)\mid X]$.  Conditional Jensen therefore gives
\begin{equation*}
    1+\chi^2(\widehat\mu_{n,X},\mu_{n,X}) \leq 1+\chi^2(\widehat\mu_n,\mu_n).
\end{equation*}

We now give the coupling argument for $W_4$. All coupling expectations below are conditional on $\mathcal F_n$. The proof of Lemma~\ref{lem:hd-general-curvature} gives the sharper time-dependent lower curvature bound
\begin{equation*}
    m(t) := \left\{ \frac{e^{-2t}}{m_0}+1-e^{-2t} \right\}^{-1} \geq 1-C_0e^{-2t}
\end{equation*}
for a fixed $C_0$.  Therefore the exact reverse drift satisfies
\begin{equation*}
    \langle z-z',b_{n,r}(z)-b_{n,r}(z')\rangle \leq \{1-2m(T-r)\}\|z-z'\|_2^2.
\end{equation*}
Adding the score error changes the one-sided Lipschitz coefficient by at most $2L_\star$.  With
\begin{equation*}
    a_0:=1-2L_\star>\frac12,
\end{equation*}
where the strict inequality follows from \eqref{eq:lipschitz-absorption}, we consequently have, for $0\leq s\leq r\leq T$,
\begin{equation}
    \exp\left\{ \int_s^r[1-2m(T-v)+2L_\star]\,\mathrm dv \right\} \leq C_1e^{-a_0(r-s)}, \label{eq:reverse-contraction-kernel}
\end{equation}
because $\int_0^\infty e^{-2t}\,\mathrm dt<\infty$.

First initialize both learned reverse equations from two optimally coupled laws, $N(0,I_{d_n})$ and $p_{n,T}$.  Drive them by the same Brownian motion. The score-error difference is controlled by $2L_\star$ and is already included in~\eqref{eq:reverse-contraction-kernel}. Hence
\begin{equation}
    W_4\bigl( \widehat P_{0,T}N(0,I_{d_n}), \widehat P_{0,T}p_{n,T} \bigr) \leq C_1e^{-a_0T}W_4\{N(0,I_{d_n}),p_{n,T}\}. \label{eq:learned-initial-transfer}
\end{equation}
Here $\widehat P$ denotes the learned reverse law map.  To bound the initial distance, use the forward coupling
\begin{equation*}
    Z_T=e^{-T}Z_0+\sqrt{1-e^{-2T}}\,G, \qquad G\sim N(0,I_{d_n}).
\end{equation*}
Uniform strong log-concavity and the covariance bounds imply $\mathbb E\|Z_0\|_2^4\leq Cd_n^2$, while $\mathbb E\|G\|_2^4\leq Cd_n^2$. Since $|1-\sqrt{1-e^{-2T}}|\leq e^{-2T}$,
\begin{equation}
    W_4\{p_{n,T},N(0,I_{d_n})\} \leq Cd_n^{1/2}e^{-T} =C(d_ne^{-2T_n})^{1/2}=o(1), \label{eq:initial-W4}
\end{equation}
where the last step uses the diffusion-horizon condition in Assumption~\ref{ass:hd-setting}.

Next let $Y_r$ be the exact reverse process initialized from $p_{n,T}$ and let $\widehat Y_r$ be the learned reverse process with the same initial value and Brownian motion.  Then $Y_r\sim p_{n,T-r}$. Put $\Delta_r=\widehat Y_r-Y_r$.  Splitting
\begin{equation*}
    e_{n,T-r}(\widehat Y_r) = e_{n,T-r}(Y_r) +\{e_{n,T-r}(\widehat Y_r)-e_{n,T-r}(Y_r)\}
\end{equation*}
and using~\eqref{eq:reverse-contraction-kernel} gives the pathwise bound
\begin{equation*}
    \|\Delta_T\|_2 \leq 2C_1\int_0^Te^{-a_0(T-r)} \|e_{n,T-r}(Y_r)\|_2\,\mathrm dr.
\end{equation*}
Minkowski's inequality followed by weighted Jensen yields
\begin{align*}
    \mathbb E\|\Delta_T\|_2^4 &\leq C\int_0^Te^{-a_0(T-r)} \mathbb E\|e_{n,T-r}(Y_r)\|_2^4\,\mathrm dr\\ &= C\int_0^Te^{-a_0t} \mathbb E_{p_{n,t}}\|e_{n,t}(Z_t)\|_2^4\,\mathrm dt \to_p 0.
\end{align*}
Thus the learned reverse process initialized from $p_{n,T}$ has terminal law converging to $\mu_n$ in $W_4$. More explicitly, if $P_{0,T}$ and $\widehat P_{0,T}$ denote the exact and learned reverse law maps, then
\begin{align*}
    W_4(\widehat\mu_n,\mu_n) &\leq W_4\bigl( \widehat P_{0,T}N(0,I_{d_n}), \widehat P_{0,T}p_{n,T} \bigr) +W_4\bigl( \widehat P_{0,T}p_{n,T}, P_{0,T}p_{n,T} \bigr)\\ &\leq C_1e^{-a_0T}W_4\{N(0,I_{d_n}),p_{n,T}\} +\left\{ C\int_0^Te^{-a_0t}a_n(t)\,\mathrm dt \right\}^{1/4} \to_p 0.
\end{align*}
This proves the last claim in~\eqref{eq:pde-conclusions}.
\end{proof}

\begin{mylemma}[$\chi^2$ convergence of the terminal OU law to the Gaussian law]
\label{lem:hd-general-init}
Suppose $Z_0\in\mathbb R^{d_n}$ is centered and $\mathbb E\|Z_0\|_2^2\leq Cd_n,d_ne^{-2T_n}\to 0$.  If $p_{n,T_n}$ is the density of
\begin{equation*}
    Z_{T_n} = e^{-T_n}Z_0+\sqrt{1-e^{-2T_n}}\,\xi,
\end{equation*}
then, for all sufficiently large $n$,
\begin{equation}
    \chi^2(\gamma_{d_n},p_{n,T_n}) \leq \exp\{C'd_ne^{-2T_n}\}-1=o(1). \label{eq:init-chi}
\end{equation}
\end{mylemma}

\begin{proof}
Put $a=e^{-T_n}$ and $q^2=1-a^2$.  The OU kernel and Jensen's inequality imply
\begin{align*}
    \frac{p_{n,T_n}(z)}{\gamma_{d_n}(z)} &= q^{-d_n} \exp\left(-\frac{a^2\|z\|_2^2}{2q^2}\right) \mathbb E\exp\left\{ \frac{a}{q^2}z^\top Z_0 -\frac{a^2}{2q^2}\|Z_0\|_2^2 \right\}\\ &\geq q^{-d_n} \exp\left\{-\frac{a^2}{2q^2} \bigl(\|z\|_2^2+\mathbb E\|Z_0\|_2^2\bigr) \right\}.
\end{align*}
The last step uses Jensen's inequality and $\mathbb E Z_0=0$.  Since $T_n\to\infty$, we have $a^2/q^2<1$ for all sufficiently large $n$. Therefore, for $G\sim N(0,I_{d_n})$,
\begin{align*}
    1+\chi^2(\gamma_{d_n},p_{n,T_n}) &= \mathbb E\frac{\gamma_{d_n}(G)}{p_{n,T_n}(G)}\\ &\leq q^{d_n} \exp\left\{ \frac{a^2}{2q^2}\mathbb E\|Z_0\|_2^2 \right\} \mathbb E\exp\left\{ \frac{a^2}{2q^2}\|G\|_2^2 \right\}\\ &= q^{d_n} \exp\left\{ \frac{a^2}{2q^2}\mathbb E\|Z_0\|_2^2 \right\} \left(1-\frac{a^2}{q^2}\right)^{-d_n/2}.
\end{align*}
For $a^2\leq1/4$, the elementary bounds $|\log(1-x)|\leq2x$ for $0\leq x\leq1/2$ and $\mathbb E\|Z_0\|_2^2\leq Cd_n$ show that the logarithm of the last display is at most $C'd_na^2$.  Hence
\begin{equation*}
    \chi^2(\gamma_{d_n},p_{n,T_n}) \leq \exp\{C'd_ne^{-2T_n}\}-1=o(1).
\end{equation*}
\end{proof}

\begin{mylemma}
\label{lem:hd-general-conditional}
Let
\begin{equation*}
    \widehat K_n(x)=\widehat\mu_n(Y\in\cdot\mid X=x), \qquad K_n(x)=N(x^\top\beta_n,\sigma_n^2).
\end{equation*}
Under Assumptions~\ref{ass:hd-setting} and~\ref{ass:hd-score-approximation},
\begin{equation}
    \mathbb E_{X\sim\widehat\mu_{n,X}} W_4^4\{\widehat K_n(X),K_n(X)\} \to_p 0. \label{eq:conditional-W4}
\end{equation}
If
\begin{equation*}
    m_n(x)=\mathbb E_{\widehat K_n(x)}Y, \quad v_n^*(x)=\Var_{\widehat K_n(x)}(Y), \quad r_n(x)=m_n(x)-x^\top\beta_n,
\end{equation*}
then
\begin{equation}
    \eta_n:=\mathbb E_{\widehat\mu_{n,X}}|r_n(X)|^4\to_p 0, \qquad \zeta_n^2 := \mathbb E_{\widehat\mu_{n,X}}|v_n^*(X)-\sigma_n^2|^2 \to_p 0. \label{eq:eta-zeta}
\end{equation}
Moreover,
\begin{equation}
    M_{4,n} := \sup_{\|u\|_2=1} \mathbb E_{\widehat\mu_{n,X}}|u^\top X|^4 =O_p(1). \label{eq:M4}
\end{equation}
\end{mylemma}

\begin{proof}
Write $P=\mu_n$, $Q=\widehat\mu_n$, and
\begin{equation*}
    R(x,y)=\frac{\mathrm dQ}{\mathrm dP}(x,y), \qquad R_X(x)=\mathbb E_P\{R(X,Y)\mid X=x\}.
\end{equation*}
Then $R_X=\mathrm dQ_X/\mathrm dP_X$.  For $R_X(x)>0$, the conditional likelihood ratio is
\begin{equation*}
    R_{Y\mid X}(x,y)=\frac{R(x,y)}{R_X(x)}.
\end{equation*}
Put
\begin{equation*}
    k_n(x) = \mathbb E_P\bigl[(R-R_X)^2\mid X=x\bigr].
\end{equation*}
Conditional Cauchy--Schwarz gives
\begin{equation*}
    \|Q(\cdot\mid x)-P(\cdot\mid x)\|_{\mathrm{TV}} \leq \frac{k_n(x)^{1/2}}{2R_X(x)}.
\end{equation*}
After integration with respect to $Q_X=R_XP_X$, the denominator cancels:
\begin{align}
    \mathbb E_{Q_X} \|Q(\cdot\mid X)-P(\cdot\mid X)\|_{\mathrm{TV}} &\leq \frac12\mathbb E_{P_X}k_n(X)^{1/2}\notag\\ &\leq \frac12\left\{ \mathbb E_P(R-R_X)^2 \right\}^{1/2}\notag\\ &\leq \frac12\chi^2(Q,P)^{1/2}. \label{eq:average-conditional-TV}
\end{align}
The final inequality uses
\begin{equation*}
    \mathbb E_P(R-R_X)^2 = \mathbb E_PR^2-\mathbb E_{P_X}R_X^2 \leq \mathbb E_PR^2-1.
\end{equation*}

We next upgrade~\eqref{eq:average-conditional-TV} to fourth-order transport.  Translate both conditional distributions by $-x^\top\beta_n$; under $P$, the translated variable is $\varepsilon\sim N(0,\sigma_n^2)$. For any two probability measures $P_0,Q_0$ on $\mathbb R$ and any $A>0$, a maximal coupling, followed by $|u-v|^4\leq8(|u|^4+|v|^4)$, gives
\begin{equation}
    W_4^4(P_0,Q_0) \leq C A^4\|P_0-Q_0\|_{\mathrm{TV}} +C\int_{|u|>A}|u|^4\,(P_0+Q_0)(\mathrm du). \label{eq:TV-to-W4}
\end{equation}
The true Gaussian residual has uniformly bounded moments of every fixed order. Moreover,
\begin{align*}
    \mathbb E_Q|Y-X^\top\beta_n|^8 &= \mathbb E_P\{R|\varepsilon|^8\}\leq (\mathbb E_PR^2)^{1/2}(\mathbb E|\varepsilon|^{16})^{1/2} =O_p(1),
\end{align*}
because Lemma~\ref{lem:hd-general-pde} gives $\mathbb E_PR^2=1+o_p(1)$.  Averaging \eqref{eq:TV-to-W4} over $Q_X$ therefore yields
\begin{equation*}
    \mathbb E_{Q_X}W_4^4\{\widehat K_n(X),K_n(X)\} \leq CA^4\chi^2(Q,P)^{1/2}+C A^{-4}O_p(1).
\end{equation*}
First let $n\to\infty$ and then $A\to\infty$.  This proves \eqref{eq:conditional-W4}.

Since the difference of the conditional means is bounded by $W_1$, we have
\begin{equation*}
    |r_n(x)| \leq W_1\{\widehat K_n(x),K_n(x)\} \leq W_4\{\widehat K_n(x),K_n(x)\}.
\end{equation*}
Also, the difference between the standard deviations of two laws is at most their $W_2$ distance. Hence
\begin{equation*}
    |v_n^*(x)-\sigma_n^2| \leq W_2^2\{\widehat K_n(x),K_n(x)\} +2\sigma_n W_2\{\widehat K_n(x),K_n(x)\}.
\end{equation*}
Equation~\eqref{eq:eta-zeta} now follows from \eqref{eq:conditional-W4}, Cauchy--Schwarz, and the uniform upper bound on $\sigma_n^2$.

Finally, write $h_{n,X}=\chi^2(Q_X,P_X)$.  Conditional Jensen gives $h_{n,X}\leq\chi^2(Q,P)=o_p(1)$.  Uniform strong log-concavity of the true design and the spectral upper bound for $\Sigma_n$ imply
\begin{equation*}
    \sup_{\|u\|_2=1}\mathbb E_{P_X}|u^\top X|^8\leq C.
\end{equation*}
Thus, uniformly over $\|u\|_2=1$,
\begin{align*}
    \mathbb E_{Q_X}|u^\top X|^4 &= \mathbb E_{P_X}\{R_X|u^\top X|^4\}\leq (1+h_{n,X})^{1/2} \left(\mathbb E_{P_X}|u^\top X|^8\right)^{1/2} =O_p(1),
\end{align*}
which is~\eqref{eq:M4}.
\end{proof}

\begin{mylemma}
\label{lem:hd-general-true-tail}
Let $M=M_n\in\{n-1,n\}$, let $X_1,\ldots,X_M\overset{\mathrm{i.i.d.}}\sim\mu_{n,X}$, and define
\begin{equation*}
    \widetilde\Sigma_M = \Sigma_n^{-1/2} \left(\frac1M\sum_{i=1}^MX_iX_i^\top\right) \Sigma_n^{-1/2}.
\end{equation*}
There exist constants $C,t_0>0$ such that, for all sufficiently large $n$,
\begin{equation}
    \mathbb P\{\lambda_{\min}(\widetilde\Sigma_M)\leq t\} \leq (C\sqrt t)^{(M-p_n)/2}, \qquad 0<t<t_0. \label{eq:true-tail}
\end{equation}
Consequently, for every fixed $Q>0$,
\begin{equation*}
    \mathbb E\lambda_{\min}^{-Q}(\widetilde\Sigma_M)=O(1).
\end{equation*}
\end{mylemma}

\begin{proof}
Set $W_i=\Sigma_n^{-1/2}X_i$ and let $W$ have rows $W_i^\top$.  The rows are isotropic.  Their potential is
\begin{equation*}
    \widetilde V_n(w)=V_n(\Sigma_n^{1/2}w),
\end{equation*}
and~\eqref{eq:relative-curvature} gives
\begin{equation*}
    \underline m I_{p_n} \preceq \nabla^2\widetilde V_n(w) \preceq \overline m I_{p_n}.
\end{equation*}
By Lemma~\ref{lem:aux-logconcave-projection}, every one-dimensional projection $\langle W_i,v\rangle$ is centered, has variance one, and has a log-concave density. Lemma~\ref{lem:aux-isotropic-logconcave-density} therefore gives
\begin{equation}
    \sup_{\|v\|_2=1} \|g_{n,v}\|_\infty\leq C, \label{eq:projection-density-bound}
\end{equation}
where $g_{n,v}$ is the density of $\langle W_i,v\rangle$ and one may take $C=1$. Consequently, for fixed $v\in\mathbb S^{p_n-1}$, the random vector
\begin{equation*}
    (\langle W_1,v\rangle,\ldots,\langle W_M,v\rangle)
\end{equation*}
has a product density bounded by $C^M$.  The volume estimate
\begin{equation*}
    \operatorname{Vol}\{x\in\mathbb R^M:\|x\|_2\leq2u\sqrt M\} = \frac{\pi^{M/2}(2u\sqrt M)^M}{\Gamma(M/2+1)} \leq(Cu)^M
\end{equation*}
gives
\begin{equation}
    \mathbb P\{\|Wv\|_2\leq2u\sqrt M\}\leq(C_1u)^M. \label{eq:small-ball}
\end{equation}

We next record the operator-norm bound used in the net argument. By Lemma~\ref{lem:aux-bakry-emery-herbst}, the lower curvature bound gives a log-Sobolev inequality with constant at most $\underline m^{-1}$, and its Herbst consequence applied to the one-Lipschitz linear functional $w\mapsto\langle w,x\rangle$ gives, for $x\in\mathbb S^{p_n-1}$,
\begin{equation}
    \mathbb E\exp\{t\langle W_i,x\rangle\} \leq \exp\left\{\frac{t^2}{2\underline m}\right\}, \qquad t\in\mathbb R. \label{eq:row-subgaussian}
\end{equation}
Thus, for fixed $x\in\mathbb S^{p_n-1}$ and $y\in\mathbb S^{M-1}$, independence of the rows implies
\begin{equation}
    \mathbb P\{|y^\top Wx|>s\} \leq2\exp(-c s^2). \label{eq:bilinear-subgaussian}
\end{equation}
Let $\mathcal N_p$ and $\mathcal N_M$ be $1/4$-nets of the corresponding unit spheres.  They may be chosen so that
\begin{equation*}
    |\mathcal N_p|\leq9^{p_n}, \qquad |\mathcal N_M|\leq9^M,
\end{equation*}
and the standard two-net approximation gives
\begin{equation*}
    \|W\|_{\op} \leq2\max_{x\in\mathcal N_p,\,y\in\mathcal N_M}|y^\top Wx|.
\end{equation*}
A union bound in~\eqref{eq:bilinear-subgaussian} now yields constants $C_0,c_0>0$ such that
\begin{equation}
    \mathbb P\left\{ \|W\|_{\op}>C_0(\sqrt M+\sqrt{p_n})+s \right\} \leq2e^{-c_0s^2}, \qquad s\geq0. \label{eq:operator-net-tail}
\end{equation}
Since $p_n/M$ is bounded away from one for all sufficiently large $n$, this implies
\begin{equation}
    \mathbb P\{\|W\|_{\op}>R\sqrt M\}\leq e^{-cR^2M}, \qquad R\geq R_0. \label{eq:op-tail}
\end{equation}

Fix $0<u<u_0$ and $R\geq R_0$.  On $\{\|W\|_{\op}\leq R\sqrt M\}$, suppose that $s_{\min}(W)\leq u\sqrt M$.  Choose $x\in\mathbb S^{p_n-1}$ with $\|Wx\|_2\leq u\sqrt M$ and let $v$ belong to an $u/(2R)$-net of $\mathbb S^{p_n-1}$ with $\|x-v\|_2\leq u/(2R)$.  Then
\begin{equation*}
    \|Wv\|_2 \leq \|Wx\|_2+\|W\|_{\op}\|x-v\|_2 \leq \frac32u\sqrt M.
\end{equation*}
The net can be chosen with cardinality at most $(1+4R/u)^{p_n}$.  A union bound and~\eqref{eq:small-ball} therefore give
\begin{align}
    \mathbb P\{s_{\min}(W)\leq u\sqrt M, \ \|W\|_{\op}\leq R\sqrt M\}\leq \left(\frac{CR}{u}\right)^{p_n}(Cu)^M \leq C^M R^{p_n}u^{M-p_n}. \label{eq:net-smallest}
\end{align}
Because $p_n/M\to\kappa<1$, both $M/(M-p_n)$ and $p_n/(M-p_n)$ are uniformly bounded.  Choose
\begin{equation*}
    R=u^{-(M-p_n)/(2p_n)}.
\end{equation*}
For sufficiently small fixed $u_0$, this choice satisfies $R\geq R_0$. Substitution into~\eqref{eq:net-smallest} gives
\begin{equation*}
    \mathbb P\{s_{\min}(W)\leq u\sqrt M, \ \|W\|_{\op}\leq R\sqrt M\} \leq (C_2u)^{(M-p_n)/2}.
\end{equation*}
The remaining probability in~\eqref{eq:op-tail} is
\begin{equation*}
    \exp\{-cMu^{-(M-p_n)/p_n}\},
\end{equation*}
which, after reducing $u_0$ if necessary, is bounded by the same right-hand side uniformly for all sufficiently large $n$.  Hence
\begin{equation}
    \mathbb P\{s_{\min}(W)\leq u\sqrt M\} \leq (C_3u)^{(M-p_n)/2}. \label{eq:smin-tail-complete}
\end{equation}
Since $\widetilde\Sigma_M=M^{-1}W^\top W$, setting $u=\sqrt t$ proves \eqref{eq:true-tail}.

For the negative moments, put $\alpha_M=(M-p_n)/2$. For every fixed $Q>0$, $\alpha_M>2Q$ for all sufficiently large $n$. The layer-cake formula gives
\begin{align*}
    \mathbb E\lambda_{\min}^{-Q}(\widetilde\Sigma_M) &= Q\int_0^\infty s^{Q-1} \mathbb P\{\lambda_{\min}^{-1}(\widetilde\Sigma_M)>s\} \,\mathrm ds\\ &\leq 1+Q\int_1^\infty s^{Q-1}\mathbb P\{\lambda_{\min}(\widetilde\Sigma_M)<s^{-1}\} \,\mathrm ds.
\end{align*}
Let
\begin{equation*}
    s_\star=\max\{1,t_0^{-1},4C^2\}.
\end{equation*}
The contribution from $1\leq s\leq s_\star$ is bounded by $s_\star^Q-1$. For $s\geq s_\star$,~\eqref{eq:true-tail} gives
\begin{align*}
    Q\int_{s_\star}^\infty s^{Q-1} \mathbb P\{\lambda_{\min}(\widetilde\Sigma_M)<s^{-1}\} \,\mathrm ds &\leq QC^{\alpha_M} \int_{s_\star}^\infty s^{Q-1-\alpha_M/2}\,\mathrm ds\\ &= \frac{Q C^{\alpha_M}s_\star^{Q-\alpha_M/2}} {\alpha_M/2-Q}.
\end{align*}
Since $C/\sqrt{s_\star}\leq1/2$ and $\alpha_M\to\infty$, the last display is bounded uniformly in $n$. This proves the claimed negative-moment estimate.
\end{proof}

\begin{mylemma}[Lower tail of the generated whitened Gram matrix]
\label{lem:hd-general-generated-tail}
Let $M =M_n\in \{n-1,n\}$. Conditionally on $\widehat\mu_n$, let $X_1^*,\ldots,X_M^*\overset{\mathrm{i.i.d.}}\sim\widehat\mu_{n,X}$ and define
\begin{equation*}
    \widetilde\Sigma_M^* = \Sigma_n^{-1/2} \left(\frac1M\sum_{i=1}^MX_i^*X_i^{*\top}\right) \Sigma_n^{-1/2}.
\end{equation*}
There is $K_n=O_p(1)$ such that, with probability tending to one,
\begin{equation}
    \mathbb P^*\{\lambda_{\min}(\widetilde\Sigma_M^*)\leq t \mid\widehat\mu_n\} \leq (K_n\sqrt t)^{(M-p_n)/4}, \qquad 0<t<t_0. \label{eq:generated-tail}
\end{equation}
For each fixed $Q>0$,
\begin{equation}
    \mathbb E^*\{\lambda_{\min}^{-Q}(\widetilde\Sigma_M^*) \mid\widehat\mu_n\} =O_p(1), \label{eq:generated-negative}
\end{equation}
and the same conclusion holds uniformly over all $n$ leave-one-out matrices formed from $n-1$ generated observations.
\end{mylemma}

\begin{proof}
Let $h_n=\chi^2(\widehat\mu_{n,X},\mu_{n,X})=o_p(1)$.  The divergence is unchanged by the invertible transformation $x\mapsto\Sigma_n^{-1/2}x$.  If $R_X$ is the marginal likelihood ratio, then the likelihood ratio of the $M$-fold product law is
\begin{equation*}
    R_X^{(M)}(x_1,\ldots,x_M)=\prod_{i=1}^M R_X(x_i),
\end{equation*}
and independence gives
\begin{equation*}
    \mathbb E_{\mu_{n,X}^{\otimes M}}\{R_X^{(M)}\}^2=(1+h_n)^M.
\end{equation*}
Thus, for every event $E\subset(\mathbb R^{p_n})^M$, Cauchy--Schwarz gives
\begin{equation*}
    \widehat\mu_{n,X}^{\otimes M}(E) \leq \mu_{n,X}^{\otimes M}(E)^{1/2}(1+h_n)^{M/2}.
\end{equation*}
Applying this bound to the event in Lemma~\ref{lem:hd-general-true-tail} gives
\begin{align*}
    \mathbb P^*\{\lambda_{\min}(\widetilde\Sigma_M^*)\leq t \mid\widehat\mu_n\} &\leq (C\sqrt t)^{(M-p_n)/4}(1+h_n)^{M/2}\\ &= (K_n\sqrt t)^{(M-p_n)/4},
\end{align*}
where
\begin{equation*}
    K_n=C(1+h_n)^{2M/(M-p_n)}=O_p(1),
\end{equation*}
because $M/(M-p_n)=O(1)$ and $h_n=o_p(1)$.  This proves \eqref{eq:generated-tail}.

Put $\alpha_M=(M-p_n)/4$.  On every event where $K_n\leq K$ for a fixed $K$, the same layer-cake calculation as in Lemma~\ref{lem:hd-general-true-tail} gives, for $\alpha_M>2Q$,
\begin{equation*}
    \mathbb E^*\{\lambda_{\min}^{-Q}(\widetilde\Sigma_M^*) \mid\widehat\mu_n\} \leq C_{Q,K}.
\end{equation*}
Since $K_n=O_p(1)$, this is precisely~\eqref{eq:generated-negative}.

For leave-one-out matrices take $M=n-1$.  For each deleted index the same tail estimate holds.  A union bound multiplies its right-hand side by $n$. Because $\alpha_{n-1}\asymp n$ and $n^{1/\alpha_{n-1}}=O(1)$, the factor $n$ can be absorbed into a fixed enlargement of $K_n$.  Repeating the layer-cake calculation therefore gives the stated uniform leave-one-out negative-moment bound.
\end{proof}

\begin{mylemma}
\label{lem:hd-general-inverse-transfer}
Let
\begin{equation*}
    S_n^*=\sum_{i=1}^nX_i^*X_i^{*\top}, \qquad S_n=\sum_{i=1}^nX_iX_i^\top.
\end{equation*}
Then
\begin{equation}
    \mathbb E^*\{c_n^\top(S_n^*)^{-1}c_n\mid\widehat\mu_n\} = \mathbb E\{c_n^\top S_n^{-1}c_n\}+o_p(n^{-1}). \label{eq:inverse-transfer}
\end{equation}
\end{mylemma}

\begin{proof}
Take i.i.d.\ rowwise optimal $W_4$ couplings
\begin{equation*}
    W_i^*=W_i+\Delta_i, \qquad W_i^*=\Sigma_n^{-1/2}X_i^*, \qquad W_i=\Sigma_n^{-1/2}X_i,
\end{equation*}
and put
\begin{equation*}
    A_n^*=\frac1n\sum_{i=1}^nW_i^*W_i^{*\top}, \qquad A_n=\frac1n\sum_{i=1}^nW_iW_i^\top.
\end{equation*}
Lemma~\ref{lem:hd-general-pde} and the spectral bounds on $\Sigma_n$ imply $\delta_n:=\mathbb E^*\|\Delta_1\|_2^4=o_p(1)$.  Let $W$ and $\Delta$ be the $n\times p_n$ matrices with rows $W_i^\top$ and $\Delta_i^\top$.  Then
\begin{equation}
    A_n^*-A_n = \frac1n\left(W^\top\Delta+\Delta^\top W+\Delta^\top\Delta\right). \label{eq:gram-perturbation}
\end{equation}
The operator-norm moment bound established in the proof of Lemma~\ref{lem:hd-general-true-tail} gives
\begin{equation*}
    \mathbb E\left(\frac{\|W\|_{\op}}{\sqrt n}\right)^4=O(1),
\end{equation*}
while independence of the rowwise couplings gives
\begin{align*}
    \mathbb E^*\left(\frac{\|\Delta\|_F^2}{n}\right)^2 &= \frac1{n^2} \mathbb E^*\left(\sum_{i=1}^n\|\Delta_i\|_2^2\right)^2\leq \frac1n\mathbb E^*\|\Delta_1\|_2^4 +\left(\mathbb E^*\|\Delta_1\|_2^2\right)^2 =o_p(1).
\end{align*}
Since $\|\Delta\|_{\op}\leq\|\Delta\|_F$, equation \eqref{eq:gram-perturbation} and Cauchy--Schwarz yield
\begin{equation}
    \mathbb E^*\|A_n^*-A_n\|_{\op}^2=o_p(1). \label{eq:A-difference}
\end{equation}
Here and below the expectation is under the joint rowwise coupling; its $A_n^*$-marginal is the bootstrap law and its $A_n$-marginal is the true law.

Lemmas~\ref{lem:hd-general-true-tail} and~\ref{lem:hd-general-generated-tail} imply
\begin{equation*}
    \mathbb E\|A_n^{-1}\|_{\op}^4=O(1), \qquad \mathbb E^*\|(A_n^*)^{-1}\|_{\op}^4=O_p(1).
\end{equation*}
Using
\begin{equation*}
    (A_n^*)^{-1}-A_n^{-1} = (A_n^*)^{-1}(A_n-A_n^*)A_n^{-1},
\end{equation*}
H\"older's inequality with exponents $4,2,4$ and \eqref{eq:A-difference} yield
\begin{equation*}
    \begin{aligned}
        \mathbb E^*\|(A_n^*)^{-1}-A_n^{-1}\|_{\op} &\leq \left(\mathbb E^*\|(A_n^*)^{-1}\|_{\op}^4\right)^{1/4} \left(\mathbb E^*\|A_n^*-A_n\|_{\op}^2\right)^{1/2} \left(\mathbb E\|A_n^{-1}\|_{\op}^4\right)^{1/4} =o_p(1).
    \end{aligned}
\end{equation*}
For $\widetilde c_n=\Sigma_n^{-1/2}c_n$, we have
\begin{equation*}
    c_n^\top(S_n^*)^{-1}c_n = \frac1n\widetilde c_n^\top(A_n^*)^{-1}\widetilde c_n, \qquad c_n^\top S_n^{-1}c_n = \frac1n\widetilde c_n^\top A_n^{-1}\widetilde c_n.
\end{equation*}
Since
\begin{equation*}
    \|\widetilde c_n\|_2^2 =c_n^\top\Sigma_n^{-1}c_n \leq c_\Sigma^{-1},
\end{equation*}
we conclude that
\begin{align*}
    \left| \mathbb E^*\{c_n^\top(S_n^*)^{-1}c_n\mid\widehat\mu_n\} -\mathbb E\{c_n^\top S_n^{-1}c_n\} \right|\leq \frac{\|\widetilde c_n\|_2^2}{n} \mathbb E^*\|(A_n^*)^{-1}-A_n^{-1}\|_{\op} =o_p(n^{-1}),
\end{align*}
which proves~\eqref{eq:inverse-transfer}.
\end{proof}

\subsection{Proof of variance consistency}

\begin{proof}[Proof of Lemma~\ref{lem:hd-joint-W4}]
The asserted joint Wasserstein convergence is the third conclusion of Lemma~\ref{lem:hd-general-pde}.
\end{proof}

\begin{proof}[Proof of Lemma~\ref{lem:hd-cond-W4}]
This is precisely Lemma~\ref{lem:hd-general-conditional}.
\end{proof}

\begin{proof}[Proof of Theorem~\ref{thm:hd-actual-var}]
All bootstrap expectations and variances are conditional on $\widehat\mu_n$.  By total variance,
\begin{align}
    \Var^*(c_n^\top\widehat\beta^*) &= \mathbb E^*\!\left[ \Var^*(c_n^\top\widehat\beta^*\mid X^*) \right] \notag\\ &\quad+ \Var^*\!\left[ \mathbb E^*(c_n^\top\widehat\beta^*\mid X^*) \right]. \label{eq:total-variance}
\end{align}
Put
\begin{equation*}
    S_n^*=\sum_{i=1}^nX_i^*X_i^{*\top}, \qquad h^*=X^*(S_n^*)^{-1}c_n.
\end{equation*}
Conditionally on $X^*$, the $Y_i^*$ are independent, hence
\begin{equation*}
    \Var^*(c_n^\top\widehat\beta^*\mid X^*) = \sum_{i=1}^n(h_i^*)^2v_n^*(X_i^*).
\end{equation*}
Therefore
\begin{align}
    \mathbb E^*\!\left[ \Var^*(c_n^\top\widehat\beta^*\mid X^*) \right] &= \sigma_n^2\mathbb E^*\{c_n^\top(S_n^*)^{-1}c_n\} + \mathbb E^*\sum_{i=1}^n (h_i^*)^2\{v_n^*(X_i^*)-\sigma_n^2\}. \label{eq:first-variance-term}
\end{align}

To control the last term, fix $i=1$, let
\begin{equation*}
    S_{-1}^*=\sum_{j=2}^nX_j^*X_j^{*\top}, \qquad a_{-1}^*=(S_{-1}^*)^{-1}c_n.
\end{equation*}
Sherman--Morrison gives $|h_1^*|\leq|X_1^{*\top}a_{-1}^*|$. Independence of $X_1^*$ and  $a_{-1}^*$, Cauchy--Schwarz, Lemma~\ref{lem:hd-general-conditional}, and the $M=n-1$ negative-moment estimate in Lemma~\ref{lem:hd-general-generated-tail} imply
\begin{equation*}
    \mathbb E^*\|a_{-1}^*\|_2^2=O_p(n^{-2}).
\end{equation*}
Indeed, if
\begin{equation*}
    A_{-1}^* = \frac1{n-1}\Sigma_n^{-1/2}S_{-1}^*\Sigma_n^{-1/2},
\end{equation*}
then
\begin{equation*}
    \|a_{-1}^*\|_2 \leq \frac{C}{n-1}\|(A_{-1}^*)^{-1}\|_{\op}.
\end{equation*}
Furthermore, conditionally on $a_{-1}^*$,
\begin{align*}
    \mathbb E^*\left[ |X_1^{*\top}a_{-1}^*|^2 |v_n^*(X_1^*)-\sigma_n^2| \mid a_{-1}^* \right]\leq M_{4,n}^{1/2}\|a_{-1}^*\|_2^2\zeta_n.
\end{align*}
Exchangeability therefore gives
\begin{equation*}
    \left| \mathbb E^*\sum_{i=1}^n (h_i^*)^2\{v_n^*(X_i^*)-\sigma_n^2\} \right| \leq nM_{4,n}^{1/2}\zeta_n \mathbb E^*\|a_{-1}^*\|_2^2 =o_p(n^{-1}).
\end{equation*}
Lemma~\ref{lem:hd-general-inverse-transfer} now yields
\begin{equation}
    \mathbb E^*\!\left[ \Var^*(c_n^\top\widehat\beta^*\mid X^*) \right] = \sigma_n^2\mathbb E\{c_n^\top S_n^{-1}c_n\}+o_p(n^{-1}). \label{eq:first-term-final}
\end{equation}

For the second term in~\eqref{eq:total-variance},
\begin{equation*}
    \mathbb E^*(c_n^\top\widehat\beta^*\mid X^*) = c_n^\top\beta_n +c_n^\top(S_n^*)^{-1} \sum_{i=1}^nX_i^*r_n(X_i^*).
\end{equation*}
We provide the leave-one-out argument in detail.  Put $Q_n=\widehat\mu_{n,X}$ and, for $M=n-1$, define
\begin{equation*}
    \overline A_M = \frac1M\sum_{j=2}^nX_j^*X_j^{*\top}, \qquad \overline g_M = \frac1M\sum_{j=2}^nX_j^*r_n(X_j^*), \qquad \theta_M=\overline A_M^{-1}\overline g_M.
\end{equation*}
We first show
\begin{equation}
    \mathbb E^*\|\theta_M\|_2^2=o_p(1). \label{eq:theta-small}
\end{equation}
For every unit vector $u$, Cauchy--Schwarz and~\eqref{eq:M4} give
\begin{equation*}
    |u^\top\mathbb E_{Q_n}\{Xr_n(X)\}|^2 \leq \mathbb E_{Q_n}|u^\top X|^2\,\mathbb E_{Q_n}r_n(X)^2 \leq O_p(1)\eta_n^{1/2}.
\end{equation*}
Hence
\begin{equation*}
    \|\mathbb E_{Q_n}\{Xr_n(X)\}\|_2^2 =O_p(\eta_n^{1/2}).
\end{equation*}
Also,~\eqref{eq:M4} implies
$\mathbb E_{Q_n}\|X\|_2^4=O_p(p_n^2)$, and therefore
\begin{equation*}
    \frac1M\mathbb E_{Q_n}\|Xr_n(X)\|_2^2 \leq \frac1M \{\mathbb E_{Q_n}\|X\|_2^4\}^{1/2}\eta_n^{1/2} =O_p(\eta_n^{1/2}).
\end{equation*}
It follows that
\begin{equation*}
    \mathbb E^*\|\overline g_M\|_2^2=O_p(\eta_n^{1/2}).
\end{equation*}
For the fourth moment, the Hilbert-space fourth-moment expansion for an i.i.d.\ average gives
\begin{align*}
    \mathbb E^*\|\overline g_M-\mathbb E^*\overline g_M\|_2^4 &\leq \frac{C}{M^2} \{\mathbb E_{Q_n}\|Xr_n(X)\|_2^2\}^2 +\frac{C}{M^3} \mathbb E_{Q_n}\|Xr_n(X)\|_2^4.
\end{align*}
Now $\mathbb E_{Q_n}|r_n(X)|^8=O_p(1)$ by conditional Jensen and the eighth-moment bound for the generated regression residual established in the proof of Lemma~\ref{lem:hd-general-conditional}. To justify the required design moment, let
\begin{equation*}
    R_{n,X}=\frac{\mathrm dQ_n}{\mathrm d\mu_{n,X}}.
\end{equation*}
By Lemma~\ref{lem:hd-general-pde}, $\mathbb E_{\mu_{n,X}}R_{n,X}^2=1+o_p(1)$. Uniform strong log-concavity and the covariance bounds imply
\begin{equation*}
    \mathbb E_{\mu_{n,X}}\|X\|_2^{16}\leq Cp_n^8.
\end{equation*}
Indeed, every coordinate has a uniformly bounded sixteenth moment and
\begin{equation*}
    \left(\sum_{j=1}^{p_n}X_j^2\right)^8 \leq p_n^7\sum_{j=1}^{p_n}|X_j|^{16}.
\end{equation*}
Hence Cauchy--Schwarz gives
\begin{align*}
    \mathbb E_{Q_n}\|X\|_2^8 &= \mathbb E_{\mu_{n,X}}\{R_{n,X}\|X\|_2^8\}\leq \{\mathbb E_{\mu_{n,X}}R_{n,X}^2\}^{1/2} \{\mathbb E_{\mu_{n,X}}\|X\|_2^{16}\}^{1/2} =O_p(p_n^4).
\end{align*}
It follows that
\begin{equation*}
    \mathbb E_{Q_n}\|Xr_n(X)\|_2^4 \leq \{\mathbb E_{Q_n}\|X\|_2^8\}^{1/2} \{\mathbb E_{Q_n}|r_n(X)|^8\}^{1/2} =O_p(p_n^2).
\end{equation*}
Moreover,
\begin{equation*}
    \|\mathbb E^*\overline g_M\|_2^4 = \|\mathbb E_{Q_n}\{Xr_n(X)\}\|_2^4 =O_p(\eta_n).
\end{equation*}
Combining these bounds with the centered fourth-moment estimate above gives
\begin{equation}
    \mathbb E^*\|\overline g_M\|_2^4 =O_p(\eta_n+n^{-1})=o_p(1). \label{eq:g-fourth}
\end{equation}
To transfer the relative-Gram negative moments to $\overline A_M$, define
\begin{equation*}
    \widetilde A_M = \Sigma_n^{-1/2}\overline A_M\Sigma_n^{-1/2}.
\end{equation*}
Then
\begin{equation*}
    \overline A_M^{-1} = \Sigma_n^{-1/2}\widetilde A_M^{-1}\Sigma_n^{-1/2}, \qquad \|\overline A_M^{-1}\|_{\op} \leq c_\Sigma^{-1}\|\widetilde A_M^{-1}\|_{\op}.
\end{equation*}
Thus Lemma~\ref{lem:hd-general-generated-tail} and H\"older's inequality give
\begin{equation*}
    \mathbb E^*\|\theta_M\|_2^2 \leq \{\mathbb E^*\|\overline A_M^{-1}\|_{\op}^4\}^{1/2} \{\mathbb E^*\|\overline g_M\|_2^4\}^{1/2} =o_p(1),
\end{equation*}
which proves~\eqref{eq:theta-small}.

Let
\begin{equation*}
    F_n(X_1^*,\ldots,X_n^*) = c_n^\top(S_n^*)^{-1} \sum_{i=1}^nX_i^*r_n(X_i^*).
\end{equation*}
For the sample with the first observation removed, write
\begin{equation*}
    F_{-1}=c_n^\top\theta_M, \qquad a_{-1}^*=(S_{-1}^*)^{-1}c_n.
\end{equation*}
Sherman--Morrison gives the exact identity
\begin{equation}
    F_n-F_{-1} = \frac{X_1^{*\top}a_{-1}^*} {1+X_1^{*\top}(S_{-1}^*)^{-1}X_1^*} \{r_n(X_1^*)-X_1^{*\top}\theta_M\}. \label{eq:loo-increment}
\end{equation}
Conditionally on the leave-one-out sample, the squared right-hand side is bounded in expectation by
\begin{equation*}
    C M_{4,n}^{1/2}\|a_{-1}^*\|_2^2\eta_n^{1/2} +C M_{4,n}\|a_{-1}^*\|_2^2\|\theta_M\|_2^2.
\end{equation*}
The first term has expectation $o_p(n^{-2})$.  For the second, use
\begin{equation*}
    \|a_{-1}^*\|_2^2\|\theta_M\|_2^2 \leq \frac{C}{M^2} \|\overline A_M^{-1}\|_{\op}^4 \|\overline g_M\|_2^2,
\end{equation*}
together with the negative moments of all fixed orders and \eqref{eq:g-fourth}.  Its expectation is also $o_p(n^{-2})$.  If $X_1^{*\prime}$ is an independent replacement, the two versions of $F_n$ share the same $F_{-1}$, so~\eqref{eq:loo-increment} and the Efron--Stein inequality imply
\begin{equation*}
    \Var^*(F_n) \leq \frac n2\mathbb E^*(F_n-F_n^{(1)})^2 \leq 2n\mathbb E^*(F_n-F_{-1})^2 =o_p(n^{-1}).
\end{equation*}
Thus
\begin{equation}
    \Var^*\!\left[ \mathbb E^*(c_n^\top\widehat\beta^*\mid X^*) \right] =o_p(n^{-1}). \label{eq:second-term-final}
\end{equation}

For the original sample, independence and homoskedasticity give the exact identity
\begin{equation}
    \Var(c_n^\top\widehat\beta) = \sigma_n^2\mathbb E\{c_n^\top S_n^{-1}c_n\}. \label{eq:true-variance}
\end{equation}
Recall that $\|c_n\|_2=1$. This quantity is of order $n^{-1}$. Indeed, if
\begin{equation*}
    A_n = \frac1n\Sigma_n^{-1/2}S_n\Sigma_n^{-1/2},
\end{equation*}
then Lemma~\ref{lem:hd-general-true-tail} gives
\begin{align*}
    \mathbb E\{c_n^\top S_n^{-1}c_n\} &\leq \frac{\|\Sigma_n^{-1/2}c_n\|_2^2}{n} \mathbb E\lambda_{\min}^{-1}(A_n)\leq \frac{1}{nc_\Sigma} \mathbb E\lambda_{\min}^{-1}(A_n) =O(n^{-1}).
\end{align*}
For the lower bound, operator convexity of $A\mapsto A^{-1}$ gives
\begin{equation*}
    \mathbb E\{c_n^\top S_n^{-1}c_n\} \geq c_n^\top(\mathbb E S_n)^{-1}c_n = \frac1n c_n^\top\Sigma_n^{-1}c_n \geq \frac1{nC_\Sigma}.
\end{equation*}
Combining~\eqref{eq:total-variance}, \eqref{eq:first-term-final}, \eqref{eq:second-term-final}, and~\eqref{eq:true-variance} proves \eqref{eq:variance-ratio}.
\end{proof}

\section{Proof of the minimax upper bound}
\begin{mylemma}[Soft-thresholding inequality]
\label{lem:soft-thresholding}
Let $u,v\in\mathbb R^p$, and let $\lambda>0$. Define the coordinatewise soft-thresholding operator
\begin{equation*}
    S_\lambda(u)_j = \sgn(u_j)(|u_j|-\lambda)_+, \qquad j=1,\dots,p.
\end{equation*}
If $\|u-v\|_\infty\leq \lambda$, then, for every $j=1,\dots,p$,
\begin{equation*}
    |S_\lambda(u)_j-v_j| \leq \min\{|v_j|,2\lambda\}.
\end{equation*}
Consequently,
\begin{equation*}
    \|S_\lambda(u)-v\|_2^2 \leq 2\lambda\|v\|_1.
\end{equation*}
\end{mylemma}

\begin{proof}
Fix $j\in\{1,\dots,p\}$. By the assumption $\|u-v\|_\infty\leq\lambda$, we have $|u_j-v_j|\leq \lambda$. We prove
\begin{equation*}
    |S_\lambda(u)_j-v_j| \leq \min\{|v_j|,2\lambda\}
\end{equation*}
by considering three cases.

First, suppose that $|u_j|\leq \lambda.$ Then, by the definition of $S_\lambda$, $S_\lambda(u)_j=0.$ Therefore
\begin{equation*}
    |S_\lambda(u)_j-v_j| = |v_j|.
\end{equation*}
Moreover,
\begin{equation*}
    |v_j| \leq |v_j-u_j|+|u_j| \leq \lambda+\lambda = 2\lambda.
\end{equation*}
Hence
\begin{equation*}
    |S_\lambda(u)_j-v_j| \leq \min\{|v_j|,2\lambda\}.
\end{equation*}

Second, suppose that $u_j>\lambda$. Then
\begin{equation*}
    S_\lambda(u)_j=u_j-\lambda>0.
\end{equation*}
Since $|u_j-v_j|\leq\lambda$, we have
\begin{equation*}
    v_j\geq u_j-\lambda=S_\lambda(u)_j.
\end{equation*}
Thus $0<S_\lambda(u)_j\leq v_j.$ Therefore
\begin{equation*}
    |S_\lambda(u)_j-v_j| = v_j-S_\lambda(u)_j \leq |v_j|.
\end{equation*}
On the other hand,
\begin{equation*}
    \begin{aligned}
        |S_\lambda(u)_j-v_j| \leq |S_\lambda(u)_j-u_j| + |u_j-v_j|= \lambda+|u_j-v_j|\leq 2\lambda.
    \end{aligned}
\end{equation*}
Hence
\begin{equation*}
    |S_\lambda(u)_j-v_j| \leq \min\{|v_j|,2\lambda\}.
\end{equation*}

Third, suppose that $u_j<-\lambda$. Then
\begin{equation*}
    S_\lambda(u)_j=u_j+\lambda<0.
\end{equation*}
Since $|u_j-v_j|\leq\lambda$, we have
\begin{equation*}
    v_j\leq u_j+\lambda=S_\lambda(u)_j.
\end{equation*}
Thus $v_j\leq S_\lambda(u)_j<0$. Therefore
\begin{equation*}
    |S_\lambda(u)_j-v_j| = S_\lambda(u)_j-v_j \leq |v_j|.
\end{equation*}
Again,
\begin{equation*}
    \begin{aligned}
        |S_\lambda(u)_j-v_j| \leq |S_\lambda(u)_j-u_j| + |u_j-v_j|= \lambda+|u_j-v_j|\leq 2\lambda.
    \end{aligned}
\end{equation*}
Hence
\begin{equation*}
    |S_\lambda(u)_j-v_j| \leq \min\{|v_j|,2\lambda\}.
\end{equation*}

Combining the three cases gives, for every $j=1,\dots,p$,
\begin{equation*}
    |S_\lambda(u)_j-v_j| \leq \min\{|v_j|,2\lambda\}.
\end{equation*}
Therefore
\begin{equation*}
    \begin{aligned}
        \|S_\lambda(u)-v\|_2^2 = \sum_{j=1}^p |S_\lambda(u)_j-v_j|^2\leq \sum_{j=1}^p \min\{v_j^2,4\lambda^2\}.
    \end{aligned}
\end{equation*}
For any $a\in\mathbb R$,
\begin{equation*}
    \min\{a^2,4\lambda^2\} \leq 2\lambda |a|.
\end{equation*}
Indeed, if $|a|\leq2\lambda$, then
\begin{equation*}
    \min\{a^2,4\lambda^2\} = a^2 \leq 2\lambda |a|;
\end{equation*}
whereas if $|a|>2\lambda$, then
\begin{equation*}
    \min\{a^2,4\lambda^2\} = 4\lambda^2 \leq 2\lambda |a|.
\end{equation*}
Applying this inequality with $a=v_j$, we obtain
\begin{equation*}
    \begin{aligned}
        \|S_\lambda(u)-v\|_2^2 &\leq 2\lambda \sum_{j=1}^p |v_j|= 2\lambda\|v\|_1.
    \end{aligned}
\end{equation*}
This proves the lemma.
\end{proof}
\begin{mylemma}
\label{lem:l1-fourth-regression}
Assume that
\begin{equation*}
    X_i\sim N(0,I_{p_n}), \qquad Y_i=X_i^\top\beta+\varepsilon_i, \qquad \varepsilon_i\sim N(0,\sigma^2),
\end{equation*}
independently for $i=1,\dots,n$, where $\varepsilon_i$ is independent of $X_i$.
Suppose that
\begin{equation*}
    \|\beta\|_1\leq B_1, \qquad 0<\sigma_{\min}^2\leq \sigma^2\leq \sigma_{\max}^2<\infty.
\end{equation*}
Then there exists an estimator $\widehat\beta\in\{b:\|b\|_1\leq B_1\}$ such that
\begin{equation*}
    \sup_{\substack{\|\beta\|_1\leq B_1\\ \sigma^2\in[\sigma_{\min}^2,\sigma_{\max}^2]}} \mathbb E_{\beta,\sigma^2} \|\widehat\beta-\beta\|_2^4 \leq C\frac{\log p}{n},
\end{equation*}
where $C<\infty$ depends only on $B_1$ and $\sigma_{\max}$. Moreover, the same estimator satisfies
\begin{equation*}
    \sup_{\substack{\|\beta\|_1\leq B_1\\ \sigma^2\in[\sigma_{\min}^2,\sigma_{\max}^2]}} \mathbb E_{\beta,\sigma^2} \|\widehat\beta-\beta\|_2^8 \leq C\left(\frac{\log p}{n}\right)^2.
\end{equation*}
\end{mylemma}

\begin{proof}
Define the moment estimator
\begin{equation*}
    \widehat m = \frac1n \sum_{i=1}^n X_iY_i.
\end{equation*}
Since
\begin{equation*}
    Y_i=X_i^\top\beta+\varepsilon_i,
\end{equation*}
we have
\begin{equation*}
    \mathbb E_{\beta,\sigma^2}[X_iY_i] = \mathbb E[X_iX_i^\top]\beta + \mathbb E[X_i\varepsilon_i] = \beta.
\end{equation*}
Thus $\widehat m$ is an unbiased estimator of $\beta$.

Let
\begin{equation*}
    \lambda = A\sqrt{\frac{\log p}{n}},
\end{equation*}
where $A>0$ is a sufficiently large constant to be chosen later. Let $S_\lambda$ denote the coordinatewise soft-thresholding operator, and set
\begin{equation*}
    \widetilde\beta = S_\lambda(\widehat m).
\end{equation*}
Define
\begin{equation*}
    \widehat\beta = \Pi_{\{b:\|b\|_1\leq B_1\}}(\widetilde\beta),
\end{equation*}
where the projection is the Euclidean projection onto the closed convex set $\{b:\|b\|_1\leq B_1\}$.

We first control the entrywise fluctuation of $\widehat m$. For each $j=1,\dots,p$,
\begin{equation*}
    \widehat m_j-\beta_j = \frac1n \sum_{i=1}^n \left( X_{ij}Y_i-\beta_j \right).
\end{equation*}
The summands are independent and mean zero. Moreover, since $(X_{ij},Y_i)$ is a centered Gaussian vector and
\begin{equation*}
    \Var(Y_i) = \|\beta\|_2^2+\sigma^2 \leq B_1^2+\sigma_{\max}^2,
\end{equation*}
the random variable $X_{ij}Y_i-\beta_j$ is sub-exponential with a sub-exponential norm bounded by a constant depending only on $B_1$ and $\sigma_{\max}$. Hence Bernstein's inequality gives
\begin{equation*}
    \mathbb P_{\beta,\sigma^2} \left( |\widehat m_j-\beta_j|>t \right) \leq 2\exp \left[ -cn\min(t^2,t) \right],
\end{equation*}
uniformly over $j$, $\|\beta\|_1\leq B_1$, and $\sigma^2\in[\sigma_{\min}^2,\sigma_{\max}^2]$.

Taking $t=\lambda$, and using $\lambda\to0$ for large $n$, we obtain
\begin{equation*}
    \mathbb P_{\beta,\sigma^2} \left( |\widehat m_j-\beta_j|>\lambda \right) \leq 2\exp(-cA^2\log p).
\end{equation*}
By the union bound,
\begin{equation*}
    \mathbb P_{\beta,\sigma^2} \left( \|\widehat m-\beta\|_\infty>\lambda \right) \leq 2p^{1-cA^2}.
\end{equation*}
Thus, by choosing $A$ sufficiently large, for any fixed $K>2$,
\begin{equation*}
    \mathbb P_{\beta,\sigma^2} \left( \|\widehat m-\beta\|_\infty>\lambda \right) \leq Cp^{-K}.
\end{equation*}
Let
\begin{equation*}
    \mathcal E_n = \left\{ \|\widehat m-\beta\|_\infty\leq\lambda \right\}.
\end{equation*}

On $\mathcal E_n$, by Lemma~\ref{lem:soft-thresholding},
\begin{equation*}
    \|\widetilde\beta-\beta\|_2^2 = \|S_\lambda(\widehat m)-\beta\|_2^2 \leq 2\lambda\|\beta\|_1 \leq 2B_1\lambda.
\end{equation*}
Since $\beta\in\{b:\|b\|_1\leq B_1\}$, and since Euclidean projection onto a closed convex set is non-expansive,
\begin{equation*}
    \|\widehat\beta-\beta\|_2 = \left\| \Pi_{\{b:\|b\|_1\leq B_1\}}(\widetilde\beta) - \Pi_{\{b:\|b\|_1\leq B_1\}}(\beta) \right\|_2 \leq \|\widetilde\beta-\beta\|_2.
\end{equation*}
Therefore, on $\mathcal E_n$,
\begin{equation*}
    \|\widehat\beta-\beta\|_2^2 \leq 2B_1\lambda.
\end{equation*}
Consequently,
\begin{equation*}
    \|\widehat\beta-\beta\|_2^4 \leq C B_1^2\lambda^2,
\end{equation*}
and
\begin{equation*}
    \|\widehat\beta-\beta\|_2^8 \leq C B_1^4\lambda^4.
\end{equation*}

On the complement $\mathcal E_n^c$, both $\widehat\beta$ and $\beta$ belong to the $\ell_1$-ball of radius $B_1$. Hence
\begin{equation*}
    \|\widehat\beta-\beta\|_2 \leq \|\widehat\beta-\beta\|_1 \leq \|\widehat\beta\|_1+\|\beta\|_1 \leq 2B_1.
\end{equation*}
Thus
\begin{equation*}
    \|\widehat\beta-\beta\|_2^4 \leq C B_1^4, \qquad \|\widehat\beta-\beta\|_2^8 \leq C B_1^8
\end{equation*}
on $\mathcal E_n^c$.

Combining the bounds on $\mathcal E_n$ and $\mathcal E_n^c$, we get
\begin{equation*}
    \begin{aligned}
        \mathbb E_{\beta,\sigma^2} \|\widehat\beta-\beta\|_2^4 &\leq C B_1^2\lambda^2 + C B_1^4 \mathbb P_{\beta,\sigma^2}(\mathcal E_n^c)\leq C\frac{\log p}{n} + Cp^{-K}.
    \end{aligned}
\end{equation*}
Choose $K>1$. Since $p/n\to\kappa\in(0,1)$,
\begin{equation*}
    p^{-K} = o\left( \frac{\log p}{n} \right).
\end{equation*}
Therefore
\begin{equation*}
    \mathbb E_{\beta,\sigma^2} \|\widehat\beta-\beta\|_2^4 \leq C\frac{\log p}{n}.
\end{equation*}

Similarly, for the eighth moment,
\begin{equation*}
    \begin{aligned}
        \mathbb E_{\beta,\sigma^2} \|\widehat\beta-\beta\|_2^8 &\leq C B_1^4\lambda^4 + C B_1^8 \mathbb P_{\beta,\sigma^2}(\mathcal E_n^c)\leq C\left(\frac{\log p}{n}\right)^2 + Cp^{-K}.
    \end{aligned}
\end{equation*}
Choose $K>2$. Since $p/n\to\kappa\in(0,1)$,
\begin{equation*}
    p^{-K} = o\left[ \left( \frac{\log p}{n} \right)^2 \right].
\end{equation*}
Thus
\begin{equation*}
    \mathbb E_{\beta,\sigma^2} \|\widehat\beta-\beta\|_2^8 \leq C\left(\frac{\log p}{n}\right)^2.
\end{equation*}
Taking the supremum over $\|\beta\|_1\leq B_1$ and $\sigma^2\in[\sigma_{\min}^2,\sigma_{\max}^2]$ proves the lemma.
\end{proof}

\begin{proof}[Proof of Theorem~\ref{thm:gaussian-linear-score-upper-unknown-var}]
Split the sample into two independent parts $I_1$ and $I_2$, with
\begin{equation*}
    n_1=|I_1|, \qquad n_2=|I_2|, \qquad n_1\asymp n_2\asymp n.
\end{equation*}
Using only the first subsample $I_1$, construct the estimator $\widehat\beta$ as in Lemma~\ref{lem:l1-fourth-regression}, with $n$ replaced by $n_1$. Since $n_1\asymp n$, Lemma~\ref{lem:l1-fourth-regression} gives
\begin{equation*}
    \sup_{\substack{\|\beta\|_1\leq B_1\\ \sigma^2\in[\sigma_{\min}^2,\sigma_{\max}^2]}} \mathbb E_{\beta,\sigma^2} \|\widehat\beta-\beta\|_2^4 \leq C\frac{\log p}{n},
\end{equation*}
and
\begin{equation*}
    \sup_{\substack{\|\beta\|_1\leq B_1\\ \sigma^2\in[\sigma_{\min}^2,\sigma_{\max}^2]}} \mathbb E_{\beta,\sigma^2} \|\widehat\beta-\beta\|_2^8 \leq C\left(\frac{\log p}{n}\right)^2.
\end{equation*}

Next estimate the noise variance using the second subsample $I_2$. Define
\begin{equation*}
    \overline\sigma^2 = \frac1{n_2} \sum_{i\in I_2} \left( Y_i-X_i^\top\widehat\beta \right)^2.
\end{equation*}
Then truncate it to the parameter interval:
\begin{equation*}
    \widehat\sigma^2 = \Pi_{[\sigma_{\min}^2,\sigma_{\max}^2]} \left( \overline\sigma^2 \right).
\end{equation*}
Since $\sigma^2\in[\sigma_{\min}^2,\sigma_{\max}^2]$, the projection is non-expansive, and hence
\begin{equation*}
    |\widehat\sigma^2-\sigma^2| \leq |\overline\sigma^2-\sigma^2|.
\end{equation*}

Let
\begin{equation*}
    \delta = \widehat\beta-\beta.
\end{equation*}
Conditionally on the first subsample $I_1$, the vector $\delta$ is fixed and the second subsample is independent of $\delta$. For $i\in I_2$,
\begin{equation*}
    Y_i-X_i^\top\widehat\beta = \varepsilon_i-X_i^\top\delta.
\end{equation*}
Since $X_i\sim N(0,I_{p_n})$, $\varepsilon_i\sim N(0,\sigma^2)$, and they are independent, conditionally on $I_1$,
\begin{equation*}
    \varepsilon_i-X_i^\top\delta \sim N(0,\sigma^2+\|\delta\|_2^2).
\end{equation*}
Hence
\begin{equation*}
    \overline\sigma^2 = \left( \sigma^2+\|\delta\|_2^2 \right) \frac{\chi^2_{n_2}}{n_2}
\end{equation*}
conditionally on $I_1$. Since
\begin{equation*}
    \|\delta\|_2 \leq \|\widehat\beta-\beta\|_1 \leq 2B_1,
\end{equation*}
we have
\begin{equation*}
    \sigma^2+\|\delta\|_2^2 \leq \sigma_{\max}^2+4B_1^2.
\end{equation*}
Using the fourth moment bound for the centered chi-square average,
\begin{equation*}
    \mathbb E \left| \frac{\chi^2_{n_2}}{n_2}-1 \right|^4 \leq \frac{C}{n_2^2},
\end{equation*}
we obtain
\begin{equation*}
    \mathbb E_{\beta,\sigma^2} \left[ \left| \overline\sigma^2- \left( \sigma^2+\|\delta\|_2^2 \right) \right|^4 \,\middle|\, I_1 \right] \leq \frac{C}{n_2^2}.
\end{equation*}
Moreover,
\begin{equation*}
    \left| \left( \sigma^2+\|\delta\|_2^2 \right) - \sigma^2 \right|^4 = \|\delta\|_2^8.
\end{equation*}
Therefore
\begin{equation*}
    \begin{aligned}
        \mathbb E_{\beta,\sigma^2} |\widehat\sigma^2-\sigma^2|^4 &\leq \mathbb E_{\beta,\sigma^2} |\overline\sigma^2-\sigma^2|^4  \\ &\leq Cn^{-2} + C \mathbb E_{\beta,\sigma^2} \|\widehat\beta-\beta\|_2^8  \\ &\leq Cn^{-2} + C\left( \frac{\log p}{n} \right)^2.
    \end{aligned}
\end{equation*}
Since $p/n\to\kappa\in(0,1)$, this implies
\begin{equation*}
    \mathbb E_{\beta,\sigma^2} |\widehat\sigma^2-\sigma^2|^4 \leq C\frac{\log p}{n}.
\end{equation*}

We now control the induced score error. For parameters $(\beta,\sigma^2)$, write
\begin{equation*}
    A_{\beta,\sigma^2} =
    \begin{pmatrix}
        I_p & \beta\\ \beta^\top & \|\beta\|_2^2+\sigma^2
    \end{pmatrix}
    .
\end{equation*}
The corresponding OU covariance is
\begin{equation*}
    A_{\beta,\sigma^2,t} = I_{p+1} + e^{-2t} \left( A_{\beta,\sigma^2}-I_{p+1} \right).
\end{equation*}
Thus
\begin{equation*}
    s_{\beta,\sigma^2,t}(z) = -A_{\beta,\sigma^2,t}^{-1}z.
\end{equation*}
Define the plug-in score estimator
\begin{equation*}
    \widehat s_t(z) = s_{\widehat\beta,\widehat\sigma^2,t}(z).
\end{equation*}

Since
\begin{equation*}
    \|\beta\|_1\leq B_1, \qquad \|\widehat\beta\|_1\leq B_1, \qquad \sigma^2,\widehat\sigma^2\in[\sigma_{\min}^2,\sigma_{\max}^2],
\end{equation*}
the covariance matrices $A_{\beta,\sigma^2,t}$ and $A_{\widehat\beta,\widehat\sigma^2,t}$ have eigenvalues uniformly bounded above and below by constants depending only on $B_1,\sigma_{\min}$, and $\sigma_{\max}$. Hence their inverse operator norms are uniformly bounded.

By the inverse perturbation identity,
\begin{equation*}
    \begin{aligned}
        &A_{\widehat\beta,\widehat\sigma^2,t}^{-1} - A_{\beta,\sigma^2,t}^{-1}= A_{\widehat\beta,\widehat\sigma^2,t}^{-1} \left( A_{\beta,\sigma^2,t} - A_{\widehat\beta,\widehat\sigma^2,t} \right) A_{\beta,\sigma^2,t}^{-1}.
    \end{aligned}
\end{equation*}
Furthermore,
\begin{equation*}
    A_{\widehat\beta,\widehat\sigma^2,t} - A_{\beta,\sigma^2,t} = e^{-2t}
    \begin{pmatrix}
        0 & \delta\\ \delta^\top & \|\widehat\beta\|_2^2-\|\beta\|_2^2 + \widehat\sigma^2-\sigma^2
    \end{pmatrix}
    .
\end{equation*}
Since
\begin{equation*}
    \left| \|\widehat\beta\|_2^2-\|\beta\|_2^2 \right| \leq \left( \|\widehat\beta\|_2+\|\beta\|_2 \right) \|\widehat\beta-\beta\|_2 \leq 2B_1\|\delta\|_2,
\end{equation*}
we get
\begin{equation*}
    \left\| A_{\widehat\beta,\widehat\sigma^2,t} - A_{\beta,\sigma^2,t} \right\|_{\op} \leq Ce^{-2t} \left( \|\widehat\beta-\beta\|_2 + |\widehat\sigma^2-\sigma^2| \right).
\end{equation*}
Therefore
\begin{equation*}
    \left\| A_{\widehat\beta,\widehat\sigma^2,t}^{-1} - A_{\beta,\sigma^2,t}^{-1} \right\|_{\op} \leq Ce^{-2t} \left( \|\widehat\beta-\beta\|_2 + |\widehat\sigma^2-\sigma^2| \right).
\end{equation*}

Let
\begin{equation*}
    H_t = A_{\widehat\beta,\widehat\sigma^2,t}^{-1} - A_{\beta,\sigma^2,t}^{-1}.
\end{equation*}
The matrix $A_{\widehat\beta,\widehat\sigma^2,t}-A_{\beta,\sigma^2,t}$ has rank at most two, and hence $H_t$ also has rank at most two. Since $Z_t\sim N(0,A_{\beta,\sigma^2,t})$, the vector $H_tZ_t$ is centered Gaussian with covariance $H_tA_{\beta,\sigma^2,t}H_t^\top$. Therefore
\begin{equation*}
    \mathbb E_{Z_t\sim p_{\beta,\sigma^2,t}} \|H_tZ_t\|_2^4 \leq C\|H_t\|_F^4.
\end{equation*}
Using $\operatorname{rank}(H_t)\leq2$, we have
\begin{equation*}
    \|H_t\|_F^2 \leq 2\|H_t\|_{\op}^2.
\end{equation*}
Consequently,
\begin{equation*}
    \begin{aligned}
        \mathbb E_{Z_t\sim p_{\beta,\sigma^2,t}} \left\| s_{\widehat\beta,\widehat\sigma^2,t}(Z_t) - s_{\beta,\sigma^2,t}(Z_t) \right\|_2^4&= \mathbb E_{Z_t\sim p_{\beta,\sigma^2,t}} \|H_tZ_t\|_2^4  \\ &\leq Ce^{-8t} \left( \|\widehat\beta-\beta\|_2 + |\widehat\sigma^2-\sigma^2| \right)^4.
    \end{aligned}
\end{equation*}
Integrating over $t\in[0,T_n]$, we obtain
\begin{equation*}
    \begin{aligned}
        \int_0^{T_n} \mathbb E_{Z_t\sim p_{\beta,\sigma^2,t}} \left\| \widehat s_t(Z_t) - s_{\beta,\sigma^2,t}(Z_t) \right\|_2^4\mathrm dt &\leq C \left( \|\widehat\beta-\beta\|_2 + |\widehat\sigma^2-\sigma^2| \right)^4  \\ &\leq C\|\widehat\beta-\beta\|_2^4 + C|\widehat\sigma^2-\sigma^2|^4.
    \end{aligned}
\end{equation*}
Taking expectation and using the bounds above yields
\begin{equation*}
    \begin{aligned}
        &\sup_{\substack{\|\beta\|_1\leq B_1\\ \sigma^2\in[\sigma_{\min}^2,\sigma_{\max}^2]}} \mathbb E_{\beta,\sigma^2} \left[ \int_0^{T_n} \mathbb E_{Z_t\sim p_{\beta,\sigma^2,t}} \left\| \widehat s_t(Z_t)-s_{\beta,\sigma^2,t}(Z_t) \right\|_2^4\mathrm dt \right]  \leq C\frac{\log p}{n}.
    \end{aligned}
\end{equation*}
Since $\mathcal R_{n,4}(T_n)$ is the infimum over all score estimators, it follows that
\begin{equation*}
    \mathcal R_{n,4}(T_n) \leq C\frac{\log p}{n}.
\end{equation*}
Finally, since $p/n\to\kappa\in(0,1)$,
\begin{equation*}
    \frac{\log p}{n}\to0.
\end{equation*}
Therefore
\begin{equation*}
    \mathcal R_{n,4}(T_n)\to0.
\end{equation*}

Moreover, the plug-in score also satisfies the Lipschitz part of Assumption~\ref{ass:hd-score-approximation}. Indeed,
\begin{equation*}
    e_t(z) = \widehat s_t(z)-s_t(z) = -H_tz,
\end{equation*}
so
\begin{equation*}
    \Lip(e_t) = \|H_t\|_{\op} \leq Ce^{-2t} \left( \|\widehat\beta-\beta\|_2+ |\widehat\sigma^2-\sigma^2| \right).
\end{equation*}
Hence
\begin{equation*}
    \sup_{0\leq t\leq T_n}\Lip(e_t) \leq C \left( \|\widehat\beta-\beta\|_2+ |\widehat\sigma^2-\sigma^2| \right) = o_p(1).
\end{equation*}
For example, with
\begin{equation*}
    \delta_n=\left(\frac{\log p}{n}\right)^{1/8},
\end{equation*}
Markov's inequality gives
\begin{equation*}
    \mathbb P \left( \sup_{0\leq t\leq T_n}\Lip(e_t)>\delta_n \right) \to0.
\end{equation*}
\end{proof}

\subsection{Proofs for the structured minimax examples}
\begin{proof}[Proof of Proposition~\ref{prop:minimax-atomic-w4-gaussian}]
We first prove the lower bound. Fix a probability measure $Q$ supported on at most $n$ points, and write
\begin{equation*}
    \supp(Q)\subseteq\{q_1,\ldots,q_n\}=:\mathcal C.
\end{equation*}
For every coupling of $Z\sim\gamma_d$ and $U\sim Q$,
\begin{equation*}
    \|Z-U\|_2 \geq \operatorname{dist}(Z,\mathcal C).
\end{equation*}
Consequently,
\begin{equation}
    W_4^4(Q,\gamma_d) \geq \mathbb E_{Z\sim\gamma_d} \operatorname{dist}(Z,\mathcal C)^4. \label{eq:atomic-distance-lower}
\end{equation}

Choose a sufficiently small universal constant $a>0$. Since the standard Gaussian density is bounded by $(2\pi)^{-d/2}$, the volume formula for a Euclidean ball and Stirling's bound give, uniformly in $x\in\mathbb R^d$,
\begin{align*}
    \gamma_d\{B(x,a\sqrt d)\} &\leq (2\pi)^{-d/2} \frac{\pi^{d/2}(a\sqrt d)^d}{\Gamma(d/2+1)}\leq (C_0a)^d \leq e^{-c_0d}
\end{align*}
for universal constants $C_0,c_0>0$. A union bound therefore yields
\begin{equation*}
    \gamma_d \left\{ \min_{1\leq j\leq n}\|Z-q_j\|_2 \leq a\sqrt d \right\} \leq ne^{-c_0d}.
\end{equation*}
After reducing the constant $c>0$, the condition $n\leq e^{cd}$ implies that the right-hand side is at most $e^{-c_0d/2}$. Hence
\begin{equation*}
    \mathbb P \{\operatorname{dist}(Z,\mathcal C)>a\sqrt d\} \geq 1-e^{-c_0d/2}.
\end{equation*}
Using~\eqref{eq:atomic-distance-lower}, and decreasing the universal constant once more if necessary, gives
\begin{equation*}
    W_4^4(Q,\gamma_d) \geq a^4d^2(1-e^{-c_0d/2}) \geq cd^2.
\end{equation*}
This bound holds for every realization of a possibly randomized $\widehat Q_n\in\mathcal A_{n,d}$, so conditioning on its randomness and then averaging proves the atomic lower bound.

For the upper bound, let
\begin{equation*}
    \widehat\gamma_n = \frac1n\sum_{i=1}^n\delta_{G_i}, \qquad G_1,\ldots,G_n\overset{\mathrm{i.i.d.}}\sim\gamma_d.
\end{equation*}
Conditionally on $G_1,\ldots,G_n$, couple $\widehat\gamma_n$ and $\gamma_d$ by choosing $I$ uniformly from $\{1,\ldots,n\}$, setting $U=G_I$, and drawing $Z\sim\gamma_d$ independently. Then
\begin{equation*}
    W_4^4(\widehat\gamma_n,\gamma_d) \leq \frac1n\sum_{i=1}^n \mathbb E_{Z\sim\gamma_d}\|G_i-Z\|_2^4.
\end{equation*}
Taking expectation and using $G-Z\sim N(0,2I_d)$ gives
\begin{equation*}
    \mathbb E W_4^4(\widehat\gamma_n,\gamma_d) \leq \mathbb E\|G-Z\|_2^4 = 4d(d+2) \leq Cd^2.
\end{equation*}
Since $\widehat\gamma_n\in\mathcal A_{n,d}$, this proves the upper bound.
\end{proof}

We first record a nuisance-estimation lemma used for all three model classes.

\begin{mylemma}[Regression nuisance estimation under a structured covariance]
\label{lem:structured-minimax-nuisance}
Let $X$ be centered and uniformly strongly log-concave, with covariance $\Sigma_\vartheta$, and assume
\begin{equation*}
    cI_p\preceq\Sigma_\vartheta\preceq CI_p, \qquad \|\Sigma_\vartheta\|_{\infty\to\infty} +\|\Sigma_\vartheta^{-1}\|_{\infty\to\infty} \leq C.
\end{equation*}
Suppose $p/n\to\kappa\in(0,1)$. A fixed number of independent sample splits may be used below, and every split has size $n_j\asymp n$. Suppose an estimator $\widehat\vartheta$ is projected onto its compact parameter space. Set
\begin{equation*}
    \lambda_n=A\sqrt{\frac{\log p}{n}}.
\end{equation*}
Assume that, for every fixed $K>0$ after increasing $A$,
\begin{equation}
    \mathbb P_\vartheta \left( \|\Sigma_{\widehat\vartheta}^{-1} -\Sigma_\vartheta^{-1}\|_{\infty\to\infty} >\lambda_n \right) \leq Cp^{-K}. \label{eq:structured-covariance-plugin-tail}
\end{equation}
Then there are estimators $\widehat\beta$ and $\widehat\sigma^2$, based on independent sample splits, such that
\begin{align}
    \mathbb E\|\widehat\beta-\beta\|_2^4 &\leq C\frac{\log p}{n}, & \mathbb E\|\widehat\beta-\beta\|_2^8 &\leq C\left(\frac{\log p}{n}\right)^2, \label{eq:structured-beta-moments}\\ \mathbb E|\widehat\sigma^2-\sigma^2|^4 &\leq C\left\{ \frac1{n^2} +\left(\frac{\log p}{n}\right)^2 \right\}. \label{eq:structured-sigma-moment}
\end{align}
The constants are uniform over $\|\beta\|_1\leq B_1$ and the prescribed noise-variance interval.
\end{mylemma}

\begin{proof}
On the first sample split, set
\begin{equation*}
    \widetilde\beta = \Sigma_{\widehat\vartheta}^{-1} \left(\frac1{n_1}\sum_{i\in I_1}X_iY_i\right).
\end{equation*}
Uniform strong log-concavity, the $\ell_1$ bound on $\beta$, and Gaussian noise imply a uniform sub-exponential bound for every centered coordinate of $X_iY_i$. Bernstein's inequality and a union bound therefore give
\begin{equation*}
    \left\| \frac1{n_1}\sum_{i\in I_1}X_iY_i -\Sigma_\vartheta\beta \right\|_\infty \leq C\lambda_n
\end{equation*}
outside an event of probability at most $Cp^{-K}$. Together with \eqref{eq:structured-covariance-plugin-tail} and the row-sum bounds, this implies
\begin{equation*}
    \|\widetilde\beta-\beta\|_\infty\leq C\lambda_n
\end{equation*}
on the same event. Apply coordinatewise soft thresholding at a sufficiently large multiple of $\lambda_n$ and project onto $\{b:\|b\|_1\leq B_1\}$. Lemma~\ref{lem:soft-thresholding} and non-expansiveness of the projection give
\begin{equation*}
    \|\widehat\beta-\beta\|_2^2\leq CB_1\lambda_n.
\end{equation*}
On the exceptional event the distance is at most $2B_1$. Taking $K>2$ proves \eqref{eq:structured-beta-moments}.

On an independent split $I_2$, define
\begin{equation*}
    \overline\sigma^2 = \frac1{n_2}\sum_{i\in I_2} (Y_i-X_i^\top\widehat\beta)^2
\end{equation*}
and project it onto $[\sigma_{\min}^2,\sigma_{\max}^2]$. Conditional on the first split, put $\delta=\widehat\beta-\beta$. Then
\begin{equation*}
    \mathbb E[(Y-X^\top\widehat\beta)^2\mid\delta] = \sigma^2+\delta^\top\Sigma_\vartheta\delta.
\end{equation*}
The conditional fourth moment of the centered empirical average is $O(n^{-2})$, uniformly in $\delta$, because $\|\delta\|_2\leq2B_1$ and strong log-concavity gives dimension-free moments of $X^\top\delta$. Moreover,
\begin{equation*}
    |\delta^\top\Sigma_\vartheta\delta|^4 \leq C\|\delta\|_2^8.
\end{equation*}
Equation~\eqref{eq:structured-sigma-moment} follows from \eqref{eq:structured-beta-moments} and the non-expansiveness of the scalar projection.
\end{proof}

\begin{mylemma}
\label{lem:gaussian-structured-score-transfer}
Let $X\sim N(0,\Sigma_p(\vartheta))$, where $\vartheta\in\Theta\subset\mathbb R^q$, $q=O(1)$, and suppose $\Theta$ is compact, $\log p/n\to0$, and the covariance matrices are uniformly well conditioned. Assume that the plug-in estimators are projected so that, almost surely,
\begin{equation*}
    \widehat\vartheta\in\Theta, \qquad \|\widehat\beta\|_1\leq B_1, \qquad \widehat\sigma^2 \in[\sigma_{\min}^2,\sigma_{\max}^2].
\end{equation*}
Assume further that, for deterministic scales $a_{\ell,p}\geq1$,
\begin{equation*}
    \|\partial_{\vartheta_\ell}\Sigma_p(\vartheta)\|_{\op}\leq C, \qquad \|\partial_{\vartheta_\ell}\Sigma_p(\vartheta)\|_{\mathrm F}^2 \leq Ca_{\ell,p},
\end{equation*}
uniformly along parameter line segments, with analogous bounds for the required second derivatives. If
\begin{equation*}
    \mathbb E |\widehat\vartheta_\ell-\vartheta_\ell|^4 \leq \frac{C}{n^2a_{\ell,p}^2},
\end{equation*}
and $\widehat\beta,\widehat\sigma^2$ satisfy \eqref{eq:structured-beta-moments}- \eqref{eq:structured-sigma-moment}, then the plug-in joint score satisfies
\begin{align}
    &\mathbb E \int_0^{T_n} \mathbb E \|\widehat s_t(Z_t)-s_t(Z_t)\|_2^4\,\mathrm dt \leq C\left(\frac{\log p}{n}+\frac1{n^2}\right), \label{eq:gaussian-structured-score-risk}\\ &\sup_{0\leq t\leq T_n}\Lip(\widehat s_t-s_t)\to_p 0. \label{eq:gaussian-structured-lipschitz}
\end{align}
\end{mylemma}

\begin{proof}
The covariance of $Z=(X,Y)$ is
\begin{equation*}
    \Gamma_{\vartheta,\beta,\sigma^2} = M_\beta
    \begin{pmatrix}
        \Sigma_p(\vartheta)&0\\ 0&\sigma^2
    \end{pmatrix}
    M_\beta^\top, \qquad M_\beta=
    \begin{pmatrix}
        I_p&0\\ \beta^\top&1
    \end{pmatrix}.
\end{equation*}
The parameter restrictions give uniform upper and lower spectral bounds for $\Gamma$. Its OU covariance and score are
\begin{equation*}
    \Gamma_t=I_{p+1}+e^{-2t}(\Gamma-I_{p+1}), \qquad s_t(z)=-\Gamma_t^{-1}z.
\end{equation*}
Along a line segment between the true and plug-in parameters, the covariance derivative bounds and $\|\beta\|_2\leq B_1$ give
\begin{equation*}
    \begin{aligned}
        \|\widehat\Gamma_t-\Gamma_t\|_{\mathrm F} \leq Ce^{-2t} \left\{ \sum_{\ell=1}^q a_{\ell,p}^{1/2} |\widehat\vartheta_\ell-\vartheta_\ell| +\|\widehat\beta-\beta\|_2 +|\widehat\sigma^2-\sigma^2| \right\}.
    \end{aligned}
\end{equation*}
The same inequality in operator norm holds without the factors $a_{\ell,p}^{1/2}$. The inverse perturbation identity transfers these bounds to $H_t=\widehat\Gamma_t^{-1}-\Gamma_t^{-1}$.

For $Z_t\sim N(0,\Gamma_t)$, the Gaussian fourth-moment formula gives
\begin{equation*}
    \begin{aligned}
        \mathbb E\|H_tZ_t\|_2^4 = \{\tr(H_t\Gamma_tH_t^\top)\}^2 +2\tr \{(H_t\Gamma_tH_t^\top)^2\},
    \end{aligned}
\end{equation*}
and hence
\begin{equation*}
    \mathbb E\|H_tZ_t\|_2^4\leq C\|H_t\|_{\mathrm F}^4.
\end{equation*}
Integrating $e^{-8t}$, using that $q$ is fixed, and taking the outer expectation yield
\begin{equation*}
    C\sum_{\ell=1}^q a_{\ell,p}^2 \mathbb E|\widehat\vartheta_\ell-\vartheta_\ell|^4 +C\mathbb E\|\widehat\beta-\beta\|_2^4 +C\mathbb E|\widehat\sigma^2-\sigma^2|^4,
\end{equation*}
which proves~\eqref{eq:gaussian-structured-score-risk}. The operator-norm version, parameter consistency, and the linearity of the score prove
\eqref{eq:gaussian-structured-lipschitz}.
\end{proof}

\begin{mylemma}
\label{lem:ou-hessian-curvature-perturbation}
Let $U_{\mathrm G}(z)=z^\top Hz/2$, where
\begin{equation*}
    c_0I_d\preceq H\preceq C_0I_d,
\end{equation*}
and let $U=U_{\mathrm G}+R$ with
\begin{equation*}
    \sup_{z\in\mathbb R^d}\|\nabla^2R(z)\|_{\op} \leq\delta<\frac{c_0}{2}.
\end{equation*}
Let $s_t$ and $s_{\mathrm G,t}$ be the scores of the standard OU evolutions started from the densities proportional to $e^{-U}$ and $e^{-U_{\mathrm G}}$, respectively. Then
\begin{equation}
    \sup_{t\geq0}\Lip(s_t-s_{\mathrm G,t}) \leq C\delta, \label{eq:ou-hessian-curvature-perturbation}
\end{equation}
where $C$ depends only on $c_0$ and $C_0$ and is independent of $d$.
\end{mylemma}

\begin{proof}
At $t=0$, the conclusion follows directly from
\begin{equation*}
    \nabla s_0-\nabla s_{\mathrm G,0}=-\nabla^2R.
\end{equation*}
Fix $t>0$, write $\alpha_t=e^{-t}$,
$\omega_t^2=1-e^{-2t}$, and set
\begin{equation*}
    q_t=\frac{\alpha_t^2}{\omega_t^2}, \qquad A_t=H+q_tI_d.
\end{equation*}
Conditionally on $Z_t=z$, the posterior law of $Z_0$ under the perturbed model has potential
\begin{equation*}
    W_{t,z}(u) = U(u)+\frac{\|z-\alpha_tu\|_2^2}{2\omega_t^2}.
\end{equation*}
Consequently,
\begin{equation}
    A_t-\delta I_d \preceq \nabla^2W_{t,z}(u) \preceq A_t+\delta I_d. \label{eq:posterior-curvature-sandwich}
\end{equation}
Let $K_{t,z}$ be the posterior covariance of $Z_0$. The Brascamp--Lieb inequality and the lower bound in \eqref{eq:posterior-curvature-sandwich} give
\begin{equation*}
    K_{t,z}\preceq(A_t-\delta I_d)^{-1}.
\end{equation*}
For the reverse inequality, the matrix Cram\'er--Rao information inequality and integration by parts give
\begin{equation*}
    K_{t,z} \succeq \left\{ \mathbb E(\nabla^2W_{t,z}(Z_0)\mid Z_t=z) \right\}^{-1} \succeq (A_t+\delta I_d)^{-1}.
\end{equation*}
The Gaussian posterior covariance is exactly $A_t^{-1}$. The resolvent identity therefore yields
\begin{equation}
    \|K_{t,z}-A_t^{-1}\|_{\op} \leq \frac{C\delta}{(c_0+q_t)^2}, \label{eq:posterior-covariance-curvature-bound}
\end{equation}
uniformly in $z$ and $d$.

The Gaussian-channel Hessian identity is
\begin{equation*}
    \nabla s_t(z) = -\frac1{\omega_t^2}I_d + \frac{\alpha_t^2}{\omega_t^4}K_{t,z}.
\end{equation*}
The corresponding Gaussian identity has $K_{t,z}$ replaced by $A_t^{-1}$. Since $\alpha_t^2+\omega_t^2=1$,
\begin{equation*}
    \frac{\alpha_t^2}{\omega_t^4} =q_t(1+q_t), \qquad \sup_{q\geq0} \frac{q(1+q)}{(c_0+q)^2}<\infty.
\end{equation*}
Combining this observation with \eqref{eq:posterior-covariance-curvature-bound} proves \eqref{eq:ou-hessian-curvature-perturbation}.
\end{proof}

\begin{mylemma}
\label{lem:product-ou-score-transfer}
Under~\eqref{eq:product-exponential-family}--\eqref{eq:product-exponential-information}, uniformly over true and candidate parameters on the compact parameter set,
\begin{align}
    \int_0^\infty \mathbb E_{\theta,\beta,\sigma^2} \|s_{\theta',\beta',\tau^2,t}(Z_t) -s_{\theta,\beta,\sigma^2,t}(Z_t)\|_2^4\,\mathrm dt &\leq C\left\{ p^2\|\theta'-\theta\|_2^4 +\|\beta'-\beta\|_2^4 +|\tau^2-\sigma^2|^4 \right\}. \label{eq:product-joint-score-transfer}
\end{align}
In addition, let $s_{\mathrm G,\beta,\sigma^2,t}$ denote the OU score of the Gaussian linear model with $X\sim N(0,I_p)$ and the same regression and noise parameters. There are constants $\delta_0>0$ and $C_{\mathrm{Lip}}<\infty$ such that, whenever $\delta_{\mathrm{NG}}\leq\delta_0$,
\begin{equation}
    \sup_{t\geq0} \Lip\{s_{\theta,\beta,\sigma^2,t} -s_{\mathrm G,\beta,\sigma^2,t}\} \leq C_{\mathrm{Lip}}\delta_{\mathrm{NG}}. \label{eq:product-to-gaussian-lipschitz}
\end{equation}
Consequently,
\begin{align}
    \sup_{t\geq0} \Lip\{s_{\theta',\beta',\tau^2,t} -s_{\theta,\beta,\sigma^2,t}\} &\leq 2C_{\mathrm{Lip}}\delta_{\mathrm{NG}} +C\left( \|\beta'-\beta\|_2 +|\tau^2-\sigma^2| \right). \label{eq:product-score-lipschitz-transfer}
\end{align}
\end{mylemma}

\begin{proof}
For a scalar parameter direction $h$, let $T_h(Z_0)=\partial_h\log f(Z_0)$ be the initial log-density tangent. With $\alpha_t=e^{-t}$ and $\omega_t^2=1-e^{-2t}$, Gaussian-channel differentiation gives
\begin{equation}
    \partial_hs_t(z) = \frac{\alpha_t}{\omega_t^2} \Cov\{T_h(Z_0),Z_0\mid Z_t=z\}. \label{eq:gaussian-channel-tangent}
\end{equation}
Uniform strong convexity, compactness of the parameter sets, and the bounded parameter derivatives provide an integrable dominating envelope. They justify differentiation under the Gaussian convolution integral and the fundamental theorem of calculus along each parameter path used below. For some constant $m_0>0$ uniform over the admissible parameter set, the
conditional potential has curvature at least
\begin{equation*}
    m_0I_{p+1}+\frac{\alpha_t^2}{\omega_t^2}I_{p+1}.
\end{equation*}
The vector Brascamp--Lieb covariance inequality therefore implies
\begin{equation*}
    \|\partial_hs_t(z)\|_2 \leq \frac{\alpha_t}{m_0\omega_t^2+\alpha_t^2} \left\{ \mathbb E[\|\nabla T_h(Z_0)\|_2^2\mid Z_t=z] \right\}^{1/2}.
\end{equation*}
The time coefficient is bounded near zero and is $O(e^{-t})$ at infinity.

For a structural direction $h\in\mathbb R^q$,
\begin{equation*}
    T_h(x,y) = -\sum_{j=1}^p\sum_{\ell=1}^q h_\ell\phi_\ell(x_j)+\text{constant},
\end{equation*}
so $\|\nabla T_h\|_2\leq C\sqrt p\|h\|_2$. Hence the integrated fourth moment of $\partial_hs_t(Z_t)$ is at most $Cp^2\|h\|_2^4$. For a regression direction $g\in\mathbb R^p$,
\begin{equation*}
    T_{\beta,g}(x,y) = \frac{(y-x^\top\beta)(x^\top g)}{\sigma^2},
\end{equation*}
and strong log-concavity gives $\mathbb E\|\nabla T_{\beta,g}(X,Y)\|_2^4\leq C\|g\|_2^4$. The variance tangent is a centered quadratic function of $y-x^\top\beta$ and obeys the analogous dimension-free bound.

The conditional expectations above may be evaluated at an intermediate parameter while $Z_t$ is distributed under the true parameter. We next give the cross-endpoint argument explicitly. Write $\eta_u=(1-u)\eta_0+u\eta_1$ and, for $t>0$, let $\nu_{\eta_u,t,z}$ be the conditional law of $Z_0$ given $Z_t=z$ under $\eta_u$. Its potential is
\begin{equation*}
    U_{\eta_u}(v) + \frac{\|z-\alpha_tv\|_2^2}{2\omega_t^2},
\end{equation*}
and hence has curvature at least $m_0+\alpha_t^2/\omega_t^2$, uniformly in $u$, $t$, and $z$. Denote its mean by
\begin{equation*}
    m_{u,t}(z) = \mathbb E_{\eta_u}[Z_0\mid Z_t=z].
\end{equation*}
Gaussian-channel differentiation and the Brascamp--Lieb inequality give
\begin{equation*}
    \nabla_zm_{u,t}(z) = \frac{\alpha_t}{\omega_t^2} \Cov_{\eta_u}(Z_0\mid Z_t=z), \qquad \Lip(m_{u,t}) \leq \frac{\alpha_t}{m_0\omega_t^2+\alpha_t^2} \leq C.
\end{equation*}
Every admissible joint law is centrally symmetric. Thus $m_{u,t}$ is odd, and the true endpoint law $p_{\eta_0,t}$ is centrally symmetric. Moreover, the strong-log-concavity preservation formula gives the lower curvature bound
\begin{equation*}
    m_t = \left\{ \frac{e^{-2t}}{m_0}+1-e^{-2t} \right\}^{-1} \geq \min\{m_0,1\}.
\end{equation*}
Therefore, for every fixed $k\geq1$ and every $v\in\mathbb R^{p+1}$, the Herbst bound gives
\begin{equation}
    \sup_{t>0}\sup_{u\in[0,1]} \int |v^\top m_{u,t}(z)|^k p_{\eta_0,t}(z)\,\mathrm dz \leq C_k\|v\|_2^k. \label{eq:cross-endpoint-posterior-mean}
\end{equation}
Conditionally on $Z_t=z$, the same Herbst bound applied to the uniformly strongly log-concave posterior $\nu_{\eta_u,t,z}$ yields
\begin{equation*}
    \mathbb E_{\eta_u} \left[ |v^\top\{Z_0-m_{u,t}(z)\}|^k \mid Z_t=z \right] \leq C_k\|v\|_2^k.
\end{equation*}
Combining this estimate with \eqref{eq:cross-endpoint-posterior-mean} proves the bound for $t>0$. At $t=0$, the conditional law is the point mass at $z$, and the same estimate follows directly from the uniform strong log-concavity of the true initial law. Thus
\begin{equation}
    \sup_{t\geq0}\sup_{u\in[0,1]} \int \mathbb E_{\eta_u} [|v^\top Z_0|^k\mid Z_t=z] p_{\eta_0,t}(z)\,\mathrm dz \leq C_k\|v\|_2^k. \label{eq:cross-endpoint-projection-moment}
\end{equation}

The structural gradient is deterministically bounded by $C\sqrt p\|h\|_2$. For the regression tangent, $\|\nabla T_{\beta,g}(x,y)\|_2^2$ is a polynomial of degree two in the two projections $x^\top g$ and $y-x^\top\beta$, with coefficients uniformly bounded on the parameter set. Its square is therefore controlled by \eqref{eq:cross-endpoint-projection-moment} with $k=4$. The variance tangent is handled in exactly the same way. Consequently,
\begin{align*}
    \sup_{t\geq0} \int \left\{ \mathbb E_{\eta_u} [\|\nabla T_h(Z_0)\|_2^2\mid Z_t=z] \right\}^2 p_{\eta_0,t}(z)\,\mathrm dz &\leq Cp^2\|h\|_2^4,\\ \sup_{t\geq0} \int \left\{ \mathbb E_{\eta_u} [\|\nabla T_{\beta,g}(Z_0)\|_2^2\mid Z_t=z] \right\}^2 p_{\eta_0,t}(z)\,\mathrm dz &\leq C\|g\|_2^4,
\end{align*}
uniformly for $u\in[0,1]$, and the variance tangent satisfies the same dimension-free estimate.

Apply the fundamental theorem of calculus along the line segment joining the true and candidate parameters. Compactness makes the preceding estimates uniform when the derivative is evaluated at an intermediate parameter and the OU trajectory is generated at the true endpoint. Minkowski's inequality then proves~\eqref{eq:product-joint-score-transfer}.

It remains to prove the spatial Lipschitz statement. The joint initial potential can be written as
\begin{equation*}
    U_{\theta,\beta,\sigma^2}(x,y) = U_{\mathrm G,\beta,\sigma^2}(x,y) + \sum_{j=1}^p\sum_{\ell=1}^q \theta_\ell\phi_\ell(x_j),
\end{equation*}
where
\begin{equation*}
    U_{\mathrm G,\beta,\sigma^2}(x,y) = \frac{\|x\|_2^2}{2} + \frac{(y-x^\top\beta)^2}{2\sigma^2}
\end{equation*}
is quadratic. Its Hessian has dimension-free upper and lower spectral bounds uniformly over the parameter set. Indeed, with
\begin{equation*}
    M_\beta =
    \begin{pmatrix}
        I_p&0\\ \beta^\top&1
    \end{pmatrix}
    ,
\end{equation*}
the quadratic Hessian equals
\begin{equation*}
    M_\beta^{-\top} \diag(I_p,\sigma^{-2}) M_\beta^{-1},
\end{equation*}
and both $M_\beta$ and $M_\beta^{-1}$ have uniformly bounded operator norm. Moreover,
\begin{equation*}
    \left\| \nabla^2U_{\theta,\beta,\sigma^2} -\nabla^2U_{\mathrm G,\beta,\sigma^2} \right\|_{\op} \leq \delta_{\mathrm{NG}}.
\end{equation*}
Choose $\delta_0$ smaller than half the uniform lower spectral bound of this quadratic Hessian. The preceding curvature-perturbation lemma gives
\eqref{eq:product-to-gaussian-lipschitz}.

For the two Gaussian reference models, the inverse-covariance perturbation argument gives
\begin{equation*}
    \sup_{t\geq0} \Lip\{s_{\mathrm G,\beta',\tau^2,t} -s_{\mathrm G,\beta,\sigma^2,t}\} \leq C\left( \|\beta'-\beta\|_2+|\tau^2-\sigma^2| \right).
\end{equation*}
The triangle inequality now proves \eqref{eq:product-score-lipschitz-transfer}.
\end{proof}

\begin{proof}[Proof of Theorem~\ref{thm:structured-score-upper}]
We verify the three classes separately. All estimators below are projected onto their compact parameter spaces. We use only a fixed number of independent sample splits, each containing a fixed positive fraction of the observations; thus every split has size comparable to $n$.

For the AR(1) model, take
\begin{equation*}
    \widetilde\rho = \frac{1}{n(p-1)} \sum_{i=1}^n\sum_{j=1}^{p-1}X_{ij}X_{i,j+1}.
\end{equation*}
Let $\widehat\rho$ be its projection onto $[-\rho_0,\rho_0]$. Gaussian quadratic-form concentration, uniformly over $|\rho|\leq\rho_0<1$, gives
\begin{equation*}
    \mathbb E|\widehat\rho-\rho|^4\leq\frac{C}{n^2p^2}, \qquad \mathbb P(|\widehat\rho-\rho|>u) \leq C\exp\{-cnp\min(u^2,u)\}.
\end{equation*}
Moreover,
\begin{equation*}
    \|\partial_\rho\Sigma_p(\rho)\|_{\op}\leq C, \qquad \|\partial_\rho\Sigma_p(\rho)\|_{\mathrm F}^2\leq Cp,
\end{equation*}
and the same bounds hold for $\partial_\rho^2\Sigma_p(\rho)$. Indeed, the corresponding row sums are controlled by convergent series of the form $\sum_{k\geq j}k^j\rho_0^{k-j}$, $j\in\{1,2\}$. The AR(1) precision matrix is tridiagonal with uniformly bounded row sums, and the same is true of its derivative on the compact stability region. The covariance itself also satisfies
\begin{equation*}
    \|\Sigma_p(\rho)\|_{\infty\to\infty} \leq 1+2\sum_{k=1}^\infty\rho_0^k = \frac{1+\rho_0}{1-\rho_0}.
\end{equation*}
Consequently,
\begin{equation*}
    \|\Sigma_p(\widehat\rho)^{-1} -\Sigma_p(\rho)^{-1}\|_{\infty\to\infty} \leq C|\widehat\rho-\rho|.
\end{equation*}
For $\lambda_n=A\sqrt{\log p/n}$, the preceding concentration inequality therefore gives, after increasing $A$, for every fixed $K>0$,
\begin{equation*}
    \mathbb P \left( \|\Sigma_p(\widehat\rho)^{-1} -\Sigma_p(\rho)^{-1}\|_{\infty\to\infty} >\lambda_n \right) \leq Cp^{-K}.
\end{equation*}
Thus Lemma~\ref{lem:structured-minimax-nuisance} applies, followed by Lemma~\ref{lem:gaussian-structured-score-transfer} with $a_{\rho,p}=p$.

For the fixed-rank model, estimate
\begin{equation*}
    \lambda_\ell \quad\text{by}\quad \widetilde\lambda_\ell = \frac1n\sum_{i=1}^n(u_{\ell,p}^\top X_i)^2-1, \qquad \ell=1,\ldots,r.
\end{equation*}
Let $\widehat\lambda$ be the Euclidean projection of $\widetilde\lambda$ onto $\Lambda$. Since $\Lambda$ is compact and convex, this projection is unique and non-expansive. The projected observations are Gaussian with uniformly bounded variance, and $r$ is fixed; hence
\begin{equation*}
    \mathbb E\|\widehat\lambda-\lambda\|_2^4 \leq Cn^{-2}, \qquad \mathbb P (\|\widehat\lambda-\lambda\|_2>u) \leq C\exp\{-cn\min(u^2,u)\}.
\end{equation*}
Also,
\begin{equation*}
    \partial_{\lambda_\ell}\Sigma_p =u_{\ell,p}u_{\ell,p}^\top, \qquad \|\partial_{\lambda_\ell}\Sigma_p\|_{\mathrm F}^2=1, \qquad \partial_{\lambda_k}\partial_{\lambda_\ell}\Sigma_p=0.
\end{equation*}
The Woodbury identity gives
\begin{equation*}
    \Sigma_p(\lambda)^{-1} = I_p - U_p\diag \left( \frac{\lambda_1}{1+\lambda_1},\ldots, \frac{\lambda_r}{1+\lambda_r} \right)U_p^\top.
\end{equation*}
Compactness of $\Lambda$ and~\eqref{eq:low-rank-row-sum} therefore imply
\begin{equation*}
    \|\Sigma_p(\lambda)\|_{\infty\to\infty} + \|\Sigma_p(\lambda)^{-1}\|_{\infty\to\infty} \leq C
\end{equation*}
and
\begin{equation*}
    \|\Sigma_p(\widehat\lambda)^{-1} -\Sigma_p(\lambda)^{-1}\|_{\infty\to\infty} \leq C\|\widehat\lambda-\lambda\|_2.
\end{equation*}
Thus, for $\lambda_n=A\sqrt{\log p/n}$ and every fixed $K>0$, increasing $A$ gives
\begin{equation*}
    \mathbb P \left( \|\Sigma_p(\widehat\lambda)^{-1} -\Sigma_p(\lambda)^{-1}\|_{\infty\to\infty} >\lambda_n \right) \leq Cp^{-K}.
\end{equation*}
Apply the preceding two lemmas with $a_{\lambda_\ell,p}=1$.

For the product exponential family, maximize the log-likelihood over $\Theta$. Let $\ell_{np}(\eta)$ be the negative log-likelihood normalized by the number $np$ of scalar coordinates. Its Hessian is
\begin{equation*}
    \nabla^2\ell_{np}(\eta) = \Cov_\eta \bigl\{(\phi_1(X_1),\ldots,\phi_q(X_1))^\top\bigr\} \succeq cI_q
\end{equation*}
by~\eqref{eq:product-exponential-information}. Constrained optimality and strong convexity, including when the true parameter lies on the boundary of $\Theta$, imply
\begin{equation*}
    \|\widehat\theta-\theta\|_2 \leq \frac{2}{c}\|\nabla\ell_{np}(\theta)\|_2.
\end{equation*}
The functions $\phi_\ell$ are uniformly Lipschitz, and the scalar design is uniformly strongly log-concave. Since $q$ is fixed, concentration of the $np$ scalar sufficient statistics gives
\begin{equation*}
    \mathbb E\|\widehat\theta-\theta\|_2^4 \leq \frac{C}{n^2p^2}, \qquad \mathbb P(\|\widehat\theta-\theta\|_2>u) \leq C\exp(-cnpu^2)
\end{equation*}
for $u$ in a fixed neighborhood of zero. Evenness makes the design centered, and its covariance is $v(\theta)I_p$. The Brascamp--Lieb and scalar Cram\'er--Rao inequalities, together with \eqref{eq:product-exponential-curvature}, give
\begin{equation*}
    \frac1L \leq v(\theta) \leq \frac1m.
\end{equation*}
Moreover, $v$ and $v^{-1}$ are uniformly Lipschitz on $\Theta$. Indeed,
\begin{equation*}
    \partial_{\theta_\ell}v(\theta) = -\Cov_\theta \{X_1^2,\phi_\ell(X_1)\},
\end{equation*}
which is uniformly bounded by the uniform moment bounds implied by strong log-concavity and the bounded first derivatives of $\phi_\ell$. Thus,
\begin{equation*}
    \begin{aligned}
        \|\Sigma_{\widehat\theta}^{-1} -\Sigma_\theta^{-1}\|_{\infty\to\infty} &= |v(\widehat\theta)^{-1}-v(\theta)^{-1}|\\ &\leq C\|\widehat\theta-\theta\|_2.
    \end{aligned}
\end{equation*}
For $\lambda_n=A\sqrt{\log p/n}$, proportional growth implies $\lambda_n\to0$, and the preceding tail bound yields, after increasing $A$, for every fixed $K>0$,
\begin{equation*}
    \mathbb P \left( \|\Sigma_{\widehat\theta}^{-1} -\Sigma_\theta^{-1}\|_{\infty\to\infty} >\lambda_n \right) \leq C\exp(-cnp\lambda_n^2) \leq Cp^{-K}.
\end{equation*}
Hence Lemma~\ref{lem:structured-minimax-nuisance} supplies $\widehat\beta$ and $\widehat\sigma^2$. Finally, Lemma~\ref{lem:product-ou-score-transfer} yields
\begin{equation*}
    \begin{aligned}
        \mathbb E \int_0^{T_n} \mathbb E \|\widehat s_t(Z_t)-s_t(Z_t)\|_2^4\,\mathrm dt &\leq Cp^2\mathbb E\|\widehat\theta-\theta\|_2^4 +C\mathbb E\|\widehat\beta-\beta\|_2^4 +C\mathbb E|\widehat\sigma^2-\sigma^2|^4\\ &\leq C\left(\frac{\log p}{n}+\frac1{n^2}\right).
    \end{aligned}
\end{equation*}
Equation~\eqref{eq:product-score-lipschitz-transfer} and consistency of $\widehat\beta$ and $\widehat\sigma^2$ give \eqref{eq:product-small-lipschitz}. If $\delta_{\mathrm{NG}}<\min\{\delta_0, L_\star/(4C_{\mathrm{Lip}})\}$, the right-hand side is strictly below $L_\star$ with probability tending to one. This proves \eqref{eq:small-lip} and completes the proof.
\end{proof}

\section{Auxiliary lemmas}
\begin{mylemma}[\citealp{bickel1981some}]
    Let $(B,\|\cdot\|)$ be a separable Banach space, and let $\Gamma_p(B)$ be the set of Borel probability measures $\gamma$ on $B$ such that $\int\|x\|^p\,\gamma(\mathrm dx)<\infty$. For $\alpha_n,\alpha\in\Gamma_p(B)$, the convergence $W_p(\alpha_n,\alpha)\to0$ is equivalent to each of the following:
    \begin{itemize}
        \item[(1)] $\alpha_n\Rightarrow\alpha$ and $\int\|x\|^p\,\alpha_n(\mathrm dx) \to\int\|x\|^p\,\alpha(\mathrm dx)$;
        \item[(2)] $\alpha_n\Rightarrow\alpha$ and $\|x\|^p$ is uniformly integrable with respect to $(\alpha_n)$;
        \item[(3)] $\int\phi\,\mathrm d\alpha_n \to\int\phi\,\mathrm d\alpha$ for every continuous $\phi$ satisfying $|\phi(x)|\leq C(1+\|x\|^p)$ for some $C<\infty$.
    \end{itemize}
    \label{wmm}
\end{mylemma}

\begin{mydefinition}[Log-concavity and strong log-concavity]
\label{def:log-concave-strong}
Let $P$ be a probability law on $\mathbb R^d$ with density $p(x)=e^{-U(x)}$ with respect to Lebesgue measure. We say that $P$ is log-concave if $U:\mathbb R^d\to(-\infty,\infty]$ is convex. If $U\in C^2(\mathbb R^d)$, we say that $P$ is strongly log-concave with curvature matrix $H\succ0$ if
\begin{equation*}
    \nabla^2 U(x)\succeq H, \qquad x\in\mathbb R^d.
\end{equation*}
Equivalently, $P$ is strongly log-concave with covariance proxy $\Sigma\succ0$ if
\begin{equation*}
    \nabla^2 U(x)\succeq \Sigma^{-1}, \qquad x\in\mathbb R^d.
\end{equation*}
\end{mydefinition}

\begin{mylemma}[\citealp{henningsson2006log}]
\label{lem:aux-strong-log-concavity-preservation}
Let $X$ and $Y$ be independent random vectors on $\mathbb R^d$.

\begin{enumerate}
\item[(i)] \textbf{Scaling.} Suppose $X$ has density $e^{-U_X}$ and
\begin{equation*}
    \nabla^2 U_X(x)\succeq \Sigma^{-1}, \qquad x\in\mathbb R^d,
\end{equation*}
for some positive definite matrix $\Sigma$. Then, for any scalar $a\neq0$, the law of $aX$ has density $e^{-U_{aX}}$ satisfying
\begin{equation*}
    \nabla^2 U_{aX}(x)\succeq (a^2\Sigma)^{-1}, \qquad x\in\mathbb R^d.
\end{equation*}

\item[(ii)] \textbf{Convolution.}
Suppose $X$ has density $e^{-U_X}$ and $Y$ has density $e^{-U_Y}$, with
\begin{equation*}
    \nabla^2 U_X(x)\succeq \Sigma^{-1}, \qquad \nabla^2 U_Y(y)\succeq \Gamma^{-1},
\end{equation*}
for all $x,y\in\mathbb R^d$, where $\Sigma,\Gamma\succ0$. Then the law of $X+Y$ has density $e^{-U_{X+Y}}$ satisfying
\begin{equation*}
    \nabla^2 U_{X+Y}(z)\succeq (\Sigma+\Gamma)^{-1}, \qquad z\in\mathbb R^d.
\end{equation*}
\end{enumerate}
\end{mylemma}

\begin{mylemma}[\citealp{bris2008existence}]
\label{lem:aux-lipschitz-commutator}
Let $\eta\in C_c^\infty(\mathbb R^d)$ be nonnegative, supported on the unit ball, and satisfy $\int\eta=1$. Put $\eta_\delta(x)=\delta^{-d}\eta(x/\delta)$. If $v:\mathbb R^d\to\mathbb R^d$ is Lipschitz and $q\in L^2_{\mathrm{loc}}$, then
\begin{equation*}
    \mathcal C_\delta(v,q) := \eta_\delta*(vq)-v(\eta_\delta*q)
\end{equation*}
satisfies, for $0<\delta\leq1$,
\begin{equation}
    \|\mathcal C_\delta(v,q)\|_{L^2(B_R)} \leq C_\eta\delta\,\Lip(v) \|q\|_{L^2(B_{R+1})}. \label{eq:aux-lipschitz-commutator}
\end{equation}
In particular, $\mathcal C_\delta(v,q)\to0$ in $L^2_{\mathrm{loc}}$.
\end{mylemma}

\begin{proof}
For almost every $x$,
\begin{equation*}
    \mathcal C_\delta(v,q)(x) = \int\eta_\delta(y) \{v(x-y)-v(x)\}q(x-y)\,\mathrm dy.
\end{equation*}
Minkowski's inequality, translation invariance, and $\|v(x-y)-v(x)\|_2\leq\Lip(v)\|y\|_2$ give
\begin{align*}
    \|\mathcal C_\delta(v,q)\|_{L^2(B_R)} &\leq \Lip(v)\|q\|_{L^2(B_{R+1})} \int\eta_\delta(y)\|y\|_2\,\mathrm dy\\ &\leq C_\eta\delta\,\Lip(v) \|q\|_{L^2(B_{R+1})}.
\end{align*}
\end{proof}

\begin{mylemma}[\citealp{prekopa1973logarithmic}]
\label{lem:aux-logconcave-projection}
Let $X$ have a log-concave density on $\mathbb R^d$. For every $v\in\mathbb S^{d-1}$, the law of $v^\top X$ has a log-concave density on $\mathbb R$. If, in addition,
\begin{equation*}
    \mathbb E[X]=0, \qquad \mathbb E[XX^\top]=I_d,
\end{equation*}
then $v^\top X$ has mean zero and variance one.
\end{mylemma}

\begin{mylemma}[\citealp{lovasz2007geometry}]
\label{lem:aux-isotropic-logconcave-density}
Let $g$ be a log-concave probability density on $\mathbb R$ satisfying
\begin{equation*}
    \int xg(x)\,\mathrm dx=0, \qquad \int x^2g(x)\,\mathrm dx=1.
\end{equation*}
Then
\begin{equation}
    \|g\|_\infty\leq1. \label{eq:aux-isotropic-logconcave-density}
\end{equation}
\end{mylemma}

\begin{mylemma}[\citealp{bakry2006diffusions}]
\label{lem:aux-bakry-emery-herbst}
Let $\mu(\mathrm dx)=Z^{-1}e^{-V(x)}\,\mathrm dx$ on $\mathbb R^d$, where $V\in C^2$ and
\begin{equation*}
    \nabla^2V(x)\succeq mI_d, \qquad x\in\mathbb R^d,
\end{equation*}
for some $m>0$. Then
\begin{equation}
    \Ent_\mu(h^2) \leq \frac{2}{m}\int\|\nabla h\|_2^2\,\mathrm d\mu. \label{eq:aux-bakry-emery-lsi}
\end{equation}
More generally, if a probability measure $\nu$ satisfies
\begin{equation*}
    \Ent_\nu(h^2) \leq 2C_{\mathrm{LS}} \int\|\nabla h\|_2^2\,\mathrm d\nu,
\end{equation*}
then every $L$-Lipschitz function $F$ satisfies
\begin{equation}
    \mathbb E_\nu \exp\{t(F-\mathbb E_\nu F)\} \leq \exp\left(\frac{C_{\mathrm{LS}}L^2t^2}{2}\right), \qquad t\in\mathbb R. \label{eq:aux-herbst-mgf}
\end{equation}
Consequently,
\begin{equation*}
    \nu\{|F-\mathbb E_\nu F|>s\} \leq 2\exp\left(-\frac{s^2}{2C_{\mathrm{LS}}L^2}\right).
\end{equation*}
\end{mylemma}

\begin{proof}
For $L_\mu=\Delta-\nabla V^\top\nabla$, the iterated carr\'e du champ is
\begin{equation*}
    \Gamma_2(\phi) = \|\nabla^2\phi\|_{\mathrm F}^2 +\nabla\phi^\top\nabla^2V\nabla\phi \geq m\Gamma(\phi).
\end{equation*}
The Bakry--\'Emery criterion therefore gives \eqref{eq:aux-bakry-emery-lsi}; see \citep[Corollary~2 and Proposition~4]{bakry2006diffusions}.

For the second assertion, apply the log-Sobolev inequality to $h=e^{tF/2}$. If $\psi(t)=\log\mathbb E_\nu e^{tF}$, then
\begin{equation*}
    t\psi'(t)-\psi(t) \leq \frac{C_{\mathrm{LS}}L^2t^2}{2}.
\end{equation*}
Integration of this differential inequality gives $\psi(t)\leq t\mathbb E_\nu F+C_{\mathrm{LS}}L^2t^2/2$, which is \eqref{eq:aux-herbst-mgf}. The tail estimate follows from Chernoff's bound.
\end{proof}

\begin{mylemma}
    Suppose $\mathcal G$ is a $\sigma$-field. Conditionally on $\mathcal G$, let $\xi_1,\ldots,\xi_n$ be i.i.d.\ mean-zero random vectors in $\mathbb R^p$. Then, for every $m\geq 2$, there exists a constant $C_{m,p}<\infty$, depending only on $m$ and $p$, such that
    \begin{equation*}
        \mathbb E[\|\dfrac 1{\sqrt n}\sum_{i=1}^n \xi_i\|^m\mid \mathcal G]\leq C_{m,p}(\mathbb E[\|\xi_1\|^2\mid \mathcal G]^{m/2}+n^{1-m/2}\mathbb E[\|\xi_1\|^m\mid \mathcal G])
    \end{equation*}
    \label{Rosenthal}
\end{mylemma}

\begin{proof}
    For every vector $v\in \mathbb R^p$, there exists $C_{m,p}<\infty$ such that
    \begin{equation*}
        \|v\|^m \leq C_{m,p}\sum_{j=1}^p |v_j|^m,
    \end{equation*}
    therefore
    \begin{equation*}
        \mathbb E[\|\dfrac 1{\sqrt n}\sum_{i=1}^n \xi_i \|^m\mid \mathcal G]\leq C_{m,p}\sum_{j=1}^p \mathbb E[|\dfrac 1{\sqrt n}\sum_{i=1}^n \xi_{ij}|^m \mid \mathcal G].
    \end{equation*}
    For each coordinate $j$, conditionally on $\mathcal G$, the variables $\xi_{1j},\ldots,\xi_{nj}$ are i.i.d.\ mean-zero real random variables. Rosenthal's inequality gives, for $m\geq 2$,
    \begin{equation*}
        \mathbb E[|\sum_{i=1}^n \xi_{ij}|^m\mid \mathcal G]\leq C_m((\sum_{i=1}^n \mathbb E[\xi_{ij}^2\mid \mathcal G])^{m/2}+\sum_{i=1}^n \mathbb E[|\xi_{ij}|^m \mid \mathcal G]).
    \end{equation*}
    Dividing by $n^{m/2}$, using conditional i.i.d., we have
    \begin{equation*}
        \mathbb E[|\dfrac 1{\sqrt n}\sum_{i=1}^n \xi_{ij}|^m \mid \mathcal G]\leq C_m ((\mathbb E[\xi_{1j}^2\mid \mathcal G])^{m/2}+n^{1-m/2}\mathbb E[|\xi_{1j}|^m\mid \mathcal G]).
    \end{equation*}
    Summing over $j=1,\cdots,p$, we obtain
    \begin{equation*}
        \mathbb E[\|\dfrac 1{\sqrt n} \sum_{i=1}^n \xi_i\|^m \mid \mathcal G]\leq C_{m,p}\sum_{j=1}^p (\mathbb E[\xi_{1j}^2\mid \mathcal G])^{m/2}+C_{m,p}n^{1-m/2}\sum_{j=1}^p \mathbb E[|\xi_{1j}|^m\mid \mathcal G].
    \end{equation*}
    Since $|\xi_{1j}|\leq \|\xi_1\|$, we have
    \begin{equation*}
        \sum_{j=1}^p \mathbb E[|\xi_{1j}|^m \mid \mathcal G]\leq p\mathbb E[\|\xi_1\|^m\mid \mathcal G].
    \end{equation*}
    Moreover, because $m/2\geq1$ and $p$ is fixed,
    \begin{equation*}
        \sum_{j=1}^p (\mathbb E[\xi_{1j}^2\mid \mathcal G])^{m/2}\leq C_{m,p}(\sum_{j=1}^p \mathbb E[\xi_{1j}^2\mid \mathcal G])^{m/2}.
    \end{equation*}
    Since $\sum_{j=1}^p \mathbb E[\xi_{1j}^2\mid \mathcal G]=\mathbb E[\|\xi_1\|^2\mid \mathcal G]$, we get
    \begin{equation*}
        \sum_{j=1}^p (\mathbb E[\xi_{1j}^2\mid \mathcal G])^{m/2}\leq C_{m,p}(\mathbb E[\|\xi_1\|^2\mid \mathcal G])^{m/2}.
    \end{equation*}
    Combining the preceding inequalities yields
    \begin{equation*}
        \mathbb E[\|\dfrac 1{\sqrt n}\sum_{i=1}^n \xi_i\|^m\mid \mathcal G]\leq C_{m,p}((\mathbb E[\|\xi_1\|^2\mid \mathcal G])^{m/2}+n^{1-m/2}\mathbb E[\|\xi_1\|^m\mid \mathcal G]).
    \end{equation*}
\end{proof}

\begin{mylemma}[Conditional P\'olya theorem, \citep{durrett2019probability}]
\label{lem:conditional-polya}
Let $Q_n$ be random probability measures on $\mathbb R^r$, and let $Q$ be a deterministic probability measure whose distribution function
\begin{equation*}
    F(v)=Q\{(-\infty,v_1]\times\cdots\times(-\infty,v_r]\}
\end{equation*}
is continuous. If $Q_n\Rightarrow_p Q$, in the sense that
\begin{equation*}
    \int g\,\mathrm dQ_n \to_p \int g\,\mathrm dQ
\end{equation*}
for every bounded continuous $g$, then, writing $F_n$ for the distribution function of $Q_n$,
\begin{equation*}
    \sup_{v\in\mathbb R^r}|F_n(v)-F(v)| \to_p 0.
\end{equation*}
The same implication holds almost surely if $Q_n\Rightarrow Q$ almost surely.
\end{mylemma}

\begin{mylemma}
    Let $I=[r_0,r_1]$. Let $E,D\geq 0$ be locally integrable functions on $I$. Let $c\geq 0$ satisfy $c\in L^1(I)$ and $cE\in L^1_{\mathrm{loc}}(I)$. Fix $\lambda>0$ and define
    \begin{equation*}
        C(r)=\int_{r_0}^r c(s)\,\mathrm ds .
    \end{equation*}
    Assume that for Lebesgue-a.e. $r_0<a<b<r_1$,
    \begin{equation*}
        E(b)+\int_a^b D(r)\,\mathrm dr \leq E(a)+\lambda\int_a^b c(r)E(r)\,\mathrm dr .
    \end{equation*}
    Then for Lebesgue-a.e. $r_0<a<b<r_1$,
    \begin{equation*}
        e^{-\lambda C(b)}E(b) +\int_a^b e^{-\lambda C(r)}D(r)\,\mathrm dr \leq e^{-\lambda C(a)}E(a).
    \end{equation*}
\end{mylemma}

\begin{proof}
    By Fubini's theorem, for Lebesgue-a.e. $a\in(r_0,r_1)$, the assumed inequality holds for Lebesgue-a.e. $b\in(a,r_1)$. Fix such an $a$. (Note that $E\in L_{\mathrm{loc}}^1$, so $E(a)<\infty$ for $a.e. a$).

    Define, for $t\in(a,r_1)$,
    \begin{equation*}
        G(t) = E(a) +\lambda\int_a^t c(r)E(r)\,\mathrm dr -\int_a^t D(r)\,\mathrm dr .
    \end{equation*}
    Since $cE,D\in L^1_{\mathrm{loc}}(I)$, the function $G$ is absolutely continuous on compact subintervals of $(a,r_1)$. The assumed inequality  gives
    \begin{equation*}
        E(t)\leq G(t)
    \end{equation*}
    for Lebesgue-a.e. $t\in(a,r_1)$. Moreover,
    \begin{equation*}
        G'(t)=\lambda c(t)E(t)-D(t)
    \end{equation*}
    for Lebesgue-a.e. $t\in(a,r_1)$. Since $c\geq 0$ and $E(t)\leq G(t)$ for a.e. $t$, we obtain
    \begin{equation*}
        G'(t) \leq \lambda c(t)G(t)-D(t)
    \end{equation*}
    for Lebesgue-a.e. $t\in(a,r_1)$.

    Set
    \begin{equation*}
        w_a(t) = \exp\left(-\lambda\int_a^t c(s)\,\mathrm ds\right).
    \end{equation*}
    Since $c\in L^1(I)$, we have $w_a\in W^{1,1}_{\mathrm{loc}}((a,r_1))$ and
    \begin{equation*}
        w_a'(t)=-\lambda c(t)w_a(t)
    \end{equation*}
    for a.e. $t$. Therefore $w_aG$ is absolutely continuous and
    \begin{equation*}
        (w_aG)'(t) = w_a(t)\bigl(G'(t)-\lambda c(t)G(t)\bigr) \leq -w_a(t)D(t)
    \end{equation*}
    for a.e. $t\in(a,r_1)$.

    Integrating from $a$ to $b$, we get
    \begin{equation*}
        w_a(b)G(b)+\int_a^b w_a(r)D(r)\,\mathrm dr \leq w_a(a)G(a).
    \end{equation*}
    Since $w_a(a)=1$ and $G(a)=E(a)$, this becomes
    \begin{equation*}
        w_a(b)G(b)+\int_a^b w_a(r)D(r)\,\mathrm dr \leq E(a).
    \end{equation*}
    Using $E(b)\leq G(b)$ for a.e. $b\in(a,r_1)$, we obtain
    \begin{equation*}
        w_a(b)E(b)+\int_a^b w_a(r)D(r)\,\mathrm dr \leq E(a)
    \end{equation*}
    for Lebesgue-a.e. $b\in(a,r_1)$.

    Finally,
    \begin{equation*}
        w_a(r) = \exp\left(-\lambda\int_a^r c(s)\,\mathrm ds\right) = e^{-\lambda C(r)}e^{\lambda C(a)}.
    \end{equation*}
    Multiplying the previous inequality by $e^{-\lambda C(a)}$, we get
    \begin{equation*}
        e^{-\lambda C(b)}E(b) + \int_a^b e^{-\lambda C(r)}D(r)\,\mathrm dr \leq e^{-\lambda C(a)}E(a).
    \end{equation*}
    Since the exceptional set of $a$'s has measure zero and, for each admissible $a$, the exceptional set of $b$'s has measure zero, the conclusion holds for Lebesgue-a.e. $r_0<a<b<r_1$.
\end{proof}

\begin{mylemma}
\label{lem:coming-down-from-infinity}
Let $I=(r_0,r_1)$, and let $Y:I\to[0,\infty)$ be finite almost everywhere. Assume that $Y$ admits a non-increasing representative. Let $K>0$ and $\alpha>0$. Suppose that, for Lebesgue-a.e. $r_0<a<b<r_1$,
\begin{equation*}
    Y(b) + K\int_a^bY(r)^{1+\alpha}\,\mathrm dr \leq Y(a).
\end{equation*}
Then, for Lebesgue-a.e. $r_0<a<b<r_1$,
\begin{equation*}
    Y(b)^{-\alpha} \geq Y(a)^{-\alpha} + \alpha K(b-a),
\end{equation*}
where $0^{-\alpha}=+\infty$. In particular, for Lebesgue-a.e. $b\in(r_0,r_1)$,
\begin{equation*}
    Y(b) \leq \{\alpha K(b-r_0)\}^{-1/\alpha}.
\end{equation*}
\end{mylemma}

\begin{proof}
Choose the non-increasing representative of $Y$. Since $Y$ is finite almost everywhere, it is a finite-valued BV function on every compact subinterval on which it is finite at the left endpoint. Define the nonnegative Radon measure
\begin{equation*}
    \mu=-\mathrm dY.
\end{equation*}
The assumed integral inequality implies
\begin{equation*}
    \mu \geq K Y(r)^{1+\alpha}\,\mathrm dr
\end{equation*}
in the sense of measures.

For $\epsilon>0$, define
\begin{equation*}
    Z_\epsilon(r) = \{Y(r)+\epsilon\}^{-\alpha}.
\end{equation*}
Since $Y$ is non-increasing, $Z_\epsilon$ is non-decreasing. Write $\mu=g\,\mathrm dr+\mu^{\mathrm s}$. The preceding measure inequality gives $g\geq KY^{1+\alpha}$ almost everywhere. The BV chain rule applied to the absolutely continuous part gives
\begin{equation*}
    (\mathrm dZ_\epsilon)^{\mathrm{ac}} = \alpha (Y+\epsilon)^{-\alpha-1}g\,\mathrm dr.
\end{equation*}
The singular continuous and jump parts of $\mathrm dZ_\epsilon$ are nonnegative because $Z_\epsilon$ is non-decreasing. They may therefore be discarded, and we obtain, in the sense of measures,
\begin{equation*}
    \mathrm dZ_\epsilon \geq \alpha K (Y+\epsilon)^{-\alpha-1} Y^{1+\alpha}\,\mathrm dr.
\end{equation*}
Integrating over $(a,b]$, we obtain
\begin{equation*}
    Z_\epsilon(b)-Z_\epsilon(a) \geq \alpha K \int_a^b \left( \frac{Y(r)}{Y(r)+\epsilon} \right)^{1+\alpha} \mathrm dr.
\end{equation*}

If $Y(b)=0$, the desired conclusion is immediate. Otherwise, by monotonicity, $Y(r)>0$ for a.e. $r\in(a,b)$, and monotone convergence as $\epsilon\downarrow0$ gives
\begin{equation*}
    Y(b)^{-\alpha} - Y(a)^{-\alpha} \geq \alpha K(b-a).
\end{equation*}
Dropping the nonnegative term $Y(a)^{-\alpha}$ yields
\begin{equation*}
    Y(b) \leq \{\alpha K(b-a)\}^{-1/\alpha}.
\end{equation*}
Finally, choose admissible times $a\downarrow r_0$ to obtain
\begin{equation*}
    Y(b) \leq \{\alpha K(b-r_0)\}^{-1/\alpha}.
\end{equation*}
\end{proof}

\begin{mylemma}
\label{lem:canonical-dissipative-representative}
Let $I=(r_0,r_1)$. Let $E,D:I\to[0,\infty)$ be locally integrable. Assume that, for Lebesgue-a.e. $r_0<a<b<r_1$,
\begin{equation*}
    E(b)+\int_a^b D(r)\,\mathrm dr \leq E(a).
\end{equation*}
Then there exists a finite-valued, non-increasing function $\widetilde E:I\to[0,\infty)$ such that
\begin{equation*}
    \widetilde E(r)=E(r)
\end{equation*}
for Lebesgue-a.e. $r\in I$, and such that, for every $r_0<a<b<r_1$,
\begin{equation*}
    \widetilde E(b) + \int_a^b D(r)\,\mathrm dr \leq \widetilde E(a).
\end{equation*}
Moreover, $\widetilde E$ may be chosen right-continuous.
\end{mylemma}

\begin{proof}
Fix $\bar r\in(r_0,r_1)$ and define
\begin{equation*}
    J(t) = \int_{\bar r}^t D(r)\,\mathrm dr, \qquad F(t)=E(t)+J(t),
\end{equation*}
where the integral is understood as a signed integral when $t<\bar r$. Since $E,D\in L^1_{\mathrm{loc}}(I)$, we have $F\in L^1_{\mathrm{loc}}(I)$.

For Lebesgue-a.e. $r_0<a<b<r_1$, the assumed inequality gives
\begin{equation*}
    F(b)\leq F(a).
\end{equation*}
Hence $F$ is essentially non-increasing. Therefore there exists a right-continuous non-increasing function $\widetilde F:I\to\mathbb R$ such that
\begin{equation*}
    \widetilde F(t)=F(t)
\end{equation*}
for Lebesgue-a.e. $t\in I$. For instance, one may take
\begin{equation*}
    \widetilde F(t) = \lim_{\substack{s\downarrow t\\ s\in\mathcal L_F}}F(s),
\end{equation*}
where $\mathcal L_F$ denotes the set of Lebesgue points of $F$.

Define
\begin{equation*}
    \widetilde E(t) = \widetilde F(t)-J(t).
\end{equation*}
Then $\widetilde E=E$ almost everywhere. Moreover, for every $r_0<a<b<r_1$,
\begin{align*}
    \widetilde E(b)+\int_a^bD(r)\,\mathrm dr &= \widetilde F(b)-J(b)+J(b)-J(a) \\ &= \widetilde F(b)-J(a) \\ &\leq \widetilde F(a)-J(a) \\ &= \widetilde E(a).
\end{align*}
Since $D\geq0$, this also implies
\begin{equation*}
    \widetilde E(b)\leq\widetilde E(a),
\end{equation*}
so $\widetilde E$ is non-increasing. Finally, since $\widetilde E=E\geq0$ almost everywhere and $\widetilde E$ is right-continuous and non-increasing, it follows that $\widetilde E\geq0$ everywhere.
\end{proof}

\begin{mylemma}
\label{lem:time-inhomogeneous-sde-wellposed}

Let \(T>0\), and let
\[
    v:[0,T]\times\mathbb R^d\to\mathbb R^d
\]
be Borel measurable. Let \(\alpha:[0,T]\to(0,\infty)\) be deterministic and Borel measurable, with
\[
    0<\underline\alpha\leq\alpha_t\leq\overline\alpha<\infty.
\]
Assume that there exist nonnegative functions
\[
    \ell,G\in L^1(0,T)
\]
such that, for Lebesgue-a.e. \(t\in(0,T)\) and every \(x,y\in\mathbb R^d\),
\begin{equation}
    \|v_t(x)-v_t(y)\| \leq \ell(t)\|x-y\|, \label{eq:aux-sde-lipschitz}
\end{equation}
and
\begin{equation}
    \|v_t(x)\| \leq G(t)(1+\|x\|). \label{eq:aux-sde-growth}
\end{equation}

Fix \(s\in[0,T]\). Let \(W\) be a \(d\)-dimensional Brownian motion and let \(\xi\) be an \(\mathcal F_s\)-measurable \(\mathbb R^d\)-valued random variable. Then the SDE
\begin{equation}
    X_t = \xi + \int_s^t v_r(X_r)\,\mathrm dr + \int_s^t\sqrt{2\alpha_r}\,\mathrm dW_r, \qquad s\leq t\leq T, \label{eq:aux-time-inhomogeneous-sde}
\end{equation}
admits a unique non-explosive strong solution with continuous paths.

If, in addition, \(\mathbb E\|\xi\|<\infty\), then
\begin{equation}
    \mathbb E \left[ \sup_{s\leq t\leq T}\|X_t\| \right] <\infty. \label{eq:aux-sde-first-moment}
\end{equation}

Moreover, for deterministic initial states \(x,y\in\mathbb R^d\), let \(X^{s,x}\) and \(X^{s,y}\) be the solutions driven by the same Brownian motion. Then
\begin{equation}
    \sup_{s\leq u\leq t} \|X_u^{s,x}-X_u^{s,y}\| \leq \exp\left\{ \int_s^t\ell(r)\,\mathrm dr \right\} \|x-y\|, \qquad s\leq t\leq T. \label{eq:aux-sde-initial-stability}
\end{equation}
In particular, for each \(s\leq t\), the map
\[
    x\longmapsto X_t^{s,x}(\omega)
\]
may be chosen continuous, and the solution can be chosen jointly measurable in \((x,\omega)\).
\end{mylemma}

The classical time-homogeneous result under local Lipschitz and linear-growth assumptions is standard; see, for example, \citep[Theorem~B.3.1]{bakry2014analysis}. The proof below records the time-inhomogeneous version with deterministic bounded volatility and time-integrable Lipschitz and growth coefficients.

\begin{proof}
Changing \(v\) on a Lebesgue-null set of times does not change the integral equation. We may therefore assume that \eqref{eq:aux-sde-lipschitz} and \eqref{eq:aux-sde-growth} hold for every \(t\in[0,T]\).

Set
\[
    B_t^{(s)} = \int_s^t\sqrt{2\alpha_r}\,\mathrm dW_r.
\]
Define the Picard iterates by
\begin{equation*}
    X_t^{(0)} = \xi+B_t^{(s)},
\end{equation*}
and, recursively,
\begin{equation*}
    X_t^{(m+1)} = \xi+B_t^{(s)} + \int_s^t v_r(X_r^{(m)})\,\mathrm dr.
\end{equation*}
Every iterate is adapted and has continuous paths.

Put
\begin{equation*}
    L_{s,t} = \int_s^t\ell(r)\,\mathrm dr
\end{equation*}
and
\begin{equation*}
    H = \int_s^T G(r) \left( 1+\|\xi\| + \sup_{s\leq u\leq T}\|B_u^{(s)}\| \right)\mathrm dr.
\end{equation*}
Since \(G\in L^1(s,T)\), the stochastic integral has continuous paths, and \(\xi\) is finite almost surely, we have \(H<\infty\) almost surely. The growth assumption gives
\begin{equation*}
    \sup_{s\leq u\leq t} \|X_u^{(1)}-X_u^{(0)}\| \leq H.
\end{equation*}

We claim that, for every \(m\geq0\),
\begin{equation}
    \sup_{s\leq u\leq t} \|X_u^{(m+1)}-X_u^{(m)}\| \leq H\frac{L_{s,t}^{m}}{m!}. \label{eq:aux-picard-factorial}
\end{equation}
The case \(m=0\) was just proved. If it holds for \(m\), then
\begin{align*}
    \sup_{s\leq u\leq t} \|X_u^{(m+2)}-X_u^{(m+1)}\| &\leq \int_s^t \ell(r) \sup_{s\leq q\leq r} \|X_q^{(m+1)}-X_q^{(m)}\| \,\mathrm dr\\ &\leq \frac{H}{m!} \int_s^t \ell(r)L_{s,r}^{m}\,\mathrm dr\\ &= H\frac{L_{s,t}^{m+1}}{(m+1)!}.
\end{align*}
Thus \eqref{eq:aux-picard-factorial} follows by induction.

Consequently, almost surely,
\[
    \sum_{m=0}^{\infty} \sup_{s\leq t\leq T} \|X_t^{(m+1)}-X_t^{(m)}\| \leq H e^{L_{s,T}} <\infty.
\]
Hence \(X^{(m)}\) converges uniformly on \([s,T]\) to an adapted continuous process \(X\). By \eqref{eq:aux-sde-lipschitz}, the uniform convergence implies
\begin{equation*}
    \int_s^t v_r(X_r^{(m)})\,\mathrm dr \longrightarrow \int_s^t v_r(X_r)\,\mathrm dr,
\end{equation*}
uniformly in \(t\), and therefore \(X\) solves \eqref{eq:aux-time-inhomogeneous-sde}.

If \(X\) and \(\widetilde X\) are two solutions with the same initial condition and Brownian motion, then
\begin{equation*}
    \sup_{s\leq u\leq t} \|X_u-\widetilde X_u\| \leq \int_s^t \ell(r) \sup_{s\leq q\leq r} \|X_q-\widetilde X_q\| \,\mathrm dr.
\end{equation*}
Gronwall's inequality gives \(X=\widetilde X\). This proves pathwise uniqueness and hence uniqueness of the strong solution.

For the moment estimate, set
\[
    S_t = \sup_{s\leq u\leq t}\|X_u\|.
\]
By the integral equation and the growth assumption,
\begin{equation*}
    S_t \leq \|\xi\| + \sup_{s\leq u\leq t}\|B_u^{(s)}\| + \int_s^tG(r)(1+S_r)\,\mathrm dr.
\end{equation*}
Gronwall's inequality yields
\begin{equation*}
    S_T \leq \left\{ \|\xi\| + \sup_{s\leq u\leq T}\|B_u^{(s)}\| + \int_s^TG(r)\,\mathrm dr \right\} \exp\left\{ \int_s^TG(r)\,\mathrm dr \right\}.
\end{equation*}
The Burkholder--Davis--Gundy inequality and the upper bound on \(\alpha\) give \(\mathbb E\sup_{s\leq u\leq T}\|B_u^{(s)}\|<\infty\). Taking expectations therefore proves
\eqref{eq:aux-sde-first-moment}.

Finally, solutions started from \(x\) and \(y\), driven by the same Brownian motion, satisfy
\begin{equation*}
    \sup_{s\leq u\leq t} \|X_u^{s,x}-X_u^{s,y}\| \leq \|x-y\| + \int_s^t \ell(r) \sup_{s\leq q\leq r} \|X_q^{s,x}-X_q^{s,y}\| \,\mathrm dr.
\end{equation*}
Another application of Gronwall proves \eqref{eq:aux-sde-initial-stability}.

Each Picard iterate is jointly measurable in \((x,\omega)\), and the limit is their pointwise limit. The pathwise estimate shows continuity in \(x\).
\end{proof}

\begin{mylemma}
\label{lem:sde-markov-kernel}

Suppose that the assumptions of Lemma~\ref{lem:time-inhomogeneous-sde-wellposed} hold. For \(0\leq s\leq t\leq T\), let \(X^{s,x}\) denote the unique solution of
\begin{equation*}
    \mathrm dX_r^{s,x} = v_r(X_r^{s,x})\,\mathrm dr+\sqrt{2\alpha_r}\,\mathrm dW_r, \qquad X_s^{s,x}=x.
\end{equation*}
Define
\begin{equation}
    P_{s,t}(x,A) = \mathbb P\bigl(X_t^{s,x}\in A\bigr), \qquad A\in\mathcal B(\mathbb R^d). \label{eq:transition-kernel-definition}
\end{equation}

Then the following statements hold.

\begin{enumerate}
    \item For every \(s\leq t\), \(P_{s,t}\) is a Markov kernel: \(P_{s,t}(x,\cdot)\) is a probability measure for every \(x\), and \(x\mapsto P_{s,t}(x,A)\) is Borel measurable for every Borel set \(A\).

    \item For every bounded Borel function \(\varphi:\mathbb R^d\to\mathbb R\), define
    \begin{equation}
        P_{s,t}\varphi(x) = \int_{\mathbb R^d} \varphi(y)P_{s,t}(x,\mathrm dy) = \mathbb E\bigl[\varphi(X_t^{s,x})\bigr]. \label{eq:transition-operator-definition}
    \end{equation}
    The family \((P_{s,t})_{0\leq s\leq t\leq T}\) satisfies
    \begin{equation}
        P_{s,s}=I, \qquad P_{s,t}=P_{s,r}P_{r,t}, \qquad 0\leq s\leq r\leq t\leq T. \label{eq:chapman-kolmogorov-operator}
    \end{equation}
    Equivalently,
    \begin{equation}
        P_{s,t}(x,A) = \int_{\mathbb R^d} P_{r,t}(y,A)P_{s,r}(x,\mathrm dy). \label{eq:chapman-kolmogorov-kernel}
    \end{equation}

    \item If \(X\) is a solution started at time \(s\), then for every \(s\leq r\leq t\) and every bounded Borel \(\varphi\),
    \begin{equation}
        \mathbb E \left[ \varphi(X_t) \mid\mathcal F_r \right] = P_{r,t}\varphi(X_r) \qquad\text{a.s.} \label{eq:time-inhomogeneous-markov-property}
    \end{equation}
    Thus \(X\) is a time-inhomogeneous Markov process with transition kernels \(P_{s,t}\).

    \item For a finite signed Borel measure \(\nu\), define its forward image by
    \begin{equation}
        (P_{s,t}^*\nu)(A) = \int_{\mathbb R^d} P_{s,t}(x,A)\,\nu(\mathrm dx). \label{eq:dual-markov-evolution}
    \end{equation}
    Equivalently,
    \begin{equation}
        \int_{\mathbb R^d} \varphi\,d(P_{s,t}^*\nu) = \int_{\mathbb R^d} P_{s,t}\varphi\,\mathrm d\nu \label{eq:dual-markov-duality}
    \end{equation}
    for every bounded Borel \(\varphi\).

    The dual maps satisfy
    \begin{equation}
        P_{s,t}^* = P_{r,t}^*P_{s,r}^*, \qquad s\leq r\leq t. \label{eq:dual-evolution-property}
    \end{equation}

    \item If \(\nu\) is a finite nonnegative measure, then \(P_{s,t}^*\nu\) is nonnegative and
    \begin{equation}
        (P_{s,t}^*\nu)(\mathbb R^d) = \nu(\mathbb R^d). \label{eq:dual-mass-preservation}
    \end{equation}
    Moreover, for finite signed measures \(\nu,\nu'\),
    \begin{equation}
        \|P_{s,t}^*\nu-P_{s,t}^*\nu'\|_{\mathrm{TV}} \leq \|\nu-\nu'\|_{\mathrm{TV}}. \label{eq:dual-TV-contraction}
    \end{equation}
    Here
    \begin{equation*}
        \|\eta\|_{\mathrm{TV}} := |\eta|(\mathbb R^d) = \sup_{\|\varphi\|_\infty\leq1} \left| \int\varphi\,\mathrm d\eta \right|.
    \end{equation*}

    In particular, if \(\nu=f\,\mathrm dx\), \(\nu'=g\,\mathrm dx\), and \(P_{s,t}^*\nu\), \(P_{s,t}^*\nu'\) have densities \(f_t,g_t\), respectively, then
    \begin{equation}
        \|f_t-g_t\|_1 \leq \|f-g\|_1. \label{eq:dual-L1-contraction}
    \end{equation}
\end{enumerate}
\end{mylemma}

For the standard notions of Markov processes, transition kernels, dual semigroups, and Chapman--Kolmogorov equations, see \citep[Sections~1.1--1.3]{bakry2014analysis}. The strong Markov property of well-posed stochastic differential equations is discussed in \citep[Appendix~B.4]{bakry2014analysis}. We include the proof because the present coefficients are time-inhomogeneous and the resulting transition family is two-parameter.

\begin{proof}
For fixed \(x\), the map
\[
    A\longmapsto P_{s,t}(x,A)
\]
is the law of \(X_t^{s,x}\), and hence is a probability measure.

By the Picard construction in Lemma~\ref{lem:time-inhomogeneous-sde-wellposed}, the map
\[
    (x,\omega)\longmapsto X_t^{s,x}(\omega)
\]
can be chosen jointly measurable. Therefore, for every Borel set \(A\),
\[
    x\longmapsto \mathbb E \left[ \mathbf 1_A(X_t^{s,x}) \right] = P_{s,t}(x,A)
\]
is Borel measurable. Thus \(P_{s,t}\) is a Markov kernel.

Fix \(s\leq r\leq t\). The process
\[
    \widetilde W_u=W_{r+u}-W_r, \qquad u\geq0,
\]
is a Brownian motion independent of \(\mathcal F_r\). On the interval \([r,t]\), the process \(X^{s,x}\) satisfies
\begin{equation*}
    X_u^{s,x} = X_r^{s,x} + \int_r^u v_q(X_q^{s,x})\,dq + \int_r^u\sqrt{2\alpha_q}\,\mathrm dW_q.
\end{equation*}
By pathwise uniqueness, the post-\(r\) segment is the unique solution started from \(X_r^{s,x}\) at time \(r\), driven by the future Brownian increments. Consequently, for every bounded Borel \(\varphi\),
\begin{equation*}
    \mathbb E \left[ \varphi(X_t^{s,x}) \mid\mathcal F_r \right] = P_{r,t}\varphi(X_r^{s,x}).
\end{equation*}
This proves the Markov property \eqref{eq:time-inhomogeneous-markov-property}.

Taking expectations gives
\begin{align*}
    P_{s,t}\varphi(x) &= \mathbb E \left[ P_{r,t}\varphi(X_r^{s,x}) \right]= P_{s,r}(P_{r,t}\varphi)(x),
\end{align*}
which proves \eqref{eq:chapman-kolmogorov-operator}. Taking \(\varphi=\mathbf 1_A\) gives \eqref{eq:chapman-kolmogorov-kernel}.

The definition \eqref{eq:dual-markov-evolution} and Fubini's theorem give \eqref{eq:dual-markov-duality}. Dualizing \eqref{eq:chapman-kolmogorov-operator} yields \eqref{eq:dual-evolution-property}.

Since the SDE is non-explosive,
\begin{equation*}
    P_{s,t}(x,\mathbb R^d)=1.
\end{equation*}
Thus, for \(\nu\geq0\),
\begin{align*}
    (P_{s,t}^*\nu)(\mathbb R^d) &= \int_{\mathbb R^d} P_{s,t}(x,\mathbb R^d)\,\nu(\mathrm dx)= \nu(\mathbb R^d),
\end{align*}
which proves mass preservation.

Finally, for every bounded Borel \(\varphi\),
\begin{equation*}
    |P_{s,t}\varphi(x)| \leq P_{s,t}|\varphi|(x) \leq \|\varphi\|_\infty.
\end{equation*}
Hence
\[
    \|P_{s,t}\varphi\|_\infty \leq \|\varphi\|_\infty.
\]
Writing \(\eta=\nu-\nu'\), we obtain
\begin{align*}
    \|P_{s,t}^*\eta\|_{\mathrm{TV}} &= \sup_{\|\varphi\|_\infty\leq1} \left| \int P_{s,t}\varphi\,\mathrm d\eta \right|\leq \sup_{\|\psi\|_\infty\leq1} \left| \int\psi\,\mathrm d\eta \right|= \|\eta\|_{\mathrm{TV}}.
\end{align*}
This proves \eqref{eq:dual-TV-contraction}. If all the measures involved have densities, the total-variation norm is the corresponding \(L^1\)-norm, which gives \eqref{eq:dual-L1-contraction}.
\end{proof}

\begin{mylemma}[Superposition principle
{\citep[Remark~2.3 and Theorem~2.5]{trevisan2016well}}]
\label{lem:trevisan-superposition}

Let $b:(0,T)\times\mathbb R^d\to\mathbb R^d$ and $a:(0,T)\times\mathbb R^d\to\operatorname{Sym}_+(\mathbb R^d)$ be Borel measurable, and define
\begin{equation*}
    L_t\phi(x)=b_t(x)^\top\nabla\phi(x)+\frac12a_t(x):\nabla^2\phi(x).
\end{equation*}
Define
\begin{equation*}
    \begin{aligned}
        \mathcal A = \Bigl\{ \psi\in C^{1,2}\bigl((0,T)\times\mathbb R^d\bigr): \ \psi,\partial_t\psi,\nabla\psi,\nabla^2\psi \text{ are uniformly bounded} \Bigr\},
    \end{aligned}
\end{equation*}
and $\mathcal A_c=C_c^{1,2}\bigl((0,T)\times\mathbb R^d\bigr).$ Every $\psi\in\mathcal A$ is understood through its continuous extension to $[0,T]\times\mathbb R^d$.

Let $(\nu_t)_{t\in(0,T)}\subset\mathcal P(\mathbb R^d)$ be a Borel curve such that
\begin{equation}
    \int_0^T \int_{\mathbb R^d} \bigl( \|b_t(x)\|_2+\|a_t(x)\|_{\mathrm F} \bigr) \,\mathrm d\nu_t(x)\,\mathrm dt <\infty. \label{eq:trevisan-integrability}
\end{equation}
Assume that $(\nu_t)$ solves the Fokker--Planck equation $\partial_t\nu_t=L_t^*\nu_t$ in the weak sense:
\begin{equation}
    \int_0^T \int_{\mathbb R^d} \left[ \partial_t\psi(t,x) + L_t\psi(t,\cdot)(x) \right] \mathrm d\nu_t(x)\,\mathrm dt = 0 \label{eq:trevisan-weak-FPE}
\end{equation}
for every $\psi\in\mathcal A_c$.

Then $(\nu_t)_{t\in(0,T)}$ admits a unique narrowly continuous representative $(\widetilde\nu_t)_{t\in[0,T]}\subset\mathcal P(\mathbb R^d)$ such that
\begin{equation}
    \widetilde\nu_t=\nu_t \qquad \text{for Lebesgue-a.e. }t\in(0,T). \label{eq:trevisan-representative-ae}
\end{equation}
Moreover, there exists $\boldsymbol\eta\in\mathcal P\bigl(C([0,T];\mathbb R^d)\bigr)$ such that
\begin{equation}
    (e_t)_\#\boldsymbol\eta = \widetilde\nu_t \qquad \text{for every }t\in[0,T], \label{eq:trevisan-marginals}
\end{equation}
where $e_t(\omega)=\omega(t)$.

In addition,
\begin{equation}
    \int_{C([0,T];\mathbb R^d)} \int_0^T \bigl( \|b_t(X_t)\|_2+\|a_t(X_t)\|_{\mathrm F} \bigr) \,\mathrm dt\,\mathrm d\boldsymbol\eta <\infty. \label{eq:trevisan-mp-integrability}
\end{equation}
Moreover, by
\eqref{eq:trevisan-marginals},
\begin{equation}
    \begin{aligned}
        \int_{C([0,T];\mathbb R^d)} \int_0^T \bigl( \|b_t(X_t)\|_2+\|a_t(X_t)\|_{\mathrm F} \bigr) \,\mathrm dt\,\mathrm d\boldsymbol\eta&= \int_0^T \int_{\mathbb R^d} \bigl( \|b_t(x)\|_2+\|a_t(x)\|_{\mathrm F} \bigr) \,\mathrm d\widetilde\nu_t(x)\,\mathrm dt\\ &= \int_0^T \int_{\mathbb R^d} \bigl( \|b_t(x)\|_2+\|a_t(x)\|_{\mathrm F} \bigr) \,\mathrm d\nu_t(x)\,\mathrm dt.
    \end{aligned}
    \label{eq:trevisan-marginal-integrability-identity}
\end{equation}

Finally, if $X_t(\omega)=\omega(t)$ denotes the canonical process, then for every $\psi\in\mathcal A$,
\begin{equation}
    \begin{aligned}
        M_t^\psi = \psi(t,X_t) - \psi(0,X_0)- \int_0^t \left[ \partial_s\psi(s,X_s) + L_s\psi(s,\cdot)(X_s) \right]\mathrm ds
    \end{aligned}
    \label{eq:trevisan-martingale}
\end{equation}
is a $\boldsymbol\eta$-martingale with respect to the canonical filtration.
\end{mylemma}

\begin{proof}
The existence and uniqueness of the narrowly continuous representative follow from \citep[Remark~2.3]{trevisan2016well}, using the argument of \citep{ambrosio2005gradient}. The existence of $\boldsymbol\eta$ satisfying \eqref{eq:trevisan-marginals} and solving the corresponding martingale problem follows from \citep[Theorem~2.5]{trevisan2016well}.

Since $\boldsymbol\eta$ is a solution of the martingale problem in the sense of \citep[Definition~2.4]{trevisan2016well}, it satisfies the coefficient-integrability condition \eqref{eq:trevisan-mp-integrability}.

Moreover, by \eqref{eq:trevisan-marginals} and Tonelli's theorem,
\begin{align*}
    &\int_{C([0,T];\mathbb R^d)} \int_0^T \bigl( \|b_t(X_t)\|_2+\|a_t(X_t)\|_{\mathrm F} \bigr) \,\mathrm dt\,\mathrm d\boldsymbol\eta= \int_0^T \int_{\mathbb R^d} \bigl( \|b_t(x)\|_2+\|a_t(x)\|_{\mathrm F} \bigr) \,\mathrm d\widetilde\nu_t(x)\,\mathrm dt.
\end{align*}
Since
\begin{equation*}
    \widetilde\nu_t=\nu_t \qquad \text{for Lebesgue-a.e. }t\in(0,T),
\end{equation*}
we further obtain
\begin{align*}
    &\int_0^T \int_{\mathbb R^d} \bigl( \|b_t(x)\|_2+\|a_t(x)\|_{\mathrm F} \bigr) \,\mathrm d\widetilde\nu_t(x)\,\mathrm dt= \int_0^T \int_{\mathbb R^d} \bigl( \|b_t(x)\|_2+\|a_t(x)\|_{\mathrm F} \bigr) \,\mathrm d\nu_t(x)\,\mathrm dt <\infty,
\end{align*}
where the final inequality follows from \eqref{eq:trevisan-integrability}. This proves \eqref{eq:trevisan-marginal-integrability-identity}.

The martingale property \eqref{eq:trevisan-martingale} is exactly the martingale-problem formulation in \citep[Definition~2.4]{trevisan2016well}.
\end{proof}

\begin{remark}
For bounded coefficients, an earlier representation theorem is \citep[Theorem~2.6]{figalli2008existence}; see also \citep[Lemma~2.3]{figalli2008existence} for the corresponding relation between uniqueness of martingale marginals and uniqueness of measure-valued Fokker--Planck solutions.
\end{remark}

\begin{mylemma}
\label{lem:superposition-uniqueness-transfer}

Let $a$, $b$, $L_t$, and $(\nu_t)_{t\in(0,T)}$ satisfy the assumptions of Lemma~\ref{lem:trevisan-superposition}, and let $(\widetilde\nu_t)_{t\in[0,T]}$ be its unique narrowly continuous representative.

Assume that the martingale problem associated with $(L_t)$ and initial law $\widetilde\nu_0$ is well posed, meaning that there exists a unique probability measure $\boldsymbol\eta\in\mathcal P\bigl(C([0,T];\mathbb R^d)\bigr)$ such that $(e_0)_\#\boldsymbol\eta=\widetilde\nu_0$, $\boldsymbol\eta$ satisfies \eqref{eq:trevisan-mp-integrability}, and $M^\psi$ defined in \eqref{eq:trevisan-martingale} is a martingale for every $\psi\in\mathcal A$.

Then $(\widetilde\nu_t)_{t\in[0,T]}$ is the marginal curve of the unique martingale solution. Consequently, there exists at most one weak probability-valued Fokker--Planck solution satisfying \eqref{eq:trevisan-integrability} and having initial law $\widetilde\nu_0$.

In particular, suppose that
\begin{equation*}
    a_t(x)=2I_d
\end{equation*}
and that, for almost every $t\in(0,T)$,
\begin{equation}
    \|b_t(x)-b_t(y)\|_2 \leq \ell(t)\|x-y\|_2, \qquad \ell\in L^1(0,T), \label{eq:auxiliary-global-Lipschitz}
\end{equation}
and
\begin{equation}
    \|b_t(x)\|_2 \leq G(t)(1+\|x\|_2), \qquad G\in L^1(0,T). \label{eq:auxiliary-linear-growth}
\end{equation}
If
\begin{equation}
    \sup_{0\leq t\leq T} \int_{\mathbb R^d} \|x\|_2\,\mathrm d\widetilde\nu_t(x) <\infty, \label{eq:auxiliary-first-moment}
\end{equation}
then $(\widetilde\nu_t)$ is the marginal curve of the unique solution of
\begin{equation}
    \mathrm dX_t = b_t(X_t)\,\mathrm dt+\sqrt2\,\mathrm dW_t, \qquad X_0\sim\widetilde\nu_0. \label{eq:auxiliary-additive-SDE}
\end{equation}
\end{mylemma}

\begin{proof}
By Lemma~\ref{lem:trevisan-superposition}, there exists a solution $\boldsymbol\eta$ of the martingale problem associated with $(L_t)$ such, that
\begin{equation*}
    (e_t)_\#\boldsymbol\eta = \widetilde\nu_t \qquad \text{for every }t\in[0,T].
\end{equation*}
In particular,
\begin{equation*}
    (e_0)_\#\boldsymbol\eta = \widetilde\nu_0.
\end{equation*}

Let $\boldsymbol\eta^{\,0}$ denote the unique solution of the martingale problem with initial law $\widetilde\nu_0$. By well-posedness we have $\boldsymbol\eta=\boldsymbol\eta^{\,0}$. Therefore
\begin{equation*}
    \widetilde\nu_t = (e_t)_\#\boldsymbol\eta^{\,0} \qquad \text{for every }t\in[0,T].
\end{equation*}

If \((\nu_t^1)\) and \((\nu_t^2)\) are two weak probability-valued Fokker--Planck solutions satisfying \eqref{eq:trevisan-integrability} and whose narrowly continuous representatives both have initial law \(\widetilde\nu_0\), apply Lemma~\ref{lem:trevisan-superposition} to the two representatives. This gives two martingale solutions with the same initial law. Well-posedness of the martingale problem implies that the two path-space laws coincide, and hence all their time marginals coincide. This proves uniqueness of the Fokker--Planck solution.

For the final assertion, $a_t=2I_d$ gives
\begin{equation*}
    \frac12a_t:\nabla^2\phi = \Delta\phi,
\end{equation*}
so the associated martingale problem corresponds to \eqref{eq:auxiliary-additive-SDE}. Moreover,
\begin{equation*}
    \int_0^T\int_{\mathbb R^d} \|a_t(x)\|_{\mathrm F} \,\mathrm d\widetilde\nu_t(x)\,\mathrm dt = T\|2I_d\|_{\mathrm F} <\infty,
\end{equation*}
and, by \eqref{eq:auxiliary-linear-growth} and \eqref{eq:auxiliary-first-moment},
\begin{align*}
    &\int_0^T \int_{\mathbb R^d} \|b_t(x)\|_2 \,\mathrm d\widetilde\nu_t(x)\,\mathrm dt\leq \int_0^T G(t) \left( 1+ \int_{\mathbb R^d}\|x\|_2\,\mathrm d\widetilde\nu_t(x) \right)\mathrm dt<\infty.
\end{align*}
Thus the Fokker--Planck coefficient-integrability condition \eqref{eq:trevisan-integrability} holds.

By Lemma~\ref{lem:time-inhomogeneous-sde-wellposed}, the SDE
\[
    \mathrm dX_t=b_t(X_t)\,\mathrm dt+\sqrt2\,\mathrm dW_t, \qquad X_0\sim\widetilde\nu_0,
\]
admits a unique non-explosive strong solution with continuous paths. Since \(\widetilde\nu_0\) has finite first moment,
\eqref{eq:aux-sde-first-moment} gives
\[
    \mathbb E \left[ \sup_{0\leq t\leq T}\|X_t\| \right] <\infty.
\]
Consequently,
\[
    \mathbb E \int_0^T \bigl( \|b_t(X_t)\|+\|2I_d\| \bigr)\,\mathrm dt <\infty,
\]
so the SDE solution belongs to the martingale-problem class considered above.

By Itô's formula, the law of the SDE solution is a solution of the associated martingale problem; cf. \citep[Section~1.10.1 and Appendix~B.4]{bakry2014analysis}. Conversely, the standard equivalence between weak solutions of the SDE and solutions of its martingale problem shows that every martingale solution in this class is the law of a weak solution of the same SDE. Pathwise uniqueness therefore implies uniqueness in law, and hence the martingale problem with initial law \(\widetilde\nu_0\) is well posed.

The first part of the lemma now identifies \((\widetilde\nu_t)_{t\in[0,T]}\) with the marginal curve of this unique SDE solution.
\end{proof}

\clearpage
\section{Appendix tables}
The tables below report the complete numerical results underlying the simulation figures in Section~\ref{sec:simulation}. Each table corresponds to one design--error combination. Within each of the two column groups, the five columns give the results for $\kappa=p/n\in\{0.1,0.2,0.3,0.4,0.5\}$. The left group reports the empirical Type~I error at the nominal level $0.05$, while the right group reports the ratio of the average estimated variance to the corresponding theoretical variance benchmark. Thus, values close to $0.05$ and $1$, respectively, indicate accurate inferential and variance calibration. The sample size, number of Monte Carlo replications, number of bootstrap samples, and precise definitions of the design distributions and variance benchmarks are given in Appendix~\ref{app:numerical-experiments}.

Across the Gaussian and i.i.d.\ Laplace designs, the diffusion pairs bootstrap generally keeps the variance ratio close to one over the full range of $\kappa$. By contrast, the classical pairs bootstrap and the jackknife become increasingly conservative as $\kappa$ grows, whereas the residual bootstrap increasingly underestimates the variance and can become anti-conservative. The same qualitative pattern is observed under both Gaussian and Laplace errors, indicating that the main distortion is driven by the high-dimensional design geometry rather than by the error distribution alone. The elliptical designs are more challenging, especially with exponential or Gaussian radial scaling. At moderate-to-large aspect ratios, diffusion pairs often substantially reduces the variance distortion relative to classical pairs, but it does not do so uniformly across aspect ratios and does not uniformly attain exact calibration.

\label{app:simulation-tables}
\begin{table}[ht]
\centering
\caption{Laplace design with Laplace errors.}
\label{tab:laplacex-laplace}
\small
\begin{tabular}{lccccc@{\hspace{1em}}ccccc}
\toprule
\multirow{2}{*}{Method}
& \multicolumn{5}{c}{Type~I error}
& \multicolumn{5}{c}{Variance ratio} \\
\cmidrule(lr){2-6}\cmidrule(lr){7-11}
& 0.1 & 0.2 & 0.3 & 0.4 & 0.5
& 0.1 & 0.2 & 0.3 & 0.4 & 0.5 \\
\midrule
Pairs
& 0.057 & 0.042 & 0.028 & 0.008 & 0.001
& 1.047 & 1.149 & 1.352 & 1.792 & 2.994 \\
Residual
& 0.067 & 0.088 & 0.106 & 0.102 & 0.207
& 0.908 & 0.810 & 0.708 & 0.604 & 0.497 \\
Jackknife
& 0.044 & 0.034 & 0.025 & 0.008 & 0.013
& 1.130 & 1.283 & 1.453 & 1.688 & 2.024 \\
Diffusion pairs
& 0.061 & 0.071 & 0.057 & 0.026 & 0.046
& 1.024 & 1.023 & 1.021 & 1.017 & 1.006 \\
\bottomrule
\end{tabular}
\end{table}

\begin{table}[ht]
\centering
\caption{Elliptical-exponential design with Gaussian errors.}
\label{tab:hd-eel-normal-type1-var-ratio}
\small
\begin{tabular}{lccccc@{\hspace{1em}}ccccc}
\toprule
\multirow{2}{*}{Method}
& \multicolumn{5}{c}{Type~I error}
& \multicolumn{5}{c}{Variance ratio} \\
\cmidrule(lr){2-6}\cmidrule(lr){7-11}
& 0.1 & 0.2 & 0.3 & 0.4 & 0.5
& 0.1 & 0.2 & 0.3 & 0.4 & 0.5 \\
\midrule
Pairs
& 0.026 & 0.015 & 0.003 & 0.000 & 0.000
& 1.316 & 1.866 & 2.783 & 4.891 & 12.982 \\
Residual
& 0.060 & 0.085 & 0.115 & 0.114 & 0.170
& 0.896 & 0.811 & 0.710 & 0.602 & 0.506 \\
Jackknife
& 0.019 & 0.012 & 0.002 & 0.000 & 0.000
& 1.430 & 1.874 & 2.356 & 2.885 & 3.723 \\
Diffusion pairs
& 0.042 & 0.030 & 0.014 & 0.010 & 0.002
& 1.809 & 1.626 & 1.288 & 0.975 & 0.731 \\
\bottomrule
\end{tabular}
\end{table}

\begin{table}[ht]
\centering
\caption{Laplace design with Gaussian errors.}
\label{tab:hd-laplacex-normal-type1-var-ratio}
\small
\begin{tabular}{lccccc@{\hspace{1em}}ccccc}
\toprule
\multirow{2}{*}{Method}
& \multicolumn{5}{c}{Type~I error}
& \multicolumn{5}{c}{Variance ratio} \\
\cmidrule(lr){2-6}\cmidrule(lr){7-11}
& 0.1 & 0.2 & 0.3 & 0.4 & 0.5
& 0.1 & 0.2 & 0.3 & 0.4 & 0.5 \\
\midrule
Pairs
& 0.046 & 0.037 & 0.025 & 0.009 & 0.000
& 1.044 & 1.151 & 1.354 & 1.790 & 3.002 \\
Residual
& 0.064 & 0.080 & 0.113 & 0.126 & 0.175
& 0.905 & 0.817 & 0.708 & 0.600 & 0.498 \\
Jackknife
& 0.039 & 0.025 & 0.021 & 0.012 & 0.005
& 1.133 & 1.282 & 1.457 & 1.695 & 2.018 \\
Diffusion pairs
& 0.064 & 0.055 & 0.056 & 0.031 & 0.033
& 1.030 & 1.038 & 1.026 & 1.011 & 1.012 \\
\bottomrule
\end{tabular}
\end{table}

\begin{table}[ht]
\centering
\caption{Elliptical-uniform design with Gaussian errors.}
\label{tab:hd-eul-normal-type1-var-ratio}
\small
\begin{tabular}{lccccc@{\hspace{1em}}ccccc}
\toprule
\multirow{2}{*}{Method}
& \multicolumn{5}{c}{Type~I error}
& \multicolumn{5}{c}{Variance ratio} \\
\cmidrule(lr){2-6}\cmidrule(lr){7-11}
& 0.1 & 0.2 & 0.3 & 0.4 & 0.5
& 0.1 & 0.2 & 0.3 & 0.4 & 0.5 \\
\midrule
Pairs
& 0.051 & 0.034 & 0.023 & 0.012 & 0.000
& 1.058 & 1.197 & 1.462 & 2.014 & 3.428 \\
Residual
& 0.065 & 0.073 & 0.111 & 0.124 & 0.199
& 0.904 & 0.805 & 0.697 & 0.602 & 0.500 \\
Jackknife
& 0.040 & 0.026 & 0.017 & 0.010 & 0.004
& 1.154 & 1.324 & 1.532 & 1.825 & 2.194 \\
Diffusion pairs
& 0.059 & 0.051 & 0.044 & 0.031 & 0.027
& 1.091 & 1.111 & 1.093 & 1.071 & 1.036 \\
\bottomrule
\end{tabular}
\end{table}

\begin{table}[ht]
\centering
\caption{Elliptical-normal design with Gaussian errors.}
\label{tab:hd-enl-normal-type1-var-ratio}
\small
\begin{tabular}{lccccc@{\hspace{1em}}ccccc}
\toprule
\multirow{2}{*}{Method}
& \multicolumn{5}{c}{Type~I error}
& \multicolumn{5}{c}{Variance ratio} \\
\cmidrule(lr){2-6}\cmidrule(lr){7-11}
& 0.1 & 0.2 & 0.3 & 0.4 & 0.5
& 0.1 & 0.2 & 0.3 & 0.4 & 0.5 \\
\midrule
Pairs
& 0.045 & 0.016 & 0.004 & 0.001 & 0.000
& 1.195 & 1.571 & 2.264 & 3.856 & 9.692 \\
Residual
& 0.068 & 0.082 & 0.118 & 0.134 & 0.162
& 0.907 & 0.813 & 0.705 & 0.607 & 0.506 \\
Jackknife
& 0.031 & 0.013 & 0.008 & 0.002 & 0.002
& 1.311 & 1.650 & 2.034 & 2.551 & 3.271 \\
Diffusion pairs
& 0.050 & 0.052 & 0.036 & 0.014 & 0.007
& 1.526 & 1.518 & 1.338 & 1.130 & 0.921 \\
\bottomrule
\end{tabular}
\end{table}

\begin{table}[ht]
\centering
\caption{Elliptical-uniform design with Laplace errors.}
\label{tab:hd-eul-laplace-type1-var-ratio}
\small
\begin{tabular}{lccccc@{\hspace{1em}}ccccc}
\toprule
\multirow{2}{*}{Method}
& \multicolumn{5}{c}{Type~I error}
& \multicolumn{5}{c}{Variance ratio} \\
\cmidrule(lr){2-6}\cmidrule(lr){7-11}
& 0.1 & 0.2 & 0.3 & 0.4 & 0.5
& 0.1 & 0.2 & 0.3 & 0.4 & 0.5 \\
\midrule
Pairs
& 0.049 & 0.043 & 0.021 & 0.007 & 0.001
& 1.058 & 1.210 & 1.475 & 2.019 & 3.426 \\
Residual
& 0.066 & 0.090 & 0.122 & 0.133 & 0.170
& 0.901 & 0.806 & 0.702 & 0.605 & 0.498 \\
Jackknife
& 0.038 & 0.031 & 0.018 & 0.010 & 0.004
& 1.155 & 1.336 & 1.543 & 1.837 & 2.200 \\
Diffusion pairs
& 0.076 & 0.060 & 0.047 & 0.033 & 0.019
& 1.077 & 1.108 & 1.094 & 1.075 & 1.030 \\
\bottomrule
\end{tabular}
\end{table}

\begin{table}[ht]
\centering
\caption{Elliptical-exponential design with Laplace errors.}
\label{tab:hd-eel-laplace-type1-var-ratio}
\small
\begin{tabular}{lccccc@{\hspace{1em}}ccccc}
\toprule
\multirow{2}{*}{Method}
& \multicolumn{5}{c}{Type~I error}
& \multicolumn{5}{c}{Variance ratio} \\
\cmidrule(lr){2-6}\cmidrule(lr){7-11}
& 0.1 & 0.2 & 0.3 & 0.4 & 0.5
& 0.1 & 0.2 & 0.3 & 0.4 & 0.5 \\
\midrule
Pairs
& 0.034 & 0.012 & 0.004 & 0.000 & 0.000
& 1.317 & 1.864 & 2.792 & 4.897 & 12.861 \\
Residual
& 0.073 & 0.093 & 0.102 & 0.121 & 0.153
& 0.898 & 0.811 & 0.712 & 0.603 & 0.505 \\
Jackknife
& 0.023 & 0.015 & 0.003 & 0.000 & 0.000
& 1.441 & 1.883 & 2.351 & 2.879 & 3.666 \\
Diffusion pairs
& 0.061 & 0.027 & 0.026 & 0.009 & 0.001
& 1.805 & 1.625 & 1.288 & 0.973 & 0.727 \\
\bottomrule
\end{tabular}
\end{table}

\begin{table}[ht]
\centering
\caption{Elliptical-normal design with Laplace errors.}
\label{tab:hd-enl-laplace-type1-var-ratio}
\small
\begin{tabular}{lccccc@{\hspace{1em}}ccccc}
\toprule
\multirow{2}{*}{Method}
& \multicolumn{5}{c}{Type~I error}
& \multicolumn{5}{c}{Variance ratio} \\
\cmidrule(lr){2-6}\cmidrule(lr){7-11}
& 0.1 & 0.2 & 0.3 & 0.4 & 0.5
& 0.1 & 0.2 & 0.3 & 0.4 & 0.5 \\
\midrule
Pairs
& 0.043 & 0.014 & 0.004 & 0.001 & 0.000
& 1.187 & 1.581 & 2.266 & 3.868 & 9.658 \\
Residual
& 0.075 & 0.069 & 0.103 & 0.128 & 0.187
& 0.907 & 0.814 & 0.706 & 0.608 & 0.506 \\
Jackknife
& 0.031 & 0.010 & 0.004 & 0.001 & 0.000
& 1.302 & 1.652 & 2.036 & 2.562 & 3.252 \\
Diffusion pairs
& 0.041 & 0.025 & 0.024 & 0.011 & 0.011
& 1.514 & 1.511 & 1.339 & 1.129 & 0.918 \\
\bottomrule
\end{tabular}
\end{table}

\begin{table}[ht]
\centering
\caption{Gaussian design with Gaussian errors.}
\label{tab:hd-gaussian-normal-type1-var-ratio}
\small
\begin{tabular}{lccccc@{\hspace{1em}}ccccc}
\toprule
\multirow{2}{*}{Method}
& \multicolumn{5}{c}{Type~I error}
& \multicolumn{5}{c}{Variance ratio} \\
\cmidrule(lr){2-6}\cmidrule(lr){7-11}
& 0.1 & 0.2 & 0.3 & 0.4 & 0.5
& 0.1 & 0.2 & 0.3 & 0.4 & 0.5 \\
\midrule
Pairs
& 0.056 & 0.040 & 0.027 & 0.015 & 0.000
& 1.028 & 1.128 & 1.331 & 1.774 & 2.999 \\
Residual
& 0.070 & 0.093 & 0.108 & 0.133 & 0.160
& 0.899 & 0.801 & 0.699 & 0.600 & 0.502 \\
Jackknife
& 0.045 & 0.031 & 0.021 & 0.016 & 0.003
& 1.112 & 1.254 & 1.433 & 1.683 & 2.026 \\
Diffusion pairs
& 0.067 & 0.068 & 0.062 & 0.056 & 0.030
& 1.009 & 1.010 & 1.009 & 1.013 & 1.009 \\
\bottomrule
\end{tabular}
\end{table}

\begin{table}[ht]
\centering
\caption{Gaussian design with Laplace errors.}
\label{tab:hd-gaussian-laplace-type1-var-ratio}
\small
\begin{tabular}{lccccc@{\hspace{1em}}ccccc}
\toprule
\multirow{2}{*}{Method}
& \multicolumn{5}{c}{Type~I error}
& \multicolumn{5}{c}{Variance ratio} \\
\cmidrule(lr){2-6}\cmidrule(lr){7-11}
& 0.1 & 0.2 & 0.3 & 0.4 & 0.5
& 0.1 & 0.2 & 0.3 & 0.4 & 0.5 \\
\midrule
Pairs
& 0.045 & 0.034 & 0.026 & 0.009 & 0.002
& 1.022 & 1.123 & 1.339 & 1.771 & 2.970 \\
Residual
& 0.061 & 0.074 & 0.100 & 0.130 & 0.152
& 0.897 & 0.794 & 0.701 & 0.600 & 0.496 \\
Jackknife
& 0.033 & 0.024 & 0.018 & 0.012 & 0.008
& 1.108 & 1.250 & 1.448 & 1.680 & 2.003 \\
Diffusion pairs
& 0.064 & 0.051 & 0.047 & 0.046 & 0.036
& 0.994 & 1.000 & 1.010 & 1.004 & 0.999 \\
\bottomrule
\end{tabular}
\end{table}

\end{document}